\documentclass{IEEEtran}
\usepackage{amsmath,amssymb,amsfonts}
\usepackage{graphicx}
\usepackage{textcomp,nicefrac}
\usepackage{longtable}

\DeclareMathSizes{5.6}{5.6}{5}{5}

\newtheorem{theorem}{Theorem}
\newtheorem{corollary}{Corollary}
\newtheorem{proposition}{Proposition}
\newtheorem{lemma}{Lemma}
\newtheorem{definition}{Definition}
\newtheorem{remark}{Remark}

\begin{document}
\title{A Nonuniform-Sampling Theory for Static Multi-Threshold Sampling}

\author{Ao Qiu, and Qingguo Xie
    \thanks{This work was supported in part by the National Natural Science Foundation of China, under Grant 625B2079, 61927801, 62250002, and 62050288. (\textit{Corresponding author: Qingguo Xie})}
    \thanks{Ao Qiu is with the Department of Biomedical Engineering, Huazhong University of Science and Technology, Wuhan, 430074 China (e-mail: aqiu@hust.edu.cn).}
    \thanks{Qingguo Xie is with the Department of Biomedical Engineering, Huazhong University of Science and Technology, Wuhan, 430074 China; Wuhan National Laboratory for Optoelectronics, Wuhan, 430074 China; and the Department of Electronic Engineering and Information Science, University of Science and Technology of China, Hefei, 230026 China (e-mail: qgxie@hust.edu.cn).}}

\maketitle

\begin{abstract}
Multi-Threshold (MT) sampling fixes voltage levels and records when a waveform crosses them. The resulting signal-generated event stream is governed not by a prescribed clock but by the nonuniform geometry of the realized crossing set. Within classical nonuniform sampling, we develop a retained-event MT theory for static, uniformly spaced thresholds. For a realized retained MT time set, Beurling's weak-limit theorem gives the exact necessary-and-sufficient condition for stable sampling of $PW_\Omega$. For a fixed signal class, the same principle gives the a priori criterion: uniform stable MT sampling is equivalent to Bernstein uniqueness for every set in the class weak-limit hull. This criterion reduces to computable forms under structured priors. Periodic templates give finite fiber/rank tests, periodic-gap streams give a Shannon-type complete-interpolation boundary and a stable-sampling density rule, and finite-state gap priors reduce noncritical certification to a maximum-cycle-mean condition. Kadec-type local certificates, retained-density bounds, and finite-matrix tests provide practical sufficient conditions for finite active records, while conditional perturbation bounds quantify how jitter, front-end noise, and threshold error consume a clean sampling margin under event correspondence. The theory also identifies the unavoidable boundary of this model: unrestricted finite-energy retained MT records have zero full-line lower Beurling density and therefore cannot provide a universal prior-free sampling theorem for $PW_\Omega$. Within retained-event point sampling, the MT recovery question is resolved precisely: exact recovery is governed by event-set geometry, computable guarantees require structural priors or finite-dimensional/tail assumptions, and the unrestricted static finite-energy problem is impossible.
\end{abstract}


\section{Introduction}
\label{sec:introduction}
\IEEEPARstart{C}{onventional} analog-to-digital converters (ADCs) sample amplitudes at prescribed time instants. Multi-Threshold (MT) sampling reverses this geometry: it fixes amplitude levels and timestamps the crossings. Introduced for PET pulse processing by Xie et al.~\cite{Xie2005}, MT is attractive because it spends events where the waveform moves and stays quiet where little changes, a useful feature for low-power event-driven acquisition~\cite{Tsividis2010,Miskowicz2006}.

The same feature also creates the main technical challenge. In uniform sampling, the time grid is known before the signal arrives. In static MT sampling, the grid is the signal's own footprint. Fast segments can produce dense crossing streams; near-quiescent segments can leave long silent intervals. Thus the question is not simply how many thresholds are used, but whether the realized crossing set is geometrically good enough to support stable recovery.

The surrounding literature supplies many pieces of this puzzle. Classical sampling and communication theory establish the Nyquist--Shannon benchmark~\cite{Nyquist1928,Shannon1948}; asynchronous and level-crossing converters provide the circuit context~\cite{Allier2003,Allier2005}. Threshold-based ADC architectures have been studied in hardware settings~\cite{Sayiner1996,Kozmin2009,Koscielnik2008}, and multichannel MT systems have been developed for nuclear instrumentation~\cite{Xi2013,Palka2017,Qiu2023}. Reconstruction theory enters through nonuniform algorithms~\cite{Grochenig1992,Feichtinger1992}, zero- and sine-crossing recovery~\cite{Logan1977,BocheMonich2012}, and level-crossing or implicit-information methods~\cite{Senay2012,Rzepka2018,RiazifarStocks2026}. The exact geometry of nonuniform sampling is governed by Beurling weak limits and complete-interpolation theory~\cite{beurling1966local,Seip2004,GrochenigRomeroStoeckler2018,LyubarskiiSeip1997}, while time-encoding machines address related event streams generated by dynamical encoders~\cite{LazarToth2004,GontierVetterli2014,Adam2020}. Standard sampling surveys and monographs~\cite{Jerri1977,Unser2000,Higgins1996}, together with nonuniform-sampling treatments~\cite{Benedetto1992}, frame the general landscape but do not settle the static, signal-generated MT case considered here.

The problem therefore has four linked layers:
\begin{enumerate}
    \item \textbf{Exact post-realization geometry.} Once the retained event set $\Lambda_{\Delta V}(x)$ has been realized, stable sampling is decided by the Beurling weak-limit hull of that fixed time set.
    \item \textbf{Exact class-level geometry.} Once a signal class is declared, the a priori object is the class weak-limit hull $\mathfrak W_{\Delta V}(\mathcal X)$ rather than the scalar threshold density alone.
    \item \textbf{Computable structured reductions.} Periodic templates, periodic-gap streams, finite-state gap priors, and calibrated noncritical density models are the cases where the abstract hull criterion becomes finite or scalar.
    \item \textbf{Boundary and implementation.} Static finite-energy retained records cannot sample all of $PW_\Omega$ without additional information; finite records require a tail prior or a declared finite-dimensional model; hardware error bounds are conditional on event correspondence.
\end{enumerate}

We answer these questions in the clean setting needed for a sampling certificate: real-valued effectively bandlimited signals, a static, uniformly spaced threshold set, the retained adjacent-transition convention, and ideal timing unless hardware perturbations are explicitly introduced. Throughout this paper, the threshold levels are chosen before acquisition, remain unchanged during acquisition, and have fixed spacing $\Delta V$; the retained event topology is likewise specified before the sampling theorems are invoked. Time-varying or adaptive threshold sets would require a different hardware and control model, including reconfiguration latency, threshold-update dynamics, and comparator calibration. They are therefore outside the theory developed here.

Actual MT hardware produces finite records on active windows. Such records cannot by themselves be sampling sets for the full infinite-dimensional space $PW_\Omega$. The finite-record result in Subsection~\ref{sec:finite_exact_condition} therefore uses the appropriate finite-dimensional notion: after a reconstruction space $V_M$ is declared, exact and stable recovery is equivalent to full column rank, or equivalently $\sigma_{\min}(H_{\mathrm{fin}})>0$. This finite-matrix condition is the operative acceptance test in numerical reconstruction.

Four classical tools anchor the analysis. Beurling's weak-limit theorem supplies the post-realization sampling condition through uniqueness sets for Bernstein spaces~\cite{beurling1966local,Seip2004,GrochenigRomeroStoeckler2018}. Complete-interpolation theory yields the critical-density analogue of Shannon's cardinal series through the Pavlov/Lyubarskii--Seip Muckenhoupt condition~\cite{LyubarskiiSeip1997}. Kadec's $1/4$ theorem measures local displacement from a Nyquist lattice~\cite{Kadec1964,Young2001}; the lattice itself comes from the Nyquist--Shannon benchmark~\cite{Nyquist1928,Shannon1948}. Landau--Beurling density theory controls the long-window information rate~\cite{Landau1967,beurling1966local,Seip2004,GrochenigRomeroStoeckler2018}. These results are used here to isolate what fixed voltage thresholds imply, and what they do not imply, for signal-generated event sets.

The contributions are organized by logical status.
\begin{enumerate}
    \item[\textbf{C1.}] \textbf{Exact realized-set criterion.}
    For a fixed retained MT time set, stable sampling of $PW_\Omega$ is equivalent to the Beurling weak-limit/Bernstein-uniqueness condition. This is an exact post-realization theorem, not a scalar ladder-design rule.

    \item[\textbf{C2.}] \textbf{Exact class-level criterion.}
    For a uniformly separated MT signal class generated by static, uniformly spaced thresholds, uniform stable sampling is equivalent to Bernstein uniqueness for every set in the class weak-limit hull. This identifies the correct a priori object.

    \item[\textbf{C3.}] \textbf{Structured computable reductions.}
    We show when the abstract class-hull criterion reduces to finite or scalar tests: finite periodic-template hulls, periodic-gap complete interpolation and redundant stable sampling, finite-state gap priors, and noncritical calibrated-density conversion.

    \item[\textbf{C4.}] \textbf{Practical sufficient certificates.}
    Kadec-type local placement bounds and retained-density estimates give checkable sufficient conditions for realized active records and structured classes. These certificates are deliberately separated from the exact weak-limit criteria.

    \item[\textbf{C5.}] \textbf{Static finite-energy no-go boundary and finite-record status.}
    We prove that unrestricted finite-energy static retained MT records have zero full-line lower density, so no static, uniformly spaced threshold set gives a universal prior-free sampling theorem on all of $PW_\Omega$. Finite records become meaningful only with a tail prior or a declared finite-dimensional reconstruction space.

    \item[\textbf{C6.}] \textbf{Conditional hardware perturbation theory.}
    Under event correspondence, timing jitter, input noise, and threshold residuals are converted into sampling-geometry perturbations and finite-matrix error budgets.
\end{enumerate}

These results should be read as a scope-delimited theorem with a boundary, not as a universal reconstruction guarantee for every finite-energy waveform. The exact positive statements concern realized or class-level event-set geometry; the computable positive statements require structured priors, exact retained-density calibration, tail information, or finite-dimensional models; and the unrestricted static finite-energy full-line problem is ruled out by the no-go theorem.

\begin{table*}[!t]
\caption{Logical status of the main MT sampling results.}
\label{tab:result_taxonomy}
\centering
\footnotesize
\setlength{\tabcolsep}{3pt}
\begin{tabular}{p{0.15\textwidth}p{0.21\textwidth}p{0.15\textwidth}p{0.21\textwidth}p{0.18\textwidth}}
\hline
Layer & Main result & Status & Required object & Meaning for static MT \\
\hline
Realized event set & Beurling weak-limit criterion & Exact iff & Fixed separated retained time set $\Lambda$ & Stable recovery is decided by event geometry after acquisition \\
Class prior & Class weak-limit hull criterion & Exact iff & Uniformly separated class hull $\mathfrak W_{\Delta V}(\mathcal X)$ & The correct a priori object is the hull, not threshold density alone \\
Structured priors & Periodic-template and finite-state reductions & Finite/scalar criteria under structure & Periodic or finite-state weak-limit geometry & The abstract hull condition becomes computable \\
Density calibration & Noncritical calibrated-density criterion & Noncritical iff & Exact retained-density calibration and separation & Scalar threshold density is exact only after calibration and away from criticality \\
Local certificates & Kadec and retained-density bounds & Sufficient & Velocity, curvature, retained topology, window priors & Practical design and post-realization checks \\
Finite records & Matrix rank/singular-value condition & Exact on $V_M$ & Declared finite-dimensional model space or tail prior & Finite data do not sample all of $PW_\Omega$ without extra information \\
Static finite-energy boundary & No-go theorem & Negative theorem & Unrestricted finite-energy retained MT with static, uniformly spaced thresholds & No universal prior-free full-line theorem exists \\
Hardware & Perturbation bounds & Conditional sufficient & Event correspondence and bounded impairments & Clean sampling margins are consumed by real hardware errors \\
\hline
\end{tabular}
\end{table*}

This organization also explains what the paper does not claim. It does not assert that static, uniformly spaced thresholds recover every finite-energy bandlimited signal over the full line. That statement is false under the retained event convention, as proved by the no-go theorem. Instead, the paper identifies the exact nonuniform-sampling object for any realized or class-level MT geometry, the structured priors under which that object becomes computable, and the finite-record or hardware conditions under which the theory can be used in practice.

The no-go theorem is not a weakness of the theory; it is the boundary condition that prevents overclaiming. It shows exactly why static MT sampling requires persistent activity, structured priors, tail information, or finite-dimensional reconstruction when one asks for stable recovery rather than merely event detection.

The paper moves from exact nonuniform-sampling geometry to computable MT certificates. Section~\ref{sec:preliminaries} recalls the sampling benchmarks, fixes the retained-event convention, and separates static finite-energy records, persistent/structured bi-infinite event geometries, and finite active records. Section~\ref{sec:exact_iff} develops the weak-limit characterization, the class-hull criterion, structured prior reductions, Shannon-type subcases, and the finite-record condition. Sections~\ref{sec:sufficient}--\ref{sec:density_analysis} derive local and long-window sufficient certificates; Section~\ref{sec:no_go} states the static finite-energy boundary; and Section~\ref{sec:hardware_error} gives conditional hardware perturbation bounds for static MT records. Section~\ref{sec:reconstruction} describes reconstruction from certified nonuniform samples, and Section~\ref{sec:numerical} reports numerical diagnostics. Sections~\ref{sec:discussion} and~\ref{sec:limitations} summarize the scope and limitations of the static-threshold theory; Appendices~\ref{app:landau}--\ref{app:pedantic_extras} provide the detailed derivations, with notation summarized in Appendix~\ref{sec:notation}.

\section{Mathematical Preliminaries}
\label{sec:preliminaries}
A summary of principal notation used throughout this paper is provided in Tables~\ref{tab:notation} and~\ref{tab:notation-extended} in Appendix~\ref{sec:notation}. Here and below, $\mathbb{R}$ denotes the real line and $\mathbb{Z}$ denotes the integer lattice.

Finite bandwidth imposes a simple physical constraint: the waveform cannot move arbitrarily fast. After the adjacent-transition convention is applied to the raw horizontal intersections, this constraint turns the recorded MT event stream into quantitative timing information. The Paley-Wiener space encodes the bandwidth constraint. The total-variation identity converts it into a bound on vertical travel, and the frame inequality specifies when the resulting event record is \emph{informationally complete}. Together, these ingredients let the Kadec and Landau--Beurling benchmarks operate on the signal-dependent retained MT set under the modeling scope stated in the Introduction. Throughout, we use standard asymptotic notation $o(\cdot)$ and $\mathcal{O}(\cdot)$ together with the comparison symbols $\lesssim$ and $\gtrsim$.

\subsection{The Paley-Wiener Space and Frame Bounds}
Throughout this paper, $x:\mathbb{R} \to \mathbb{R}$ denotes a continuous, real-valued physical signal. The scalar real-valued model corresponds to one physical comparator channel. We assume $x$ belongs to the finite-energy, bandlimited Paley--Wiener space~\cite{paley1934fourier} $PW_\Omega$:
\begin{equation}
    PW_\Omega = \left\{ x \in L^2(\mathbb{R}) : \text{supp}(\hat{x}) \subseteq [-\Omega, \Omega] \right\}
    \label{eq:pw_space}
\end{equation}
where $L^2(\mathbb{R})$ is the square-integrable function space. The Fourier transform is $\hat{x}(\omega) = \int_{-\infty}^{\infty} x(t)\,e^{-j\omega t}\,dt$, written in the IEEE convention, where $j^2=-1$ denotes the IEEE imaginary unit. The operators $\mathcal{F}$ and $\mathcal{F}^{-1}$ denote the Fourier and inverse Fourier transform operators when operator notation is more convenient. The parameter $f_{max}$ is the absolute maximum frequency component present in the physical signal, measured in Hertz (Hz). Equivalently, $\Omega = 2\pi f_{max}$ is the maximum angular bandwidth (rad/s) and sets the strict support boundary of $\hat{x}$ in the frequency domain. Throughout, $\mathrm{sinc}(u) \triangleq \sin(\pi u)/(\pi u)$ denotes the normalized sinc function. By the Paley--Wiener theorem, any function $x \in PW_\Omega$ extends to an entire function of exponential type~\cite{Boas1954,Levin1996} at most $\Omega$ on the complex plane. Consequently, signals in $PW_\Omega$ are bounded, uniformly continuous, and infinitely differentiable on $\mathbb{R}$. Bernstein's inequality~\cite{Boas1954,Levin1996} further gives the derivative bound $\|x^{(k)}\|_\infty \le \Omega^k \|x\|_\infty$ for all orders $k \ge 1$, where $x^{(k)}$ denotes the $k$-th derivative of $x$. For continuously differentiable signals moving between finite energy limits, we define the total variation on an interval $[a,b]$ as:
\begin{equation}
    \mathrm{TV}(x;[a,b]) = \int_a^b |x'(t)|\,dt.
    \label{eq:tv_def}
\end{equation}

Physically, the assumption $x \in PW_\Omega$ is used as an \emph{effective bandlimitation} model: spectral content above $f_{max}$ has either been sufficiently attenuated by the analog front-end or is immaterial at the modeling accuracy sought here. The mathematical theorems below, however, use the exact support condition in~\eqref{eq:pw_space}. Approximate bandlimitation should be handled by writing $x=x_\Omega+r$ with $x_\Omega\in PW_\Omega$ and carrying the out-of-band residual $r$ into the finite-record or sample-value error budget. Thus $PW_\Omega$ is the standard ideal limit that makes the bandwidth constraint explicit, rather than a literal brick-wall hardware spectrum. Under that exact model, finite bandwidth controls derivative norms through Bernstein-type inequalities for a fixed amplitude envelope. It does \emph{not} by itself impose a finite-window zero-density or local-oscillation bound. Superoscillatory bandlimited functions may oscillate locally faster than their largest Fourier component, with the cost appearing in dynamic range or energy concentration elsewhere~\cite{Kempf2000}. Consequently, the local crossing-density claims below rely on explicit anti-superoscillation, monotonicity, or finite-window count certificates rather than on bandwidth alone. In real-world instrumentation, $f_{max}$ is the usable front-end analog bandwidth of the acquisition hardware, typically set by parasitic capacitance together with any explicit anti-alias or anti-noise filtering before the comparator bank. Geometrically, Bernstein's inequality prevents arbitrarily steep slopes or infinitely sharp corners under a fixed amplitude envelope. Once the additional local-density certificates have been supplied, this deterministic derivative control converts vertical MT threshold traversals into quantitative bounds on horizontal timing displacement.

For the MT architecture, the functional $\mathrm{TV}(x;[a,b])$ has a direct physical interpretation: it is the cumulative absolute vertical voltage distance traveled by the waveform's trajectory over the time interval $[a,b]$. In the idealized event model with static, uniformly spaced thresholds used here, each monotone traversal of one threshold interval of size $\Delta V$ contributes one raw threshold hit, while the retained hardware record may lose explicitly counted same-threshold returns under the adjacent-transition convention below. Thus the macroscopic ratio $\mathrm{TV}/\Delta V$ forms the natural first-order proxy for the raw traversal count, and the retained-count theorems subtract endpoint, floor, and suppression losses. These regularity properties---global energy finiteness, boundedness, smoothness, and controlled derivative growth---are used throughout the subsequent Kadec and Landau--Beurling arguments.

To isolate the static MT regime stated in the Introduction, we assume throughout that the threshold levels are chosen before acquisition, remain unchanged during acquisition, and are uniformly spaced by $\Delta V > 0$. This is the natural hardware choice for a resistor-ladder voltage reference, the simplest and most common comparator-bank architecture. It is also the prior-free static-threshold choice. A nonuniform threshold set allocates more event resolution to selected voltage ranges and therefore encodes an amplitude-domain prior about where the waveform is expected to spend time, where small changes matter most, or where the hardware should spend comparator resources. In the absence of such a value-domain prior, translation invariance in the voltage variable and the desire for a single traversal-to-count calibration lead to the uniform threshold set. This choice also makes raw crossing counts transparent: each monotone traversal of one threshold step contributes exactly one raw hit, regardless of which threshold is crossed; the retained record is obtained only after the topology convention below is applied. The present theory therefore reduces threshold design to the single scalar parameter $\Delta V$ and treats all sampling theorems after that threshold set and retained topology have been specified. Static nonuniform threshold sets are not ruled out, but they require an explicit value-domain prior and a model for the local threshold density. In particular, such sets are not merely a notational variant of the uniform case: after realization their event times can still be tested by classical nonuniform sampling, but the design of the voltage-level density itself is a value-domain allocation problem outside the classical time-sampling framework. Let $\{I_k\}$ denote the open maximal monotone intervals of $x(t)$ (on which $x$ is strictly increasing or decreasing). First define the raw horizontal-intersection set
\[
    \Lambda_{\Delta V}^{\mathrm{raw}}
    = \bigcup_k \{ t \in I_k : x(t) \in \Delta V \cdot \mathbb{Z} \}.
\]
The raw set above is not yet the hardware record used in the theorems. A physical comparator bank may report every edge, may include hysteresis, or may use digital filtering. We therefore make the recorded-event topology an explicit modeling assumption rather than a neutral definition.

\textbf{Recorded adjacent-transition convention.}
The ideal record in this paper is the ordered, labeled, adjacent-transition subsequence
\[
    E_{\Delta V}(x)
    =
    \bigl\{(t_n,m_n): t_n\in\Lambda_{\Delta V}^{\mathrm{raw}},
    \ x(t_n)=m_n\Delta V\bigr\}_{n\in\mathcal{N}},
\]
obtained by applying the retained adjacent-transition map to the raw crossings. The associated time-coordinate set used in the sampling analysis is
\begin{equation}
    \Lambda_{\Delta V}
    =
    \mathcal{A}_{\Delta V}\!\left(\Lambda_{\Delta V}^{\mathrm{raw}}\right),
    \label{eq:mvt_set}
\end{equation}
Here $\mathcal{A}_{\Delta V}$ is the following deterministic state map. On any finite interval, list the locally finite raw events in increasing time order as
\[
    R=\{(r_j,\ell_j)\}_{j\in J},
    \qquad
    x(r_j)=\ell_j\Delta V.
\]
For a finite active observation window, the retained map also depends on the retained state immediately before the left endpoint. Write this state as
\[
    m_{\mathrm{init}}^-\in\mathbb Z\cup\{\varnothing\}.
\]
The finite-window map is denoted by $\mathcal A_{\Delta V}^{m_{\mathrm{init}}^-}(R)$. If $R$ is empty, the retained record is empty. If $m_{\mathrm{init}}^-=\varnothing$, we use local initialization and retain the first raw event in the chosen time orientation. If $m_{\mathrm{init}}^-=m_0\in\mathbb Z$, the first retained event is instead
\[
    j_1=\min\{j\in J:\ |\ell_j-m_0|=1\},
\]
when this set is nonempty; raw same-state returns before this first adjacent transition are suppressed. Having retained $j_k$, define
\begin{equation}
    j_{k+1}
    =
    \min\{j>j_k:\ |\ell_j-\ell_{j_k}|=1\},
    \label{eq:retained_state_forward}
\end{equation}
when this set is nonempty, and repeat. All raw returns with $\ell_j=\ell_{j_k}$ before the next adjacent-level transition are suppressed. Continuity of $x$ prevents a first different raw label from skipping an intermediate ladder level; the explicit adjacent condition in~\eqref{eq:retained_state_forward} records this fact as part of the event topology. The notation $\mathcal A_{\Delta V}(R)$ means the globally initialized bi-infinite map when a retained origin is fixed, or local initialization $m_{\mathrm{init}}^-=\varnothing$ when only an isolated finite record is available. If the true pre-window state is unavailable, the possible discrepancy between $\mathcal A_{\Delta V}^{m_{\mathrm{init}}^-}$ and local initialization can affect only the first retained event of the finite record, and all finite-window density, Kadec-extraction, and boundary-collar estimates absorb this $O(1)$ uncertainty into the stated endpoint, monotone-count, support, or $S_{\mathrm{bdry}}$ corrections. For a bi-infinite retained stream one fixes one retained origin and applies the same recursion forward and backward in time; changing the origin can alter only a finite prefix in a finite record and does not change the local adjacent-transition rule used in the sampling theorems. Because the branches $I_k$ are open and strictly monotone, a point at which $x$ only touches a threshold at a turning point is not a raw crossing. If $x\equiv0$, or if a threshold is followed on a flat interval, no raw event is produced by this convention. For nonzero $PW_\Omega$ signals, analyticity gives only finitely many monotone branches and raw threshold hits on each compact interval, so the ordered state recursion is well defined locally.
\noindent Here $\mathcal{N}$ is the ordered index set of retained labeled events in $E_{\Delta V}(x)$.

All static MT theorems in this paper are conditional on this recorded adjacent-transition convention. If an implementation reports tangential contacts, same-threshold returns, or noisy comparator bounces as separate events, the event topology changes and the corresponding sampling set must be reanalyzed under that different record model. The mapping $x \mapsto \Lambda_{\Delta V}(x)$ is nonlinear because the sampling set itself depends on the signal. Accordingly, whenever frame or density theorems are invoked below, they are applied \emph{after fixing a signal instance} and studying the realized set $\Lambda_{\Delta V}(x)$.

Each element of $\Lambda_{\Delta V}$ is therefore a retained time coordinate at which the waveform crosses one of the horizontal lines $v = m\Delta V$ along a monotone branch and completes an adjacent threshold-index transition. The decomposition into open maximal monotone intervals prevents ambiguous double counting at turning points: on a strictly monotone branch, the equation $x(t)=m\Delta V$ has at most one solution for each threshold level, while endpoint contacts at extrema are excluded by definition.

The monotone-crossing convention is physically reasonable when the comparator output is filtered into a retained adjacent-transition record, either by analog hysteresis or by digital edge filtering, so that each traversal between adjacent threshold rungs produces one timestamped event. Fast bouncing near a threshold is therefore excluded from the core model and would require an explicit nonideal event model. This assumption is used essentially in the separation bound, the finite-energy obstruction, and the no-go construction; it should be regarded as part of the acquisition model, not as a theorem about all possible level-crossing hardware.

\textbf{Idealized hardware assumptions.} Throughout this paper, the MT system is modeled under three idealizations: (i)~each comparator detects threshold crossings with infinite time precision (zero jitter); (ii)~the time-to-digital converter has infinite resolution; and (iii)~transient effects such as comparator metastability and propagation delay are neglected. In practice, these effects introduce bounded timing errors $\epsilon_k^{(j)}$ at each crossing, using the jitter-error notation formalized later in~\eqref{eq:jitter_model}. If a selected crossing subsequence already satisfies Kadec's inequality with margin $\eta > 0$, namely $\sup_n |t_n-(\tau_0+n\Delta T)| \le \Delta T/4 - \eta$ for some grid offset $\tau_0$, then any additional timing error satisfying $\sup_k |\epsilon_k^{(j)}| < \eta$ preserves Kadec compatibility.

These idealizations are standard first-order abstractions for a sampling theorem: the core requirement is only that the total timing uncertainty stay below the available Kadec margin. When that perturbative margin is not small, timing nonidealities must be incorporated into the model explicitly rather than treated as negligible.

\subsection{Three Uses of the Retained MT Time Set}
\label{sec:event_set_regimes}
The same notation $\Lambda_{\Delta V}$ is used in three logically distinct ways.

\textbf{Static finite-energy retained record.}
When $x\in PW_\Omega$ is an unrestricted finite-energy full-line waveform, the static, uniformly spaced threshold set and retained adjacent-transition convention produce a static retained record. Theorem~\ref{thm:no_go} shows that this full-line retained set has zero lower Beurling density; indeed, under the retained convention it is finite after the finite-energy tails fall below one threshold cell. Therefore such a record cannot sample all of $PW_\Omega$ without additional information.

\textbf{Persistent or structured bi-infinite event geometry.}
The post-realization weak-limit theorem, periodic-gap theory, finite-template hulls, and finite-state gap priors are theorems about fixed separated event geometries, persistent streaming models, periodic steady-state extensions, idealized bi-infinite extensions of active records, or class-level event-set priors. These structured event models should not be read as claims that every such bi-infinite event set is generated by an unrestricted finite-energy full-line waveform.

\textbf{Finite active record.}
Actual hardware data live on an observation interval $[0,T_{\mathrm{obs}}]$. Exact finite-record recovery claims therefore pass through either a tail prior/padding values or a declared finite-dimensional reconstruction space $V_M$. In the latter case the exact stability object is the finite sampling matrix and its smallest singular value $\sigma_{\min}(H_{\mathrm{fin}})$, not a full-line frame for all of $PW_\Omega$.

Thus, when $x\in PW_\Omega$ is an unrestricted finite-energy full-line waveform, the static retained record cannot be a full-line sampling set. The infinite sampling theorems below should instead be read as theorems about a fixed separated retained event geometry, a persistent/periodic extension, or a class-level weak-limit prior. The finite-record reconstruction statements are separate and pass through either a tail prior or the finite-dimensional rank condition.

For a fixed realized nonuniform set $\Lambda = \{t_n\}$, stable reconstruction is understood in the usual frame sense~\cite{DuffinSchaeffer1952,Christensen2003}: $\Lambda$ is a valid sampling set if it satisfies the frame inequality with global bounds $0 < A \le B < \infty$:
\begin{equation}
    A \|f\|_{L^2}^2 \le \sum_{n \in \mathbb{Z}} |f(t_n)|^2 \le B \|f\|_{L^2}^2,
    \qquad f\in PW_\Omega,
    \label{eq:frame}
\end{equation}
The upper bound $B$ ensures operator boundedness (numerical stability), while the lower bound $A > 0$ guarantees injectivity (no information loss) and invertibility of the reconstruction.

From an engineering viewpoint, the lower frame bound rules out indistinguishable value data from two distinct bandlimited functions $f_1,f_2$ on the already realized set $\Lambda$. The upper frame bound prevents sample perturbations from being amplified without control. Stable sampling therefore requires uniqueness with quantitative robustness. The ratio $B/A$ acts as an effective condition number of the sampling operator: small $B/A$ corresponds to a well-conditioned inverse problem, while $A \to 0$ signals the onset of ill-posed reconstruction regardless of how many crossings are recorded.

\subsection{The Landau--Beurling Theorem}
A basic theorem governing nonuniform sampling is the Landau--Beurling density theorem~\cite{Landau1967,beurling1966local}, with modern sampling-theoretic treatments in analytic and shift-invariant settings~\cite{Seip2004,GrochenigRomeroStoeckler2018}. The lower Beurling density of $\Lambda$ is defined as
\begin{equation}
    D^-(\Lambda) = \liminf_{r \to \infty} \inf_{t \in \mathbb{R}} \frac{n(t, r)}{r}
    \label{eq:density_def}
\end{equation}
where $n(t, r)$ is the number of sampling points within the interval $[t, t+r]$.

The quantity $D^-(\Lambda)$ measures the worst-case average number of temporal constraints per unit time. It is a macroscopic information-rate descriptor: if the sampler undersupplies points on some sufficiently long interval, no excess of points elsewhere can recover the missing local degrees of freedom. The double limit---first an infimum over window positions $t$, then a liminf as the window length $r \to \infty$---selects the slowest long-term arrival rate of samples and therefore represents the most conservative reading of the effective sampling rate on $\Lambda$.

\textbf{Necessary condition (Landau).} If $\Lambda$ is a sampling set for $PW_\Omega$, then $D^-(\Lambda) \ge 2f_{max}$.

\textbf{Sufficient condition (Beurling).} If $\Lambda$ is uniformly discrete (i.e., $\inf_{m \neq n} |t_m - t_n| \ge \delta_{\mathrm{sep}} > 0$) and $D^-(\Lambda) > 2f_{max}$, then $\Lambda$ is a sampling set for $PW_\Omega$.

Appendix~\ref{app:landau} records the normalization, the Landau concentration argument used for necessity, and the Beurling strict-density input used for sufficiency. The auxiliary prolate transition theorem and localized Plancherel-P\'olya input are stated separately in Appendix~\ref{app:pedantic_extras}.

\textbf{The critical boundary case.} When $D^-(\Lambda) = 2f_{max}$, stable reconstruction depends on microscopic global point geometry beyond macroscopic density. As a canonical example, the uniform Shannon grid~\cite{Shannon1948} $\Lambda_0 = \{n/(2f_{max})\}_{n \in \mathbb{Z}}$ has $D^-(\Lambda_0) = 2f_{max}$ and is a stable sampling set. Removing even a single point (e.g., $t = 0$) leaves the density unchanged at $D^- = 2f_{max}$. Yet the frame lower bound $A$ collapses to zero: the signal $f(t) = \mathrm{sinc}(2f_{max}\, t)$ vanishes at all remaining grid points while carrying nonzero energy. The sampling system is therefore unable to detect this signal, and the boundary case $D^- = 2f_{max}$ is undecidable from density alone.

This example is especially relevant for MT. A threshold design may generate a crossing rate that looks Nyquist-compatible on average, while still leaving a problematic microscopic gap in the realized event pattern. Density alone is therefore not the whole story.

\subsection{Kadec's 1/4 Theorem}
\label{sec:kadec}
One convenient sufficient route to stable sampling is to compare the nonuniform sequence with a Nyquist lattice. Kadec's $1/4$ theorem gives a sharp perturbative certificate: if a sequence can be indexed so that its displacement from some Nyquist lattice is uniformly below one quarter of the Nyquist spacing, then it is a complete interpolating sequence, hence a stable sampling set. Given a uniform Nyquist interval $\Delta T = 1/(2f_{max})$, this sufficient certificate is
\begin{equation}
    \sup_n |t_n - n\Delta T| < \frac{\Delta T}{4} = \frac{1}{8f_{max}}
    \label{eq:kadec_base}
\end{equation}

Physically, Kadec's theorem shows that average density alone is insufficient for this particular lattice-perturbation route; the certified sample locations must also remain uniformly close to the Nyquist lattice. Many stable sampling sets are not Kadec perturbations of a single lattice, so the theorem is used here as a computable sufficient certificate rather than as a necessary condition. The constant $1/4$ is the sharp local placement threshold for that certificate.

Appendix~\ref{app:constraints} gives the detailed proof. After normalizing the bandwidth to $\Omega = \pi$, it analyzes the bounded synthesis operator that maps the orthonormal Fourier basis of $L^2[-\pi,\pi]$ onto the perturbed exponential system generated by $\{t_n\}$. The classical Kadec estimate yields the perturbation bound $\theta(L) = 1 - \cos(\pi L) + \sin(\pi L)$. Since $\theta(L) < 1$ exactly when $L < 1/4$, a Neumann-series argument provides a bounded inverse and establishes the Riesz-basis property.

\section{Post-Realization Necessary-and-Sufficient Characterization}
\label{sec:exact_iff}
The preceding density and Kadec statements are deliberately one-sided: Landau gives a necessary density bound, Beurling gives a strict-density sufficient condition, and Kadec gives a local sufficient condition. An exact infinite-dimensional condition is available, but it is global and microscopic rather than a finite-parameter design inequality. We state that condition for MT-generated point sets only after the events have been realized. This distinction is essential. The weak-limit theorem decides whether a fixed set of time nodes samples $PW_\Omega$; it is not, by itself, a stability theorem for the nonlinear acquisition map that sends a waveform to its crossing times.

\subsection{Weak Limits and Bernstein Uniqueness}
\label{sec:weak_limit_condition}
Let $\mathcal{B}_\Omega$ denote the Bernstein space of bounded real-line restrictions of entire functions of exponential type at most $\Omega$:
\begin{align*}
    \mathcal{B}_\Omega
    =
    \{F:\mathbb{R}\to\mathbb{C}:\
    &F \text{ is bounded on }\mathbb{R},\\
    &F \text{ is entire of type }\le\Omega\}.
\end{align*}
This is the $L^\infty$ counterpart of $PW_\Omega$; the symbol should not be confused with a frame upper bound $B$.

For two closed discrete sets $\Lambda_j,\Gamma\subset\mathbb{R}$, write $\Lambda_j\xrightarrow{w}\Gamma$ if, for every bounded open interval $I$ and every $\varepsilon>0$, all sufficiently large $j$ satisfy
\[
    \Lambda_j\cap I\subseteq \Gamma+(-\varepsilon,\varepsilon),
    \qquad
    \Gamma\cap I\subseteq \Lambda_j+(-\varepsilon,\varepsilon).
\]
The weak-limit hull of $\Lambda$ is
\[
    W(\Lambda)
    =
    \left\{
        \Gamma:\exists a_j\in\mathbb{R}
        \text{ such that } \Lambda-a_j\xrightarrow{w}\Gamma
    \right\}.
\]
A set $\Gamma$ is a uniqueness set for $\mathcal{B}_\Omega$ if
\begin{equation}
    F\in\mathcal{B}_\Omega,\qquad F|_\Gamma=0
    \quad\Longrightarrow\quad
    F\equiv0.
    \label{eq:bernstein_uniqueness}
\end{equation}

\begin{theorem}[Post-realization sampling certificate for retained MT time sets]
\label{thm:mvt_weak_limit_iff}
\textit{Let $\Lambda$ be a separated retained time-coordinate set produced by the adjacent-transition convention for static, uniformly spaced thresholds, or by a specified persistent/bi-infinite extension of such a retained record. Then the following statements about the fixed point set $\Lambda$ are equivalent:}
\begin{enumerate}
    \item[\textup{(i)}] \textit{The classical point-sampling operator $S_\Lambda f=(f(t_n))_{t_n\in\Lambda}$ satisfies}
    \begin{equation}
        A\|f\|_{L^2}^2
        \le
        \sum_{t_n\in\Lambda}|f(t_n)|^2
        \le
        B\|f\|_{L^2}^2,
        \qquad f\in PW_\Omega,
        \label{eq:exact_frame_condition}
    \end{equation}
    \textit{for some $0<A\le B<\infty$.}
    \item[\textup{(ii)}] \textit{The exponential system $\{e^{-j\omega t_n}\}_{t_n\in\Lambda}$ is a Fourier frame for $L^2[-\Omega,\Omega]$.}
    \item[\textup{(iii)}] \textit{Every weak limit $\Gamma\in W(\Lambda)$ is a uniqueness set for $\mathcal{B}_\Omega$.}
\end{enumerate}
\end{theorem}

\begin{remark}[Finite static records]
For an unrestricted finite-energy static signal $x\in PW_\Omega$, the retained full-line set generated by a static, uniformly spaced threshold set does not satisfy the positive sampling alternatives above; Theorem~\ref{thm:no_go} shows that its lower Beurling density is zero. Thus Theorem~\ref{thm:mvt_weak_limit_iff} is exact for a fixed separated event geometry, but finite static MT records require the finite-dimensional or tail-prior formulations in Subsection~\ref{sec:finite_bridge}.
\end{remark}

\textit{Proof.}
The theorem is an application of classical nonuniform sampling theory to the fixed retained time set. First, by the inverse Fourier representation,
\[
    f(t_n)
    =
    \frac{1}{2\pi}
    \int_{-\Omega}^{\Omega}
        \hat{f}(\omega)e^{j\omega t_n}\,d\omega.
\]
Plancherel's theorem identifies the sample map $f\mapsto\{f(t_n)\}$ with the analysis map generated by the exponential family $\{e^{-j\omega t_n}\}$ on $L^2[-\Omega,\Omega]$, up to the fixed Fourier-normalization constant. Hence~\eqref{eq:exact_frame_condition} is equivalent to the Fourier-frame property.

Second, the theorem assumes the separated geometry required by Beurling's weak-limit theorem. When $\Lambda$ is generated by a finite-slew retained MT record, this separation follows from the adjacent-transition convention and Bernstein's inequality. Indeed, by the mean value theorem,
\begin{equation}
    |t_{n+1}-t_n|
    =
    \frac{\Delta V}{|x'(\xi_n)|}
    \ge
    \frac{\Delta V}{\|x'\|_\infty}
    \ge
    \frac{\Delta V}{\Omega\|x\|_\infty}.
    \label{eq:exact_mvt_separation}
\end{equation}
For a specified persistent or bi-infinite extension, separation is part of the stated hypothesis.
For separated point sets, Beurling's weak-limit theorem states that stable sampling in $PW_\Omega$ is equivalent to uniqueness in $\mathcal{B}_\Omega$ for every translate weak limit of the set~\cite{beurling1966local,Seip2004,GrochenigRomeroStoeckler2018}. Applying that theorem to $\Lambda$ proves the equivalence between~\eqref{eq:exact_frame_condition} and~\eqref{eq:bernstein_uniqueness} for all $\Gamma\in W(\Lambda)$. The reductions and scaling constants are expanded in Appendix~\ref{app:exact_iff}. \hfill $\square$

\textbf{Interpretation.}
The theorem is an exact iff for a realized or explicitly extended event geometry. It is not a finite scalar design rule for $\Delta V$, and it is not a global stability theorem for the nonlinear acquisition map $x\mapsto E_{\Delta V}(x)$.

\begin{corollary}[Conditional recovery of the generating waveform from ideal MT labels]
\label{cor:mvt_acquisition_recovery}
\textit{Let $E_{\Delta V}(x)=\{(t_n,m_n)\}$ be an ideal retained MT record generated by $x\in PW_\Omega$ under the recorded adjacent-transition convention. If the realized time set $\Lambda=\{t_n\}$, or a specified persistent/bi-infinite extension of that set, satisfies the equivalent conditions of the post-realization sampling certificate above, then the generating waveform $x$ is stably recoverable from the labeled record through the exact sample values}
\[
    y_n=m_n\Delta V=x(t_n).
\]
\textit{In particular, any $g\in PW_\Omega$ satisfying $g(t_n)=m_n\Delta V$ for all retained events must equal $x$.}
\end{corollary}

\textit{Proof.}
The labels provide exact point samples of the generating waveform at the realized event times. Since $S_\Lambda$ is bounded below on $PW_\Omega$, it is injective and has a bounded left inverse on its range. Applying that inverse to the sample vector $\{m_n\Delta V\}$ recovers $x$. If another $g\in PW_\Omega$ has the same labeled samples, then $S_\Lambda(g-x)=0$, and the lower frame bound forces $g=x$. \hfill $\square$

The corollary is the acquisition statement used in this paper. It is conditional on the event topology and on the realized point-set certificate. It does not assert that the nonlinear map $x\mapsto E_{\Delta V}(x)$ is globally one-to-one or Lipschitz on all of $PW_\Omega$; small waveform perturbations can create, delete, or relabel crossings unless an event-correspondence hypothesis such as the one in Section~\ref{sec:hardware_error} is imposed.

\textbf{Uniqueness versus stability.} If one asks only for algebraic uniqueness, the condition weakens to $\ker S_\Lambda=\{0\}$, i.e., no nonzero $f\in PW_\Omega$ vanishes on all points of $\Lambda$. This is not the recovery criterion used here. MT reconstruction in hardware must tolerate timing error, threshold error, finite windows, and numerical inversion, so the relevant notion is the stable frame condition~\eqref{eq:exact_frame_condition}.

\subsection{A Priori Weak-Limit Hull Criterion for MT Classes}
\label{sec:apriori_class_weak_limit_hull}

The post-realization theorem applies to one realized MT time set. A genuinely a priori statement becomes possible only after a signal class has been fixed. The right object is not a scalar density alone, but the weak-limit hull of all time sets that the class can generate.

\begin{definition}[Class weak-limit hull]
\label{def:class_weak_limit_hull}
Fix a threshold spacing $\Delta V>0$ and a signal class $\mathcal X$. Write
\[
    \Lambda_x=\Lambda_{\Delta V}(x),\qquad x\in\mathcal X,
\]
for the retained time-coordinate set produced by the recorded adjacent-transition convention. The class is called \emph{uniformly separated} if
\begin{equation}
    \operatorname{sep}(\mathcal X,\Delta V)
    :=
    \inf_{x\in\mathcal X}\inf_{\lambda\ne\lambda'\in\Lambda_x}
    |\lambda-\lambda'|>0.
    \label{eq:class_uniform_separation}
\end{equation}
Its MT weak-limit hull is
\begin{equation}
    \mathfrak W_{\Delta V}(\mathcal X)
    :=
    \left\{
    \Gamma:
    \exists x_j\in\mathcal X,
    \exists a_j\in\mathbb R,
    \ \Lambda_{x_j}-a_j\xrightarrow{w}\Gamma
    \right\}.
    \label{eq:class_weak_limit_hull}
\end{equation}
The class is called a \emph{uniform stable MT sampling class} for $PW_\Omega$ if there exist constants $0<A\le B<\infty$, independent of $x$, such that
\begin{equation}
    \begin{aligned}
    A\|f\|_{L^2}^2
    &\le
    \sum_{\lambda\in\Lambda_x}|f(\lambda)|^2
    \le
    B\|f\|_{L^2}^2,\\
    &\hspace{0.45in}
    f\in PW_\Omega,\quad x\in\mathcal X.
    \end{aligned}
    \label{eq:uniform_class_sampling}
\end{equation}
\end{definition}

The uniform-separation hypothesis is a point-set hypothesis. For an MT-generated class it must be checked from the signal model. A simple sufficient condition is an amplitude envelope:
\[
    \sup_{x\in\mathcal X}\|x\|_\infty\le M<\infty,\qquad M>0 .
\]
Together with the retained adjacent-transition convention, Bernstein's inequality gives
\begin{equation}
    \operatorname{sep}(\mathcal X,\Delta V)
    \ge
    \frac{\Delta V}{\Omega M}.
    \label{eq:class_sep_amplitude_envelope}
\end{equation}
If no such envelope, or no equivalent class-level separation certificate, is available, the theorem below remains a point-set family theorem but does not automatically apply to that MT signal class.

\begin{theorem}[A priori class-hull MT criterion]
\label{thm:apriori_class_hull_mvt}
\textit{Assume that $\mathcal X$ is uniformly separated in the sense of \eqref{eq:class_uniform_separation}. Then $\mathcal X$ is a uniform stable MT sampling class for $PW_\Omega$ if and only if every set $\Gamma\in\mathfrak W_{\Delta V}(\mathcal X)$ is a uniqueness set for the Bernstein space $\mathcal B_\Omega$:}
\begin{equation}
    F\in\mathcal B_\Omega,
    \quad F|_\Gamma=0
    \quad\Longrightarrow\quad
    F\equiv0.
    \label{eq:class_hull_uniqueness}
\end{equation}
\end{theorem}

\textit{Proof.}
The ordinary Plancherel--P\'olya inequality gives a uniform upper frame bound in~\eqref{eq:uniform_class_sampling} from the uniform separation constant in~\eqref{eq:class_uniform_separation}. Hence only the lower bound is at issue.

We use Beurling's weak-limit theorem in its uniform-family form: a uniformly separated family of point sets is uniformly sampling for $PW_\Omega$ if and only if every translate weak limit of every set in the family is a uniqueness set for $\mathcal B_\Omega$. This is the same compactness theorem used in Section~\ref{sec:weak_limit_condition}, with the sequence of translates allowed to come from different members of the family.

Apply this external theorem to the uniformly separated family $\{\Lambda_x:x\in\mathcal X\}$. Its family translate weak-limit hull is exactly $\mathfrak W_{\Delta V}(\mathcal X)$ by definition. Hence the uniform lower sampling bound holds precisely when every $\Gamma\in\mathfrak W_{\Delta V}(\mathcal X)$ is a uniqueness set for $\mathcal B_\Omega$. Together with the uniform Plancherel--P\'olya upper bound, this proves the equivalence. \hfill $\square$

\begin{remark}[Relation to the finite-energy no-go theorem]
The full finite-energy class $PW_\Omega$ with static, uniformly spaced thresholds is not a uniform stable MT sampling class. Indeed, Theorem~\ref{thm:no_go} shows that every retained finite-energy static MT record has zero full-line lower density; equivalently, the class weak-limit hull contains non-uniqueness limits such as the empty set. The class-hull theorem therefore does not weaken the no-go result; it identifies the extra a priori information needed to escape it.
\end{remark}

\begin{corollary}[Finite periodic-template reduction modulo translations]
\label{cor:finite_periodic_template_mvt}
\textit{Assume the hypotheses of Theorem~\ref{thm:apriori_class_hull_mvt}. For a separated set $\Gamma\subset\mathbb R$, write}
\[
    [\Gamma]
    =
    \{\Gamma-\tau:\tau\in\mathbb R\}
\]
\textit{for its translation orbit. Suppose that the class weak-limit hull has only finitely many translation orbits,}
\begin{equation}
    \mathfrak W_{\Delta V}(\mathcal X)/\mathbb R
    =
    \{[\Gamma_1],\ldots,[\Gamma_J]\},
    \label{eq:finite_periodic_template_quotient_hull}
\end{equation}
\textit{where each representative $\Gamma_j$ is a periodic template}
\begin{equation}
    \begin{aligned}
    \Gamma_j
    &=
    \{kT_j+a_{j,r}:k\in\mathbb Z,\ r=0,\ldots,q_j-1\},\\
    &\hspace{0.2in}
    0\le a_{j,0}<\cdots<a_{j,q_j-1}<T_j,
    \end{aligned}
    \label{eq:periodic_template_family}
\end{equation}
\textit{with positive periods $T_j$. Then $\mathcal X$ is a uniform stable MT sampling class for $PW_\Omega$ if and only if}
\begin{equation}
    \frac{q_j}{T_j}
    \ge
    \frac{\Omega}{\pi}
    =
    2f_{max},
    \qquad j=1,\ldots,J.
    \label{eq:periodic_template_density}
\end{equation}
\textit{Equality is the complete-interpolation case for the corresponding periodic orbit; strict inequality gives a redundant sampling frame.}
\end{corollary}

\textit{Proof.}
Sampling, interpolation, Bernstein uniqueness, and frame bounds are invariant under translation of the sampling set, up to the harmless modulation induced on the Fourier side. Hence it is enough to test one representative from each orbit.

Fix one template and write it as
\[
    \begin{aligned}
    \Gamma
    &=
    \{kT+a_r:k\in\mathbb Z,\ r=0,\ldots,q-1\},
    \\
    &\hspace{0.2in}
    0\le a_0<\cdots<a_{q-1}<T.
    \end{aligned}
\]
Set $P=2\pi/T$. For $f\in PW_\Omega$, Poisson summation on each coset gives, for almost every $\lambda\in[0,P)$,
\begin{equation}
    \sum_{k\in\mathbb Z} f(kT+a_r)e^{-ikT\lambda}
    =
    \frac{1}{T}
    \sum_{m\in M(\lambda)}
        \widehat f(\lambda+mP)e^{i(\lambda+mP)a_r},
    \label{eq:periodic_template_fiber}
\end{equation}
where
\[
    M(\lambda)=\{m\in\mathbb Z:\lambda+mP\in[-\Omega,\Omega]\}.
\]
Thus the sampling operator fibers into the finite matrices
\begin{equation}
    A(\lambda)_{r,m}
    =
    e^{i(\lambda+mP)a_r},
    \qquad
    r=0,\ldots,q-1,
    \quad m\in M(\lambda).
    \label{eq:periodic_template_alias_matrix}
\end{equation}

If $q/T\ge\Omega/\pi$, then $2\Omega\le qP$. Hence $|M(\lambda)|\le q$ for almost every $\lambda$. The active indices in $M(\lambda)$ are consecutive, and after removing the harmless row factors $e^{i\lambda a_r}$, every active matrix is a rectangular Vandermonde matrix in the distinct nodes
\[
    z_r=e^{iPa_r}.
\]
Therefore every active submatrix has full column rank. Since only finitely many active index sets occur, the smallest nonzero singular value is bounded below uniformly in $\lambda$. Integrating the fiber inequalities gives the sampling frame inequality for $\Gamma$, and therefore for every translate of $\Gamma$.

If $q/T<\Omega/\pi$, then $2\Omega>qP$, so on a positive-measure set of fibers one has $|M(\lambda)|\ge q+1$. On a smaller positive-measure subset the active index set is constant. There $A(\lambda)$ has more columns than rows and hence has a nontrivial nullspace. Choosing $\widehat f$ supported on that subset in a measurable null direction gives a nonzero $f\in PW_\Omega$ whose samples on $\Gamma$ vanish. Thus $\Gamma$ is not a uniqueness set, hence not a sampling set. This proves~\eqref{eq:periodic_template_density} for a single orbit.

Because the quotient hull~\eqref{eq:finite_periodic_template_quotient_hull} contains only finitely many template orbits, the positive lower frame bounds in the sampling cases can be minimized over the finitely many representatives, and the upper bounds can be maximized. The class-hull theorem then gives the asserted uniform MT conclusion. Conversely, if one orbit in the quotient hull violates~\eqref{eq:periodic_template_density}, that orbit contains a non-uniqueness weak limit, and Theorem~\ref{thm:apriori_class_hull_mvt} rules out a uniform lower sampling bound. \hfill $\square$

\begin{corollary}[Periodic-template threshold-density form]
\label{cor:periodic_template_threshold_density}
\textit{In the setting of Corollary~\ref{cor:finite_periodic_template_mvt}, assume that every gap in template $j$ is an adjacent retained MT threshold gap and that}
\begin{equation}
    T_j
    =
    \Delta V\sum_{r=0}^{q_j-1}v_{j,r}^{-1},
    \qquad
    v_{j,r}>0.
    \label{eq:template_velocity_period}
\end{equation}
\textit{Here $v_{j,r}$ is the gap-derived crossing-velocity parameter $v_{j,r}:=\Delta V/g_{j,r}$, where $g_{j,r}$ is the corresponding physical time gap. If the gap is generated by a differentiable monotone branch, the ordinary mean value theorem realizes this number as $|x'(\xi_{j,r})|$ for some point inside the gap.}
\textit{Let}
\begin{equation}
    v_{H,j}
    =
    \frac{q_j}{\sum_{r=0}^{q_j-1}v_{j,r}^{-1}}
    \label{eq:template_harmonic_velocity}
\end{equation}
\textit{be the harmonic mean of these gap-derived velocities for template $j$. Then the finite periodic-template quotient-hull class is uniformly stably sampled if and only if}
\begin{equation}
    \delta_V=\frac1{\Delta V}
    \ge
    \max_j \frac{\Omega}{\pi v_{H,j}}
    =
    \max_j \frac{2f_{max}}{v_{H,j}}.
    \label{eq:template_threshold_density_iff}
\end{equation}
\end{corollary}

\textit{Proof.}
By~\eqref{eq:template_velocity_period},
\[
    \frac{q_j}{T_j}
    =
    \frac{q_j}{\Delta V\sum_r v_{j,r}^{-1}}
    =
    \frac{v_{H,j}}{\Delta V}
    =
    v_{H,j}\delta_V.
\]
Substituting this identity into~\eqref{eq:periodic_template_density} for every $j$ gives exactly~\eqref{eq:template_threshold_density_iff}. \hfill $\square$

\subsubsection{Finite-State Gap Priors}
\label{sec:finite_state_gap_priors}

The finite periodic-template quotient-hull corollary covers the case in which every
possible weak limit lies in one of finitely many translation orbits of periodic words.
A slightly richer a priori model is a finite-state gap prior: the long-run MT gap
pattern is generated by a finite directed graph. This includes periodic templates as
single cycles, while also allowing the signal to switch among finitely many certified
operating modes.

Let $G=(V,E)$ be a finite directed graph. Each edge $e\in E$ carries a positive
normalized time gap $h_e>0$, measured in Nyquist units. Thus the physical gap is
$g_e=h_e\Delta T$, where $\Delta T=\pi/\Omega$. A bi-infinite admissible path
$\gamma=(e_n)_{n\in\mathbb Z}$ generates a normalized point set
$\mu_\gamma=\{\mu_n\}_{n\in\mathbb Z}$ by
\begin{equation}
    \mu_{n+1}-\mu_n=h_{e_n}.
    \label{eq:finite_state_gap_generation}
\end{equation}
Let $\mathcal S_G$ denote the translation-invariant family of physical point sets
\begin{equation}
    \begin{aligned}
    \mathcal S_G
    &=
    \left\{
        a+\Delta T\,\mu_\gamma:
        a\in\mathbb R,\right.\\
    &\hspace{0.3in}\left.
        \gamma \text{ is a bi-infinite admissible path in }G
    \right\}.
    \end{aligned}
    \label{eq:finite_state_translation_invariant_family}
\end{equation}
Since $G$ is finite and all edge weights are positive, every set in $\mathcal S_G$ is separated with a common separation constant. For a directed cycle $C$ in $G$, define its mean normalized gap
\begin{equation}
    \bar h(C)
    =
    \frac{1}{|C|}\sum_{e\in C}h_e,
    \qquad
    \bar h_G
    =
    \max_C\bar h(C),
    \label{eq:max_cycle_mean_gap}
\end{equation}
where the maximum is over directed cycles that can occur in a bi-infinite admissible
path.
Equivalently, the cycle lies in a strongly connected component that supports a
bi-infinite path; finite transient edges outside such components can affect only
finite prefixes and do not change the lower-density or weak-limit obstruction.

\begin{theorem}[Finite-state gap-prior MT criterion]
\label{thm:finite_state_gap_prior}
Assume the finite-state gap model above.
\begin{enumerate}
    \item[\textup{(i)}]
    Every $S\in\mathcal S_G$ satisfies the uniform density estimate
    \begin{equation}
        D^-(S)
        \ge
        \frac{1}{\Delta T\,\bar h_G}
        =
        \frac{\Omega}{\pi\bar h_G}.
        \label{eq:finite_state_density_lower}
    \end{equation}
    Moreover, if $C_*$ is a cycle attaining the maximum in
    \eqref{eq:max_cycle_mean_gap}, then every translate of the periodic path obtained by repeating
    $C_*$ has density exactly $(\Delta T\bar h_G)^{-1}$.

    \item[\textup{(ii)}]
    If $\bar h_G<1$, then $\mathcal S_G$ is a uniform stable sampling family for
    $PW_\Omega$. In particular, if an MT signal class $\mathcal X$ is uniformly
    separated and its class weak-limit hull satisfies
    \begin{equation}
        \mathfrak W_{\Delta V}(\mathcal X)\subseteq\mathcal S_G,
        \label{eq:hull_contained_finite_state_family}
    \end{equation}
    then $\mathcal X$ is a uniform stable MT sampling class for $PW_\Omega$.

    \item[\textup{(iii)}]
    If $\bar h_G>1$, then the graph family $\mathcal S_G$ is not a sampling family
    for $PW_\Omega$: the periodic point set generated by a maximizing cycle has
    density below $\Omega/\pi$ and fails Landau's necessary density condition. If
    this maximizing-cycle template, or any of its translates, belongs to
    $\mathfrak W_{\Delta V}(\mathcal X)$, then the MT class $\mathcal X$ cannot
    have a uniform stable MT sampling guarantee.
\end{enumerate}
At the critical value $\bar h_G=1$, the cycle mean is not decisive: one must return
to the class-hull uniqueness criterion of Theorem~\ref{thm:apriori_class_hull_mvt}
or to a structure-specific complete-interpolation test.
\end{theorem}

\textit{Proof.}
We first prove the density estimate. Let $N\ge1$ and consider any block of $N$
consecutive edges in an admissible path. Remove directed cycles from this finite path
one at a time. Each removed cycle has average gap at most $\bar h_G$. After all cycles
have been removed, the remaining path is simple and has at most $|V|-1$ edges. If
$h_{\max}=\max_{e\in E}h_e$, the residual contribution is at most
$(|V|-1)h_{\max}$. Hence there is a graph constant
\begin{equation}
    C_G=(|V|-1)h_{\max}
    \label{eq:finite_state_residual_constant}
\end{equation}
such that every block of $N$ consecutive normalized gaps has total length at most
\begin{equation}
    \sum_{k=m}^{m+N-1} h_{e_k}
    \le
    \bar h_G N+C_G.
    \label{eq:path_block_cycle_bound}
\end{equation}
Because the edge set is finite and positive, a window of normalized length $R$ can
lose only $O_G(1)$ points at its two endpoints. Using
\eqref{eq:path_block_cycle_bound} gives
\begin{equation}
    \#(\mu_\gamma\cap[u,u+R])
    \ge
    \frac{R}{\bar h_G}-O_G(1),
    \qquad R\to\infty,
    \label{eq:finite_state_count_lower}
\end{equation}
uniformly in the admissible path $\gamma$ and the window position $u$. Dividing by
$R$, taking the infimum over $u$, and then the lower limit as $R\to\infty$ gives the
normalized lower-density bound $D^-(\mu_\gamma)\ge1/\bar h_G$, and scaling by
$\Delta T=\pi/\Omega$ gives~\eqref{eq:finite_state_density_lower}. Translation does
not change density. If $C_*$ attains $\bar h_G$, the periodic path
$C_*C_*C_*\cdots$ has one point per average normalized length $\bar h_G$, so each
of its translates has density exactly $(\Delta T\bar h_G)^{-1}$.

If $\bar h_G<1$, then~\eqref{eq:finite_state_density_lower} gives the uniform strict
margin
\[
    D^-(S)
    \ge
    \frac{\Omega}{\pi\bar h_G}
    >
    \frac{\Omega}{\pi},
    \qquad S\in\mathcal S_G.
\]
The finite positive gap alphabet gives a common separation constant. Beurling's
strict-density theorem, in its uniform separated-family form, therefore gives common
sampling frame bounds for all $S\in\mathcal S_G$. If
\eqref{eq:hull_contained_finite_state_family} holds, every weak limit of every MT set
in $\mathcal X$ lies in a uniformly sampling family and is in particular a Bernstein
uniqueness set; Theorem~\ref{thm:apriori_class_hull_mvt} then gives uniform stable
MT sampling for $\mathcal X$.

If $\bar h_G>1$, the periodic point set generated by a maximizing cycle has density
$(\Delta T\bar h_G)^{-1}<\Omega/\pi$. Landau's necessary density theorem excludes
stable sampling of $PW_\Omega$ by that periodic set or by any translate of it. Thus
the full graph family cannot be uniformly sampled. If the same periodic template, or
one of its translates, occurs in the MT class weak-limit hull,
Theorem~\ref{thm:apriori_class_hull_mvt} also rules out a uniform lower sampling
bound for $\mathcal X$. \hfill $\square$

\begin{corollary}[Finite-state threshold-density form]
\label{cor:finite_state_threshold_density}
Assume the finite-state gap prior of Theorem~\ref{thm:finite_state_gap_prior}. Suppose
that every edge gap is an adjacent retained MT threshold gap of the form
\begin{equation}
    g_e=\frac{\Delta V}{v_e},
    \qquad v_e>0,
    \label{eq:finite_state_edge_velocity}
\end{equation}
where $v_e:=\Delta V/g_e$ is the edge's gap-derived crossing-velocity parameter, with the same mean-value interpretation as above when the edge is generated by a differentiable monotone branch,
so that
\begin{equation}
    h_e
    =
    \frac{g_e}{\Delta T}
    =
    \frac{\Omega\Delta V}{\pi v_e}.
    \label{eq:finite_state_normalized_velocity_gap}
\end{equation}
For a directed cycle $C$, define the harmonic mean of these gap-derived velocities
\begin{equation}
    v_H(C)
    =
    \frac{|C|}{\sum_{e\in C}v_e^{-1}}.
    \label{eq:finite_state_cycle_harmonic_velocity}
\end{equation}
Then, away from the critical value, the graph family $\mathcal S_G$ is uniformly
stably sampling for $PW_\Omega$ if and only if
\begin{equation}
    \delta_V=\frac1{\Delta V}>\delta_G^*,
    \label{eq:finite_state_threshold_rule_intro}
\end{equation}
where
\begin{equation}
    \delta_G^*
    =
    \max_C \frac{\Omega}{\pi v_H(C)}
    =
    \max_C \frac{2f_{max}}{v_H(C)},
    \label{eq:finite_state_threshold_critical_value}
\end{equation}
provided $\delta_V\ne\delta_G^*$. More explicitly, $\delta_V>\delta_G^*$ gives a
uniform stable sampling guarantee for every graph-admissible MT gap stream, whereas
$\delta_V<\delta_G^*$ produces a periodic graph-admissible stream whose density is
below Nyquist. At equality, density is critical and the class-hull or periodic
complete-interpolation geometry must be checked.
\end{corollary}

\textit{Proof.}
For a cycle $C$, substituting~\eqref{eq:finite_state_normalized_velocity_gap} into
\eqref{eq:max_cycle_mean_gap} gives
\[
    \bar h(C)
    =
    \frac{\Omega\Delta V}{\pi}
    \frac{1}{|C|}\sum_{e\in C}v_e^{-1}
    =
    \frac{\Omega\Delta V}{\pi v_H(C)}.
\]
Therefore $\bar h_G<1$ is equivalent to
\[
    \Delta V
    <
    \min_C \frac{\pi v_H(C)}{\Omega},
\]
or, in threshold-density form,
\[
    \delta_V
    >
    \max_C\frac{\Omega}{\pi v_H(C)}.
\]
The cases $\bar h_G<1$ and $\bar h_G>1$ are exactly the two noncritical alternatives
in Theorem~\ref{thm:finite_state_gap_prior}; replacing $\Omega/\pi$ by $2f_{max}$
gives the final expression. \hfill $\square$

\begin{corollary}[Interval-uncertain finite-state gap prior]
\label{cor:interval_uncertain_finite_state_gap_prior}
\textit{Let $G=(V,E)$ be a finite directed graph. Suppose that each normalized edge
gap is not known exactly but is constrained to an interval}
\[
    h_e\in[\underline h_e,\overline h_e],
    \qquad
    0<\underline h_e\le \overline h_e,
    \qquad e\in E.
\]
\textit{Let $\mathcal S_G([\underline h,\overline h])$ denote the translation-invariant
family of all physical point sets $a+\Delta T\,\mu_\gamma$ generated by bi-infinite
admissible paths in $G$, with each edge gap chosen inside its interval. Define the
worst-case maximum cycle mean}
\begin{equation}
    \overline h_G
    =
    \max_C
    \frac1{|C|}
    \sum_{e\in C}\overline h_e,
    \label{eq:interval_worst_cycle_mean}
\end{equation}
\textit{where the maximum is over directed cycles that can occur in a bi-infinite
path, equivalently cycles contained in a strongly connected component supporting
such a path. If}
\begin{equation}
    \overline h_G<1,
    \label{eq:interval_finite_state_safe}
\end{equation}
\textit{then every point set in $\mathcal S_G([\underline h,\overline h])$ is a
sampling set for $PW_\Omega$, with uniform frame bounds depending only on the graph,
the interval endpoints, and the strict margin $1-\overline h_G$.}

\textit{Conversely, if $\overline h_G>1$, then there exists an admissible choice of
edge gaps within the intervals and a periodic path in $G$ whose generated point set
has lower density below the critical value $\Omega/\pi$. Hence no uniform sampling
guarantee can hold for the entire interval-uncertain family. The critical case
$\overline h_G=1$ is not decided by the cycle mean alone.}
\end{corollary}

\textit{Proof.}
For every compatible choice of gaps and every directed cycle $C$,
\[
    \frac1{|C|}\sum_{e\in C}h_e
    \le
    \frac1{|C|}\sum_{e\in C}\overline h_e
    \le
    \overline h_G.
\]
The cycle-decomposition argument in Theorem~\ref{thm:finite_state_gap_prior} therefore
applies uniformly with $\overline h_G$ in place of the exact maximum cycle mean. If
$\overline h_G<1$, every generated normalized point set has lower density at least
$1/\overline h_G>1$, uniformly over paths, translations, and admissible interval
choices. The positive lower endpoints give a common separation constant. Beurling's
strict-density theorem then gives uniform sampling bounds after scaling by
$\Delta T=\pi/\Omega$.

If $\overline h_G>1$, choose a maximizing cycle and set each edge on that cycle to
its upper endpoint. Repeating the cycle gives a periodic admissible path with average
normalized gap $\overline h_G>1$, hence physical lower density
$(\Delta T\overline h_G)^{-1}<\Omega/\pi$. Landau's necessary density theorem rules
out stable sampling for that point set, so the whole interval-uncertain family cannot
have a uniform lower frame bound. \hfill $\square$

\textbf{Interpretation.}
Theorem~\ref{thm:apriori_class_hull_mvt} is the closest exact a priori statement available inside the general nonuniform-sampling framework: the prior object is the class-level weak-limit hull. Corollary~\ref{cor:finite_periodic_template_mvt} shows when this exact condition collapses to finitely many algebraic density/rank checks, but its iff uses the finite quotient hull~\eqref{eq:finite_periodic_template_quotient_hull}; a mere covering library gives sufficiency, and a negative conclusion requires the bad template orbit to belong to the quotient hull. The scalar threshold density $\delta_V$ becomes an if-and-only-if threshold only in the final, more restrictive situation where each admissible weak-limit orbit has a calibrated harmonic-mean MT velocity, or in the finite-state noncritical prior of Corollary~\ref{cor:finite_state_threshold_density}. Without such class-level weak-limit, finite-orbit periodic-template, finite-state, or exact-density-calibration structure, density remains only necessary below the critical value and sufficient above it, with the critical boundary decided by microscopic geometry.

\subsection{What the Post-Realization Condition Implies for Static Finite-Energy MT}
\label{sec:finite_energy_exact_implication}
The post-realization condition immediately explains a fundamental limitation of the static-threshold model.

\begin{corollary}[Static finite-energy implication]
\label{cor:finite_energy_static_density_zero}
\textit{For every finite-energy $x\in PW_\Omega$ and every fixed $\Delta V>0$ under the recorded adjacent-transition convention, the retained static MT record has}
\begin{equation}
    D^-(\Lambda_{\Delta V}(x))=0.
    \label{eq:finite_energy_mvt_density_zero}
\end{equation}
\textit{Consequently, without external padding values, tail information, a persistent streaming model, or a finite-dimensional prior, such a record cannot globally sample the full infinite-dimensional space $PW_\Omega$.}
\end{corollary}

This corollary is included here to connect the post-realization sampling condition with the static finite-energy boundary. The full proof is the global part of Theorem~\ref{thm:no_go}: Riemann--Lebesgue tail decay removes informative tail events, the retained adjacent-transition convention suppresses repeated same-threshold returns in the final threshold cell, the separation estimate~\eqref{eq:exact_mvt_separation} makes the compact part finite, and Landau's necessary density condition then rules out sampling of $PW_\Omega$.

This conclusion does not invalidate the finite-record reconstruction pipeline. It identifies its exact logical status: finite MT data support recovery only after the target space has been restricted or after unmeasured tail information has been supplied.

\subsection{Shannon-Type Complete-Interpolation Subcase}
\label{sec:critical_complete_interp}
The weak-limit condition is the most general exact statement for possibly redundant sampling frames. A more explicit theorem is available if one asks for the closest analogue of the Shannon cardinal series: a non-redundant, critical-density subsequence with one sample per Nyquist degree of freedom and a Lagrange interpolation formula. This is a stronger structural requirement than stable sampling by the full MT set, but it gives the most literal MT analogue of the Nyquist--Shannon theorem.

Let $S=\{s_n\}_{n\in\mathbb{Z}}\subseteq\Lambda_{\Delta V}(x)$ be a strictly ordered, separated MT event subsequence, and normalize it to the $\pi$-bandlimited scale by
\[
    \mu_n=\frac{\Omega}{\pi}s_n.
\]
Let $G$ be the canonical product whose real zeros are $\{\mu_n\}$, with the standard primary factors needed for convergence. Define
\[
    W_G(u)=\frac{|G(u)|}{\mathrm{dist}(u,\{\mu_n\})}.
\]

\begin{theorem}[Shannon-type MT--Pavlov condition]
\label{thm:mvt_pavlov_condition}
\textit{Under the standard real separated generating-function hypotheses of Pavlov's theorem, the following statements are equivalent:}
\begin{enumerate}
    \item[\textup{(i)}] \textit{The MT subsequence $S$ gives Shannon-type stable exact recovery on $PW_\Omega$: the restriction map}
    \[
        T_S f=(f(s_n))_{n\in\mathbb{Z}}:
        PW_\Omega\to\ell^2(\mathbb{Z})
    \]
    \textit{is an isomorphism onto $\ell^2(\mathbb{Z})$.}
    \item[\textup{(ii)}] \textit{The normalized node sequence $\{\mu_n\}$ is a complete interpolating sequence for $PW_\pi$.}
    \item[\textup{(iii)}] \textit{The exponential system $\{e^{-j\mu_n u}\}_{n\in\mathbb{Z}}$ is a Riesz basis for $L^2[-\pi,\pi]$, equivalently $\{e^{-j\omega s_n}\}_{n\in\mathbb{Z}}$ is a Riesz basis for $L^2[-\Omega,\Omega]$.}
    \item[\textup{(iv)}] \textit{The canonical-product weight satisfies the Pavlov/Lyubarskii--Seip Muckenhoupt condition}
\begin{equation}
    \sup_I
    \left(\frac{1}{|I|}\int_I W_G(u)^2\,du\right)
    \left(\frac{1}{|I|}\int_I W_G(u)^{-2}\,du\right)
    <\infty
    \label{eq:muckenhoupt_a2_condition}
\end{equation}
    \textit{where the supremum is taken over all bounded real intervals $I$~\cite{LyubarskiiSeip1997}.}
\end{enumerate}
\textit{When these conditions hold, every $f\in PW_\Omega$ admits the Lagrange reconstruction}
\begin{equation}
    f(t)
    =
    \sum_{n\in\mathbb{Z}}
        f(s_n)
        \frac{G(\Omega t/\pi)}
             {G'(\mu_n)(\Omega t/\pi-\mu_n)}.
    \label{eq:mvt_shannon_pavlov_formula}
\end{equation}
\textit{For the generating MT waveform itself, the coefficients in~\eqref{eq:mvt_shannon_pavlov_formula} are the recorded threshold values $f(s_n)=x(s_n)=m_n\Delta V$.}
\end{theorem}

\textit{Proof.}
The dilation $g(u)=f(\pi u/\Omega)$ carries $PW_\Omega$ to $PW_\pi$ and maps the samples $f(s_n)$ to $g(\mu_n)$. Therefore $T_S$ is an isomorphism on $PW_\Omega$ exactly when the normalized restriction map $g\mapsto(g(\mu_n))$ is an isomorphism on $PW_\pi$, which is the definition of a complete interpolating sequence. By the inverse Fourier representation, the latter restriction map is the analysis operator associated with $\{e^{-j\mu_n u}\}$ on $L^2[-\pi,\pi]$; hence it is an isomorphism exactly when the exponential system is a Riesz basis. Pavlov's theorem in the Lyubarskii--Seip formulation gives the equivalence with~\eqref{eq:muckenhoupt_a2_condition}. Finally, complete interpolation gives the cardinal functions
\[
    \ell_n(u)=
    \frac{G(u)}{G'(\mu_n)(u-\mu_n)},
    \qquad
    \ell_n(\mu_m)=\delta_{nm},
\]
so $g(u)=\sum_n g(\mu_n)\ell_n(u)$. Substituting $u=\Omega t/\pi$ gives~\eqref{eq:mvt_shannon_pavlov_formula}. Appendix~\ref{app:pavlov_a2} expands each step, including the scaling constants and Shannon-grid reduction. \hfill $\square$

If $\mu_n=n+\alpha$, then one may take $G(u)=\sin\pi(u-\alpha)$. Equation~\eqref{eq:mvt_shannon_pavlov_formula} reduces to the shifted Shannon series
\begin{equation}
    f(t)
    =
    \sum_{n\in\mathbb{Z}}
        f\!\left(\tau+\frac{\pi n}{\Omega}\right)
        \mathrm{sinc}\!\left(
            \frac{\Omega(t-\tau)}{\pi}-n
        \right),
    \qquad
    \alpha=\frac{\Omega\tau}{\pi}.
    \label{eq:shifted_shannon_from_pavlov}
\end{equation}
Thus the canonical product $G$ is the nonuniform replacement for $\sin\pi u$, and~\eqref{eq:mvt_shannon_pavlov_formula} is the corresponding MT-generated generalized sinc formula.

The theorem also marks the boundary of explicitness. At the exact Nyquist density, a scalar equality $D^-(S)=\Omega/\pi$ is not sufficient; the decisive object is the full canonical-product geometry of the realized crossing subsequence. The complete-interpolation theorem is therefore more explicit than the weak-limit condition because it gives a reconstruction kernel, but it is not a finite-parameter design rule in $\Delta V$ alone.

\subsubsection{Periodic-Gap MT--Shannon Theorem}
\label{sec:periodic_gap_shannon}
There is, however, one important complete-interpolation subclass in which the canonical-product geometry collapses to a finite list of gaps. Assume that the normalized MT nodes have periodic gaps:
\begin{equation}
    \mu_{n+1}-\mu_n=h_{n\bmod q},
    \qquad
    h_r>0,\quad r=0,\ldots,q-1,
    \label{eq:periodic_gap_def}
\end{equation}
and define the period-block length
\begin{equation}
    H=\sum_{r=0}^{q-1}h_r.
    \label{eq:periodic_gap_H}
\end{equation}
Then $\mu_{n+q}=\mu_n+H$. This model should be read as an ideal persistent stream, a periodic steady-state stream, a periodic extension, or an extracted bi-infinite periodic subsequence. It does not contradict Subsection~\ref{sec:finite_energy_exact_implication}: a retained finite-energy static MT record still eventually stops producing informative tail events.

\begin{theorem}[Periodic-gap MT--Shannon condition]
\label{thm:periodic_gap_mvt_shannon_condition}
\textit{Under~\eqref{eq:periodic_gap_def}, the following statements are equivalent within the periodic-gap class:}
\begin{enumerate}
    \item[\textup{(i)}] \textit{$S$ gives Shannon-type stable exact recovery on $PW_\Omega$, i.e., $f\mapsto(f(s_n))$ is an isomorphism from $PW_\Omega$ onto $\ell^2(\mathbb{Z})$.}
    \item[\textup{(ii)}] \textit{$\{\mu_n\}$ is a complete interpolating sequence for $PW_\pi$.}
    \item[\textup{(iii)}] \textit{$\{e^{-j\mu_n u}\}_{n\in\mathbb{Z}}$ is a Riesz basis for $L^2[-\pi,\pi]$.}
    \item[\textup{(iv)}] \textit{The normalized period block is critical:}
    \begin{equation}
        H=q.
        \label{eq:periodic_mvt_shannon_condition}
    \end{equation}
\end{enumerate}
\textit{When these conditions hold, write}
\begin{equation}
    \beta_0=\mu_0,\qquad
    \beta_r=\mu_0+\sum_{\ell=0}^{r-1}h_\ell,
    \quad r=1,\ldots,q-1.
    \label{eq:periodic_beta_def}
\end{equation}
\textit{Then}
\begin{equation}
    \{\mu_n:n\in\mathbb{Z}\}
    =
    \bigcup_{r=0}^{q-1}(q\mathbb{Z}+\beta_r),
    \label{eq:periodic_union_nodes}
\end{equation}
\textit{with distinct residues $\beta_r\bmod q$. A generating function is}
\begin{equation}
    G_{\mathrm{per}}(u)
    =
    \prod_{r=0}^{q-1}
    \sin\!\left(\frac{\pi}{q}(u-\beta_r)\right),
    \label{eq:periodic_generating_function}
\end{equation}
\textit{and every $f\in PW_\Omega$ has the explicit expansion}
\begin{equation}
    f(t)
    =
    \sum_{n\in\mathbb{Z}}
    f(s_n)
    \frac{G_{\mathrm{per}}(\Omega t/\pi)}
         {G_{\mathrm{per}}'(\mu_n)(\Omega t/\pi-\mu_n)}.
    \label{eq:periodic_mvt_shannon_formula}
\end{equation}
\textit{For the generating ideal MT waveform, $f(s_n)=x(s_n)=m_n\Delta V$.}
\end{theorem}

The same condition has a direct event-gap interpretation. Since $\Delta T=\pi/\Omega$,
\begin{equation}
    h_{n\bmod q}
    =
    \mu_{n+1}-\mu_n
    =
    \frac{s_{n+1}-s_n}{\Delta T}
    =
    \frac{\Omega(s_{n+1}-s_n)}{\pi}.
    \label{eq:periodic_gap_time_normalization}
\end{equation}
When $s_n$ and $s_{n+1}$ are adjacent retained MT crossings on one monotone branch, the mean-value theorem gives a point $\xi_n\in(s_n,s_{n+1})$ such that
\begin{equation}
    |x'(\xi_n)|
    =
    \frac{\Delta V}{s_{n+1}-s_n},
    \qquad
    h_{n\bmod q}
    =
    \frac{\Omega\Delta V}{\pi |x'(\xi_n)|}.
    \label{eq:periodic_mvt_gap_velocity}
\end{equation}
Thus the explicit critical condition~\eqref{eq:periodic_mvt_shannon_condition} is equivalently
\begin{equation}
    \sum_{r=0}^{q-1}(s_{kq+r+1}-s_{kq+r})
    =
    q\Delta T
    =
    \frac{q\pi}{\Omega},
    \label{eq:periodic_block_time_condition}
\end{equation}
or, on monotone adjacent-threshold segments,
\begin{equation}
    \sum_{r=0}^{q-1}
    \frac{\Delta V}{|x'(\xi_{k,r})|}
    =
    \frac{q\pi}{\Omega}.
    \label{eq:periodic_block_velocity_condition}
\end{equation}
For $q=1$, $H=1$ forces $\mu_n=\mu_0+n$ and~\eqref{eq:periodic_mvt_shannon_formula} reduces to the shifted Shannon formula~\eqref{eq:shifted_shannon_from_pavlov}. For $q>1$, the theorem is still fully explicit, but it is explicit because the entire microscopic pattern is known: the finite gap block repeats forever. Appendix~\ref{app:periodic_gap_shannon} gives the complete proof, including the finite Vandermonde/fiber Riesz-basis calculation.

\begin{corollary}[Periodic-gap stable sampling form]
\label{cor:periodic_gap_stable_sampling}
\textit{Under the periodic-gap hypothesis~\eqref{eq:periodic_gap_def}, let $g_r=s_{n+r+1}-s_{n+r}$ be the physical time gaps over one period and let}
\[
    T_{\mathrm{per}}=\sum_{r=0}^{q-1}g_r,
    \qquad
    H=\sum_{r=0}^{q-1}h_r=\frac{T_{\mathrm{per}}}{\Delta T}.
\]
\textit{Then the periodic fixed set $S=\{s_n\}$ is a stable sampling set for $PW_\Omega$ if and only if}
\begin{equation}
    \frac{q}{T_{\mathrm{per}}}
    \ge
    \frac{\Omega}{\pi}
    =
    2f_{max},
    \qquad\text{equivalently}\qquad
    \frac{q}{H}\ge1.
    \label{eq:periodic_stable_density_condition}
\end{equation}
\textit{Equality is exactly the complete-interpolation case $H=q$ in Theorem~\ref{thm:periodic_gap_mvt_shannon_condition}. Strict inequality gives a redundant stable sampling frame, while $q/T_{\mathrm{per}}<\Omega/\pi$ is excluded by Landau's necessary density theorem.}

\textit{If each period gap is an adjacent retained threshold-crossing gap with}
\begin{equation}
    g_r=\frac{\Delta V}{v_r},
    \qquad
    v_r>0,
    \quad r=0,\ldots,q-1,
    \label{eq:periodic_gap_velocity}
\end{equation}
\textit{where $v_r:=\Delta V/g_r$ is the gap-derived crossing-velocity parameter. Define the harmonic mean of these parameters}
\begin{equation}
    v_H
    =
    \frac{q}{\sum_{r=0}^{q-1}v_r^{-1}}.
    \label{eq:periodic_harmonic_velocity}
\end{equation}
\textit{Then~\eqref{eq:periodic_stable_density_condition} becomes}
\begin{equation}
    \delta_V=\frac{1}{\Delta V}
    \ge
    \frac{\Omega}{\pi v_H}
    =
    \frac{2f_{max}}{v_H}.
    \label{eq:periodic_threshold_density_condition}
\end{equation}
\end{corollary}

\textit{Proof.}
It is enough to prove the normalized statement. The normalized periodic set can be written, after a translation, as
\[
    \begin{aligned}
    \mu=
    \{kH+\beta_r:k\in\mathbb Z,
    \ r=0,\ldots,q-1\},\\
    0\le \beta_0<\cdots<\beta_{q-1}<H.
    \end{aligned}
\]
where $H=T_{\mathrm{per}}/\Delta T$. Set $P=2\pi/H$. For $g\in PW_\pi$, Poisson summation on each coset gives, for almost every $\lambda\in[0,P)$,
\begin{equation}
    \sum_{k\in\mathbb Z} g(kH+\beta_r)e^{-ikH\lambda}
    =
    \frac{1}{H}
    \sum_{m\in M(\lambda)}
        \widehat g(\lambda+mP)e^{i(\lambda+mP)\beta_r},
    \label{eq:periodic_stable_fiber_poisson}
\end{equation}
where
\[
    M(\lambda)=\{m\in\mathbb Z:\lambda+mP\in[-\pi,\pi]\}.
\]
Thus the sampling operator fibers into the matrices
\begin{equation}
    A(\lambda)_{r,m}
    =
    e^{i(\lambda+mP)\beta_r},
    \qquad
    r=0,\ldots,q-1,
    \quad m\in M(\lambda).
    \label{eq:periodic_stable_alias_matrix}
\end{equation}

If $q/H\ge1$, equivalently $H\le q$, then $2\pi\le qP$. Hence $|M(\lambda)|\le q$ for almost every $\lambda$. The active indices in $M(\lambda)$ are consecutive. After removing the unit-modulus row factors $e^{i\lambda\beta_r}$, every active matrix is a rectangular Vandermonde matrix with nodes
\[
    z_r=e^{iP\beta_r},
    \qquad r=0,\ldots,q-1.
\]
These nodes are distinct because the residues $\beta_r\bmod H$ are distinct. Therefore every active submatrix has full column rank. Only finitely many active index sets occur as $\lambda$ varies, so the smallest nonzero singular value is bounded below uniformly. Integrating the resulting fiber inequalities gives the stable sampling frame inequality for $PW_\pi$, and scaling back gives stable sampling for $PW_\Omega$.

If $q/H<1$, equivalently $H>q$, then $2\pi>qP$. Therefore, on a subset of $[0,P)$ of positive measure, $|M(\lambda)|\ge q+1$. On a smaller positive-measure subset the active index set is constant. There the matrix $A(\lambda)$ has more columns than rows and hence a nontrivial nullspace. Choosing $\widehat g$ supported on that subset in a fixed null direction gives a nonzero $g\in PW_\pi$ whose samples on $\mu$ all vanish. Thus the set is not a sampling set. This proves $\mu$ samples $PW_\pi$ if and only if $q/H\ge1$, equivalently $q/T_{\mathrm{per}}\ge\Omega/\pi$ in physical units.

At equality $H=q$, the fiber matrices are square Vandermonde matrices almost everywhere; this is exactly the complete-interpolation case proved in Theorem~\ref{thm:periodic_gap_mvt_shannon_condition}. If $H<q$, the matrices are tall and the frame is redundant. Finally,~\eqref{eq:periodic_gap_velocity} gives
\[
    T_{\mathrm{per}}
    =
    \Delta V\sum_{r=0}^{q-1}v_r^{-1}
    =
    \Delta V\frac{q}{v_H},
\]
and substituting this identity into $q/T_{\mathrm{per}}\ge\Omega/\pi$ yields~\eqref{eq:periodic_threshold_density_condition}. \hfill $\square$

\subsection{Finite-Record Necessary-and-Sufficient Condition}
\label{sec:finite_exact_condition}
For the finite records actually produced by MT hardware, the exact condition is linear algebraic. Let
\[
    V_M=\mathrm{span}\{\phi_1,\ldots,\phi_M\}\subset PW_\Omega
\]
be the declared reconstruction space and let $\{t_n\}_{n=1}^{N}$ be the retained MT event times in the active window. Define
\begin{equation}
    (H_{\mathrm{fin}})_{n,k}=\phi_k(t_n),
    \qquad
    1\le n\le N,\quad 1\le k\le M.
    \label{eq:finite_exact_matrix}
\end{equation}
Then finite-record MT data determine every $p\in V_M$ stably and exactly if and only if
\begin{equation}
    \begin{aligned}
    \mathrm{rank}(H_{\mathrm{fin}})=M
    &\quad\Longleftrightarrow\quad
    H_{\mathrm{fin}}^*H_{\mathrm{fin}}\succ0\\
    &\quad\Longleftrightarrow\quad
    \sigma_{\min}(H_{\mathrm{fin}})>0.
    \end{aligned}
    \label{eq:finite_exact_rank_condition}
\end{equation}
If~\eqref{eq:finite_exact_rank_condition} holds, the finite frame constants are
\[
    A_M=\sigma_{\min}(H_{\mathrm{fin}})^2,\qquad
    B_M=\sigma_{\max}(H_{\mathrm{fin}})^2,
\]
and the coefficients of $p=\sum_{k=1}^M c_k\phi_k$ are recovered from $\mathbf{y}=H_{\mathrm{fin}}\mathbf{c}$ by $\mathbf{c}=H_{\mathrm{fin}}^\dagger\mathbf{y}$. If the rank condition fails, there exists a nonzero $\mathbf{c}_0\in\ker H_{\mathrm{fin}}$, so $p_0=\sum_k(c_0)_k\phi_k$ vanishes on all recorded event times. Then $p$ and $p+p_0$ generate identical finite sample values on the record and cannot be distinguished. This is the finite-dimensional counterpart of the weak-limit theorem and is the acceptance test used in Subsection~\ref{sec:finite_bridge} and Section~\ref{sec:reconstruction}.

\textbf{Interpretation.}
The condition $\sigma_{\min}(H_{\mathrm{fin}})>0$ is exact only after the finite model space $V_M$ has been declared. A finite active record may be perfectly well conditioned on $V_M$ and still say nothing universal about all of $PW_\Omega$ without a tail prior or additional model restriction.

\section{Rigorous Sufficient Conditions via Kadec's Theorem}
\label{sec:sufficient}
Before deriving sufficient certificates, we emphasize the boundary proved later in Theorem~\ref{thm:no_go}: for unrestricted finite-energy $PW_\Omega$ signals, a retained record generated by static, uniformly spaced thresholds has zero full-line lower density. Thus the Kadec and density conditions below are not universal prior-free full-line theorems. They are certificates for active windows, persistent/structured event models, tail-qualified records, or finite-dimensional reconstruction spaces.

The logic of this section proceeds on two complementary spatial scales. Macroscopically, each Nyquist-sized neighborhood of the time axis must receive at least one MT event. Microscopically, the received event must lie close enough to the corresponding Nyquist anchor to satisfy the perturbation bound in Kadec's theorem. The key departure from generic nonuniform sampling analysis is the \emph{structural regularity} of MT crossing sets. We first identify properties that arbitrary nonuniform point sets need not possess, and then use them to obtain tighter sufficient conditions and signal-instance frame bounds.

\textbf{Scope.} The sufficient conditions derived below apply over a finite active observation window $[0, T_{\text{obs}}]$ during which the signal maintains non-trivial amplitude, in the sense that $\|x\|_{L^\infty(\mathbb{R} \setminus [0, T_{\text{obs}}])} < \Delta V$. Because the MT set is signal-dependent, the theorem is used here locally and instance-wise: we seek conditions under which the realized crossing set contains a subsequence that can be paired with a reference Nyquist grid on the active interval.

This is an effective finite-record assumption, not a literal compact-support claim: it only requires the out-of-window amplitude to stay below one threshold step so that no additional informative crossing can occur there. It does not make the out-of-window energy observable; finite-record uniqueness still requires the tail prior or finite-dimensional model stated in Subsection~\ref{sec:finite_bridge}.

\subsection{Structural Regularity of MT Crossing Sets}
\label{sec:mvt_regularity}
Before deriving sufficient conditions, we establish four structural properties that distinguish MT crossing sets from arbitrary nonuniform sampling sequences. These properties are \emph{not} assumed; they are \emph{proved consequences} of the MT definition~\eqref{eq:mvt_set} together with the Paley-Wiener regularity of~$x$. They will be exploited throughout the remainder of this section and in Section~\ref{sec:density_analysis}.

\begin{proposition}[Local inter-crossing spacing]
\label{prop:local_intercrossing_spacing}
On a maximal monotone interval $I_k$ where $x$ is strictly monotone, let $t_j, t_{j+1} \in \Lambda_{\Delta V} \cap I_k$ be consecutive MT crossings. Then
\begin{equation}
    t_{j+1} - t_j = \frac{\Delta V}{|x'(\xi_j)|}
    \label{eq:local_spacing}
\end{equation}
for some $\xi_j \in (t_j, t_{j+1})$.
\end{proposition}

\textit{Proof.} By the MT definition, $|x(t_{j+1})-x(t_j)| = \Delta V$. By the mean value theorem, there exists $\xi_j \in (t_j, t_{j+1})$ with $x(t_{j+1})-x(t_j) = x'(\xi_j)(t_{j+1}-t_j)$. Since $x$ is strictly monotone on $I_k$, $x'(\xi_j) \neq 0$, giving~\eqref{eq:local_spacing}. \hfill $\square$

The physical content of~\eqref{eq:local_spacing} is simple but consequential: the local crossing spacing is \emph{deterministically locked} to the reciprocal of the local signal velocity. Fast signal segments produce closely spaced crossings; slow segments produce widely spaced crossings; and the transition between the two regimes is governed by the smoothness of $x'$. This deterministic coupling is a structural feature that is not encoded in generic nonuniform sampling models.

\begin{proposition}[Smooth variation of inter-crossing spacing]
\label{prop:smooth_gap_variation}
On the same monotone branch, define the inverse function $\tau(v) = x^{-1}(v)$ (well-defined since $x$ is strictly monotone), and let $\sigma_k=\operatorname{sgn}(x')$ on $I_k$. To keep the time index increasing on both rising and falling branches, enumerate the crossed threshold levels in traversal order as $v_m=v_{m_0}+\sigma_k m\Delta V$ and set $t_m=\tau(v_m)$. The inter-crossing spacing admits the second-order expansion
\begin{equation}
    t_{m+1} - t_m = \frac{\Delta V}{|x'(t_m)|} - \frac{(\Delta V)^2}{2}\,\frac{x''(\tau(\eta_m))}{(x'(\tau(\eta_m)))^3} 
    \label{eq:spacing_expansion}
\end{equation}
for some $\eta_m$ between $v_m$ and $v_{m+1}$, and the variation between consecutive spacings satisfies
\begin{equation}
    \bigl|(t_{m+2}-t_{m+1})-(t_{m+1}-t_m)\bigr| \le \frac{(\Delta V)^2\,\|x''\|_\infty}{v_{\mathrm{loc}}^3}
    \label{eq:spacing_variation}
\end{equation}
where $v_{\mathrm{loc}} = \inf_{I_k}|x'|$ is the minimum velocity on the monotone branch.
\end{proposition}

\textit{Proof.} Since $\tau'(v) = 1/x'(\tau(v))$ and $\tau''(v) = -x''(\tau(v))/(x'(\tau(v)))^3$, a Taylor expansion of $\tau$ around $v_m$ with step $\sigma_k\Delta V$ gives
\begin{multline*}
    t_{m+1} - t_m = \tau(v_m+\sigma_k\Delta V) - \tau(v_m) \\
    = \sigma_k\Delta V\,\tau'(v_m) + \frac{(\Delta V)^2}{2}\tau''(\eta_m)
\end{multline*}
for some $\eta_m$ between $v_m$ and $v_{m+1}$. Because $\sigma_k\tau'(v_m)=1/|x'(t_m)|$, this yields~\eqref{eq:spacing_expansion}. For the variation bound, apply the mean value theorem to the map $m \mapsto t_{m+1}-t_m = \sigma_k\Delta V\,\tau'(v_m) + \mathcal{O}(\Delta V^2)$. By the mean value theorem for finite differences,
\begin{align*}
    |(t_{m+2}-t_{m+1})-(t_{m+1}-t_m)| &= |\Delta V^2\,\tau''(\eta')| \\
    &\le \frac{(\Delta V)^2 \|x''\|_\infty}{v_{\mathrm{loc}}^3}
\end{align*}
since $|\tau''(v)| = |x''(\tau(v))|/|x'(\tau(v))|^3 \le \|x''\|_\infty/v_{\mathrm{loc}}^3$. \hfill $\square$

Proposition~\ref{prop:smooth_gap_variation} quantifies what an event-driven engineer observes empirically: on a steep signal segment, consecutive MT timestamps arrive at nearly uniform intervals, and the inter-arrival time drifts slowly as the slope changes. The drift rate is controlled by $\|x''\|_\infty/v_{\mathrm{loc}}^3$, which is small when the signal curvature is modest relative to the cube of the local velocity. This smooth modulation structure is absent from arbitrary nonuniform sampling sets, where consecutive inter-sample gaps can change by an arbitrary amount between any two indices.

\begin{proposition}[Asymptotic phase-aware retained-event pattern near extrema]
\label{prop:phase_aware_extremum_pattern}
Near a non-degenerate extremum $t^*$ with $x'(t^*)=0$, $x''(t^*)=\kappa \neq 0$, let $\sigma_v=\operatorname{sgn}(\kappa)$ be the direction in value space in which the waveform moves away from the extremum. Define the directional raw threshold-phase gap
\begin{equation}
\begin{aligned}
    \delta_0(t^*;\Delta V)
    &=
    \min_{m\in\mathbb{Z}}
    \bigl\{\sigma_v(m\Delta V-x(t^*)):\\
    &\hspace{2.7cm}
    \sigma_v(m\Delta V-x(t^*))>0\bigr\}.
\end{aligned}
    \label{eq:extremum_phase_gap}
\end{equation}
Thus $0<\delta_0\le\Delta V$; if $x(t^*)$ lies exactly on a threshold, the tangential contact at $t^*$ is not recorded under the adjacent-transition convention and $\delta_0=\Delta V$.

The first raw crossing after the turn may be a same-threshold return to the last retained level before the extremum, in which case the retained adjacent-transition map suppresses it. For each side $\sigma\in\{-1,+1\}$ of the extremum, define
\begin{equation}
    \delta_{\mathrm{ret}}^{(\sigma)}(t^*;\Delta V)
    \triangleq
    \delta_0(t^*;\Delta V)
    + \chi_{\mathrm{ret}}^{(\sigma)}(t^*)\,\Delta V,
    \qquad
    \chi_{\mathrm{ret}}^{(\sigma)}\in\{0,1\},
    \label{eq:retained_extremum_gap}
\end{equation}
where $\chi_{\mathrm{ret}}^{(\sigma)}=1$ exactly when the first raw crossing on that side is suppressed as a same-threshold return. Hence
\begin{equation}
    \delta_0
    \le
    \delta_{\mathrm{ret}}^{(\sigma)}
    \le
    \delta_0+\Delta V
    \le 2\Delta V.
    \label{eq:retained_gap_bounds}
\end{equation}
As $\delta_{\mathrm{ret}}^{(\sigma)}+j\Delta V\to0$ inside the isolated extremal neighborhood, the retained MT crossings on side $\sigma$ of $t^*$ satisfy, for $j=0,1,2,\ldots$,
\begin{equation}
    \begin{aligned}
    t_{j,\mathrm{ret}}^{(\sigma)} - t^*
    &=
    \sigma\!\sqrt{\frac{2(\delta_{\mathrm{ret}}^{(\sigma)}+j\Delta V)}{|\kappa|}}\\
    &\quad\times
    \left(1 + \mathcal{O}\!\left(
    \frac{\delta_{\mathrm{ret}}^{(\sigma)}+j\Delta V}{|\kappa|/\Omega^2}
    \right)\right),
    \end{aligned}
    \label{eq:extremum_pattern}
\end{equation}
where $\sigma=\pm1$ selects the side. With
$r_j^{(\sigma)}=\sqrt{\delta_{\mathrm{ret}}^{(\sigma)}+j\Delta V}$, the inter-crossing spacing near the extremum obeys
\begin{equation}
\begin{aligned}
    t_{j+1,\mathrm{ret}}^{(\sigma)} - t_{j,\mathrm{ret}}^{(\sigma)}
    &=
    \sqrt{\frac{2}{|\kappa|}}
    \left(r_{j+1}^{(\sigma)}-r_j^{(\sigma)}\right)\\
    &\quad\times
    \left(1+\mathcal{O}\!\left(\frac{1}{j+1}\right)\right),
\end{aligned}
    \label{eq:extremum_spacing_exact_phase}
\end{equation}
and, for $j\Delta V\gg\delta_{\mathrm{ret}}^{(\sigma)}$,
\begin{equation}
    \begin{aligned}
    t_{j+1,\mathrm{ret}}^{(\sigma)} - t_{j,\mathrm{ret}}^{(\sigma)}
    &=
    \frac{\Delta V}
    {\sqrt{2|\kappa|\,(\delta_{\mathrm{ret}}^{(\sigma)}+j\Delta V)}}\\
    &\quad\times
    \left(1+\mathcal{O}(1/j)\right).
    \end{aligned}
    \label{eq:extremum_spacing}
\end{equation}
\end{proposition}

\textit{Proof.} From the Taylor approximation
\[
    x(t) = x(t^*) + \frac{\kappa}{2}(t-t^*)^2 + \mathcal{O}((t-t^*)^3),
\]
the raw threshold levels encountered while moving away from the extremum are at signed value excursions $\delta_0+k\Delta V$. The adjacent-transition convention can suppress only the first such raw hit on a side, because after the next threshold rung is crossed the retained index differs from the last retained index before the turn. Thus the retained stream starts at the excursion $\delta_{\mathrm{ret}}^{(\sigma)}$ in~\eqref{eq:retained_extremum_gap} and then advances in steps of $\Delta V$. Hence the $j$-th retained adjacent-level crossing on a fixed side satisfies
\[
    \frac{|\kappa|}{2}(t_{j,\mathrm{ret}}^{(\sigma)}-t^*)^2
    =
    \delta_{\mathrm{ret}}^{(\sigma)}+j\Delta V
    +\mathcal{O}((t_{j,\mathrm{ret}}^{(\sigma)}-t^*)^3).
\]
Solving gives~\eqref{eq:extremum_pattern}. Subtracting the two adjacent square-root positions gives~\eqref{eq:extremum_spacing_exact_phase}; the large-$j$ expansion gives~\eqref{eq:extremum_spacing}. \hfill $\square$

Proposition~\ref{prop:phase_aware_extremum_pattern} reveals the characteristic square-root law of MT events near turning points, but it also shows that the first crossing is phase-dependent. Its leading distance is
\[
    h_{\mathrm{ret}}^{(\sigma)}(t^*)=
    \sqrt{2\delta_{\mathrm{ret}}^{(\sigma)}(t^*;\Delta V)/|\kappa|},
\]
not necessarily $\sqrt{2\Delta V/|\kappa|}$. For retained events the phase-uniform worst case is $\sqrt{4\Delta V/|\kappa|}$, reflecting the possibility that the nearest raw hit is suppressed as a same-threshold return. As $j$ increases, the spacing contracts and the local velocity $|x'(t_{j,\mathrm{ret}})| \approx \sqrt{2|\kappa|(\delta_{\mathrm{ret}}^{(\sigma)}+j\Delta V)}$ increases, recovering the monotone-segment regime of Proposition~\ref{prop:local_intercrossing_spacing} where spacing $\approx \Delta V/|x'|$. This retained-aware transition from the extremal regime to the monotone regime will be used below to derive both realized retained-phase and retained phase-uniform curvature certificates.

\begin{proposition}[Structural regularity under fixed retained topology]
\label{prop:structural_regularity_fixed_topology}
Fix a finite active interval and a retained MT record whose selected crossings are transverse: $|x'(t_n)|\ge v_0>0$ on the selected monotone branches. Assume further that the retained label sequence is fixed under the perturbations being considered: no tangential contact is created, no threshold event is created or deleted, and no same-threshold return changes its suppressed/retained status. Under these fixed-topology hypotheses:
\begin{enumerate}
    \item[(P1)] \textbf{Local event-time continuity.} Each retained crossing time is a locally continuous, indeed differentiable, functional of the waveform perturbation by the implicit function theorem applied to $x(t)=m_n\Delta V$.
    \item[(P2)] \textbf{Controlled gap variation.} On a monotone branch with $v_{\mathrm{loc}}>0$, consecutive retained gaps obey the inverse-function expansion of Proposition~\ref{prop:smooth_gap_variation}, with variation bounded as in~\eqref{eq:spacing_variation}.
    \item[(P3)] \textbf{Restricted Kadec perturbations.} For a fixed grid offset and retained topology, the perturbation sequence $\delta_n=s_n-(\tau_0+n\Delta T)$ is determined by the waveform through these local inverse maps; it is not an arbitrary bounded sequence chosen independently at each index.
\end{enumerate}
\end{proposition}

These are local regularity statements, not global exclusion theorems. At tangencies, near-quiescent returns, or event-topology changes, the retained record can jump: events may be created, deleted, or suppressed. The certificates below therefore always operate either after the retained record has been realized, or under explicit transverse/fixed-topology and retained-eligibility hypotheses.

\subsection{Velocity-Calibrated Kadec Perturbation Bound}
\label{sec:velocity_kadec}
In the MT geometry, the certification goal is to ensure that each Nyquist neighborhood contains at least one threshold crossing that can serve as a surrogate sample, since the Nyquist grid point itself is not directly sampled. Kadec's theorem permits this surrogate construction: a sample within the admissible safety radius $\Delta T/4$ suffices. Having established the structural regularity of MT crossings, we now show that the \emph{actual} perturbation from the Nyquist grid is controlled by $\Delta V$ and the local velocity, and is typically much smaller than the worst-case Kadec bound $\Delta T/4$.

Fix a grid offset $\tau_0 \in \mathbb{R}$ and define the centered half-Nyquist windows
\begin{equation}
    J_n(\tau_0) = \left[\tau_0 + n\Delta T - \frac{\Delta T}{4},\; \tau_0 + n\Delta T + \frac{\Delta T}{4}\right],
    \label{eq:centered_window}
\end{equation}
where $\Delta T = 1/(2f_{max})$ is the Nyquist interval. The task is to guarantee one crossing in every centered window and to bound how far that crossing lies from the window center.

\begin{lemma}[Retained MT perturbation bound on monotone segments]
\label{lem:retained_mvt_perturbation_monotone}
Let $J_n(\tau_0) \subseteq [0, T_{\text{obs}}]$ be a centered window that lies entirely inside a single maximal monotone interval $I_k$ of $x$, and suppose
\begin{equation}
    v_n \triangleq \inf_{t \in I_k \cap J_n} |x'(t)| > 0.
    \label{eq:local_velocity_def}
\end{equation}
Let $g_n=\tau_0+n\Delta T$. Define the retained-eligible threshold without assuming that a crossing already lies in $J_n(\tau_0)$. For each side $\sigma\in\{+,-\}$, where $+$ means searching to the right of $g_n$ and $-$ means searching to the left, initialize the retained state by the adjacent-transition record at the boundary of that one-sided search. Equivalently, each side uses its own one-sided finite-window map $\mathcal A_{\Delta V}^{m_{n,\sigma}^{\mathrm{bdry}}}$, with the search orientation inherited from the retained recursion on that side. Among all ladder levels $q=m\Delta V$ reachable from $x(g_n)$ along side $\sigma$ of the monotone branch and admissible for that side's boundary state, let $q_{n,\sigma}^{(\mathrm{ret})}$ be one with minimal voltage distance; if no such level exists before the side leaves the branch, set $\eta_{n,\sigma}^{(\mathrm{ret})}=+\infty$. Thus
\[
    \eta_{n,\sigma}^{(\mathrm{ret})}
    =
    |q_{n,\sigma}^{(\mathrm{ret})}-x(g_n)|
    \qquad(\sigma\in\{+,-\})
\]
whenever the admissible set is nonempty. Define the two-sided retained-eligible voltage distance by
\[
    \eta_n^{(\mathrm{ret})}
    =
    \min_{\sigma\in\{+,-\}}\eta_{n,\sigma}^{(\mathrm{ret})}.
\]
If $x(g_n)$ itself is an admissible retained threshold, then $\eta_n^{(\mathrm{ret})}=0$. This definition uses only the value geometry, the monotone branch, and the retained state supplied by $\mathcal A_{\Delta V}$; it does not presuppose that the selected threshold is crossed inside the centered window. Write
\begin{equation}
    \alpha_n^{(\mathrm{ret})}
    \triangleq
    \frac{\eta_n^{(\mathrm{ret})}}{\Delta V}.
    \label{eq:retained_alpha_def}
\end{equation}
The value $\alpha_n^{(\mathrm{ret})}$ is a \emph{realized} retained-topology quantity: it is computed only after a spacing $\Delta V$, grid offset $\tau_0$, and adjacent-transition convention have been fixed. In the exact-anchor case, a retained crossing may occur at $g_n$ itself, so $\alpha_n^{(\mathrm{ret})}=0$. To avoid turning this harmless boundary case into a singular design formula, all scalar threshold-spacing bounds below use the clipped envelope
\begin{equation}
    \bar{\alpha}_n^{(\mathrm{ret})}
    \triangleq
    \max\!\left\{\frac12,\alpha_n^{(\mathrm{ret})}\right\}.
    \label{eq:retained_alpha_clipped}
\end{equation}
This clipping is conservative: it leaves the actual perturbation estimate unchanged, but it prevents division by zero and reduces to the clean-topology factor $1/2$ whenever the nearest retained threshold is already available.
If
\begin{equation}
    \min_{\sigma\in\{+,-\}}\eta_{n,\sigma}^{(\mathrm{ret})}
    \le
    \frac{v_n\Delta T}{4},
    \label{eq:retained_velocity_condition}
\end{equation}
then there exists a retained crossing $s_n \in \Lambda_{\Delta V} \cap J_n(\tau_0)$ satisfying
\begin{equation}
    |s_n-g_n|
    \le
    \frac{\eta_n^{(\mathrm{ret})}}{v_n}
    =
    \frac{\alpha_n^{(\mathrm{ret})}\Delta V}{v_n}.
    \label{eq:retained_velocity_perturbation}
\end{equation}
In the clean fixed-topology case where the nearest raw threshold is retained, this reduces to
\begin{equation}
    |s_n - (\tau_0 + n\Delta T)| \le \frac{\Delta V}{2v_n}.
    \label{eq:velocity_perturbation}
\end{equation}
\end{lemma}

\textit{Proof.} Choose a side $\sigma_*$ attaining $\eta_n^{(\mathrm{ret})}=\min_{\sigma}\eta_{n,\sigma}^{(\mathrm{ret})}$, and let $q_{\mathrm{ret}}=q_{n,\sigma_*}^{(\mathrm{ret})}$. If $x(g_n)=q_{\mathrm{ret}}$ and that threshold is admissible for the retained state on side $\sigma_*$, take $s_n=g_n$. Otherwise, since $x$ is strictly monotone on $I_k$, the side of $g_n$ on which the value $q_{\mathrm{ret}}$ can be reached is the selected side $\sigma_*$. For every point $u$ on that side inside $J_n$, the mean-value theorem and the definition of $v_n$ give
\[
    |x(u)-x(g_n)| \ge v_n |u-g_n|.
\]
At the point $u$ on the appropriate side with $|u-g_n|=\eta_n^{(\mathrm{ret})}/v_n$, the value of $x(u)$ has reached or passed $q_{\mathrm{ret}}$. Condition~\eqref{eq:retained_velocity_condition} keeps this point inside $J_n$. The intermediate value theorem therefore gives $s_n\in J_n$ with $x(s_n)=q_{\mathrm{ret}}$ and
\[
    |s_n-g_n|\le \frac{|x(g_n)-q_{\mathrm{ret}}|}{v_n}
    =
    \frac{\eta_n^{(\mathrm{ret})}}{v_n}.
\]
By construction, this threshold hit is retained by the one-sided state map and therefore by $\mathcal{A}_{\Delta V}$ on the realized record, so $s_n\in\Lambda_{\Delta V}$. If the nearest raw threshold is retained, then $\eta_n^{(\mathrm{ret})}\le\Delta V/2$, which recovers~\eqref{eq:velocity_perturbation}. Without this fixed-topology information, the nearest ladder level may be a suppressed same-threshold return. Taking the nearer of the two adjacent voltage directions and then moving one additional threshold step if suppression occurs gives the two-sided retained-safe envelope $\eta_n^{(\mathrm{ret})}\le3\Delta V/2$. This envelope is invoked only in the centered-window monotone setting together with the strict retained-safe velocity inequality, which guarantees that the required adjacent level is reached before leaving the window. If the branch terminates before that level is reachable, the realized one-sided distance is $+\infty$ and the local certificate simply fails. \hfill $\square$

The normalized perturbation ratio on monotone segments is therefore
\begin{equation}
    L_n^{(\mathrm{mono})}
    =
    \frac{|s_n-g_n|}{\Delta T}
    \le
    \frac{\alpha_n^{(\mathrm{ret})}\Delta V}{v_n\Delta T}
    =
    \frac{2\alpha_n^{(\mathrm{ret})}\Delta V f_{max}}{v_n}.
    \label{eq:normalized_perturbation}
\end{equation}
This result sharpens the generic Kadec picture in two ways:
\begin{enumerate}
    \item The perturbation ratio scales \emph{linearly} with $\Delta V$: finer threshold ladders produce smaller perturbations and hence better-conditioned reconstruction.
    \item The perturbation is inversely proportional to the local velocity $v_n$: fast signal segments contribute near-zero perturbation, concentrating the Kadec budget on the slowest segments.
\end{enumerate}
Neither of these refinements is available in the generic Kadec framework, which treats $L$ as a fixed worst-case constant.

The condition for retained-event Kadec compatibility on monotone segments is $L_n^{(\mathrm{mono})} < 1/4$, which gives
\begin{equation}
    \Delta V
    <
    \frac{v_n}{8\bar{\alpha}_n^{(\mathrm{ret})}f_{max}}.
    \label{eq:velocity_kadec_condition}
\end{equation}
Equation~\eqref{eq:velocity_kadec_condition} is a sufficient design form of the sharper realized inequality $2\alpha_n^{(\mathrm{ret})}\Delta V f_{max}/v_n<1/4$. In the clean fixed-topology case $\alpha_n^{(\mathrm{ret})}\le1/2$, clipping gives $\bar{\alpha}_n^{(\mathrm{ret})}=1/2$ and the familiar $\Delta V<v_n/(4f_{max})$. With no retained-topology information beyond the adjacent-transition rule, $\alpha_n^{(\mathrm{ret})}\le3/2$, giving the conservative retained-safe condition $\Delta V<v_n/(12f_{max})$.

The monotone-window part of the certificate is now fully mechanical. For a fixed grid offset $\tau_0$, define
\begin{equation}
    \begin{aligned}
    \mathcal{N}_{\mathrm{mono}}(\tau_0)
    =
    \{\,n:\;&J_n(\tau_0)\subseteq[0,T_{\text{obs}}],\\
    &J_n(\tau_0)\subset I_k\ \text{for some }k\,\}.
    \end{aligned}
    \label{eq:mono_window_set}
\end{equation}
For $n\in\mathcal{N}_{\mathrm{mono}}(\tau_0)$ set $v_n$ by~\eqref{eq:local_velocity_def}, and define
\begin{equation}
    \begin{aligned}
    v_{min}^{(\mathrm{mono})}(x;\tau_0)
    &\triangleq
    \min_{n\in\mathcal{N}_{\mathrm{mono}}(\tau_0)} v_n,\\
    \alpha_{\mathrm{mono}}^{(\mathrm{ret})}(x;\tau_0)
    &\triangleq
    \max_{n\in\mathcal{N}_{\mathrm{mono}}(\tau_0)}
    \alpha_n^{(\mathrm{ret})},\\
    \bar{\alpha}_{\mathrm{mono}}^{(\mathrm{ret})}(x;\tau_0)
    &\triangleq
    \max_{n\in\mathcal{N}_{\mathrm{mono}}(\tau_0)}
    \bar{\alpha}_n^{(\mathrm{ret})}.
    \end{aligned}
    \label{eq:global_velocity_condition}
\end{equation}
If $\mathcal{N}_{\mathrm{mono}}(\tau_0)$ is empty, the velocity branch is omitted. Otherwise all monotone windows are retained-Kadec compatible when
\begin{equation}
    \Delta V
    <
    \frac{v_{min}^{(\mathrm{mono})}(x;\tau_0)}
    {8\bar{\alpha}_{\mathrm{mono}}^{(\mathrm{ret})}(x;\tau_0)f_{max}}.
    \label{eq:global_retained_velocity_condition}
\end{equation}

\textbf{Remark.} The clean-topology version of~\eqref{eq:global_retained_velocity_condition} is the MT-specific counterpart of the range condition $R_x^{(c)}(\tau)>\Delta V$ used in the original centered-window framework. The retained-safe version is stronger because it provisions for one suppressed same-threshold return before the next retained adjacent-level transition. In either case, the certificate identifies a retained sample and bounds its displacement from the grid center, beyond registering that a raw threshold has been hit.
\subsection{Unified Velocity-Curvature Sufficient Condition}
\label{sec:unified_condition}
The retained velocity condition~\eqref{eq:global_retained_velocity_condition} governs centered windows that lie inside a monotone region. Near a non-degenerate extremum $t^*$ where $x'(t^*)=0$, the velocity $|x'|$ passes through zero, and the velocity condition necessarily fails in any window containing $t^*$. In this regime the local geometry is governed by second-order curvature and by the retained threshold phase in~\eqref{eq:retained_extremum_gap}. The relevant question is whether the retained directional phase gap is small enough for a retained crossing to occur before the trajectory leaves the centered window.

Let $\rho_-(t^*)$ and $\rho_+(t^*)$ denote the maximal left and right monotone radii around $t^*$ within the active observation interval: $x'$ has one fixed nonzero sign on $(t^*-\rho_-(t^*),t^*)$ and on $(t^*,t^*+\rho_+(t^*))$, respectively. These radii are clipped by the observation-window boundary and, in particular, by the distances from $t^*$ to the nearest neighboring critical points. Define the extremum isolation radius
\begin{equation}
    \rho(t^*) \triangleq
    \min\!\left\{\rho_-(t^*),\,\rho_+(t^*),\,\frac{\Delta T}{4}\right\}.
    \label{eq:extremum_isolation_radius}
\end{equation}
This extra radius is not cosmetic: if two extrema are closer than the first-crossing distance given by the quadratic law, the local monotone branch can terminate before an adjacent threshold is crossed. The curvature branch is therefore certified only on the radius $\rho(t^*)$, not automatically on the full Kadec half-window scale.

The asymptotic square-root law in Proposition~\ref{prop:phase_aware_extremum_pattern} explains the local scale of the first retained crossing. The certificate used below is the following non-asymptotic radius lemma.

\begin{lemma}[Retained crossing before the isolation radius]
\label{lem:retained_crossing_before_radius}
Let $t^*$ be a non-degenerate extremum, let $\sigma\in\{-1,+1\}$ be a monotone side whose branch is defined for distance at least $\rho(t^*)$, and set
\[
    M_3(t^*)=\sup_{|u-t^*|\le \rho(t^*)}|x'''(u)|
\]
and
\begin{equation}
    B_\rho(t^*)=
    \frac{1}{2}|x''(t^*)|\rho(t^*)^2
    -\frac{1}{6}M_3(t^*)\rho(t^*)^3 .
    \label{eq:curvature_budget}
\end{equation}
If $B_\rho(t^*)>0$ and
\[
    \delta_{\mathrm{ret}}^{(\sigma)}(t^*;\Delta V)<B_\rho(t^*),
\]
then the corresponding monotone side contains a retained adjacent-threshold crossing at distance strictly less than $\rho(t^*)$ from $t^*$.
\end{lemma}

\textit{Proof.} Let $\sigma_v=\operatorname{sgn}(x''(t^*))$ and write $t=t^*+\sigma s$, $0\le s\le\rho(t^*)$. Taylor's theorem gives
\[
    \sigma_v\bigl(x(t)-x(t^*)\bigr)
    \ge
    \frac{1}{2}|x''(t^*)|s^2-\frac{1}{6}M_3(t^*)s^3 .
\]
At $s=\rho(t^*)$ the right-hand side equals $B_\rho(t^*)$, which is larger than the retained directional phase gap. Since the signed excursion is continuous and starts at zero, the intermediate value theorem gives some $s<\rho(t^*)$ at which the waveform reaches the retained-eligible threshold on side $\sigma$. By the definition of $\delta_{\mathrm{ret}}^{(\sigma)}$, this threshold hit is retained by the adjacent-transition map. \hfill $\square$

In any centered window $J_n$ containing $t^*$, at least one side of the window has boundary distance at least $\Delta T/4$; since both monotone sides are guaranteed for distance $\rho(t^*)$, Lemma~\ref{lem:retained_crossing_before_radius} certifies the selected side whenever the retained phase gap lies below $B_\rho(t^*)$. The leading-order comparison behind this condition is
\begin{equation}
    \delta_{\mathrm{ret}}^{(\sigma)}(t^*;\Delta V)
    <
    \frac{1}{2}|\kappa|\,\rho(t^*)^2,
    \label{eq:curvature_leading}
\end{equation}
under the realized threshold phase, at leading order.
Including the third-order Taylor remainder explicitly, the realized retained-phase local sufficient condition is
\begin{equation}
    \delta_{\mathrm{ret}}^{(\sigma)}(t^*;\Delta V) <
    \frac{1}{2}|x''(t^*)|\rho(t^*)^2
    - \frac{1}{6}M_3(t^*)\rho(t^*)^3.
    \label{eq:curvature_exact}
\end{equation}
Since $\delta_{\mathrm{ret}}^{(\sigma)}(t^*;\Delta V)\le2\Delta V$, the retained phase-uniform sufficient condition
\begin{equation}
    2\Delta V <
    \frac{1}{2}|x''(t^*)|\rho(t^*)^2
    - \frac{1}{6}M_3(t^*)\rho(t^*)^3
    \label{eq:curvature_phase_uniform}
\end{equation}
implies~\eqref{eq:curvature_exact} for every threshold phase. When $\rho_-(t^*)$ and $\rho_+(t^*)$ are both at least $\Delta T/4$, and $M_3(t^*)$ is replaced by the global envelope $\|x'''\|_\infty$, the phase-uniform condition reduces to the simpler formula
$\Delta V<|x''(t^*)|/(256f_{max}^2)\bigl(1-\|x'''\|_\infty/(24f_{max}|x''(t^*)|)\bigr)$. If the right-hand side of~\eqref{eq:curvature_exact} is non-positive, the curvature branch supplies no certificate at that extremum.
The non-degeneracy assumption is physically reasonable for generic smooth waveforms, because flat higher-order stationary points are structurally unstable.

\textbf{Perturbation bound near extrema.} When either the realized retained-phase condition~\eqref{eq:curvature_exact} or the retained phase-uniform condition~\eqref{eq:curvature_phase_uniform} is satisfied, the crossing used for the Kadec subsequence must also account for the location of the extremum relative to the Nyquist anchor. Let $g_n=\tau_0+n\Delta T$, let $d_n=|t^*-g_n|$, and let $\varsigma_n=\operatorname{sgn}(g_n-t^*)$ denote the side pointing from the extremum toward the anchor. Let $h_n$ be the actual distance from $t^*$ to the first retained threshold crossing on that side. The realized certificate uses this actual $h_n$; the quadratic law gives only the asymptotic scale
\[
    h_n
    \sim
    \sqrt{\frac{2\delta_{\mathrm{ret}}^{(\varsigma_n)}(t^*;\Delta V)}{|\kappa|}}.
\]
The selected crossing then satisfies the exact geometric bound
\begin{equation}
    \begin{aligned}
        L_n^{(\mathrm{ext})}
        &= \frac{|s_n-g_n|}{\Delta T}
        \le \frac{|d_n-h_n|}{\Delta T}, \\
        h_n&\sim
        \sqrt{\frac{2\delta_{\mathrm{ret}}^{(\varsigma_n)}(t^*;\Delta V)}{|\kappa|}}.
    \end{aligned}
    \label{eq:extremum_perturbation}
\end{equation}
The curvature-isolation condition guarantees $h_n<\rho(t^*)\le\Delta T/4$, and the centered-window geometry gives $d_n\le\Delta T/4$, so the selected retained crossing lies inside the Kadec window on whichever monotone side has room inside $J_n$. Equation~\eqref{eq:extremum_perturbation} also shows two dependencies: retained threshold phase controls $h_n$, while grid offset controls $d_n$. The retained phase-uniform scaling is still $\mathcal{O}(\sqrt{\Delta V})$, but the constant is larger than in the raw-hit model because one same-threshold return may be suppressed. Without offset alignment, the extremal contribution to the perturbation budget is controlled by the residual offset $d_n/\Delta T$ rather than by threshold refinement alone.

\textbf{Design bound versus realized certificate.} The quantities $\alpha_n^{(\mathrm{ret})}$ and $\delta_{\mathrm{ret}}^{(\sigma)}(t^*;\Delta V)$ are not design constants: they change when the threshold spacing or ladder phase changes. We therefore use two distinct objects. The first is a phase-uniform, retained-safe \emph{design} bound that can be evaluated from the waveform geometry and the chosen grid offset without inspecting the retained topology at the deployed $\Delta V$. It budgets the worst retained monotone excursion factor $3/2$ and the phase-uniform extremal bound $\delta_{\mathrm{ret}}\le2\Delta V$; realized values of $\alpha_n^{(\mathrm{ret})}$ and $\delta_{\mathrm{ret}}$ are used only in post-realization verification. Let
\[
    \mathcal{E}_x=\{t^*\in[0,T_{\text{obs}}]\mid x'(t^*)=0\}.
\]
The retained phase-uniform curvature margin is
\begin{align*}
    \Delta V_{\mathrm{curv,ret}}^*
    \triangleq
    \frac{1}{2}\inf_{t^* \in \mathcal{E}_x}
    &\left[
    \frac{1}{2}|x''(t^*)|\rho(t^*)^2
    -\frac{1}{6}M_3(t^*)\rho(t^*)^3
    \right].
\end{align*}
If there are no extrema in the active interval, the infimum over extrema is omitted; if the displayed curvature margin is non-positive at some extremum, the retained phase-uniform curvature branch supplies no positive certificate there. Combining this curvature branch with the retained-safe monotone factor $3/2$ gives the non-circular scalar design inequality
\begin{equation}
    \boxed{\begin{aligned}
    \Delta V
    &< \Delta V_{\mathrm{Kadec,ret}}^*(x;\tau_0)\\
    &\triangleq
    \min\!\left(
    \frac{v_{min}^{(\mathrm{mono})}(x;\tau_0)}
    {12f_{max}},
    \Delta V_{\mathrm{curv,ret}}^*\right).
    \end{aligned}}
    \label{eq:unified_condition}
\end{equation}
with the velocity term omitted if $\mathcal{N}_{\mathrm{mono}}(\tau_0)$ is empty. The clean-topology design margin
\begin{equation}
    \Delta V_{\mathrm{clean}}^*(x;\tau_0)
    \triangleq
    \min\!\left(
    \frac{v_{min}^{(\mathrm{mono})}(x;\tau_0)}
    {4f_{max}},
    \Delta V_{\mathrm{curv,clean}}^*\right)
    \label{eq:clean_design_margin}
\end{equation}
where $\Delta V_{\mathrm{curv,clean}}^*$ is the same curvature infimum as $\Delta V_{\mathrm{curv,ret}}^*$ but without the outer prefactor $1/2$. This clean margin~\eqref{eq:clean_design_margin} omits the retained-suppression factors and is valid only after verifying that the selected nearest crossings are retained. The simpler retained-safe expression is $\kappa_{min}/(256f_{max}^2)$, up to the same third-order correction, when every relevant extremum is isolated by at least $\Delta T/4$ on both monotone sides.

The second object is a fixed-$\Delta V$ post-realization check, using the local curvature budget $B_\rho(t^*)$ from~\eqref{eq:curvature_budget}.
For a deployed spacing $\Delta V$ and grid offset $\tau_0$, define $\mathsf{Cert}_{\mathrm{Kadec}}(x,\Delta V,\tau_0)$ to hold when the extracted retained subsequence satisfies
\begin{equation}
    \begin{cases}
    \displaystyle
    \frac{2\alpha_n^{(\mathrm{ret})}\Delta V f_{max}}{v_n}<\frac14,
    & n\in\mathcal{N}_{\mathrm{mono}}(\tau_0),\\[2ex]
    \displaystyle
    \delta_{\mathrm{ret}}^{(\varsigma_n)}(t^*;\Delta V)
    < B_\rho(t^*),
    & J_n(\tau_0)\ \text{contains }t^*,\\[2ex]
    \displaystyle
    \frac{|d_n-h_n|}{\Delta T}<\frac14,
    & J_n(\tau_0)\ \text{contains }t^*,
    \end{cases}
    \label{eq:realized_kadec_certificate}
\end{equation}
and the selected threshold hits are retained by $\mathcal{A}_{\Delta V}$. This predicate is often much less conservative than~\eqref{eq:unified_condition}, but it is no longer a spacing formula independent of the deployed ladder. Thus the superscript-$*$ quantities in Table~\ref{tab:threshold_symbol_taxonomy} are design-time scalar margins, while $\mathsf{Cert}_{\mathrm{Kadec}}(x,\Delta V,\tau_0)$ is the fixed-$\Delta V$ post-realization verification object. If $\tau_0$ is fixed before acquisition,~\eqref{eq:unified_condition} is an acquisition-side design guarantee for that offset. If $\tau_0$ is optimized after seeing the event record, $\mathsf{Cert}_{\mathrm{Kadec}}$ is a post-realization diagnostic certificate, not a hardware acquisition guarantee.

The design bound~\eqref{eq:unified_condition} and the realized predicate~\eqref{eq:realized_kadec_certificate} encode the same two physical mechanisms. Moving segments are governed by velocity, including possible same-threshold suppression. Turning points are governed by curvature, because motion stalls and event generation must be sustained by the local second-order shape.

As long as the retained topology is fixed, the transition between the two regimes is smooth and quantitative. Near an extremum $t^*$, the local velocity grows as $|x'(t)| \approx |\kappa|\,|t-t^*|$. The retained-safe velocity condition $3\Delta V f_{max}/|x'|<1/4$ takes over at distance $|t-t^*| \gtrsim 12f_{max}\Delta V/|\kappa|$ from the extremum. Closer to the extremum, the first retained crossings are governed by the retained phase-aware square-root law~\eqref{eq:extremum_pattern}. Farther away, the monotone spacing approximation $\Delta V/|x'|$ becomes the correct local scale. The retained phase-uniform curvature branch is conservative because it provisions for the worst possible value of $\delta_{\mathrm{ret}}$.

\textbf{Sliding-window no-gap certificate.} The Kadec extraction above is tied to a fixed grid offset $\tau_0$. A static-threshold no-gap prior asks a related question: can an arbitrary silent interval of length $\Delta T/2$ persist under the deployed spacing? Let
\[
    J(g)=\left[g-\frac{\Delta T}{4},\,g+\frac{\Delta T}{4}\right].
\]
For this purpose define the monotone sliding-window velocity
\begin{equation}
    v_{min}^{(\mathrm{slide})}
    \triangleq
    \inf_{\substack{J(g)\subseteq[0,T_{\text{obs}}]\\ J(g)\ \mathrm{monotone}}}
    \inf_{t\in J(g)} |x'(t)|,
    \label{eq:sliding_velocity_margin}
\end{equation}
and the retained sliding no-gap margin
\begin{equation}
    \Delta V_{\mathrm{gap}}^*(x)
    \triangleq
    \min\!\left(\frac{v_{min}^{(\mathrm{slide})}}{8f_{max}},\,
    \Delta V_{\mathrm{curv,ret}}^*\right).
    \label{eq:sliding_gap_margin}
\end{equation}
If the set of monotone sliding windows in~\eqref{eq:sliding_velocity_margin} is empty, we set $v_{min}^{(\mathrm{slide})}=+\infty$ and omit the velocity branch. The no-gap certificate is then supplied only by the curvature branch. If both branches are absent or non-positive, $\Delta V_{\mathrm{gap}}^*(x)$ is not a positive certificate. If $\Delta V<\Delta V_{\mathrm{gap}}^*(x)$, every interval of length $\Delta T/2$ contained in the active window contains at least one retained adjacent-transition MT crossing. Such an interval is either monotone or contains a non-degenerate extremum. In the monotone case, its range is at least $v_{min}^{(\mathrm{slide})}\Delta T/2>2\Delta V$, so it crosses at least two threshold rungs and at most one can be suppressed as a same-threshold return. In the extremal case, the retained phase-uniform curvature condition~\eqref{eq:curvature_phase_uniform} forces a retained adjacent-level crossing within the isolated monotone side of that extremum. Thus a silent interval of duration $\Delta T/2$ contradicts the retained-event certificate. When both the sliding no-gap property and the fixed-grid Kadec extraction are needed, we use the combined margin
\begin{equation}
    \Delta V_{\mathrm{fg}}^*(x;\tau_0)
    \triangleq
    \min\!\left(\Delta V_{\mathrm{gap}}^*(x),\,
    \Delta V_{\mathrm{Kadec,ret}}^*(x;\tau_0)\right).
    \label{eq:combined_fixed_grid_margin}
\end{equation}

\begin{table}[!t]
\caption{Threshold-spacing symbols and their status.}
\label{tab:threshold_symbol_taxonomy}
\centering
\footnotesize
\setlength{\tabcolsep}{3pt}
\begin{tabular}{p{0.31\columnwidth}p{0.25\columnwidth}p{0.31\columnwidth}}
\hline
Symbol & Depends on deployed $\Delta V$? & Role \\
\hline
$\Delta V_{\mathrm{Kadec,ret}}^*(x;\tau_0)$ & No & Phase-uniform retained-safe fixed-grid design bound; used only with strict $\Delta V<\Delta V_{\mathrm{Kadec,ret}}^*$. \\
$\mathsf{Cert}_{\mathrm{Kadec}}(x,\Delta V,\tau_0)$ & Yes & Post-realization retained-subsequence verification. \\
$\Delta V_{\mathrm{clean}}^*(x;\tau_0)$ & Only through topology check & Conditional clean-topology normalization/design bound. \\
$\Delta V_{\mathrm{gap}}^*(x)$ & No & Sliding-window no-gap watchdog margin; used with strict spacing below the margin. \\
$\Delta V_{\mathrm{fg}}^*(x;\tau_0)$ & No & Combined static-threshold no-gap and Kadec design margin. \\
\hline
\end{tabular}
\end{table}

\textbf{Remark (incompleteness and the no-go boundary).} The unified condition~\eqref{eq:unified_condition} can fail when neither term provides a positive lower bound. If the signal contains a nearly monotone segment with $|x'| \ll 12f_{max}\Delta V$ and no extremum, the retained-safe velocity branch has no margin and no curvature rescue is available. Such ``near-quiescent monotone'' intervals arise whenever the signal's activity is concentrated well below $f_{max}$. This gap reflects an intrinsic limitation of the static-threshold model: the no-go theorem of Section~\ref{sec:no_go} shows that, for any fixed $\Delta V > 0$, the signal class $PW_\Omega$ contains waveforms with arbitrarily long near-quiescent stretches that are incompatible with fixed local threshold certificates.

\subsection{Signal-Instance Frame Bounds}
\label{sec:signal_instance_frame}
In the generic Kadec framework, the frame bounds~\eqref{eq:padded_frame} depend on the worst-case normalized perturbation $L = \sup_n|\delta_n|/\Delta T$, which is bounded by a constant $< 1/4$. For MT crossing sets, the regularity results of Subsections~\ref{sec:mvt_regularity}--\ref{sec:velocity_kadec} show that the actual perturbation is controlled by $\Delta V$ and the signal dynamics. The resulting frame bounds are therefore \emph{signal-instance} bounds: they improve continuously as the threshold spacing is refined.

Let $\tau_0$ be a grid offset for which~\eqref{eq:unified_condition} holds on every active centered window, and define the overall worst-case retained MT perturbation ratio
\begin{equation}
    L_{MT} = \max\!\left(\sup_{\substack{n \in \mathcal{I} \\ J_n \text{ mono.}}} \frac{\alpha_n^{(\mathrm{ret})}\Delta V}{v_n \Delta T},\;
    \sup_{\substack{n \in \mathcal{I} \\ J_n \ni t^*}} \frac{|d_n-h_n|}{\Delta T}\right)
    \label{eq:L_mvt}
\end{equation}
where $v_n$ is the local velocity on monotone windows, $\alpha_n^{(\mathrm{ret})}$ is the retained-eligibility factor from~\eqref{eq:retained_alpha_def}, and $d_n$ and $h_n$ are the extremum-grid offset and retained first-crossing distance from~\eqref{eq:extremum_perturbation}. By construction, $L_{MT} < 1/4$ whenever~\eqref{eq:unified_condition} holds with a nonzero centered-window margin. With $L_{MT}$ defined by~\eqref{eq:L_mvt}, the signal-instance frame bounds are
\begin{align}
    A_{MT} &= \frac{(1-\theta(L_{MT}))^2}{\Delta T}, \notag\\
    B_{MT} &= \frac{(1+\theta(L_{MT}))^2}{\Delta T},
    \label{eq:signal_instance_frame_bounds}
\end{align}
with $\theta(L) = 1-\cos(\pi L)+\sin(\pi L)$. The auxiliary MT condition number associated with these padded frame bounds is $\kappa_{MT}=B_{MT}/A_{MT}$.

The key point is the \emph{scaling behavior}. On a signal with minimum velocity $v_{min}$, retained monotone factor $\alpha_{\mathrm{ret}}$, and a positive retained curvature margin $\Delta V_{\mathrm{curv,ret}}^*$:
\begin{itemize}
    \item If the velocity term dominates: $L_{MT} \le 2\alpha_{\mathrm{ret}}\Delta V f_{max}/v_{min}$, and for $\Delta V \ll v_{min}/(8\alpha_{\mathrm{ret}}f_{max})$ we have $L_{MT} \ll 1/4$. Then $\theta(L_{MT}) \approx \pi L_{MT} = 2\pi\alpha_{\mathrm{ret}}\Delta V f_{max}/v_{min}$, and the condition number
    \begin{equation}
        \begin{aligned}
        \kappa_{MT}
        &=
        \frac{B_{MT}}{A_{MT}}
        =
        \left(\frac{1+\theta}{1-\theta}\right)^2\\
        &\approx
        1 + 4\theta
        \approx
        1 + \frac{8\pi\alpha_{\mathrm{ret}}\Delta V f_{max}}{v_{min}} .
        \end{aligned}
        \label{eq:condition_number_scaling}
    \end{equation}
    approaches $1$ linearly as $\Delta V \to 0$.
    \item If the curvature term dominates and the grid offset is chosen so that the extremum offsets satisfy $d_n=\mathcal{O}(\sqrt{\Delta V})$, then $L_{MT}=\mathcal{O}(\sqrt{\Delta V})$ and the condition number approaches $1$ as $\sqrt{\Delta V} \to 0$. Without this alignment, the limiting condition number is set by the residual offsets $d_n/\Delta T$ even though Kadec compatibility remains certified.
\end{itemize}

\textbf{Physical interpretation.} Inequality~\eqref{eq:condition_number_scaling} says that the reconstruction conditioning improves continuously as the threshold ladder is refined. A designer who doubles the number of comparator levels (halving $\Delta V$) roughly halves the excess condition number $B/A - 1$. This quantitative design guidance is not available from a generic Kadec analysis, which only certifies ``$L<1/4$ or not'' without revealing how far inside the bound the MT crossing set actually sits.

\subsection{From Finite Kadec Geometry to Finite-Record Reconstruction}
\label{sec:finite_bridge}

The classical Kadec theorem~\cite{Kadec1964,Avdonin1995} (Section~\ref{sec:kadec}, Appendix~\ref{app:constraints}) yields a Riesz basis~\cite{Christensen2003} for $PW_\Omega$ from a \emph{bi-infinite} perturbed sequence $\{t_n\}_{n \in \mathbb{Z}}$. The unified sufficient condition above produces only a \emph{finite} Kadec-compatible subsequence $\{s_n\}_{n \in \mathcal{I}}$ on the active interval, where $\mathcal{I} = \{n : J_n(\tau_0) \subseteq [0, T_{\text{obs}}]\}$ and $|\mathcal{I}| \approx 2f_{max}T_{\text{obs}}$. A direct invocation of Kadec's theorem on the finite MT record is therefore impossible: the MT hardware has not measured the values of $x$ outside the observation window. The padded frame inequality is never an acquisition theorem for the finite MT device unless the outside grid values, or equivalent tail/collar bounds, are supplied externally. We separate two distinct statements below. The first is an auxiliary padding certificate that quantifies what would be true with an external tail prior. The second is the actual finite-record reconstruction statement used by the numerical algorithms.

\textbf{Step 1: Auxiliary bi-infinite padding.} Define the padded sequence $\{\tilde{s}_n\}_{n \in \mathbb{Z}}$ by
\begin{equation}
    \tilde{s}_n = \begin{cases}
        s_n                & \text{if } n \in \mathcal{I}, \\
        \tau_0 + n\Delta T & \text{otherwise}.
    \end{cases}
    \label{eq:padded_sequence}
\end{equation}
For $n \in \mathcal{I}$, the inside-window points satisfy $|\tilde{s}_n - (\tau_0+n\Delta T)| \le L_{MT}\,\Delta T < \Delta T/4$ by the retained velocity-calibrated bound~\eqref{eq:retained_velocity_perturbation} or the extremum bound~\eqref{eq:extremum_perturbation}. For $n \notin \mathcal{I}$, the outside-window points sit exactly on the Nyquist grid with zero perturbation. Hence $\sup_n |\tilde{s}_n - (\tau_0+n\Delta T)| \le L_{MT}\,\Delta T < \Delta T/4$, and the classical Kadec theorem applies to the full bi-infinite sequence $\{\tilde{s}_n\}_{n \in \mathbb{Z}}$, yielding the sampling inequality
\begin{equation}
    A_{MT}\|x\|_{L^2}^2 \le \sum_{n \in \mathbb{Z}} |x(\tilde{s}_n)|^2 \le B_{MT}\|x\|_{L^2}^2,
    \label{eq:padded_frame}
\end{equation}
with the signal-instance bounds $A_{MT}$, $B_{MT}$ from~\eqref{eq:signal_instance_frame_bounds}.

\textbf{Step 2: Isolate the unmeasured tail and boundary-collar terms.} Split the sum in \eqref{eq:padded_frame}:
\begin{equation}
    \sum_{n \in \mathbb{Z}} |x(\tilde{s}_n)|^2
    = \underbrace{\sum_{n \in \mathcal{I}} |x(s_n)|^2}_{S_{\mathrm{in}}}
    + \underbrace{\sum_{n \notin \mathcal{I}} |x(\tau_0 + n\Delta T)|^2}_{S_{\mathrm{out}}}.
    \label{eq:split_sum}
\end{equation}
The inside sum $S_{\mathrm{in}}$ consists of known MT data (each $x(s_n) = m_n\Delta V$). The outside sum $S_{\mathrm{out}}$ involves Nyquist-grid evaluations of $x$ not supplied by the finite MT record. Some of these grid centers may still lie in the boundary collar of $[0,T_{\text{obs}}]$, because $\mathcal I$ is defined by full centered-window inclusion rather than by center inclusion. Set
\begin{equation}
    S_{\mathrm{bdry}}
    =
    \sum_{\substack{n\notin\mathcal I\\
    \tau_0+n\Delta T\in[0,T_{\text{obs}}]}}
    |x(\tau_0+n\Delta T)|^2.
    \label{eq:boundary_collar_samples}
\end{equation}
For the remaining unmeasured grid centers, the localized Plancherel--P\'olya inequality (Appendix~\ref{app:pp_local}) controls samples by a collar-enlarged tail energy:
\begin{equation}
    \begin{aligned}
    S_{\mathrm{out}}
    &\le
    C_{\Delta T,\Omega}\, E_{\mathrm{tail}}^{\mathrm{collar}}
    +S_{\mathrm{bdry}},\\
    E_{\mathrm{tail}}^{\mathrm{collar}}
    &\triangleq
    \|x\|_{L^2((\mathbb{R}\setminus[0,T_{\text{obs}}])
    +[-r_{\mathrm{loc}},r_{\mathrm{loc}}])}^2,
    \end{aligned}
    \label{eq:tail_bound}
\end{equation}
where $r_{\mathrm{loc}}$ is the localization radius in Appendix~\ref{app:pp_local} and $C_{\Delta T,\Omega}$ is a positive constant depending only on the separation and bandwidth.

\textbf{Step 3: Tail-conditional lower certificate.} From the lower bound in \eqref{eq:padded_frame} and the tail estimate \eqref{eq:tail_bound}:
\begin{equation}
    S_{\mathrm{in}} \ge A_{MT}\|x\|_{L^2}^2
    - C_{\Delta T,\Omega}\, E_{\mathrm{tail}}^{\mathrm{collar}}
    - S_{\mathrm{bdry}}.
    \label{eq:effective_lower}
\end{equation}
This is not a frame inequality for the finite MT record on all of $PW_\Omega$. It becomes a positive lower certificate only after an external modeling assumption supplies usable bounds on both the collar-enlarged tail and the finite boundary-collar samples, for example $C_{\Delta T,\Omega}E_{\mathrm{tail}}^{\mathrm{collar}}+S_{\mathrm{bdry}}\le \alpha A_{MT}\|x\|_{L^2}^2$ with some $\alpha<1$. Without such a prior, two bandlimited functions that agree with the recorded finite crossing values can differ by a component concentrated outside the active window or in the unrecorded boundary collar, and the finite MT record alone cannot identify that component.

\textbf{Step 4: Actual finite-dimensional reconstruction guarantee.} In practice the reconstruction stage operates on a finite model space, such as a truncated sinc basis, a trigonometric-polynomial space, or a PSWF space localized to the active window. Let
\[
    V_M=\operatorname{span}\{\phi_1,\ldots,\phi_M\}\subset PW_\Omega
\]
with an orthonormal basis in the reconstruction norm, and form the sampling matrix
\[
    H_{\mathcal{I},M}=(\phi_m(s_n))_{n\in\mathcal{I},\,1\le m\le M}.
\]
\begingroup
\setlength{\abovedisplayskip}{6pt}
\setlength{\belowdisplayskip}{6pt}
\setlength{\abovedisplayshortskip}{0pt}
\setlength{\belowdisplayshortskip}{6pt}
\setlength{\parskip}{0pt}
If $\sigma_{\min}(H_{\mathcal{I},M})>0$, then the finite MT data give the genuine finite-record frame inequality
\begin{equation}
    \begin{aligned}
    \sigma_{\min}^2(H_{\mathcal{I},M})\,\|p\|_{L^2}^2
    &\le \sum_{n\in\mathcal{I}} |p(s_n)|^2 \\
    &\le \|H_{\mathcal{I},M}\|_2^2\,\|p\|_{L^2}^2,
    \qquad p\in V_M.
    \end{aligned}
    \label{eq:finite_matrix_frame}
\end{equation}
Consequently, least-squares or pseudo-inverse reconstruction on $V_M$ is stable with condition number
\[
    \kappa_M=\|H_{\mathcal{I},M}\|_2/\sigma_{\min}(H_{\mathcal{I},M}).
\]
For a general $x=p+r$ with $p=P_{V_M}x$, the reconstruction error contains three separate terms:
\begin{equation}
    \begin{aligned}
    \|x-x_{\mathrm{rec}}\|_{L^2}
    \le{}& \|r\|_{L^2}
    + \sigma_{\min}(H_{\mathcal{I},M})^{-1} \\
    &\times
      \bigl(\|\epsilon\|_{\ell^2(\mathcal{I})}
      +\|(r(s_n))_{n\in\mathcal{I}}\|_{\ell^2}\bigr),
    \end{aligned}
    \label{eq:finite_model_error}
\end{equation}
where $x_{\mathrm{rec}}$ denotes the finite-record reconstructed signal and $\epsilon$ denotes sample noise or hardware-induced value error. Thus the finite-data theorem is an approximation theorem on a declared finite model class, not a uniqueness theorem for arbitrary global elements of $PW_\Omega$.

\textbf{Interpretation.} The Kadec-compatible geometry remains useful because it gives a computable diagnostic for when $H_{\mathcal{I},M}$ should be well conditioned and how conditioning changes as $\Delta V$ is refined. The mathematically valid finite-record claim, however, is the matrix inequality~\eqref{eq:finite_matrix_frame} plus the model/tail error bound~\eqref{eq:finite_model_error}. Exact infinite-dimensional recovery is recovered only in an idealized limiting regime in which the active window expands, the tail term vanishes, and the finite reconstruction spaces exhaust $PW_\Omega$.

\endgroup

\medskip
\noindent
\rule{\columnwidth}{0.4pt}
\begingroup
\small
\noindent\textbf{Summary of Local Sufficient Conditions for Stable MT Reconstruction.}

\medskip
\textit{Modeling assumptions.}
\begin{enumerate}
    \renewcommand{\labelenumi}{(A\arabic{enumi})}
    \item $x \in PW_\Omega$ with exact Fourier support in $[-\Omega,\Omega]$; approximate bandlimitation is handled by adding an out-of-band residual term.
    \item Static uniform threshold ladder with spacing $\Delta V > 0$.
    \item Ideal timing (zero jitter, infinite TDC resolution).
    \item Monotone-crossing event model (retained adjacent-transition comparator output).
    \item Non-degenerate extrema: $x'(t^*) = 0 \implies x''(t^*) \neq 0$ on $[0, T_{\text{obs}}]$.
    \item Finite active window: $\|x\|_{L^\infty(\mathbb{R} \setminus [0, T_{\text{obs}}])} < \Delta V$ prevents unmodeled outside-window events, but does not control outside-window $L^2$ energy.
    \item Either an external tail/collar prior $C_{\Delta T,\Omega}\, E_{\mathrm{tail}}^{\mathrm{collar}}+S_{\mathrm{bdry}} < A_{MT}\|x\|_{L^2}^2$ or an explicitly chosen finite reconstruction space $V_M$ with $\sigma_{\min}(H_{\mathcal{I},M})>0$.
\end{enumerate}

\medskip
\textit{MT crossing regularity (Propositions 1--4).}
The realized retained crossing set $\Lambda_{\Delta V}(x)$ has deterministic regularity absent from generic nonuniform sets under the stated event topology: local spacing locked to $\Delta V/|x'|$ (Proposition~\ref{prop:local_intercrossing_spacing}), smooth spacing variation (Proposition~\ref{prop:smooth_gap_variation}), retained phase-aware $\sqrt{m}$-law near extrema (Proposition~\ref{prop:phase_aware_extremum_pattern}), and fixed-topology event-time regularity (Proposition~\ref{prop:structural_regularity_fixed_topology}).

\medskip
\textit{Unified sufficient condition.}
$\Delta V < \Delta V_{\mathrm{Kadec,ret}}^*(x;\tau_0)$ from~\eqref{eq:unified_condition} is the phase-uniform retained-safe design version; for a deployed ladder one may instead verify the fixed-$\Delta V$ predicate $\mathsf{Cert}_{\mathrm{Kadec}}(x,\Delta V,\tau_0)$ in~\eqref{eq:realized_kadec_certificate}, using the actual retained phases, clipped exact-anchor convention, third-order Taylor remainder, and extremum-isolation radius~\eqref{eq:extremum_isolation_radius}.

\medskip
\textit{Signal-instance reconstruction guarantee.}
Under (A1)--(A7): the auxiliary padded sequence satisfies the Kadec frame inequality~\eqref{eq:padded_frame}; the finite MT record itself supports either the tail-conditional certificate~\eqref{eq:effective_lower} or the finite-dimensional matrix frame~\eqref{eq:finite_matrix_frame}. The sufficient bound implies an improving conditioning trend as the threshold ladder is refined, with the rate quantified by~\eqref{eq:condition_number_scaling}.

\medskip
\textit{Scope boundary.}
These conditions are \emph{signal-instance} statements. The no-go theorem of Section~\ref{sec:no_go} proves that no signal-independent guarantee is possible for static MT.
\par
\endgroup
\noindent\rule{\columnwidth}{0.4pt}
\section{Asymptotic Density Analysis on Long Active Windows}
\label{sec:density_analysis}
While the Kadec conditions of Section~\ref{sec:sufficient} provide rigorous local sufficient criteria, they are per-window statements. We now derive complementary long-window design inequalities by comparing the retained MT event rate against the critical Nyquist density from Landau--Beurling theory. A key additional object is the \emph{pointwise threshold-traversal envelope} $\rho_{MT}(t) = |x'(t)|/\Delta V$. This Lipschitz, signal-determined envelope is not itself the retained counting measure: endpoint truncation, floor effects, and same-threshold suppression all create discrete losses. The main theorem of this section (Theorem~\ref{thm:kadec_density_convergence}) then shows that, on monotone segments, the Kadec path (local displacement control) and the retained-density path (long-window event rate) share conservative design scales. This connection is specific to MT-generated sets and is generally unavailable for arbitrary nonuniform sampling sets.

\textbf{Scope of analysis.} Since signals of practical interest have finite effective duration, the analysis in this section is restricted to a long active observation window $[0, T_{\text{obs}}]$ with $T_{\text{obs}} \gg \Delta T$, where $\|x\|_{L^\infty(\mathbb{R} \setminus [0, T_{\text{obs}}])} < \Delta V$. The Landau--Beurling theorem is used here only as an \emph{asymptotic benchmark} for the retained event density observed on long sub-intervals fully contained in the active window.

Throughout this section, the signal instance $x$ is fixed. The quantity $n_{\mathrm{ret}}(t,r)$ denotes the number of retained adjacent-transition events in $[t,t+r]$, while $N_{\mathrm{mono}}(t,r)$, $\bar{\nu}(t,r)$, and $\nu_{min}$ are computed from the same waveform on the same active window.

\subsection{Pointwise Threshold-Traversal Envelope}
\label{sec:density_profile}
The structural regularity of MT crossing sets (Propositions~\ref{prop:local_intercrossing_spacing}--\ref{prop:structural_regularity_fixed_topology}) implies that the monotone-branch crossing rate has an explicitly computable envelope. Because of the absolute value in $|x'|$, this envelope is Lipschitz rather than necessarily differentiable at extrema.

\begin{definition}[MT threshold-traversal envelope]
The \emph{MT threshold-traversal envelope} of the signal $x$ with threshold spacing $\Delta V$ is
\begin{equation}
    \begin{aligned}
    \rho_{MT}(t) &= \frac{|x'(t)|}{\Delta V},\\
    \bar{\rho}(t,r)
    &=
    \frac{1}{r}\int_t^{t+r}\rho_{MT}(\tau)\,d\tau,
    \qquad t \in [0, T_{\text{obs}}].
    \end{aligned}
    \label{eq:density_profile}
\end{equation}
\end{definition}

The physical interpretation is immediate: $\rho_{MT}(t)$ is the continuum rate at which the waveform traverses threshold intervals per unit time at position $t$. On a monotone segment near $t$, the local raw crossing spacing is $\Delta V/|x'(t)| = 1/\rho_{MT}(t)$ (Proposition~\ref{prop:local_intercrossing_spacing}). It is only an envelope for retained event counts; the retained adjacent-transition convention can remove a same-threshold return, as in Experiment~7.

\begin{proposition}[Lipschitz regularity of the traversal envelope]
\label{prop:traversal_envelope_lipschitz}
For $x \in PW_\Omega$, the traversal envelope $\rho_{MT}$ is Lipschitz continuous with constant controlled by the bandwidth, amplitude envelope, and threshold spacing:
\begin{multline}
    |\rho_{MT}(t_1) - \rho_{MT}(t_2)| \le \frac{\|x''\|_\infty}{\Delta V}|t_1{-}t_2| \\
    \le \frac{\Omega^2\|x\|_\infty}{\Delta V}|t_1{-}t_2|
    \label{eq:density_lipschitz}
\end{multline}
for all $t_1, t_2 \in [0, T_{\text{obs}}]$, where the second inequality follows from Bernstein's inequality.
\end{proposition}

\textit{Proof.} Since $|x'|$ is Lipschitz with constant $\|x''\|_\infty$ (by the mean value theorem applied to $x'$), the claim follows from dividing by $\Delta V > 0$. \hfill $\square$

\begin{corollary}[Windowed density variation bound]
\label{cor:windowed_density_variation}
On any sub-interval of length $r$:
\begin{equation}
    \sup_{[t,t+r]} \rho_{MT} - \inf_{[t,t+r]} \rho_{MT} \le \frac{\Omega^2\|x\|_\infty}{\Delta V}\,r.
    \label{eq:density_variation}
\end{equation}
\end{corollary}

This bound should not be read as a substitute for a sliding-window lower-density certificate. A high average value of $\rho_{MT}$ on a long interval can still coexist with a local valley near an extremum or near-quiescent segment; the excess density elsewhere compensates for the deficit. Lipschitz regularity only limits how abruptly such a valley can begin or end. Consequently, all density design rules below use an explicit infimum over windows, such as $\nu_{min}$ or the extremum-density certificate, rather than inferring local admissibility from a single high window average.

\subsection{Separation Condition via Bernstein's Inequality}
Beurling's sufficient condition requires the realized sampling set to be uniformly discrete (i.e., $\inf_{m \neq n}|t_m - t_n| > 0$). The retained static MT sequence automatically satisfies this property. Let $(t_k,m_k)$ and $(t_{k+1},m_{k+1})$ be consecutive retained events. The adjacent-transition convention gives $m_{k+1}\neq m_k$; by continuity, intermediate threshold hits would have been recorded before any non-adjacent jump, so in particular $|x(t_{k+1})-x(t_k)|\ge\Delta V$. The ordinary mean value theorem on $[t_k,t_{k+1}]$ gives a point $\xi_k$ such that
\[
    |x(t_{k+1})-x(t_k)|=|x'(\xi_k)|\,|t_{k+1}-t_k|.
\]
Since $|x'(\xi_k)| \le \|x'\|_\infty$, the minimum inter-sample time is bounded below. Applying Bernstein's inequality $\|x'\|_\infty \le \Omega \|x\|_\infty = 2\pi f_{max} \|x\|_\infty$:
\begin{equation}
    \delta_{\mathrm{sep}} = \inf_k |t_{k+1} - t_k| \ge \frac{\Delta V}{\|x'\|_\infty} \ge \frac{\Delta V}{2\pi f_{max} \|x\|_\infty}
    \label{eq:sep_bound}
\end{equation}
The realized MT sequence is therefore automatically separated. Finite bandwidth implies finite slew rate, and finite slew rate combined with nonzero threshold spacing forbids arbitrarily tight clusters of events.

\subsection{Structure-Adapted Density Bounds}
\label{sec:structure_density}
The physical link between bandwidth, voltage, and event count is provided by total variation. Each retained MT event is preceded by a threshold-step traversal of size $\Delta V$, but not every raw traversal is retained: the adjacent-transition convention can suppress the first same-threshold return on a monotone segment. The retained count on $[t,t+r]$ is therefore controlled by the vertical distance traveled minus explicit endpoint, floor, and suppression losses.

Let $N_{\mathrm{supp}}(t,r)$ denote the number of monotone pieces in $[t,t+r]$ whose first raw threshold hit is suppressed as a same-threshold return by $\mathcal{A}_{\Delta V}$. On a strictly monotone piece, this can happen at most once, because after the first retained adjacent-level transition the threshold indices continue monotonically. Hence
\begin{equation}
    0\le N_{\mathrm{supp}}(t,r)\le N_{\mathrm{mono}}(t,r).
    \label{eq:supp_count_bound}
\end{equation}
The retained crossing count $n_{\mathrm{ret}}(t,r)$ within $[t,t+r] \subseteq [0,T_{\text{obs}}]$ satisfies:
\begin{align}
    \text{Upper:}\quad n_{\mathrm{ret}}(t, r)
    &\le \frac{1}{\Delta V} \int_t^{t+r} |x'(\tau)| \,d\tau + 2 \label{eq:cross_upper}\\[4pt]
    \text{Lower:}\quad n_{\mathrm{ret}}(t, r)
    &\ge \frac{1}{\Delta V} \int_t^{t+r} |x'(\tau)| \,d\tau \nonumber\\
    &\quad - N_{\mathrm{mono}}(t,r)-N_{\mathrm{supp}}(t,r). \label{eq:cross_lower}
\end{align}
The upper bound is a retained-count statement. If $u_1<\cdots<u_N$ are retained events in the window, then consecutive retained labels differ by one ladder step, so
\[
    \int_{u_i}^{u_{i+1}}|x'(\tau)|\,d\tau\ge\Delta V,
    \qquad i=1,\ldots,N-1.
\]
Hence $(N-1)\Delta V\le\mathrm{TV}(x;[t,t+r])$; the displayed $+2$ is a harmless boundary cushion for open/closed endpoint conventions and subsequence extraction. This argument is not valid for raw crossings, which can repeatedly return to the same threshold with arbitrarily small vertical excursions. In the lower estimate, the usual floor loss removes one count per monotone piece; the same-threshold suppression loss then removes any first raw hit suppressed by the retained-event convention. Without a separate suppression certificate, the retained-safe replacement is the conservative bound $n_{\mathrm{ret}}\ge \mathrm{TV}/\Delta V-2N_{\mathrm{mono}}$.

\textbf{Signal-specific extremum density.} Bandwidth alone does not give a uniform finite-window bound of the form $N_{\mathrm{mono}}(t,r)\le 2f_{max}r+\mathcal{O}(1)$ for all $x\in PW_\Omega$. Superoscillatory bandlimited functions can realize arbitrarily dense prescribed local oscillations on a finite interval, at the price of large dynamic range or energy concentration elsewhere~\cite{Kempf2000}. Therefore the density theory below is stated with an explicit signal-instance extremum-density certificate. For a fixed signal instance, define
\begin{equation}
    \eta_x = \sup_{\substack{[t,t+r] \subseteq [0,T_{\text{obs}}] \\ r \ge \Delta T}} \frac{Z_x(t,r)}{r}
    \label{eq:extremum_density}
\end{equation}
where $Z_x(t,r)$ is the number of zeros of $x'$ in $[t,t+r]$. On the finite active window this quantity is finite for every nonzero analytic signal instance with isolated critical points, but it is not bounded above by a bandwidth-only constant. Equivalently, for a signal class $\mathcal{S}$ one may assume the explicit anti-superoscillation condition
\begin{equation}
    Z_x(t,r)\le \eta_0 r+C_0,\qquad x\in\mathcal{S},
    \label{eq:zero_density_assumption}
\end{equation}
with constants $\eta_0,C_0$ supplied by the model, by a dynamic-range/energy-concentration certificate, or by direct offline verification of the signal family. Under this assumption $N_{\mathrm{mono}}(t,r)\le \eta_0 r+\mathcal{O}(1)$; for a single verified instance one may set $\eta_0=\eta_x$.

The retained model also needs a suppression-density certificate. For a fixed instance define
\begin{equation}
    \eta_{\mathrm{supp},x}
    =
    \sup_{\substack{[t,t+r] \subseteq [0,T_{\text{obs}}] \\ r \ge \Delta T}}
    \frac{N_{\mathrm{supp}}(t,r)}{r},
    \label{eq:suppression_density}
\end{equation}
and for a signal class assume $N_{\mathrm{supp}}(t,r)\le \eta_{\mathrm{supp},0}r+C_{\mathrm{supp}}$. Since~\eqref{eq:supp_count_bound} gives $\eta_{\mathrm{supp},0}\le\eta_0$ after adjusting endpoint constants, one always has a retained-safe fallback with $\eta_0+\eta_{\mathrm{supp},0}\le2\eta_0$.

\textbf{Design status of $\eta_{\mathrm{supp},0}$.} A class-level value of $\eta_{\mathrm{supp},0}$ is an acquisition-side design parameter only when it is supplied by an analytic topology prior, a certified waveform library, or an offline verification of the admissible signal class before deployment. If the suppression count is computed after observing a particular record, the resulting $\eta_{\mathrm{supp},x}$ is a post-realization diagnostic quantity. In the absence of such prior topology information, the design rule must use the fallback $\eta_{\mathrm{supp},0}\le\eta_0$ rather than treating a smaller suppression density as a free margin.

Using the $\rho_{MT}$ and $\bar{\rho}$ notation from~\eqref{eq:density_profile}, the retained event-rate bounds become
\begin{multline}
    \bar{\rho}(t,r)
    - \eta_x-\eta_{\mathrm{supp},x} - \mathcal{O}(1/r) \le \frac{n_{\mathrm{ret}}(t,r)}{r} \\
    \le \bar{\rho}(t,r) + \frac{2}{r}.
    \label{eq:density_window_improved}
\end{multline}

Define the \emph{minimum mean absolute velocity} as before:
\begin{align}
    \nu_{min} &= \inf_{\substack{[t, t+r] \subseteq [0, T_{\text{obs}}] \\ r \ge 1/(2f_{max})}} \bar{\nu}(t,r), \notag\\
    \bar{\nu}(t,r) &= \frac{1}{r}\int_t^{t+r}|x'(\tau)|\,d\tau.
    \label{eq:velocity}
\end{align}
Under the explicit extremum-density certificate~\eqref{eq:zero_density_assumption}, the lower envelope of the retained event density satisfies the \emph{structure-calibrated} estimate
\begin{equation}
    \frac{\nu_{min}}{\Delta V} - \eta_0-\eta_{\mathrm{supp},0}
    \;\lesssim\; \inf \frac{n_{\mathrm{ret}}(t,r)}{r}
    \;\lesssim\; \frac{\nu_{min}}{\Delta V},
    \label{eq:density_bounds_improved}
\end{equation}
where $\eta_0,\eta_{\mathrm{supp},0}$ may be replaced by the verified instance values $\eta_x,\eta_{\mathrm{supp},x}$. Compared with a model that pessimistically provisions for frequent turning points and every possible suppressed return, the structure-calibrated bound is tighter when the certified $\eta_0$ and $\eta_{\mathrm{supp},0}$ are small. For example, a non-superoscillatory slowly oscillating signal with effective oscillation scale $f_0 \ll f_{max}$ typically has $\eta_x \approx 2f_0$; this is an additional signal-instance observation, not a consequence of the nominal bandwidth alone.

\textbf{Finite-window retained-density certificate.} The asymptotic estimate~\eqref{eq:density_bounds_improved} is not, by itself, a finite-window theorem. To make a statement over all certified windows of length at least $R\ge\Delta T$, define
\begin{equation}
    \nu_{min}^{(R)}
    =
    \inf_{\substack{[t,t+r]\subseteq[0,T_{\textup{obs}}]\\ r\ge R}}
    \bar{\nu}(t,r).
    \label{eq:finite_R_velocity}
\end{equation}
Assume nonnegative finite-window constants $C_{\mathrm{mono}}$ and $C_{\mathrm{supp}}$ such that
\begin{align}
    N_{\mathrm{mono}}(t,r)
    &\le \eta_0 r+C_{\mathrm{mono}}, \nonumber\\
    N_{\mathrm{supp}}(t,r)
    &\le \eta_{\mathrm{supp},0}r+C_{\mathrm{supp}}
    \label{eq:finite_R_count_constants}
\end{align}
on every certified window with $r\ge R$. Then~\eqref{eq:cross_lower} gives the explicit finite-window lower bound
\begin{align}
    \frac{n_{\mathrm{ret}}(t,r)}{r}
    &\ge
    \frac{\nu_{min}^{(R)}}{\Delta V}
    -\eta_0-\eta_{\mathrm{supp},0}
    -\frac{C_{\mathrm{mono}}+C_{\mathrm{supp}}}{R}, \nonumber\\
    &\hspace{0.7in} r\ge R.
    \label{eq:finite_R_density_bound}
\end{align}
Consequently, the retained density exceeds the Nyquist benchmark on every certified window of length at least $R$ whenever
\begin{equation}
    \Delta V
    <
    \frac{\nu_{min}^{(R)}}
    {2f_{max}+\eta_0+\eta_{\mathrm{supp},0}
    +(C_{\mathrm{mono}}+C_{\mathrm{supp}})/R}.
    \label{eq:finite_R_density_design}
\end{equation}
Equation~\eqref{eq:finite_R_density_design} is the finite-window statement when a suppression-topology certificate is available before acquisition. If no such acquisition-side certificate is available, the only topology-free retained-safe replacement is~\eqref{eq:supp_count_bound}: take $\eta_{\mathrm{supp},0}=\eta_0$ and $C_{\mathrm{supp}}=C_{\mathrm{mono}}$. This gives the explicit fallback
\begin{equation}
    \Delta V
    <
    \frac{\nu_{min}^{(R)}}
    {2f_{max}+2\eta_0+2C_{\mathrm{mono}}/R}.
    \label{eq:finite_R_retained_safe_fallback}
\end{equation}
This fallback is generally more conservative, but it is the finite-window acquisition-side rule that remains valid without a library-level or analytic bound on same-threshold suppression. The simpler formulas below are long-window design scales obtained after the additive constants become negligible.

\begin{corollary}[Class-level a priori retained-density certificate]
Let $\mathcal{S}\subset PW_\Omega$ be a restricted signal class on the active window $[0,T_{\textup{obs}}]$, and fix a window scale $R\ge\Delta T$. Suppose there are constants
\[
    \nu_{\mathcal{S}}^{(R)}>0,\qquad
    \eta_0,\eta_{\mathrm{supp},0}\ge0,\qquad
    C_{\mathrm{mono}},C_{\mathrm{supp}}\ge0
\]
such that, for every $x\in\mathcal{S}$ and every $[t,t+r]\subseteq[0,T_{\textup{obs}}]$ with $r\ge R$,
\begin{align*}
    \bar{\nu}_x(t,r)&=\frac{1}{r}\int_t^{t+r}|x'(\tau)|\,d\tau
        \ge \nu_{\mathcal{S}}^{(R)},\\
    N_{\mathrm{mono}}(t,r)&\le \eta_0 r+C_{\mathrm{mono}},\\
    N_{\mathrm{supp}}(t,r)&\le \eta_{\mathrm{supp},0}r+C_{\mathrm{supp}}.
\end{align*}
The first condition is the class-level exclusion of near-quiescent windows at scale $R$; the last two conditions are the class-level extremum-density and retained-topology certificates. If
\begin{equation}
    \Delta V
    <
    \frac{\nu_{\mathcal{S}}^{(R)}}
    {2f_{max}+\eta_0+\eta_{\mathrm{supp},0}
    +(C_{\mathrm{mono}}+C_{\mathrm{supp}})/R},
    \label{eq:class_level_density_certificate}
\end{equation}
then the same static, uniformly spaced threshold set gives, for every $x\in\mathcal{S}$,
\begin{equation}
    \frac{n_{\mathrm{ret}}(t,r)}{r}>2f_{max},
    \qquad
    [t,t+r]\subseteq[0,T_{\textup{obs}}],\ r\ge R.
    \label{eq:class_level_density_result}
\end{equation}
Thus~\eqref{eq:class_level_density_certificate} is an a priori sufficient retained-density certificate for the whole class $\mathcal{S}$. If no independent suppression-topology certificate is available before acquisition, the retained-safe sufficient replacement is obtained by setting $\eta_{\mathrm{supp},0}=\eta_0$ and $C_{\mathrm{supp}}=C_{\mathrm{mono}}$.
\end{corollary}

\textit{Proof.}
The hypotheses give $\nu_{min}^{(R)}\ge\nu_{\mathcal{S}}^{(R)}$ for every $x\in\mathcal{S}$ and supply the uniform count constants in~\eqref{eq:finite_R_count_constants}. Substituting these class-level constants into~\eqref{eq:finite_R_density_bound} shows that~\eqref{eq:class_level_density_certificate} forces the retained event rate to exceed $2f_{max}$ on every certified window. The retained-safe replacement follows from~\eqref{eq:supp_count_bound}. \hfill $\square$

\textbf{Status of the corollary.}
This is a class-level sufficient certificate, not a density-only necessary-and-sufficient theorem. A record outside the displayed spacing bound may still pass a sharper post-realization weak-limit, Kadec, or finite-matrix test because of favorable threshold phase or redundant microscopic geometry. Conversely, without the uniform lower-velocity and topology assumptions on $\mathcal{S}$, the no-go construction in Section~\ref{sec:no_go} can create near-quiescent windows for which no fixed positive $\Delta V$ satisfies a prior-free retained-density guarantee.

\subsection{Consistency of the Kadec and Density Paths}
\label{sec:convergence}
We now state the precise relation between the local Kadec path and the long-window density path on monotone segments. The two criteria are not equivalent in general: a particular threshold phase or a favorable waveform shape may satisfy one test even when a conservative velocity bound is violated. What can be proved without additional microscopic assumptions is a common sufficient design scale. On monotone branches, the same velocity-to-threshold ratio that controls local Kadec displacement also forces the long-window retained event density above the Nyquist benchmark.

\begin{theorem}[Clean and retained-safe Kadec-density consistency]
\label{thm:kadec_density_convergence}
\textit{Let $x \in PW_\Omega$ be fixed, and let $[a,b] \subseteq [0, T_{\textup{obs}}]$ be a sub-interval contained in a single maximal monotone branch of $x$ with $v_{\mathrm{loc}} = \inf_{t \in [a,b]}|x'(t)| > 0$. Then the following two common sufficient regimes hold.}
\begin{enumerate}
    \item[\textup{(C)}] \textit{\textbf{Clean retained topology.} Suppose the nearest raw threshold hit selected in each centered window is retained, and suppose $N_{\mathrm{supp}}(t,r)=0$ on the density sub-intervals under consideration. If}
\begin{equation}
    \Delta V < \frac{v_{\mathrm{loc}}}{4f_{max}},
    \label{eq:kadec_density_clean_sufficient}
\end{equation}
\textit{then every centered half-Nyquist window $J_n(\tau_0)$ contained in $[a,b]$ contains a retained MT crossing at distance $<\Delta T/4$ from its center, and $n_{\mathrm{ret}}(t,r)/r>2f_{max}$ for all $[t,t+r]\subseteq[a,b]$ with $r\ge\Delta T$.}
    \item[\textup{(R)}] \textit{\textbf{Retained-safe topology.} Without clean-topology information, assume only the adjacent-transition convention. If}
\begin{equation}
    \Delta V < \frac{v_{\mathrm{loc}}}{12f_{max}},
    \label{eq:kadec_density_retained_sufficient}
\end{equation}
\textit{then the same two conclusions hold for the retained event count $n_{\mathrm{ret}}$.}
\end{enumerate}
\textit{Conversely, neither Kadec compatibility nor retained Nyquist-rate density alone implies either sufficient inequality without additional assumptions on the threshold phase and the full inter-crossing gap sequence.}
\end{theorem}

\textit{Proof.}

\textit{Clean branch.} In the clean retained topology, $\alpha_n^{(\mathrm{ret})}\le1/2$ and Lemma~\ref{lem:retained_mvt_perturbation_monotone} gives a retained crossing in each centered window with displacement at most $\Delta V/(2v_{\mathrm{loc}})<\Delta T/4$. For density, the monotone branch has $N_{\mathrm{mono}}(t,r)=1$ and $N_{\mathrm{supp}}(t,r)=0$, so~\eqref{eq:cross_lower} gives
\[
    n_{\mathrm{ret}}(t,r)\ge \frac{\mathrm{TV}(x;[t,t+r])}{\Delta V}-1
    \ge \frac{v_{\mathrm{loc}}r}{\Delta V}-1.
\]
Dividing by $r\ge\Delta T$ yields
\begin{multline*}
    \frac{n_{\mathrm{ret}}(t,r)}{r}
    \ge \frac{v_{\mathrm{loc}}}{\Delta V}-\frac{1}{r}
    > 4f_{max}-2f_{max}
    =2f_{max}.
\end{multline*}

\textit{Retained-safe branch.} Under only the adjacent-transition convention, $\alpha_n^{(\mathrm{ret})}\le3/2$, so $\Delta V<v_{\mathrm{loc}}/(12f_{max})$ gives the retained Kadec displacement bound $(3/2)\Delta V/v_{\mathrm{loc}}<\Delta T/4$. For density, a monotone sub-interval can lose one floor unit and at most one same-threshold return, hence
\[
    n_{\mathrm{ret}}(t,r)\ge \frac{v_{\mathrm{loc}}r}{\Delta V}-2.
\]
Since $r\ge\Delta T$, we have $2/r\le4f_{max}$, and therefore
\[
    \frac{n_{\mathrm{ret}}(t,r)}{r}
    \ge \frac{v_{\mathrm{loc}}}{\Delta V}-\frac{2}{r}
    > 12f_{max}-4f_{max}
    >2f_{max}.
\]
The final statement follows from the fact that Kadec compatibility and retained density depend on the realized threshold phase and the complete local gap pattern, whereas~\eqref{eq:kadec_density_clean_sufficient} and~\eqref{eq:kadec_density_retained_sufficient} use only the worst-case lower velocity $v_{\mathrm{loc}}$ and topology class. A favorable phase may place crossings near the Nyquist anchors even when the lower-velocity sufficient bound is exceeded, and a high average velocity over a long interval may keep the density above Nyquist while a short local gap violates Kadec placement. Thus the displayed inequalities are conservative common certificates rather than necessary conditions. \hfill $\square$

\textbf{Discussion.} For an arbitrary nonuniform sampling set $\Lambda$, local Kadec placement and long-window density are logically independent: a set can satisfy one without the other. A high-density set with one isolated gap violates Kadec placement while still satisfying density on sufficiently long windows. Conversely, a set that is well distributed over a short active region need not satisfy a global density condition. MT does not remove this independence completely, but it links both criteria to the same physical quantity: local signal velocity. The clean scale~\eqref{eq:kadec_density_clean_sufficient} and retained-safe scale~\eqref{eq:kadec_density_retained_sufficient} are therefore common sufficient guardrails, not sharp boundaries. The retained-safe scale is smaller because it budgets both larger Kadec displacement and one possible same-threshold suppression loss per monotone piece.

This consistency accounts for the empirical observation that density-based and displacement-based design rules for MT often give thresholds of the same order. The agreement is structural at the level of scaling, while exact admissibility remains governed by the microscopic realized crossing geometry.

\subsection{A Priori Strict-Density Certificate}
\label{sec:apriori_strict_density_certificate}
The finite-window certificates above are sufficient design rules. A scalar density statement becomes a necessary-and-sufficient theorem only after the retained MT time-event density has been calibrated a priori and the critical Landau boundary has been excluded. This subsection records that exact noncritical conversion. It should not be read as a global theorem for all finite-energy static MT records; Theorem~\ref{thm:no_go} rules that out. The appropriate settings are persistent streams, periodic extensions, finite active records equipped with padding or tail priors, and any MT-admissible class for which the exact density calibration below is part of the model.

\begin{definition}[MT strict-density admissibility]
\label{def:mvt_strict_density_admissibility}
Fix $\Delta V>0$ and set $\delta_V=1/\Delta V$. A retained MT instance is \emph{strict-density admissible at bandwidth $\Omega$} if its retained time set $\Lambda_{\Delta V}(x)$ is separated, if
\begin{equation}
    \nu_x^-
    =
    \liminf_{r\to\infty}\inf_{t\in\mathbb{R}}
    \frac{1}{r}\int_t^{t+r}|x'(\tau)|\,d\tau
    >0,
    \label{eq:mvt_nu_minus_def}
\end{equation}
and if a nonnegative retained-loss calibration functional $L_x(t,r)$ is supplied a priori. This object is not required to have a canonical pointwise formula such as $\delta_V\int_t^{t+r}|x'(\tau)|\,d\tau-n_{\mathrm{ret}}(t,r)$, because endpoint phases and finite-window initialization can create $O(1)$ fluctuations of either sign in that raw difference. Instead, $L_x(t,r)$ is part of the adopted retained-event model calibration: it may absorb endpoint floors, monotone-piece boundary floors, same-threshold suppression, finite-window initialization corrections, and other convention-dependent losses, as long as the asymptotic identity~\eqref{eq:exact_density_calibration} below has been proved for the retained map under consideration. Its lower-density loss rate is
\begin{equation}
    \ell_x
    =
    \limsup_{r\to\infty}\sup_{t\in\mathbb{R}}
    \frac{L_x(t,r)}{r}.
    \label{eq:retained_loss_density_def}
\end{equation}
Finite endpoint losses vanish asymptotically unless they occur with positive long-window frequency, for example through a positive density of monotone-piece boundaries or repeated retained-topology suppressions. The instance is strict-density admissible if, in addition, the exact density calibration
\begin{equation}
    D^-(\Lambda_{\Delta V}(x))
    =
    \nu_x^-\delta_V-\ell_x
    \label{eq:exact_density_calibration}
\end{equation}
holds, with $D^-$ as defined in~\eqref{eq:density_def}. The instance is \emph{noncritical} if
\begin{equation}
    \nu_x^-\delta_V-\ell_x\ne\frac{\Omega}{\pi}.
    \label{eq:noncritical_density}
\end{equation}
\end{definition}

\textbf{Logical status.}
This theorem is not a prior-free density theorem. It is a conversion result: once a model has exactly calibrated retained threshold traversal and retained losses into the actual lower time-event density, Landau necessity and Beurling strict-density sufficiency turn the noncritical scalar comparison into an iff. Without that calibration, scalar threshold density remains only a proxy for the retained event geometry.

\begin{theorem}[Noncritical calibrated-density conversion lemma]
\label{thm:apriori_strict_density_mvt}
\textit{Let $x$ be a noncritical strict-density-admissible retained MT instance in the sense of Definition~\ref{def:mvt_strict_density_admissibility}. Then the fixed set $\Lambda_{\Delta V}(x)$ is a stable sampling set for $PW_\Omega$, equivalently the retained labeled record gives stable recovery from its exact point values under Corollary~\ref{cor:mvt_acquisition_recovery}, if and only if}
\begin{equation}
    \delta_V
    >
    \frac{\Omega/\pi+\ell_x}{\nu_x^-}
    =
    \frac{2f_{max}+\ell_x}{\nu_x^-}.
    \label{eq:strict_density_threshold}
\end{equation}
\end{theorem}

\textit{Proof.}
Let $\Lambda=\Lambda_{\Delta V}(x)$. If $\Lambda$ is a stable sampling set for $PW_\Omega$, Landau's necessary density theorem gives
\[
    D^-(\Lambda)\ge\frac{\Omega}{\pi}.
\]
Using the exact calibration~\eqref{eq:exact_density_calibration},
\[
    \nu_x^-\delta_V-\ell_x\ge\frac{\Omega}{\pi}.
\]
The noncritical hypothesis~\eqref{eq:noncritical_density} excludes equality, hence
\[
    \nu_x^-\delta_V-\ell_x>\frac{\Omega}{\pi}.
\]
Since $\nu_x^->0$, division by $\nu_x^-$ gives~\eqref{eq:strict_density_threshold}.

Conversely, assume~\eqref{eq:strict_density_threshold}. Multiplying by $\nu_x^-$ and subtracting $\ell_x$ gives
\[
    \nu_x^-\delta_V-\ell_x>\frac{\Omega}{\pi}.
\]
By~\eqref{eq:exact_density_calibration}, this is exactly
\[
    D^-(\Lambda)>\frac{\Omega}{\pi}.
\]
Strict-density admissibility includes separation of $\Lambda$, so Beurling's strict-density theorem applies and gives the sampling frame inequality on $PW_\Omega$. The retained labels satisfy $x(t_n)=m_n\Delta V$, so the acquisition recovery statement follows from Corollary~\ref{cor:mvt_acquisition_recovery}. \hfill $\square$

\textbf{Remark (why the theorem is strict).}
The noncritical assumption cannot be removed from a density-only statement. At the boundary $D^-(\Lambda)=\Omega/\pi$, stable sampling depends on fixed-set geometry beyond scalar density: the Shannon grid samples at critical density, whereas deleting one grid point leaves the lower density unchanged but destroys uniqueness. The boundary case must therefore be certified by the weak-limit condition of Section~\ref{sec:weak_limit_condition}, by a Pavlov/Lyubarskii--Seip complete-interpolation test, or by another structure-specific certificate such as the periodic-gap theorem below.

\subsection{Conditional Design Inequalities}
\label{sec:design_ineq}
Using the finite-window certificate~\eqref{eq:finite_R_density_design}, the structure-calibrated asymptotic bound~\eqref{eq:density_bounds_improved}, and the Kadec-density consistency result above, we obtain scalar design rules only after the extremum-density certificate~\eqref{eq:zero_density_assumption} has been supplied. Thus the density rules below are not universal consequences of $x\in PW_\Omega$; they are conditional rules for non-superoscillatory signal instances or classes with certified extremum and suppression densities.

\textbf{Finite-window sufficient condition.} When the intended claim is that every certified window with $r\ge R$ has retained density above Nyquist, the operative design rule is~\eqref{eq:finite_R_density_design}. It includes the endpoint, floor, and suppression constants through $(C_{\mathrm{mono}}+C_{\mathrm{supp}})/R$ and should be used for finite active records. If the acquisition design does not include an independent suppression-topology prior, the appropriate finite-record rule is the retained-safe fallback~\eqref{eq:finite_R_retained_safe_fallback}, not the more optimistic formula with a smaller post-realization value of $\eta_{\mathrm{supp},0}$.

\textbf{Long-window sufficient scale.} In the regime where $R$ is large enough that the finite-window additive constants are negligible, imposing $\nu_{min}/\Delta V - \eta_0-\eta_{\mathrm{supp},0} > 2f_{max}$ gives the asymptotic design scale
\begin{equation}
    \Delta V < \frac{\nu_{min}}{2f_{max} + \eta_0+\eta_{\mathrm{supp},0}}.
    \label{eq:strict_suff_improved}
\end{equation}
For a single verified signal instance, $\eta_0$ and $\eta_{\mathrm{supp},0}$ may be replaced by $\eta_x$ and $\eta_{\mathrm{supp},x}$; in that case the formula is a post-realization long-window diagnostic unless those instance quantities were certified before acquisition. If no acquisition-side suppression-density certificate is available, the finite-window fallback is~\eqref{eq:finite_R_retained_safe_fallback}; after taking the long-window limit in which $2C_{\mathrm{mono}}/R$ is negligible, it reduces to
\begin{equation}
    \Delta V < \frac{\nu_{min}}{2f_{max} + 2\eta_0}.
    \label{eq:strict_suff_retained_safe}
\end{equation}
The clean-topology formula $\Delta V<\nu_{min}/(2f_{max}+\eta_0)$ is recovered only when $N_{\mathrm{supp}}(t,r)=0$ on the certified windows. If an independent model justifies $\eta_0\approx 2f_0$ for an effective oscillation scale $f_0$ and no cleaner suppression information is available, the retained-safe long-window scale becomes $\Delta V<\nu_{min}/(2f_{max}+4f_0)$. Without such certificates, none of these formulas should be read as a prior-free ladder-design rule.

\textbf{Leading-density benchmark.} In a certified high-density regime ($\nu_{min}/\Delta V \gg \eta_0+\eta_{\mathrm{supp},0}$), the turning-point and suppression corrections are lower order and the leading traversal-density comparison becomes
\begin{equation}
    \Delta V \lesssim \frac{\nu_{min}}{2f_{max}}.
    \label{eq:ultimate_ns}
\end{equation}
This is a benchmark for the leading event-rate scale, not an asymptotically sharp theorem over the whole Paley-Wiener class. Superoscillatory or high-turning-density instances can violate the assumptions under which the lower-order term is negligible.

\textbf{Condition-number-aware design.} Beyond ``reconstructable or not,'' the signal-instance frame bounds~\eqref{eq:signal_instance_frame_bounds}--\eqref{eq:condition_number_scaling} provide a continuous figure of merit. A designer can target a specific condition number $\kappa_{\text{target}}$ by choosing
\begin{equation}
    \Delta V \le \frac{v_{min}(\kappa_{\text{target}}^{1/2}-1)}
    {2\pi\alpha_{\mathrm{ret}} f_{max}(\kappa_{\text{target}}^{1/2}+1)}
    \approx
    \frac{v_{min}(\kappa_{\text{target}}-1)}
    {8\pi\alpha_{\mathrm{ret}} f_{max}\kappa_{\text{target}}}
    \label{eq:condition_design}
\end{equation}
for $\kappa_{\text{target}}$ close to $1$. (This follows from $L = 2\alpha_{\mathrm{ret}}\Delta V f_{max}/v_{min}$ and $\theta \approx \pi L$, with the MT condition number $\kappa_{MT} = ((1+\theta)/(1-\theta))^2$.) This converts the abstract sufficiency question into a quantitative signal-to-noise engineering trade-off: a target reconstruction SNR determines the required condition number, which in turn determines the minimum number of comparator levels through~\eqref{eq:condition_design}. Setting $\alpha_{\mathrm{ret}}=1/2$ recovers the clean fixed-topology formula; setting $\alpha_{\mathrm{ret}}=3/2$ gives the retained-safe design.

\textbf{Physical interpretation.} The finite-window rules~\eqref{eq:finite_R_density_design} and~\eqref{eq:finite_R_retained_safe_fallback}, the long-window scale~\eqref{eq:strict_suff_improved}, and the conditioning rule~\eqref{eq:condition_design} all impose the same physical demand: the ladder must be fine enough that even the slowest useful segment of motion triggers retained events at roughly the Nyquist rate. When a uniform statement over $r\ge R$ is needed, the finite-window losses must remain explicit. Compared with a generic density argument, the MT regularity analysis quantifies \emph{how much faster than Nyquist} the event rate must be to achieve a specified reconstruction quality. The velocity-calibrated perturbation bound then bridges the gap between ``just enough crossings'' and ``well-conditioned crossings,'' a refinement unavailable without the structural regularity of the crossing set.

\section{Boundary of the Static Finite-Energy Problem}
\label{sec:no_go}
The no-go theorem is not a weakness of the static-threshold theory; it is the boundary condition that prevents overclaiming. The positive results above are exact or sufficient only after a retained event geometry, a class hull, a structured prior, a tail prior, or a finite-dimensional model has been specified. We now prove that the unrestricted static finite-energy retained problem has no universal prior-free full-line solution.

\subsection{Signal Classes and Definitions}

\begin{definition}[Near-quiescent interval]
\label{def:near_quiescent_interval}
An interval $[a,b] \subseteq [0, T_{\text{obs}}]$ with $b - a \ge \Delta T/2$ is called \emph{$\epsilon$-quiescent} for $x \in PW_\Omega$ if
\begin{equation}
    \int_a^b |x'(t)|\,dt < \epsilon.
    \label{eq:near_quiescent}
\end{equation}
The condition~\eqref{eq:near_quiescent} measures total voltage travel on the interval. When $\epsilon = 0$ the interval is called \emph{exactly quiescent}; for a nonzero $x \in PW_\Omega$ this cannot occur by the identity theorem for analytic functions. The mathematically relevant obstruction is the existence of \emph{arbitrarily} near-quiescent intervals ($\epsilon \to 0^+$), on which motion is present but too weak to sustain a robust stream of threshold events.
\end{definition}

\begin{definition}[Universal static MT guarantee]
\label{def:universal_static_mvt_guarantee}
A threshold spacing $\Delta V > 0$ provides a \emph{universal static MT guarantee} on a signal class $\mathcal{X} \subseteq PW_\Omega$ if, for every $x \in \mathcal{X}$, the crossing set $\Lambda_{\Delta V}(x)$ is a stable sampling set for $PW_\Omega$, i.e., there exist $0 < A \le B < \infty$ (possibly depending on $x$) such that
\begin{equation}
    A\|f\|_{L^2}^2 \le \sum_{t_n \in \Lambda_{\Delta V}(x)} |f(t_n)|^2 \le B\|f\|_{L^2}^2
    \label{eq:universal_frame}
\end{equation}
for all $f \in PW_\Omega$.
\end{definition}

\subsection{The No-Go Theorem}

\begin{theorem}[Static finite-energy obstruction and local certificate failure]
\label{thm:no_go}
\textit{Fix $\Delta V>0$ and the recorded adjacent-transition convention.}
\begin{enumerate}
    \item[\textup{(i)}] \textit{For every finite-energy $x\in PW_\Omega$, the retained record generated by static, uniformly spaced thresholds has lower Beurling density zero. Hence the retained static MT record, without tail samples, external padding values, or a finite-dimensional prior, is not a sampling set for the full space $PW_\Omega$. In particular, no fixed $\Delta V>0$ provides a universal static MT guarantee on $PW_\Omega$ in the sense of Definition~\ref{def:universal_static_mvt_guarantee}.}
    \item[\textup{(ii)}] \textit{For any compact interval $I=[a,b]$ with $b-a\ge \Delta T/2$ and any amplitude level $A>0$, there exists $x_{I,A}\in PW_\Omega$ with $\|x_{I,A}\|_\infty=A$ whose retained MT crossing set contains no points in $I$. Thus a static threshold set can fail any local Kadec-window or sliding no-gap certificate on a prescribed interval even when the same waveform has arbitrarily large amplitude elsewhere.}
\end{enumerate}
\end{theorem}

\textit{Proof.}
Statement~(i) is the global finite-energy obstruction. Since $x\in PW_\Omega$, its Fourier transform is supported in $[-\Omega,\Omega]$ and belongs to $L^2[-\Omega,\Omega]$. By Cauchy--Schwarz, $\hat{x}\in L^1[-\Omega,\Omega]$, and the inverse Fourier representation together with the Riemann--Lebesgue lemma gives
\[
    x(t)\to0,\qquad |t|\to\infty.
\]
Choose $T>0$ so large that $|x(t)|<\Delta V$ whenever $|t|>T$. Outside $[-T,T]$, the waveform cannot cross any nonzero threshold $m\Delta V$, $m\ne0$. It may cross the zero threshold, so the retained state has to be checked. On the right tail $(T,\infty)$ all raw labels, if any, are equal to $0$. The first raw zero after $T$ may be retained if the previous retained label is different from $0$ or if the tail is being initialized as a one-sided record. After that possible event, the retained state is $0$, and every later raw zero is a same-threshold return. Thus the right tail contributes at most one retained event. The same argument applied backward in time gives the left tail bound. Hence only finitely many tail events can be retained outside the compact interval.

It remains to rule out infinitely many retained events inside $[-T,T]$. If $x\equiv0$, there are no retained events. Otherwise $\|x\|_\infty>0$. For two consecutive retained adjacent-transition times $t_n<t_{n+1}$, the retained threshold labels differ by one, so
\[
    |x(t_{n+1})-x(t_n)|=\Delta V.
\]
The mean value theorem and Bernstein's inequality give
\[
    \Delta V
    \le \|x'\|_\infty |t_{n+1}-t_n|
    \le \Omega\|x\|_\infty |t_{n+1}-t_n|,
\]
and therefore
\[
    |t_{n+1}-t_n|
    \ge
    \frac{\Delta V}{\Omega\|x\|_\infty}>0.
\]
A compact interval contains only finitely many points with this separation. Thus the entire retained static MT set is finite, and in particular
\[
    D^-(\Lambda_{\Delta V}(x))=0.
\]
Landau's necessary condition requires $D^-(\Lambda)\ge\Omega/\pi=2f_{max}$ for any sampling set of $PW_\Omega$, proving the global obstruction and excluding any universal fixed-$\Delta V$ finite-energy theorem over the full space.

For statement~(ii), fix $I=[a,b]$, choose $\epsilon\in(0,\Delta V)$, and let $A>0$. Choose a shift $R$ so far from $I$ that, uniformly for $t\in I$,
\[
    A\,\mathrm{sinc}^2\!\left(\frac{\Omega(t-R)}{2\pi}\right)
    <\frac{\epsilon}{4},
\]
and
\[
    A\left|\frac{d}{dt}
    \mathrm{sinc}^2\!\left(\frac{\Omega(t-R)}{2\pi}\right)\right|
    <\frac{\epsilon}{4(b-a)}.
\]
This is possible because the shifted Fej\'er bump and its derivative converge uniformly to zero on the fixed compact interval $I$ as $|R|\to\infty$. Define
\[
    x_{I,A}(t)=A\,\mathrm{sinc}^2\!\left(\frac{\Omega(t-R)}{2\pi}\right).
\]
Then $x_{I,A}\in PW_\Omega$ and $\|x_{I,A}\|_\infty=A$, but on $I$ the signal remains in the single threshold cell $[0,\Delta V)$ and has total variation less than $\epsilon/4<\Delta V$. Possible contacts with the zero level are tangential local minima and are not retained adjacent-transition crossings. Hence $\Lambda_{\Delta V}(x_{I,A})\cap I=\emptyset$. Since $I$ contains a centered Kadec window whenever $b-a\ge \Delta T/2$ after a suitable grid offset, the event record can miss an entire Nyquist neighborhood even though the same bandlimited signal may have arbitrarily large amplitude outside $I$. \hfill $\square$

\textbf{Remark (Subthreshold non-injectivity).}
The preceding argument is a post-realization sampling obstruction. Under the same retained topology, the nonlinear finite-energy acquisition map also fails to be globally injective unless a subthreshold exclusion prior is imposed. Choose $\epsilon\in(0,\Delta V)$ and set
\begin{equation}
    x_{\mathrm{bump}}(t) = \frac{\epsilon}{4}\,\mathrm{sinc}^2\!\left(\frac{\Omega\,(t - T_{\text{obs}}/2)}{2\pi}\right).
    \label{eq:nogo_construction}
\end{equation}
The square of $\mathrm{sinc}(\Omega u/(2\pi))$ has Fourier transform equal to a triangle supported on $[-\Omega,\Omega]$ (the Fej\'er kernel), hence $x_{\mathrm{bump}} \in PW_\Omega$. Moreover $x_{\mathrm{bump}}(t) \ge 0$ with maximum $x_{\mathrm{bump}}(T_{\text{obs}}/2) = \epsilon/4$, so
\[
    0 \le x_{\mathrm{bump}}(t) \le \frac{\epsilon}{4} < \frac{\Delta V}{4} < \Delta V.
\]
The adjacent threshold levels around the range of $x_{\mathrm{bump}}$ are $m=0$ (at the value $0$) and $m=1$ (at the value $\Delta V$). Since $\|x_{\mathrm{bump}}\|_\infty < \Delta V$, the upper level is never reached. The lower level $0$ is touched only at the isolated zeros $t = T_{\text{obs}}/2 + 2k\pi/\Omega$ ($k \neq 0$). At each such point, $x_{\mathrm{bump}}$ attains a local minimum and then returns to positive values without descending to the next lower threshold $-\Delta V$. By the monotone-crossing convention of~\eqref{eq:mvt_set}, which excludes same-threshold return crossings, these tangential zeros are not recorded events. Hence $\Lambda_{\Delta V}(x_{\mathrm{bump}}) = \emptyset$. The zero signal has the same empty retained record, so exact reconstruction over the whole finite-energy class is impossible even before stability is considered.

\textbf{Remark 1 (What the local construction does and does not prove).}
The local construction is a certificate-level obstruction: it shows that a prescribed Nyquist neighborhood may receive no event. A single finite empty window does not, by itself, prove that an otherwise redundant bi-infinite point set fails to be a frame. To turn local silence into global sampling failure one needs the finite-energy tail argument in statement~(i), or an infinite sequence of such gaps causing the relevant lower-density or relative-density condition to fail. The no-go mechanism therefore depends on the absence of any positive lower bound on short-window voltage excursion, not on amplitude scaling alone.

\textbf{Remark 2 (Density perspective).}
The density analysis is fully consistent with this theorem. On a near-quiescent interval of duration at least one Nyquist period, the mean absolute velocity becomes arbitrarily small, hence $\nu_{min} \approx 0$ and the long-window design rule~\eqref{eq:ultimate_ns} becomes impossible to satisfy with any fixed positive $\Delta V$.

\textbf{Remark 3 (Physical interpretation).}
The obstruction arises because sufficiently weak motion is observationally indistinguishable from no motion for a static threshold set, even though exact constancy is nongeneric. If the waveform remains within a narrow voltage strip for long enough, the comparator bank has no new threshold to trigger. Static MT sampling therefore requires auxiliary prior information, padding, periodic extension, finite-window restrictions, or a restricted signal class with a strictly positive lower bound on short-window variation. In practice, many relevant signals---audio, electrocardiograms, radar returns with inter-pulse gaps---may contain such low-activity regions, which helps delineate the regime in which static MT sampling can be certified.

\section{Conditional Hardware Perturbation Bounds}
\label{sec:hardware_error}
The preceding analysis assumes ideal timing, exact threshold levels, and a noiseless input. Real MT hardware introduces three classes of small perturbations: (i)~timing jitter from finite TDC resolution, (ii)~threshold imprecision from comparator offset and DAC nonlinearity, and (iii)~additive input noise. This section develops a conditional error theory for the regime in which those impairments perturb already-existing crossings without changing the event topology. The resulting inequalities specify how much Kadec margin is consumed by timing and voltage errors. They do not, by themselves, model spurious crossings, missed crossings, or threshold-index errors; those topology-changing events are separated explicitly in Subsection~\ref{sec:event_topology}.

\subsection{Impairment Model}
\label{sec:impairment_model}
We model the three impairments as deterministic worst-case perturbations (stochastic models can be substituted for average-case analysis).

\textbf{Timing jitter.} The TDC reports each crossing time with bounded error:
\begin{equation}
    \hat{t}_k = t_k + \epsilon_k^{(j)}, \qquad |\epsilon_k^{(j)}| \le \sigma_j,
    \label{eq:jitter_model}
\end{equation}
where $\sigma_j > 0$ is the worst-case TDC jitter (half the least significant bit of the TDC).

\textbf{Input noise.} The comparator bank sees $\tilde{x}(t) = x(t) + w(t)$ instead of the clean signal $x$, where $\|w\|_\infty \le \sigma_n$. This amplitude bound alone gives only a leading displacement scale; it does not guarantee that the event topology persists. The rigorous correspondence statement below uses the stronger local regularity hypothesis $e\in C^1$ with a small derivative on each certified crossing interval. In the event-correspondence regime, the noise waveform perturbs the time of a clean crossing satisfying $x(t_k) = \hat{V}_{m_k}$ to a nearby noisy crossing $\tilde{t}_k$ satisfying $x(\tilde{t}_k) + w(\tilde{t}_k) = \hat{V}_{m_k}$, where $\hat{V}_{m_k}$ is the threshold level reported to the reconstruction algorithm after calibration. By the implicit function theorem (on a monotone segment with $|x'| \ge v > 0$ and sufficiently small perturbation):
\begin{equation}
    |\tilde{t}_k - t_k| \le \frac{\sigma_n}{v_k} + \mathcal{O}\!\left(\frac{\|x''\|_\infty\sigma_n^2}{v_k^3}\right),
    \label{eq:noise_time_shift}
\end{equation}
where $v_k = |x'(\zeta_k)|$ for some $\zeta_k$ between $t_k$ and $\tilde{t}_k$.

\textbf{Threshold imprecision and calibration.} Each physical comparator threshold deviates from its nominal level:
\begin{equation}
    V_m^{(\mathrm{actual})} = m\Delta V + \bar{q}_m + r_m,
    \qquad |r_m| \le \sigma_q,
    \label{eq:threshold_imprecision}
\end{equation}
where $\bar{q}_m$ is the calibrated offset reported by a per-threshold calibration table and $r_m$ is the unreported residual offset. The bound $\sigma_q$ accounts for residual DAC integral nonlinearity, comparator offset voltage, and reference drift after calibration. In an uncalibrated implementation one sets $\bar{q}_m=0$ and interprets $r_m$ as the full threshold offset; the same formula applies, but $\sigma_q$ must then cover the full offset range rather than the post-calibration residual.

The observed MT data is therefore a set of timestamps and threshold indices $\{(\hat{t}_k,m_k)\}$ together with an optional calibrated threshold table $\{\bar{q}_m\}$. The value passed to reconstruction is
\begin{equation}
    \hat{y}_k = \hat{V}_{m_k}
    \triangleq m_k\Delta V+\bar{q}_{m_k}.
    \label{eq:reported_threshold_value}
\end{equation}
The residual $r_{m_k}$ is not reported event-by-event; it enters only through the error budget.

\textbf{Event-correspondence assumption.} The perturbative theory below assumes that, after any analog hysteresis or digital retained-transition filtering, the noisy event stream can be matched one-to-one with the clean selected subsequence:
\begin{enumerate}
    \item[\textup{(H1)}] each clean selected crossing $s_n$ has exactly one reported crossing $\hat{s}_n$ in a neighborhood of radius smaller than the available Kadec margin;
    \item[\textup{(H2)}] no additional reported crossing inside the active windows is used by the reconstruction algorithm;
    \item[\textup{(H3)}] the threshold index attached to $\hat{s}_n$ is the same index as the clean crossing, and the reconstruction algorithm uses the calibrated table value $\bar{q}_m$ when such a table is available.
\end{enumerate}
These hypotheses are checkable engineering conditions, not consequences of the bounded-amplitude noise model. They hold, for example, when comparator hysteresis, blanking time, and threshold-index consistency filters remove chatter. They also require threshold offsets to be either calibrated or bounded as residuals, and the clean crossing slope to be large enough that a unique implicit root persists.

\begin{proposition}[A sufficient isolation condition for event correspondence]
\label{prop:event_correspondence_isolation}
\textit{Let $s_n$ be a clean retained crossing of threshold $V_{m_n}$. Suppose that
there exist disjoint intervals $I_n=[s_n-r_n,s_n+r_n]$ such that, on each $I_n$,
$x(t)-V_{m_n}$ is strictly monotone and}
\begin{equation}
    |x'(t)|\ge v_n>0,
    \qquad t\in I_n .
    \label{eq:local_crossing_slope_margin}
\end{equation}
\textit{Assume further that outside the union of these intervals the clean waveform
has vertical separation}
\begin{equation}
    \operatorname{dist}\bigl(x(t),\Delta V\mathbb Z\bigr)
    \ge m_{\mathrm{out}}>0.
    \label{eq:outside_threshold_margin}
\end{equation}
\textit{Let the observed waveform be $x(t)+e(t)$ and let the residual threshold error
be bounded by $|\theta_m|\le \sigma_q$. Suppose that $e$ is continuously
differentiable on each $I_n$ and}
\[
    \|e\|_{\infty}+\sigma_q
    <
    \frac12 m_{\mathrm{out}},
    \qquad
    \|e\|_{\infty}+\sigma_q
    <
    \frac12 v_n r_n,
\]
\textit{and}
\begin{equation}
    \|e'\|_{L^\infty(I_n)}
    \le
    \frac12 v_n,
    \qquad \text{for all }n .
    \label{eq:noise_derivative_margin}
\end{equation}
\textit{Then each clean selected crossing persists as a unique perturbed crossing in
$I_n$, no additional threshold crossing is created outside the union of the $I_n$,
and}
\begin{equation}
    |\tilde s_n-s_n|
    \le
    \frac{\|e\|_{\infty}+\sigma_q}
    {v_n-\|e'\|_{L^\infty(I_n)}}
    \le
    \frac{2(\|e\|_{\infty}+\sigma_q)}{v_n}.
    \label{eq:isolated_crossing_shift_bound}
\end{equation}
\end{proposition}

\textit{Proof.}
On $I_n$, the derivative of the perturbed crossing function
$x(t)+e(t)-V_{m_n}-\theta_{m_n}$ has the same sign as $x'(t)$, because
$\|e'\|_{L^\infty(I_n)}\le v_n/2$. Hence the perturbed crossing function is strictly
monotone on $I_n$. At the endpoints $s_n\pm r_n$, the clean vertical distance from
the threshold is at least $v_n r_n$. The perturbation magnitude is less than
$v_n r_n/2$, so the endpoint signs remain opposite. The intermediate value theorem
gives existence of one perturbed crossing, and strict monotonicity gives uniqueness.
The inverse-function estimate yields the first inequality in
\eqref{eq:isolated_crossing_shift_bound}; the second follows from
\eqref{eq:noise_derivative_margin}. Outside the union of the $I_n$, the clean
vertical threshold distance is at least $m_{\mathrm{out}}$, while the perturbation is
less than $m_{\mathrm{out}}/2$, so no additional threshold crossing can occur there.
\hfill $\square$

Without such an isolation or topology certificate, spurious, missed, or mislabeled
crossings are not small perturbations of the sampling set; they are changes of the
event topology and are outside the conditional perturbation theory of this paper.

\subsection{Topology-Changing Noise Events}
\label{sec:event_topology}
When (H1)--(H3) fail, the sampling set itself changes. A useful non-perturbative model decomposes the reported event set into matched, missed, and spurious components:
\begin{align*}
    \widehat{\Lambda}
    &=
    \mathcal{M}(\Lambda_{\Delta V})\;\dot{\cup}\;\Lambda_{\mathrm{spur}},\\
    \Lambda_{\mathrm{miss}}
    &=
    \Lambda_{\Delta V}\setminus
    \mathcal{M}^{-1}(\mathcal{M}(\Lambda_{\Delta V})).
\end{align*}
Here $\mathcal{M}$ is a partial matching from clean to reported crossings, $\Lambda_{\mathrm{spur}}$ are noise-driven events with no clean counterpart, and $\Lambda_{\mathrm{miss}}$ are clean events that were not reported. Stable reconstruction then requires a different certificate: after filtering or robust fitting, the retained matched set must still contain a Kadec-compatible subsequence, and the reconstruction algorithm must tolerate outliers. A deterministic sufficient condition is:
\begin{equation}
    \begin{aligned}
    &\exists\,\mathcal{J}\subseteq \mathcal{M}(\Lambda_{\Delta V})
    \ \text{such that}\\
    &\sup_{n\in\mathcal{J}}
    \frac{|\hat{s}_n-g_n|}{\Delta T}<\frac14,\\
    &|\Lambda_{\mathrm{spur}}\cap J_n|=0 .
    \end{aligned}
    \label{eq:event_topology_certificate}
\end{equation}
for the windows used in the Kadec extraction. If spurious events remain, the problem becomes an outlier-robust nonuniform sampling problem rather than a small-perturbation problem; one must use rejection rules based on threshold-index monotonicity, refractory/blanking intervals, hysteresis, or robust estimators. For stochastic front-end noise, the rate of such events should be modeled separately, for example through Rice crossing-rate formulae for Gaussian processes~\cite{Rice1944,Rice1945}. The rest of this section assumes the event-correspondence regime and therefore does not claim to bound topology-changing failures.

\subsection{Jitter- and Noise-Tolerant Kadec Condition}
\label{sec:robust_kadec}
The retained velocity-calibrated Kadec bound (Section~\ref{sec:velocity_kadec}) established that the perturbation from the Nyquist grid on monotone segments satisfies $|s_n - g_n| \le \alpha_n^{(\mathrm{ret})}\Delta V/v_n$. Each hardware impairment adds to this perturbation.

\begin{proposition}[Impairment-augmented perturbation bound under event correspondence]
\label{prop:impairment_augmented_perturbation}
On a monotone segment with local velocity $v_n > 0$, the total perturbation of the observed crossing from the Nyquist grid satisfies
\begin{equation}
    \begin{aligned}
    |\hat{s}_n - g_n|
    &\le
    \underbrace{\frac{\alpha_n^{(\mathrm{ret})}\Delta V}{v_n}}_{\text{retained MT geometry}}
    + \underbrace{\frac{\sigma_n}{v_n}}_{\text{input noise}}\\
    &\quad
    + \underbrace{\frac{\sigma_q}{v_n}}_{\text{threshold}}
    + \underbrace{\sigma_j}_{\text{TDC jitter}}.
    \end{aligned}
    \label{eq:total_perturbation}
\end{equation}
\noindent
up to second-order implicit-shift terms of size $\mathcal{O}(\|x''\|_\infty(\sigma_n+\sigma_q)^2/v_n^3)$.
\end{proposition}

\textit{Proof.} The noise-shifted crossing adds $\sigma_n/v_n$ (Eq.~\eqref{eq:noise_time_shift}). The unreported threshold residual $r_m$ shifts the physical crossing level away from the reported calibrated level by $r_m$, which on a segment with velocity $v_n$ translates to a time shift of $|r_m|/v_n \le \sigma_q/v_n$. The TDC jitter adds $\sigma_j$ directly. These perturbations accumulate (worst case, same sign). \hfill $\square$

For a fully deterministic margin check, introduce a certified remainder bound
\[
    R_{2,n}\ge
    C_n^{(2)}\frac{\|x''\|_{L^\infty(I_n)}(\sigma_n+\sigma_q)^2}{v_n^3},
\]
with $C_n^{(2)}$ supplied by the local implicit-function smallness estimate. The leading-order formula below is valid when $R_{2,n}$ is either negligible or explicitly subtracted from the Kadec margin. With this understood, the condition $|\hat{s}_n - g_n| < \Delta T/4$ becomes:
\begin{equation}
    \frac{\alpha_n^{(\mathrm{ret})}\Delta V+\sigma_n+\sigma_q}{v_n} + \sigma_j+R_{2,n} < \frac{\Delta T}{4} = \frac{1}{8f_{max}}.
    \label{eq:robust_kadec}
\end{equation}

Solving the exact inequality with $\alpha_n^{(\mathrm{ret})}>0$ gives the sharp fixed-topology value. For a denominator-safe design rule that also covers exact-anchor windows, we replace $\alpha_n^{(\mathrm{ret})}$ by the clipped envelope $\bar{\alpha}_n^{(\mathrm{ret})}$ from~\eqref{eq:retained_alpha_clipped}:
\begin{equation}
    \boxed{\begin{aligned}
    \Delta V
    &<
    \frac{1}{\bar{\alpha}_n^{(\mathrm{ret})}}
    \left(\frac{v_n}{8f_{max}} - v_n\sigma_j - v_nR_{2,n} - \sigma_n - \sigma_q\right)\\
    &\triangleq \Delta V_{\mathrm{hw}}^*(v_n,\bar{\alpha}_n^{(\mathrm{ret})},R_{2,n}).
    \end{aligned}}
    \label{eq:hw_design}
\end{equation}

Within the event-correspondence regime, this is the \emph{hardware design inequality}: the admissible threshold spacing is the retained-geometry limit $v_n/(8\bar{\alpha}_n^{(\mathrm{ret})}f_{max})$ minus the hardware and implicit-shift penalty terms. It reduces to the earlier clean-topology expression when $\bar{\alpha}_n^{(\mathrm{ret})}=1/2$ and to the retained-safe expression when $\bar{\alpha}_n^{(\mathrm{ret})}=3/2$. If a window is exactly anchored, the realized MT geometry term in~\eqref{eq:robust_kadec} is zero, but the design rule still budgets it as a clean-topology window. The inequality has a positive right-hand side (i.e., the perturbative Kadec certificate remains possible) if and only if
\begin{equation}
    \sigma_j+R_{2,n} < \frac{1}{8f_{max}} - \frac{\sigma_n + \sigma_q}{v_n}.
    \label{eq:feasibility}
\end{equation}

\textbf{Physical interpretation.} Each impairment ``eats into'' the Kadec safety margin:
\begin{itemize}
    \item \textbf{TDC jitter} $\sigma_j$ directly reduces the tolerable perturbation budget, regardless of signal velocity. A TDC with resolution $\sigma_j = \Delta T/8$ consumes half the entire Kadec margin.
    \item \textbf{Input noise} $\sigma_n$ induces a velocity-dependent time shift $\sigma_n/v_n$: slow signals suffer more because noise-induced timing errors are amplified by $1/v_n$.
    \item \textbf{Threshold imprecision} $\sigma_q$ acts identically to input noise at the comparator level: a threshold offset of $\sigma_q$ is indistinguishable from an input DC offset of the same magnitude.
    \item \textbf{Implicit-shift remainder} $R_{2,n}$ is not an independent leading-order impairment; it is the certified higher-order time-shift budget left by the local implicit-function estimate.
\end{itemize}

\subsection{Noise-Aware Frame Bounds and Condition Number}
\label{sec:noisy_frame}
Substituting the impaired perturbation~\eqref{eq:total_perturbation} into the signal-instance Kadec geometry gives a local normalized perturbation for each selected window. We separate the leading-order diagnostic from the deterministic certificate. The leading-order local quantity is
\begin{equation}
    \begin{aligned}
    L_{\mathrm{hw},n}^{(1)}
    \triangleq
    \frac{(2\alpha_n^{(\mathrm{ret})}\Delta V
    + 2\sigma_n + 2\sigma_q)f_{max}}{v_n}\\
    &\quad + 2f_{max}\sigma_j .
    \end{aligned}
    \label{eq:L_hw_local_leading}
\end{equation}
The deterministic local perturbation used for certified inequalities is
\begin{equation}
    \frac{|\hat{s}_n-g_n|}{\Delta T}
    \le
    L_{\mathrm{hw},n}^{\mathrm{det}}
    \triangleq
    L_{\mathrm{hw},n}^{(1)}+2f_{max}R_{2,n}.
    \label{eq:L_hw_local}
\end{equation}
The global leading-order and deterministic perturbation constants over the selected active windows are
\begin{equation}
    \begin{aligned}
    L_{\mathrm{hw}}^{(1)}
    \triangleq
    \sup_{n\in\mathcal{I}} L_{\mathrm{hw},n}^{(1)}
    &\le
    \frac{(2\alpha_{\mathrm{hw}}^{(\mathrm{ret})}\Delta V
    + 2\sigma_n + 2\sigma_q)f_{max}}{v_{min}}\\
    &\quad + 2f_{max}\sigma_j ,
    \end{aligned}
    \label{eq:L_hw_leading}
\end{equation}
and
\begin{equation}
    \begin{aligned}
    L_{\mathrm{hw}}^{\mathrm{det}}
    &\triangleq
    \sup_{n\in\mathcal I} L_{\mathrm{hw},n}^{\mathrm{det}}
    \le
    L_{\mathrm{hw}}^{(1)}
    +2f_{max}R_{2,\mathrm{hw}},\\
    R_{2,\mathrm{hw}}
    &\triangleq
    \sup_{n\in\mathcal I}R_{2,n}.
    \end{aligned}
    \label{eq:L_hw}
\end{equation}
where $v_{min}=\inf_{n\in\mathcal{I}}v_n$ on the selected event-correspondence windows and $\alpha_{\mathrm{hw}}^{(\mathrm{ret})}=\sup_{n\in\mathcal{I}}\alpha_n^{(\mathrm{ret})}$. For a denominator-safe design envelope, use $\bar{\alpha}_{\mathrm{hw}}^{(\mathrm{ret})}=\sup_{n\in\mathcal{I}}\bar{\alpha}_n^{(\mathrm{ret})}$ in place of $\alpha_{\mathrm{hw}}^{(\mathrm{ret})}$. Only when $L_{\mathrm{hw}}^{\mathrm{det}}<1/4$ does the auxiliary Kadec-padded hardware sequence have the impaired frame bounds
\begin{equation}
    A_{\mathrm{hw}}^{\mathrm{det}} =
    \frac{(1-\theta(L_{\mathrm{hw}}^{\mathrm{det}}))^2}{\Delta T},
    \qquad
    B_{\mathrm{hw}}^{\mathrm{det}} =
    \frac{(1+\theta(L_{\mathrm{hw}}^{\mathrm{det}}))^2}{\Delta T},
    \label{eq:hw_frame_bounds}
\end{equation}
with $\theta(L) = 1 - \cos(\pi L) + \sin(\pi L)$ as before.

These constants are available only under the deterministic condition $L_{\mathrm{hw}}^{\mathrm{det}}<1/4$. They have the same status as the Nyquist-padded certificate in Subsection~\ref{sec:finite_bridge}: they are infinite-sequence frame constants and become reconstruction constants only under an external tail prior. For finite records, the corresponding stability constant is instead the smallest singular value of the hardware-perturbed sampling matrix defined in Theorem~\ref{thm:error_budget}. The auxiliary hardware-aware condition number is
\begin{equation}
    \kappa_{\mathrm{hw}}^{\mathrm{det}}
    =
    \frac{B_{\mathrm{hw}}^{\mathrm{det}}}{A_{\mathrm{hw}}^{\mathrm{det}}}
    =
    \left(\frac{1+\theta(L_{\mathrm{hw}}^{\mathrm{det}})}
    {1-\theta(L_{\mathrm{hw}}^{\mathrm{det}})}\right)^2
    \approx 1 + 4\pi L_{\mathrm{hw}}^{\mathrm{det}}
    \label{eq:hw_condition}
\end{equation}
for $L_{\mathrm{hw}}^{\mathrm{det}} \ll 1/4$.

The leading-order quantity $L_{\mathrm{hw}}^{(1)}$ remains useful for diagnostic plots and first-order sensitivity calculations. Certified frame statements use $L_{\mathrm{hw}}^{\mathrm{det}}$, which adds the implicit-shift remainder. For instance, improving TDC resolution from $\sigma_j$ to $\sigma_j/2$ reduces $L_{\mathrm{hw}}^{(1)}$ by $f_{max}\sigma_j$ and reduces the leading part of $\kappa_{\mathrm{hw}}^{\mathrm{det}}-1$ by approximately $4\pi f_{max}\sigma_j$, while the certified value also depends on the separately bounded $R_{2,\mathrm{hw}}$.

\begin{proposition}[Geometry-margin reconstruction error under matched timing perturbations]
\label{prop:geometry_margin_reconstruction_error}
\textit{Let $S=\{s_n\}$ be a clean selected MT subsequence paired with grid anchors
$g_n=\tau_0+n\Delta T$ and suppose}
\[
    |s_n-g_n|
    \le
    \left(\frac14-\rho\right)\Delta T,
    \qquad
    0<\rho<\frac14 .
\]
\textit{Assume event correspondence and let $\widehat S=\{\hat s_n\}$ satisfy}
\[
    \varepsilon_t
    \triangleq
    \sup_n|\hat s_n-s_n|
    <
    \rho\Delta T .
\]
\textit{Then $\widehat S$ remains Kadec-compatible. If $\widehat S$ has lower frame
bound $A_{\widehat S}$ and the reconstruction data are
$\hat y_n=x(\hat s_n)+e_n$, then the canonical frame reconstruction from
$(\hat s_n,\hat y_n)$ obeys}
\begin{equation}
    \|x-\widehat x\|_{L^2}
    \le
    A_{\widehat S}^{-1/2}
    \|(e_n)\|_{\ell^2}.
    \label{eq:hardware_reconstruction_error}
\end{equation}
\end{proposition}

\textit{Proof.}
The triangle inequality gives
\[
    |\hat s_n-g_n|
    \le
    |s_n-g_n|+|\hat s_n-s_n|
    <
    \Delta T/4,
\]
so Kadec's theorem gives a sampling frame for $\widehat S$ with lower bound
$A_{\widehat S}>0$. The reconstruction error is the synthesis of the sample-error
sequence through the canonical dual frame of $\widehat S$, whose operator norm is at
most $A_{\widehat S}^{-1/2}$. This proves~\eqref{eq:hardware_reconstruction_error}.
\hfill $\square$

The timing budget in Proposition~\ref{prop:geometry_margin_reconstruction_error} is
obtained from the hardware parameters by the leading-order bound
\begin{equation}
    \varepsilon_t
    \le
    \sigma_j
    +
    \sup_n\frac{\sigma_n+\sigma_q}{v_n}
    +
    \text{higher-order curvature terms}.
    \label{eq:hardware_timing_budget}
\end{equation}
Thus the fixed-topology hardware design rule is
\begin{equation}
    \sigma_j
    +
    \sup_n\frac{\sigma_n+\sigma_q}{v_n}
    +
    \text{higher-order terms}
    <
    \rho\Delta T.
    \label{eq:hardware_margin_rule}
\end{equation}

\subsection{Reconstruction Error Budget}
\label{sec:error_budget}
We now derive the end-to-end reconstruction error from the impaired sample data $\{(\hat{t}_k, m_k, \bar{q}_{m_k})\}$, with the convention $\bar{q}_{m_k}=0$ in the uncalibrated mode.

The reconstruction algorithm receives sample pairs $(\hat{t}_k, \hat{y}_k)$ where $\hat{y}_k = m_k\Delta V + \bar{q}_{m_k}$ is the reported crossing level from~\eqref{eq:reported_threshold_value}. Let $r_{m_k}=V_{m_k}^{(\mathrm{actual})}-\hat{y}_k$ denote the unreported residual threshold offset. At the observed position $\hat{t}_k$, the true signal value is
\begin{align}
    x(\hat{t}_k) &= x(t_k) + x'(t_k)(\hat{t}_k - t_k) + \mathcal{O}((\hat{t}_k-t_k)^2) \notag \\
    &= \hat{y}_k + r_{m_k}
    + x'(t_k)(\hat{t}_k - t_k) + \mathcal{O}((\hat{t}_k-t_k)^2).
    \label{eq:value_at_observed}
\end{align}
Hence the per-sample value error is
\begin{equation}
    e_k \triangleq \hat{y}_k - x(\hat{t}_k)
    = -r_{m_k} - x'(t_k)(\hat{t}_k - t_k)
    + \mathcal{O}((\hat{t}_k-t_k)^2).
    \label{eq:sample_error}
\end{equation}
If $|\hat t_k-t_k|\le\varepsilon_{t,k}$ and $\|x''\|_{L^\infty([t_k,\hat t_k])}\le M_{2,k}$, Taylor's theorem gives the explicit remainder bound
\begin{equation}
    |R_{v,k}|
    \le
    \frac12 M_{2,k}\varepsilon_{t,k}^2,
    \qquad
    e_k=-r_{m_k}-x'(t_k)(\hat t_k-t_k)+R_{v,k}.
    \label{eq:sample_error_remainder_bound}
\end{equation}
On the selected event-correspondence windows, write the deterministic remainder budget as
\begin{equation}
    E_R
    \triangleq
    \left(\sum_{k\in\mathcal I} R_{v,k}^2\right)^{1/2}.
    \label{eq:sample_remainder_l2}
\end{equation}

\begin{theorem}[Tail-prior and finite-dimensional error bounds under event correspondence]
\label{thm:error_budget}
\textit{Assume (H1)--(H3) and let $\mathbf{e}=(e_k)_{k\in\mathcal{I}}$ be the sample-value error vector from~\eqref{eq:sample_error}.}

\textit{(i) Tail-prior version. Suppose the matched hardware sequence satisfies the deterministic global Kadec condition $L_{\mathrm{hw}}^{\mathrm{det}}<1/4$, so that the auxiliary impaired frame lower bound $A_{\mathrm{hw}}^{\mathrm{det}}>0$ from~\eqref{eq:hw_frame_bounds} is available. If an external tail prior supplies the same finite-window qualification as in Subsection~\ref{sec:finite_bridge}, then the tail-conditional reconstruction $x_{\mathrm{tail}}$ obeys}
\begin{equation}
    \|x - x_{\mathrm{tail}}\|_{L^2}
    \le
    \frac{\|\mathbf{e}\|_{\ell^2(\mathcal{I})}}{\sqrt{A_{\mathrm{hw}}^{\mathrm{det}}}}
    + \mathcal{O}\!\left((E_{\mathrm{tail}}^{\mathrm{collar}})^{1/2}\right)
    + \mathcal{O}(S_{\mathrm{bdry}}^{1/2}).
    \label{eq:recon_error}
\end{equation}

\textit{(ii) Finite-dimensional version. Let $V_M=\operatorname{span}\{\phi_1,\ldots,\phi_M\}\subset PW_\Omega$ be the declared reconstruction space from Subsection~\ref{sec:finite_bridge}. Define the hardware-perturbed sampling matrix}
\begin{equation}
    H_{\mathrm{hw},M}
    =
    \bigl(\phi_m(\hat{s}_n)\bigr)_{n\in\mathcal{I},\,1\le m\le M}.
    \label{eq:finite_hw_matrix}
\end{equation}
\textit{If $\sigma_{\min}(H_{\mathrm{hw},M})>0$, then for $x=p+r$ with $p=P_{V_M}x$, least-squares reconstruction $\hat{p}$ on $V_M$ satisfies}
\begin{equation}
    \begin{aligned}
    \|x-\hat{p}\|_{L^2}
    \le{}& \|r\|_{L^2}
    +\sigma_{\min}(H_{\mathrm{hw},M})^{-1} \\
    &\times
    \left(\|\mathbf{e}\|_{\ell^2(\mathcal{I})}
    + \|(r(\hat{s}_n))_{n\in\mathcal{I}}\|_{\ell^2}\right).
    \end{aligned}
    \label{eq:finite_hw_recon_error}
\end{equation}
\textit{In both versions, the per-sample error $e_k$ from~\eqref{eq:sample_error} satisfies}
\begin{equation}
    \begin{aligned}
    |e_k| \le{}& \sigma_q
    + \|x'\|_\infty(\sigma_j + \sigma_n/v_k + \sigma_q/v_k) \\
    &+\frac12 M_{2,k}(\sigma_j+\sigma_n/v_k+\sigma_q/v_k)^2,
    \end{aligned}
    \label{eq:per_sample_bound}
\end{equation}
\textit{On a monotone segment with $N_{\mathrm{in}}$ samples and minimum velocity $v_{min}$:}
\begin{multline}
    \sqrt{\sum_k |e_k|^2} \le \sqrt{N_{\mathrm{in}}}\bigl(\sigma_q + 2\pi f_{max}\|x\|_\infty \sigma_j \\
    + \tfrac{2\pi f_{max}\|x\|_\infty}{v_{min}}(\sigma_n + \sigma_q)\bigr)
    +E_R.
    \label{eq:total_sample_error}
\end{multline}
The first term is the leading-order contribution; $E_R$ is the deterministic $\ell^2$ budget for the quadratic Taylor remainders.
\end{theorem}

\textit{Proof.} From~\eqref{eq:sample_error_remainder_bound},
$e_k$ is the sum of the reported threshold residual, the first-order timing term, and the Taylor remainder $R_{v,k}$. The time perturbation relative to the nominal selected crossing is bounded conservatively by $|\hat{t}_k - t_k| \le \sigma_j + \sigma_n/v_k + \sigma_q/v_k$: the last term accounts for the displacement of the physical crossing induced by the unreported threshold residual. Using $|x'(t_k)| \le \|x'\|_\infty \le 2\pi f_{max}\|x\|_\infty$ (Bernstein) for the jitter term, and bounding the velocity-dependent terms by replacing $v_k \ge v_{min}$, gives the displayed leading per-sample bound. Applying the $\ell^2$ triangle inequality to the leading error vector and the remainder vector, whose norm is $E_R$ by~\eqref{eq:sample_remainder_l2}, yields~\eqref{eq:total_sample_error}.

For (i), $L_{\mathrm{hw}}^{\mathrm{det}}<1/4$ gives the auxiliary frame lower bound $A_{\mathrm{hw}}^{\mathrm{det}}$ for the padded impaired sequence. The pseudo-inverse of a frame with lower bound $A_{\mathrm{hw}}^{\mathrm{det}}$ has operator norm $1/\sqrt{A_{\mathrm{hw}}^{\mathrm{det}}}$, and the finite-window tail contribution is exactly the external qualification already isolated in~\eqref{eq:effective_lower}. This gives~\eqref{eq:recon_error}. For (ii), the finite record is the linear system $H_{\mathrm{hw},M}\mathbf{a}=\mathbf{y}+\mathbf{e}$ on $V_M$. If $\sigma_{\min}(H_{\mathrm{hw},M})>0$, the least-squares pseudo-inverse has norm $1/\sigma_{\min}(H_{\mathrm{hw},M})$, giving the sample-error term in~\eqref{eq:finite_hw_recon_error}. The two residual terms are the model truncation error $\|r\|_{L^2}$ and the mismatch between the true signal and the model samples at the observed hardware times. \hfill $\square$

\begin{proposition}[Local finite-matrix perturbation certificate]
\label{prop:finite_matrix_perturb}
\textit{Let $H_M=(\phi_m(s_n))_{n\in\mathcal{I},\,1\le m\le M}$ be the clean finite sampling matrix on $V_M$, assume $H_M$ has full column rank, and let $H_{\mathrm{hw},M}=H_M+\delta H_M$ be the matrix built from the reported event times. Suppose $p=\sum_m a_m\phi_m\in V_M$ and the reported data, expressed relative to the clean finite model, satisfy}
\[
    \hat{\mathbf{y}} = H_M\mathbf{a}+\boldsymbol{\eta},
\]
\textit{where $\boldsymbol{\eta}$ collects threshold-value residuals, input-noise value residuals, and any retained model-sample mismatch. The reconstruction algorithm solves the perturbed least-squares system using $H_{\mathrm{hw},M}$. If}
\begin{equation}
    \mu_M \triangleq \|H_M^\dagger\|_2\,\|\delta H_M\|_2 < 1,
    \label{eq:matrix_perturb_condition}
\end{equation}
\textit{then the least-squares coefficient estimate $\hat{\mathbf{a}}=H_{\mathrm{hw},M}^\dagger\hat{\mathbf{y}}$ satisfies the local deterministic bound}
\begin{equation}
    \|\hat{\mathbf{a}}-\mathbf{a}\|_2
    \le
    \frac{\|H_M^\dagger\|_2}{1-\mu_M}
    \left(
        \|\boldsymbol{\eta}\|_2+\|\delta H_M\|_2\|\mathbf{a}\|_2
    \right).
    \label{eq:matrix_perturb_bound}
\end{equation}
\end{proposition}

\textit{Proof.}
By Weyl's singular-value perturbation inequality,
\[
    \sigma_{\min}(H_{\mathrm{hw},M})
    \ge
    \sigma_{\min}(H_M)-\|\delta H_M\|_2
    =
    \frac{1-\mu_M}{\|H_M^\dagger\|_2}.
\]
Thus $H_{\mathrm{hw},M}$ has full column rank and
\[
    \|H_{\mathrm{hw},M}^\dagger\|_2
    =
    \frac{1}{\sigma_{\min}(H_{\mathrm{hw},M})}
    \le
    \frac{\|H_M^\dagger\|_2}{1-\mu_M}.
\]
Since $\hat{\mathbf{y}}=H_M\mathbf{a}+\boldsymbol{\eta}
=H_{\mathrm{hw},M}\mathbf{a}+(\boldsymbol{\eta}-\delta H_M\mathbf{a})$ and $H_{\mathrm{hw},M}^\dagger H_{\mathrm{hw},M}=I$ on the full-column-rank model space,
\[
    \hat{\mathbf{a}}-\mathbf{a}
    =
    H_{\mathrm{hw},M}^\dagger
    (\boldsymbol{\eta}-\delta H_M\mathbf{a}).
\]
Taking norms and using the preceding bound on $\|H_{\mathrm{hw},M}^\dagger\|_2$ gives~\eqref{eq:matrix_perturb_bound}. \hfill $\square$

Proposition~\ref{prop:finite_matrix_perturb} is often much tighter than the Kadec-style worst-case constants because it uses the actual finite matrix and the observed perturbation geometry. It is not a replacement for the global event-geometry certificate: it is a finite-record acceptance test once a candidate event set and model dimension have been chosen.

\subsection{Explicit Hardware Specifications from Target SNR}
\label{sec:hw_specs}
The error budget~\eqref{eq:recon_error}--\eqref{eq:total_sample_error} enables a direct translation from a reconstruction performance target to hardware specifications, but the target norm must match the reconstruction model. Define the full-signal reconstruction signal-to-noise ratio:
\begin{equation}
    \mathrm{SNR}_{\mathrm{recon}} = \frac{\|x\|_{L^2}^2}{\|x - x_{\mathrm{rec}}\|_{L^2}^2}.
    \label{eq:snr_def}
\end{equation}
For finite-dimensional least squares on $V_M$, the hardware-only part is instead naturally measured by the model-space ratio
\begin{equation}
    \mathrm{SNR}_{M}
    =
    \frac{\|p\|_{L^2}^2}{\|p-\hat{p}\|_{L^2}^2},
    \qquad p=P_{V_M}x.
    \label{eq:model_snr_def}
\end{equation}
The two coincide only when the model residual is negligible or is explicitly budgeted.

Define
\[
    C=2\pi f_{max}\|x\|_\infty,
    \qquad
    \beta = \frac{C}{v_{min}},
\]
where $\beta$ is the velocity amplification factor governing the noise sensitivity of slow signal segments.

\begin{corollary}[Deterministic hardware specification rules with leading-order split]
\label{cor:hw_specs}
Under the event-correspondence assumptions, let $S_{\mathrm{tar}}$ denote the target SNR in linear scale and use the deterministic equal worst-case split in which the three leading-order contributions---jitter, input noise, and threshold residual---share the perturbative error amplitude after the Taylor remainder budget $E_R$ has been reserved.

\textit{(a) Tail-prior full-signal SNR.} For the tail-prior version of Theorem~\ref{thm:error_budget}, let $E_{\mathrm{tc}}\ge0$ denote the externally assigned tail/collar reconstruction-error budget from~\eqref{eq:recon_error}; set $E_{\mathrm{tc}}=0$ only when those terms are handled outside the displayed hardware budget. Define
\begin{equation}
    \Gamma_A^{\mathrm{det}}
    =
    \sqrt{A_{\mathrm{hw}}^{\mathrm{det}}}
    \left(\frac{\|x\|_{L^2}}{\sqrt{S_{\mathrm{tar}}}}-E_{\mathrm{tc}}\right)
    -E_R,
    \qquad
    \gamma_A^{\mathrm{det}}=
    \frac{\Gamma_A^{\mathrm{det}}}{3\sqrt{N_{\mathrm{in}}}}.
    \label{eq:gamma_A_det}
\end{equation}
If $\Gamma_A^{\mathrm{det}}>0$, the following simultaneous sufficient conditions imply $\mathrm{SNR}_{\mathrm{recon}}\ge S_{\mathrm{tar}}$:
\begin{align}
    \text{TDC:} &\quad \sigma_j \le \frac{\gamma_A^{\mathrm{det}}}{C}, \label{eq:tdc_spec}\\[4pt]
    \text{Noise:} &\quad \sigma_n \le \frac{\gamma_A^{\mathrm{det}}}{\beta}, \label{eq:noise_spec}\\[4pt]
    \text{DAC:} &\quad \sigma_q \le \frac{\gamma_A^{\mathrm{det}}}{1+\beta}. \label{eq:dac_spec}
\end{align}
If the hardware perturbation is kept in the stronger interior Kadec regime
\begin{equation}
    \theta(L_{\mathrm{hw}}^{\mathrm{det}})\le 1-\frac{1}{\sqrt{2}},
    \label{eq:strong_hw_margin}
\end{equation}
then $A_{\mathrm{hw}}^{\mathrm{det}}\ge 1/(2\Delta T)$. A simpler explicit sufficient split is obtained by replacing $\gamma_A^{\mathrm{det}}$ with
\begin{equation}
    \gamma_{A,\mathrm{simp}}^{\mathrm{det}}
    =
    \frac{1}{3\sqrt{N_{\mathrm{in}}}}
    \left(
    \frac{\|x\|_{L^2}/\sqrt{S_{\mathrm{tar}}}-E_{\mathrm{tc}}}{\sqrt{2\Delta T}}
    -E_R
    \right),
    \label{eq:gamma_A_simple_det}
\end{equation}
provided the right-hand side is positive:
\begin{align}
    \sigma_j &\le \frac{\gamma_{A,\mathrm{simp}}^{\mathrm{det}}}{C}, \label{eq:tdc_spec_simple}\\
    \sigma_n &\le \frac{\gamma_{A,\mathrm{simp}}^{\mathrm{det}}}{\beta}, \label{eq:noise_spec_simple}\\
    \sigma_q &\le \frac{\gamma_{A,\mathrm{simp}}^{\mathrm{det}}}{1+\beta}. \label{eq:dac_spec_simple}
\end{align}
The simplified version is not implied by the bare deterministic Kadec condition $L_{\mathrm{hw}}^{\mathrm{det}}<1/4$; near the boundary $A_{\mathrm{hw}}^{\mathrm{det}}$ tends to zero and the actual-$A_{\mathrm{hw}}^{\mathrm{det}}$ formulas above must be used.
Because $A_{\mathrm{hw}}^{\mathrm{det}}$ itself depends on the proposed impairment levels through $L_{\mathrm{hw}}^{\mathrm{det}}$, the primary tail-prior specifications are implicit design inequalities: a candidate hardware budget must first be inserted into~\eqref{eq:L_hw} and~\eqref{eq:hw_frame_bounds}, and then checked against~\eqref{eq:tdc_spec}--\eqref{eq:dac_spec}. The stronger-margin simplification~\eqref{eq:strong_hw_margin} is a convenient explicit sufficient regime.

\textit{(b) Finite-dimensional model-space SNR.} For reconstruction on $V_M$, let $s_M=\sigma_{\min}(H_{\mathrm{hw},M})$. If the declared model residual is excluded from the hardware budget, the following simultaneous sufficient conditions imply $\mathrm{SNR}_{M}\ge S_{\mathrm{tar}}$:
\begin{equation}
    \Gamma_M^{\mathrm{det}}
    =
    \frac{s_M\|p\|_{L^2}}{\sqrt{S_{\mathrm{tar}}}}
    -E_R,
    \qquad
    \gamma_M^{\mathrm{det}}
    =
    \frac{\Gamma_M^{\mathrm{det}}}{3\sqrt{N_{\mathrm{in}}}}.
    \label{eq:gamma_M_det}
\end{equation}
Assume $\Gamma_M^{\mathrm{det}}>0$. Then
\begin{align}
    \text{TDC:} &\quad
    \sigma_j \le
    \frac{\gamma_M^{\mathrm{det}}}{C},
    \label{eq:tdc_spec_finite}\\[4pt]
    \text{Noise:} &\quad
    \sigma_n \le
    \frac{\gamma_M^{\mathrm{det}}}{\beta},
    \label{eq:noise_spec_finite}\\[4pt]
    \text{DAC:} &\quad
    \sigma_q \le
    \frac{\gamma_M^{\mathrm{det}}}{1+\beta}.
    \label{eq:dac_spec_finite}
\end{align}
These are model-space hardware specifications; the corresponding full-signal SNR must also account for the residual component $r$ in $x=p+r$.

\textit{(c) Finite-dimensional full-signal SNR with residual budget.} Define the residual allowance
\begin{equation}
    E_M
    \triangleq
    \|r\|_{L^2}
    +
    s_M^{-1}\|(r(\hat{s}_n))_{n\in\mathcal{I}}\|_{\ell^2}.
    \label{eq:model_residual_budget}
\end{equation}
If $E_M < \|x\|_{L^2}/\sqrt{S_{\mathrm{tar}}}$, then the same finite-dimensional derivation gives a full-signal guarantee by replacing $\|p\|_{L^2}/\sqrt{S_{\mathrm{tar}}}$ with the remaining amplitude budget
\begin{equation}
    R_M(S_{\mathrm{tar}})
    =
    \frac{\|x\|_{L^2}}{\sqrt{S_{\mathrm{tar}}}} - E_M .
    \label{eq:full_signal_remaining_budget}
\end{equation}
Equivalently, in~\eqref{eq:gamma_M_det} replace $s_M\|p\|_{L^2}/\sqrt{S_{\mathrm{tar}}}$ by $s_M R_M(S_{\mathrm{tar}})$ to obtain sufficient full-signal hardware budgets. If $s_MR_M(S_{\mathrm{tar}})-E_R\le0$, no choice of TDC, front-end-noise, or threshold-residual budget can certify $\mathrm{SNR}_{\mathrm{recon}}\ge S_{\mathrm{tar}}$ within the declared reconstruction space; the model space, observation window, or deterministic remainder budget must be changed.

\textit{Derivation.} The deterministic sample-error norm~\eqref{eq:total_sample_error} consists of the three leading additive components plus the reserved Taylor remainder budget:
\begin{align*}
    &C\sigma_j,\qquad \beta\sigma_n,\qquad (1+\beta)\sigma_q,\\
    &E_R.
\end{align*}
For the tail-prior case, the amplification factor is $1/\sqrt{A_{\mathrm{hw}}^{\mathrm{det}}}$. The available sample-error budget is therefore $\sqrt{A_{\mathrm{hw}}^{\mathrm{det}}}(\|x\|_{L^2}/\sqrt{S_{\mathrm{tar}}}-E_{\mathrm{tc}})$, and $E_R$ must be reserved before the three leading components are split equally. This gives~\eqref{eq:gamma_A_det}--\eqref{eq:dac_spec}. The simplified rules~\eqref{eq:tdc_spec_simple}--\eqref{eq:dac_spec_simple} follow only after imposing~\eqref{eq:strong_hw_margin}, which gives $A_{\mathrm{hw}}^{\mathrm{det}}\ge1/(2\Delta T)$. For model-space finite reconstruction, the amplification factor is $1/s_M$ and the target norm is $\|p\|_{L^2}$, giving~\eqref{eq:gamma_M_det}--\eqref{eq:dac_spec_finite}. For full-signal finite reconstruction, Theorem~\ref{thm:error_budget} separates the error into the hardware-amplified sample error and the residual allowance $E_M$; the hardware part must fit inside the remaining budget~\eqref{eq:full_signal_remaining_budget} after $E_R$ is reserved. In the high-amplification regime $\beta\gg 1$, the DAC bound reduces asymptotically to the same scale as the noise bound. \hfill $\square$
\end{corollary}

\textbf{Remark (Cross-domain tradeoffs).} The tail-prior specifications~\eqref{eq:tdc_spec}--\eqref{eq:dac_spec} and their finite-dimensional model-space counterparts~\eqref{eq:tdc_spec_finite}--\eqref{eq:dac_spec_finite} exhibit a fundamental \emph{time-voltage duality}: timing precision ($\sigma_j$) and voltage precision ($\sigma_n$, $\sigma_q$) are independently constrained, and meeting only one is insufficient. A TDC with femtosecond resolution paired with a noisy front-end performs no better than one limited by its noise floor. In the leading-order part of the deterministic budget, this coupling is more explicit for event-driven MT than for conventional uniform-rate ADC models, where timing (clock jitter) and voltage (quantization noise) are often treated separately. For finite-dimensional full-signal claims, this tradeoff applies only after the residual allowance~\eqref{eq:model_residual_budget} and the Taylor remainder budget~\eqref{eq:sample_remainder_l2} are small enough to leave a positive hardware budget.

\subsection{Comparison with Ideal Bounds}
Let $R_{2,\mathrm{hw}}=\sup_{n\in\mathcal I}R_{2,n}$ on the selected event-correspondence windows. The hardware design inequality~\eqref{eq:hw_design} makes the \emph{penalty of non-ideal hardware} explicit and separable:
\begin{equation}
    \begin{aligned}
    \Delta V_{\mathrm{hw}}^*
    =
    \frac{1}{\bar{\alpha}_{\mathrm{hw}}^{(\mathrm{ret})}}
    \left(
    \underbrace{\frac{v_{min}}{8f_{max}}}_{\text{Kadec margin}}
    - \underbrace{v_{min}\sigma_j}_{\text{jitter penalty}}
    - \underbrace{v_{min}R_{2,\mathrm{hw}}}_{\text{implicit-shift remainder}}
    \right.\\
    \left.
    - \underbrace{\sigma_n}_{\text{noise penalty}}
    - \underbrace{\sigma_q}_{\text{DAC penalty}}
    \right).
    \end{aligned}
    \label{eq:penalty_decomposition}
\end{equation}
Each budget-consuming term has clear units and clear physical origin. The prefactor converts the available timing margin into an admissible threshold spacing under the certified retained topology: $\bar{\alpha}_{\mathrm{hw}}^{(\mathrm{ret})}=1/2$ gives the clean fixed-topology limit, while $\bar{\alpha}_{\mathrm{hw}}^{(\mathrm{ret})}=3/2$ gives the retained-safe limit. The jitter, input-noise, and threshold-residual terms enter at leading order, while $R_{2,\mathrm{hw}}$ collects certified higher-order implicit-shift remainders. The decomposition has three immediate engineering consequences:
\begin{enumerate}
    \item \textbf{Bottleneck identification}: the largest budget-consuming term identifies the dominant limitation.
    \item \textbf{Marginal return}: improving any single leading-order parameter beyond the point where its penalty falls below the others yields diminishing returns, unless it also reduces the certified remainder.
    \item \textbf{Co-design}: the threshold spacing $\Delta V$ must be jointly optimized with the hardware impairments and retained topology---choosing $\Delta V$ close to the clean-topology limit without verifying event correspondence can leave no retained-safe design margin.
\end{enumerate}

\medskip
\noindent
\rule{\columnwidth}{0.4pt}
\begingroup
\small
\noindent\textbf{Summary: Conditional Hardware Perturbation Budget.}

\medskip
\textit{Impairment model:} TDC jitter $\sigma_j$, input noise $\sigma_n$, threshold imprecision $\sigma_q$, implicit-shift remainder $R_{2,\mathrm{hw}}$, plus the event-correspondence assumptions (H1)--(H3).

\medskip
\textit{Design inequality:} $\Delta V < (\bar{\alpha}_{\mathrm{hw}}^{(\mathrm{ret})})^{-1}\!\left(v_{min}/(8f_{max})-v_{min}\sigma_j-v_{min}R_{2,\mathrm{hw}}-\sigma_n-\sigma_q\right)$.

\medskip
\textit{Tail-prior reconstruction error:} with
$C_x:=\Omega\|x\|_\infty$,
\begin{multline*}
    \|x-x_{\mathrm{rec}}\| \le
    \frac{1}{\sqrt{A_{\mathrm{hw}}^{\mathrm{det}}}}
    \Bigl[
    \sqrt{N_{\mathrm{in}}}\sigma_q
    + \sqrt{N_{\mathrm{in}}}C_x\sigma_j\\
    + \sqrt{N_{\mathrm{in}}}\frac{C_x}{v_{min}}(\sigma_n{+}\sigma_q)
    +E_R\Bigr]
    + \mathcal{O}\!\left((E_{\mathrm{tail}}^{\mathrm{collar}})^{1/2}\right)\\
    + \mathcal{O}(S_{\mathrm{bdry}}^{1/2}).
\end{multline*}
Finite-record model-space reconstruction uses the matrix amplifier $1/\sigma_{\min}(H_{\mathrm{hw},M})$ instead of the tail-prior frame amplifier. Full-signal finite-record claims additionally require the residual allowance $E_M$ in~\eqref{eq:model_residual_budget} and the Taylor remainder budget $E_R$ to leave a positive remaining hardware budget.

\medskip
\textit{Practical consequence:} Given a target $\mathrm{SNR}$ and system bandwidth $f_{max}$, Eqs.~\eqref{eq:gamma_A_det}--\eqref{eq:dac_spec} prescribe tail-prior deterministic perturbative budgets after reserving $E_R$, while Eqs.~\eqref{eq:gamma_M_det}--\eqref{eq:dac_spec_finite} give finite-dimensional model-space budgets controlled by the actual matrix conditioning. Spurious, missed, or mislabeled events require the separate topology certificate~\eqref{eq:event_topology_certificate}.
\par
\endgroup
\noindent\rule{\columnwidth}{0.4pt}

\section{Signal Reconstruction from Nonuniform Samples}
\label{sec:reconstruction}
Although MT acquisition is a nonlinear map from signal to event record, the reconstruction stage is linear once the event record has been observed. This asymmetry is central: we do not need to invert the nonlinear acquisition, we only need to invert a linear sampling operator whose domain is the realized point set $\Lambda$ and whose codomain is $\ell^2(\Lambda)$. The three algorithms described below differ in the stability certificate they require. Dual-frame and infinite-dimensional projection methods use a genuine sampling-frame lower bound; finite-record pseudo-inverse reconstruction uses the smallest singular value of the realized sampling matrix; a finite implementation of an iterative method inherits whichever of these two certificates is actually being used.

The preceding sections establish conditions under which a fixed realized set $\Lambda = \Lambda_{\Delta V}(x)$ contains a well-placed subsequence and, after the finite-record qualifications of Subsection~\ref{sec:finite_bridge}, supports stable reconstruction on a declared model space or under an external tail prior. Once the crossing times are fixed, the reconstruction stage becomes a linear inverse problem in the sample values $\{x(t_n)\}$. We now discuss three constructive methods of increasing computational practicality.

For MT, the sample values are not arbitrary amplitudes; they are known threshold levels $y_n = m_n\Delta V$ attached to the observed crossing times $t_n$. Once the pair $(t_n,y_n)$ is frozen, the nonlinearity of the acquisition stage disappears and reconstruction reduces to a linear problem over the bandlimited model space.

The assumption that $y_n$ is known is physically mild: in hardware, the comparator channel that fired identifies the threshold index $m_n$ directly. Without those threshold labels, one would be facing a different inverse problem than the one analyzed in this paper.

\subsection{Dual Frame Reconstruction}
Given that $\Lambda$ satisfies the frame condition \eqref{eq:frame}, define the reproducing kernel functions $k_n(t) = K(t, t_n)$, where
\begin{equation}
    K(t, s) = \frac{\sin[\Omega(t - s)]}{\pi(t - s)}
    \label{eq:repro_kernel}
\end{equation}
is the reproducing kernel of $PW_\Omega$ (see Appendix~\ref{app:trace} for derivation). By the reproducing property, $x(t_n) = \langle x, k_n \rangle_{L^2}$ for all $x \in PW_\Omega$, where $\langle f, g \rangle_{L^2} = \int_{-\infty}^{\infty} f(t)\,\overline{g(t)}\,dt$ denotes the $L^2$ inner product. The frame condition \eqref{eq:frame} then states precisely that $\{k_n\}$ forms a frame for $PW_\Omega$, and the associated frame operator $\mathfrak{S}_{\Lambda}: PW_\Omega \to PW_\Omega$, defined by
\begin{equation}
    (\mathfrak{S}_{\Lambda}f)(t) = \sum_n f(t_n)\, K(t, t_n)
    \label{eq:frame_operator}
\end{equation}
is bounded and invertible with $AI \preceq \mathfrak{S}_{\Lambda} \preceq BI$. The unique dual frame $\{\tilde{k}_n = \mathfrak{S}_{\Lambda}^{-1}k_n\}$ yields the exact reconstruction formula:
\begin{equation}
    x(t) = \sum_{n \in \mathbb{Z}} x(t_n) \, \tilde{k}_n(t)
    \label{eq:dual_frame}
\end{equation}

The kernel $K(t,s)$ plays the role of the bandlimited point-spread function: it is the waveform enforced by a unit sample constraint at time $s$. Dual-frame reconstruction therefore says that the full signal can be synthesized by superposing these bandlimited response shapes with coefficients given by the observed threshold data.
While mathematically exact, computing the dual frame requires inverting the infinite-dimensional frame operator $\mathfrak{S}_{\Lambda}$, which generally admits no closed-form expression for arbitrary nonuniform point sets. This full-line dual-frame formula is therefore primarily theoretical unless a tail prior or persistent extension supplies the missing outside-window values. When $\Lambda$ satisfies Beurling's sufficiency condition, the external strict-density theorem recorded in Appendix~\ref{app:beurling_strict_density_input} supplies a bounded left inverse abstractly; in practical finite records, the operative test is instead the pseudo-inverse matrix conditioning below.

\subsection{Matrix Pseudo-Inverse Method}
For practical computation over a finite observation window, $x(t)$ can be approximated as a finite superposition of Nyquist-grid sinc functions:
\begin{equation}
    x(t) \approx \sum_{m=1}^{M} a_m \, \mathrm{sinc}\!\left(\frac{t - m\Delta T}{\Delta T}\right)
    \label{eq:sinc_approx}
\end{equation}
where $\Delta T = 1/(2f_{max})$ is the Nyquist interval (as defined in Section~\ref{sec:kadec}) and $M$ is the number of Nyquist-grid points within the observation window. Given $N$ nonuniform samples $y_n = x(t_n)$ (with $N > M$ for overdetermination), substitution yields a linear system $\mathbf{y} = \mathbf{H}\mathbf{a}$, where $H_{n,m} = \mathrm{sinc}((t_n - m\Delta T)/\Delta T)$. The coefficient vector is recovered via the Moore-Penrose pseudo-inverse:
\begin{equation}
    \mathbf{a} = \mathbf{H}^{\dagger}\mathbf{y} = (\mathbf{H}^* \mathbf{H})^{-1} \mathbf{H}^* \mathbf{y}
    \label{eq:pseudoinverse}
\end{equation}
where $\mathbf{H}^*$ denotes the conjugate transpose (equal to $\mathbf{H}^T$ for real-valued data). For a finite MT record the computable stability certificate is
\begin{equation}
    s_{\mathcal{I},M}
    =
    \sigma_{\min}(\mathbf{H}_{\mathcal{I},M})>0,
    \label{eq:finite_recon_smin}
\end{equation}
not the auxiliary infinite-sequence frame constant by itself. The reconstruction error caused by sample perturbations is amplified by at most $1/s_{\mathcal{I},M}$ on the declared model space. A certified implementation should therefore use an online acceptance rule such as
\begin{equation}
    \begin{aligned}
    &N_{\mathcal{I}}\ge M,\qquad
    s_{\mathcal{I},M}\ge s_{\mathrm{floor}},\\
    &\text{or equivalently}\qquad
    \kappa(\mathbf{H}_{\mathcal{I},M})\le\kappa_{\max},
    \end{aligned}
    \label{eq:finite_acceptance}
\end{equation}
for design thresholds $s_{\mathrm{floor}}>0$ or $\kappa_{\max}<\infty$. If this test fails, the finite record is not certified at the requested model dimension; the system must extend the observation window, lower $M$, refine the threshold ladder to obtain more events, or output an uncertified reconstruction flag. When the realized finite matrix is a well-conditioned finite section of a tail-qualified frame, the frame bounds help predict $s_{\mathcal{I},M}$, but the finite-record stability constant remains the measured singular value in~\eqref{eq:finite_recon_smin}. Once $\mathbf{a}$ is recovered, the signal is equivalently represented on the uniform Nyquist grid. The computational complexity is $O(M^3)$, which limits this approach to short signal segments.

Numerically, this method replaces irregular threshold-crossing times by coefficients on a regular Shannon grid. One solves for the unique bandlimited sinc superposition whose graph passes through the observed crossing constraints in the least-squares sense.

\subsection{Iterative Projection Method}
For longer signals or real-time applications, iterative algorithms based on alternating projections onto convex sets (POCS) avoid explicit matrix inversion. Starting from an initial estimate $x_0(t)$ constructed by interpolating the nonuniform samples, the iterative update rule is:
\begin{equation}
    x_{k+1}(t) = x_k(t) + \lambda \cdot P_\Omega\!\left[\sum_n \bigl(y_n - x_k(t_n)\bigr)\,\delta(t - t_n)\right]
    \label{eq:iterative}
\end{equation}
where $P_\Omega g = \mathcal{F}^{-1}(\mathbf{1}_{[-\Omega,\Omega]}\hat{g})$ denotes the ideal low-pass projection onto $PW_\Omega$, $\delta(\cdot)$ denotes the Dirac delta distribution, and $\lambda \in (0, 2/B)$ is a relaxation parameter, with $B$ the frame upper bound.

This iterative scheme compares the current estimate with the observed threshold data, injects the residual at the sample locations as Dirac impulses, and projects the correction back into $PW_\Omega$.

At each iteration, the algorithm computes the residual error at the sampling points, distributes these errors as impulses, projects them back into $PW_\Omega$, and adds the correction to the current estimate. In a discrete implementation, this projection is realized numerically by FFT-based convolution with the ideal low-pass kernel. When the infinite sampling set satisfies the frame condition, the contraction mapping principle guarantees exponential convergence (see Appendix~\ref{app:iterative} for the full derivation):
\begin{equation}
    \|x_k - x\|_{L^2} \le \rho^k \|x_0 - x\|_{L^2}
    \label{eq:iterative_convergence}
\end{equation}
where $\rho = \max(|1 - \lambda A|, |1 - \lambda B|) < 1$. The optimal relaxation parameter $\lambda^* = 2/(A+B)$ minimizes the convergence rate~\cite{Feichtinger1992,Grochenig1992} to $\rho^* = (B-A)/(B+A)$. This choice balances the slowest and fastest observable sampling modes, so the convergence speed is governed by the conditioning ratio $B/A$: near-tight sampling geometries converge rapidly, whereas poorly conditioned geometries converge more slowly. This method requires only FFT operations and pointwise arithmetic per iteration, yielding $O(M \log M)$ per-iteration complexity. It is well suited to parallel implementation on FPGA and DSP platforms.

For a finite-window discretization of the same iteration, $A$ and $B$ should be replaced by the extremal singular-value squares of the realized finite matrix (or by validated bounds on them). Thus the same online rejection rule~\eqref{eq:finite_acceptance} applies before claiming finite-record stability.

\section{Numerical Experiments}
\label{sec:numerical}
The preceding sections develop the MT sampling theory in closed form. We now present seven computational experiments that diagnose the principal certificates on concrete finite signal instances. These experiments are not a broad empirical validation of the infinite-dimensional theory: most panels use a deliberately simple two-tone waveform so that the event geometry, threshold phase, and finite-matrix conditioning can be inspected directly. Unless otherwise noted, the experiments use the test signal
\begin{equation}
    x(t) = \sin(2\pi\!\cdot\!0.5\,t) + 0.5\,\sin(2\pi\!\cdot\!0.8\,t),
    \label{eq:test_signal}
\end{equation}
with effective bandwidth $f_{max} = 1$~Hz ($\Omega = 2\pi$~rad/s), giving the Nyquist interval $\Delta T = 0.5$~s. The signal-instance parameters computed on the active window $[0, T_{\text{obs}}]$ with $T_{\text{obs}} = 6$--$10$~s are: minimum monotone-window velocity $v_{min}^{(\mathrm{mono})} = 0.185$, minimum extremal curvature $\kappa_{min} = 8.48$, minimum extremum-isolation radius $\rho_{min}=\min_{t^*\in\mathcal{E}_x}\rho(t^*)=0.125=\Delta T/4$, and numerical third-derivative envelope $\|x'''\|_\infty \approx 94.51$. The clean fixed-topology margins are $\Delta V_{\mathrm{vel,clean}}^*=0.046$ and $\Delta V_{\mathrm{curv,clean}}^*\approx0.0355$, so the plotted clean-topology normalization is $\Delta V_{\mathrm{clean}}^*\approx0.0355$. The retained-safe margins are $\Delta V_{\mathrm{vel,ret}}^*=v_{min}^{(\mathrm{mono})}/(12f_{max})\approx0.0154$ and $\Delta V_{\mathrm{curv,ret}}^*\approx0.0178$, giving $\Delta V_{\mathrm{Kadec,ret}}^*\approx0.0154$. The leading-order clean curvature value $\kappa_{min}/(128f_{max}^2)\approx0.066$ is reported only as a non-remainder-corrected, non-isolation-stressed reference.

The role of the grid offset $\tau_0$ is reported explicitly. When $\tau_0$ is fixed before acquisition, the scalar margins in Table~\ref{tab:threshold_symbol_taxonomy} can be read as design guarantees for that offset. When $\tau_0$ is optimized offline after the event record has been generated, the result is a post-realization search for a Kadec-compatible retained subsequence and should not be interpreted as a guarantee available to the hardware before acquisition.

\begin{table*}[!t]
\caption{Numerical experiments and the certificate layer they test.}
\label{tab:numerical_layers}
\centering
\footnotesize
\begin{tabular}{|p{0.13\textwidth}|p{0.29\textwidth}|p{0.49\textwidth}|}
\hline
Experiment & Certificate layer & Interpretation \\
\hline
1--2 & Post-realization retained Kadec geometry & Offline $\tau_0$ search and measured $L_{MT}$; not a scalar design guarantee at the plotted coarse spacings. \\
\hline
3 & Long-window density diagnostic & Checks the total-variation proxy on the tested windows only; it does not replace the sliding infimum. \\
\hline
4 & Finite-record matrix reconstruction & Tests $\sigma_{\min}(H_{\mathcal{I},M})$ and ideal threshold labels on a declared finite model space. \\
\hline
5 & Static no-go/certificate failure & Illustrates empty or locally silent retained records, not a redundant-frame impossibility from one finite gap alone. \\
\hline
6 & Conditional hardware perturbation & Applies only under checked clean/noisy event correspondence; plotted constants are finite-instance diagnostics. \\
\hline
7 & Retained-suppression stress test & Isolates the same-threshold return mechanism that forces the retained-safe $3/2$ voltage-excursion factor. \\
\hline
\end{tabular}
\end{table*}

\subsection{Experiment~1: MT Crossing Set and Kadec Extraction}
\label{sec:exp1}
Fig.~\ref{fig:crossings}(a) shows the signal~\eqref{eq:test_signal} with threshold spacing $\Delta V = 0.25$ and the realized retained crossing set $\Lambda_{\Delta V}$. The retained event density is visibly nonuniform: steep segments generate closely spaced events, while near extrema the retained crossings thin out. Fig.~\ref{fig:crossings}(b) shows a Kadec-compatible retained subsequence extracted after optimizing the grid offset $\tau_0$ offline. Even at $\Delta V/\Delta V_{\mathrm{clean}}^* \approx 7.0$ (and $\Delta V/\Delta V_{\mathrm{Kadec,ret}}^*\approx16.2$), the extraction yields $L_{MT} = 0.075 < 1/4$. This does not make $\Delta V=0.25$ a certified design value; it is a post-realization retained-set measurement showing that the sufficient conditions are conservative on this particular phase and offset. The retained-phase diagnostic confirms this distinction: at $\Delta V=0.25$, the extrema have $\delta_0\in[0.0022,0.239]$ and outgoing retained gaps $\delta_{\mathrm{ret}}\in[0.252,0.489]$, far exceeding the local curvature margin $0.0355$; the scalar design certificate fails, but the realized retained extraction still happens to pass.

\begin{figure*}[!t]
    \centering
    \includegraphics[width=\textwidth]{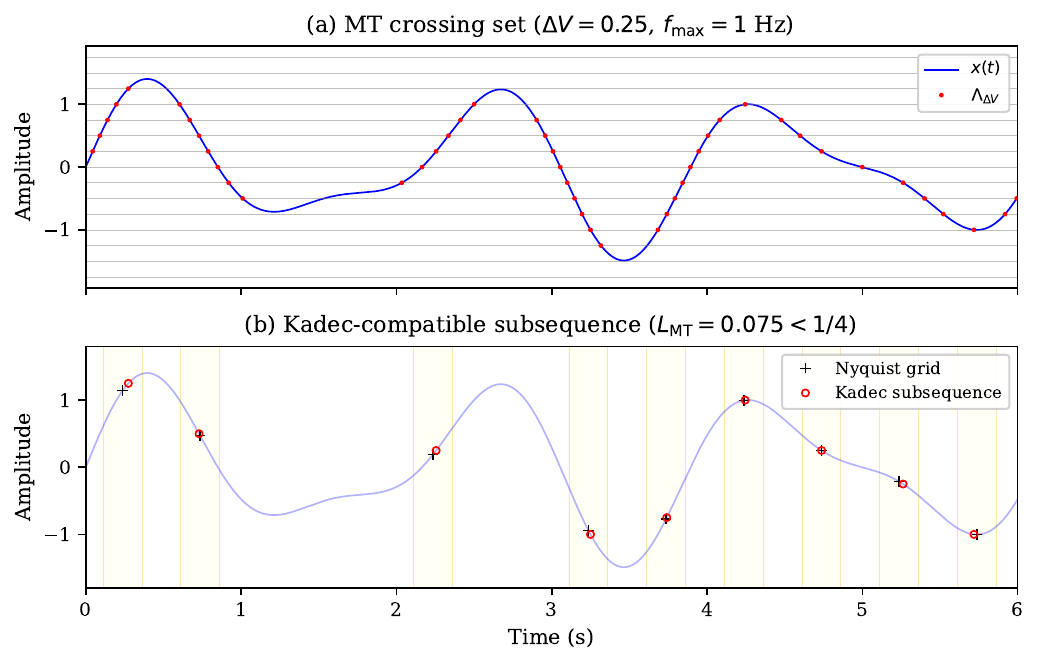}
    \caption{(a)~Signal~\eqref{eq:test_signal} with threshold ladder $\Delta V = 0.25$ and crossing set $\Lambda_{\Delta V}$ (red dots). (b)~Kadec-compatible subsequence (red circles) extracted from $\Lambda_{\Delta V}$, with an offline optimal Nyquist grid (black crosses) and centered windows $J_n(\tau_0)$ (shaded). The measured perturbation is $L_{MT} = 0.075$ on this finite instance.}
    \label{fig:crossings}
\end{figure*}

\subsection{Experiment~2: Kadec Perturbation Ratio Scaling}
\label{sec:exp2}
Fig.~\ref{fig:perturbation} plots the worst-case normalized perturbation $L_{MT}$ as a function of $\Delta V/\Delta V_{\mathrm{clean}}^*$ over the tested sweep range, together with the clean-topology theoretical upper bound $\Delta V \cdot f_{max}/v_{min}^{(\mathrm{mono})}$ from~\eqref{eq:normalized_perturbation} with $\alpha_{\mathrm{ret}}=1/2$. The retained-safe bound is three times larger and is not the curve plotted in this figure. Two observations are consistent with the local theory on this range:
\begin{enumerate}
    \item For all $\Delta V < \Delta V_{\mathrm{clean}}^*$, the numerical retained-set $L_{MT}$ remains below $1/4$, consistent with the clean fixed-topology sufficient condition of Section~\ref{sec:sufficient}.
    \item The numerical $L_{MT}$ is substantially smaller than the upper bound, consistent with the fact that the bound is a worst-case estimate over all windows, whereas the optimal grid offset $\tau_0$ avoids the worst window placement. The gap between the numerical curve and the bound quantifies the conservatism of the velocity-based Kadec condition.
\end{enumerate}

\begin{figure}[!t]
    \centering
    \includegraphics[width=\columnwidth]{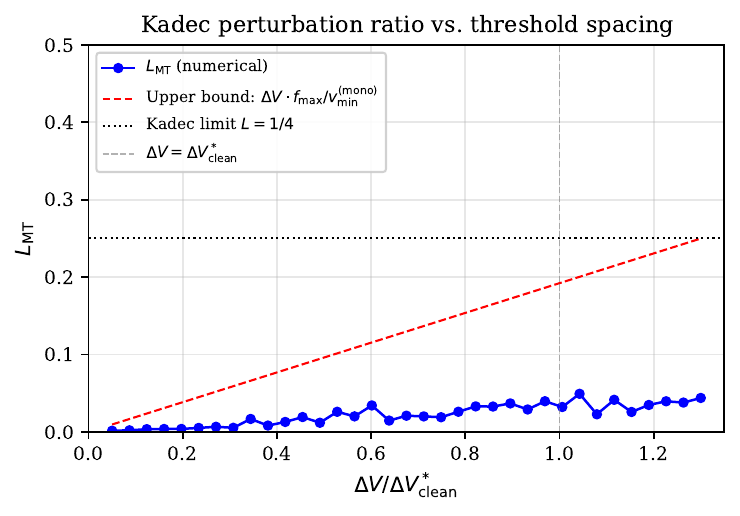}
    \caption{Worst-case Kadec perturbation ratio $L_{MT}$ versus threshold spacing, normalized by the clean-topology threshold $\Delta V_{\mathrm{clean}}^*$. Blue: numerical retained-set extraction with optimal grid offset. Red dashed: clean-topology upper bound $\Delta V f_{max}/v_{min}^{(\mathrm{mono})}$; the retained-safe upper bound would be three times larger. The numerical perturbation stays below the Kadec limit $L=1/4$ (dotted line) on this finite sweep; no wide-range safety claim is made beyond the tested spacings.}
    \label{fig:perturbation}
\end{figure}

\subsection{Experiment~3: Crossing Density Evaluation}
\label{sec:exp3}
Fig.~\ref{fig:density} compares the measured retained windowed event density $n_{\mathrm{ret}}(t,r)/r$ (blue) with the total-variation proxy $\bar{\nu}(t,r)/\Delta V$ (red dashed), using $\Delta V = 0.15$ and a sliding window of length $r = 2$~s. For this non-superoscillatory synthetic signal, the extremum count and same-threshold suppression count can be verified directly from the known waveform, so the conditional density bounds~\eqref{eq:density_window_improved} apply with instance-specific $\eta_x$ and $\eta_{\mathrm{supp},x}$. The measured retained density exceeds the Nyquist benchmark $2f_{max} = 2$~events/s on the plotted windows, reflecting the high-velocity character of the test signal after the retained-event losses are included. This figure should not be read as evidence that high average density on one window protects every subwindow; that logical step is supplied only by the explicit sliding infimum/certificate assumptions in Section~\ref{sec:density_analysis}.

\begin{figure}[!t]
    \centering
    \includegraphics[width=\columnwidth]{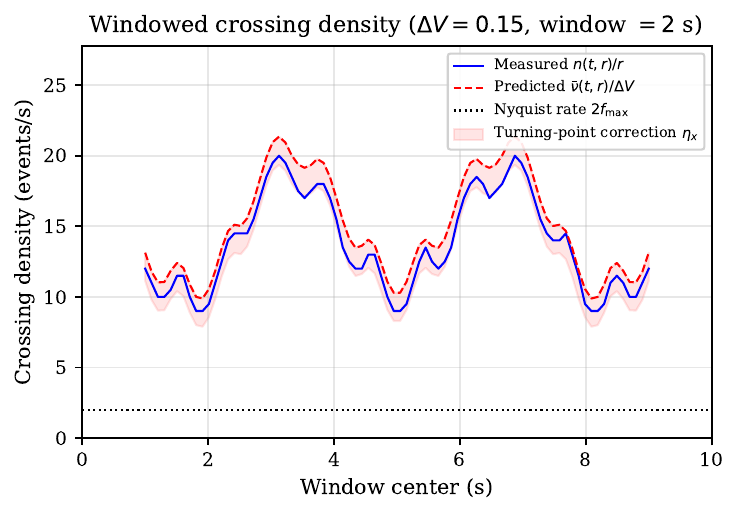}
    \caption{Windowed retained event density for $\Delta V = 0.15$ and window length $r = 2$~s. The measured retained event rate (blue) is compared with the total-variation proxy $\bar{\nu}(t,r)/\Delta V$ (red dashed) and the Nyquist benchmark $2f_{max}$ (dotted).}
    \label{fig:density}
\end{figure}

\subsection{Experiment~4: Pseudo-Inverse Reconstruction}
\label{sec:exp4}
Fig.~\ref{fig:recon}(a) plots the normalized root-mean-square error (NRMSE) of the pseudo-inverse reconstruction~\eqref{eq:pseudoinverse} as a function of $\Delta V/\Delta V_{\mathrm{clean}}^*$, using both nominal MT threshold labels $y_n = m_n\Delta V$ (blue squares) and exact signal values $x(t_n)$ at the computed crossing times (red triangles). The finite model is a truncated sinc space with $M=13$ Nyquist basis functions over $T_{\mathrm{obs}}=6$~s; crossings are fitted on the interior interval $t_n\in[\Delta T,T_{\mathrm{obs}}-\Delta T]$, and the displayed NRMSE is evaluated on $[1,T_{\mathrm{obs}}-1]$ to reduce boundary dominance. The two curves overlap because, in the ideal crossing model, the threshold label is exactly the signal value at the event time. This is therefore an ideal label-chain test, not a noisy hardware-chain reconstruction.

The residual error of approximately $3\%$ is attributable primarily to the finite observation window and sinc truncation effects. Across the plotted sweep, the finite matrix remains full rank in the usable range, but it is not a near-tight frame: representative singular values are $\sigma_{\min}(\mathbf{H})\approx3.3\times10^{-2}$ and $\mathrm{cond}(\mathbf{H})\approx1.4\times10^2$ at $\Delta V=0.10$ (about $N=106$ interior retained crossings), while $\sigma_{\min}(\mathbf{H})$ decreases to the $10^{-2}$ scale at coarser spacings. Thus the experiment verifies a finite-matrix acceptance pipeline for this instance; it does not establish sharp infinite-dimensional constants.

Fig.~\ref{fig:recon}(b) shows an example reconstructed waveform at $\Delta V = 0.10$ ($\Delta V/\Delta V_{\mathrm{clean}}^* \approx 2.8$, $N = 106$ retained crossings). The reconstruction closely tracks the original signal, with the largest discrepancies occurring near the observation-window boundaries where the sinc basis functions are truncated.

\begin{figure*}[!t]
    \centering
    \includegraphics[width=\textwidth]{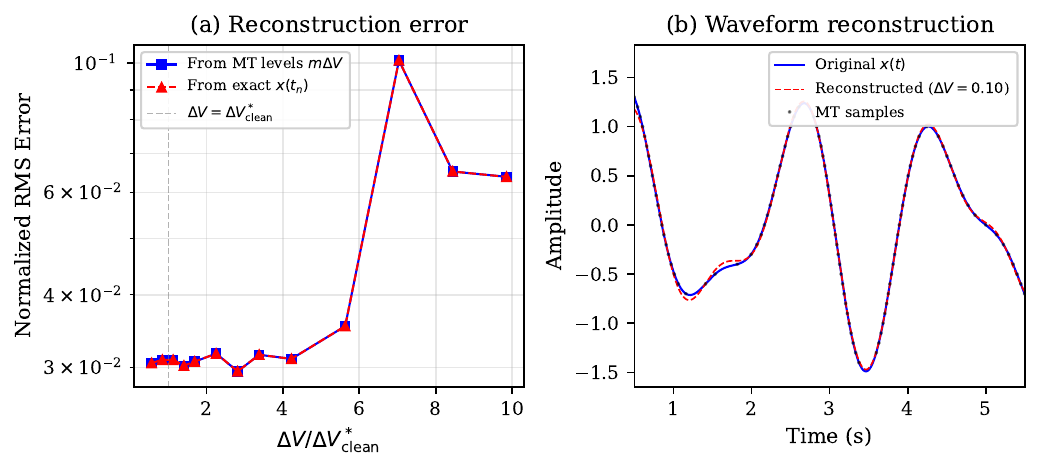}
    \caption{(a)~Finite sinc-model reconstruction NRMSE versus $\Delta V/\Delta V_{\mathrm{clean}}^*$ from ideal MT threshold labels (blue) and from exact crossing-time values (red). (b)~Example reconstructed waveform at $\Delta V = 0.10$ (about 106 retained crossing samples over 6~s).}
    \label{fig:recon}
\end{figure*}

\subsection{Experiment~5: No-Go Theorem Illustration}
\label{sec:exp5}
Fig.~\ref{fig:nogo}(a) illustrates the elementary empty-record example in Theorem~\ref{thm:no_go}: the Fej\'er bump~\eqref{eq:nogo_construction} with $\epsilon = 0.8\Delta V$ and $\Delta V = 0.30$ yields $\|x_{\mathrm{bump}}\|_\infty = 0.060 < \Delta V$, so the waveform remains inside a single threshold cell and the retained adjacent-transition record is empty. Fig.~\ref{fig:nogo}(b) illustrates the localized certificate failure: a beat-frequency signal $x(t) = \sin(2\pi \cdot 0.45\,t) + \sin(2\pi \cdot 0.55\,t)$ produces near-quiescent intervals near the beat nulls, where the total variation over a Nyquist-length window drops below $\Delta V$ and the retained crossing stream is interrupted. Fig.~\ref{fig:nogo}(c) shows the corresponding local retained event density dipping below the Nyquist rate $2f_{max}$ near the null. This local dip demonstrates failure of the plotted Kadec/density certificates on that interval; by itself it is not a proof that an otherwise redundant infinite sampling set fails globally.

\begin{figure*}[!t]
    \centering
    \includegraphics[width=\textwidth]{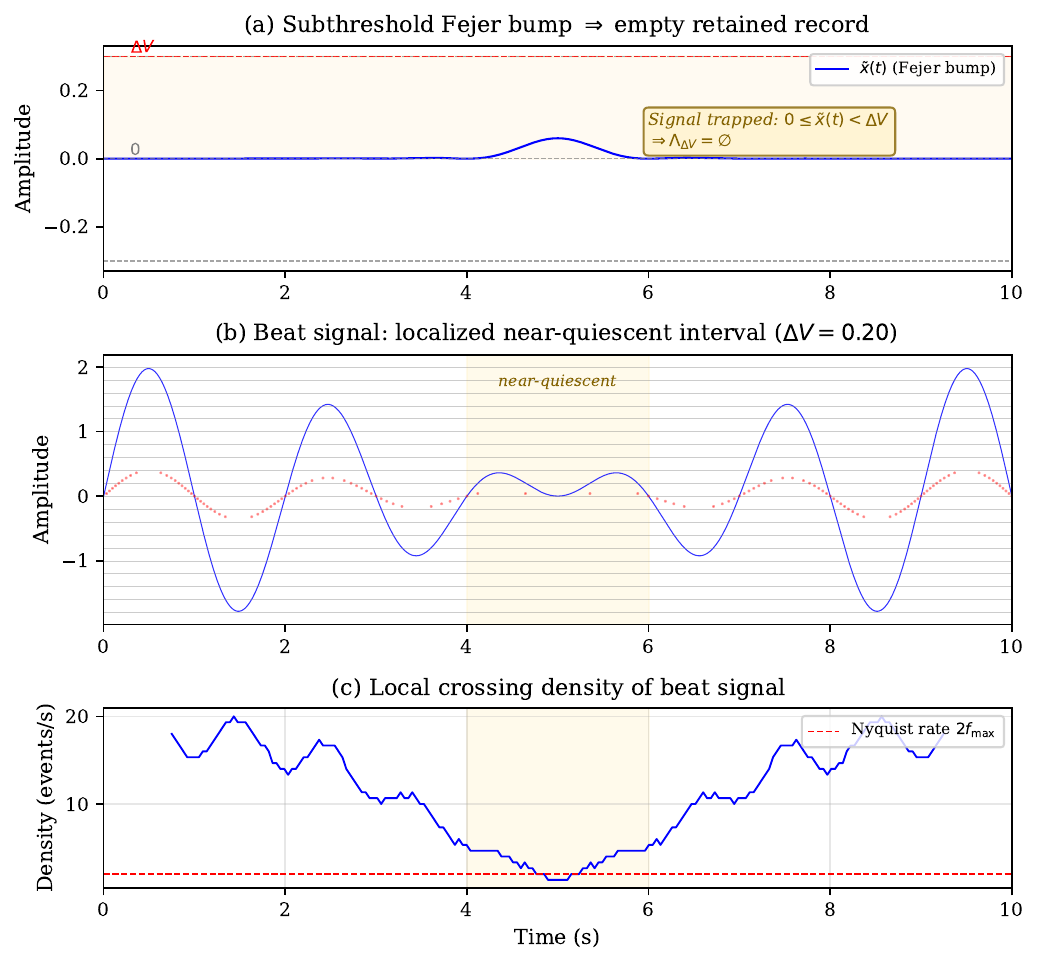}
    \caption{No-go theorem illustration. (a)~Fej\'er bump with $\|x_{\mathrm{bump}}\|_\infty < \Delta V$: the retained adjacent-transition record is empty. (b)~Beat signal with a near-quiescent interval near the beat null. (c)~Local retained event density of the beat signal: the density dips below the Nyquist rate $2f_{max}$ in the near-quiescent region, illustrating local certificate failure.}
    \label{fig:nogo}
\end{figure*}

\subsection{Experiment~6: Hardware Error Theory Evaluation}
\label{sec:exp6}
This experiment tests the conditional hardware perturbation bounds of Section~\ref{sec:hardware_error}. Using the test signal with $\Delta V = 0.02$ ($\Delta V/\Delta V_{\mathrm{clean}}^* \approx 0.56$, but slightly above the retained-safe scalar design margin $\Delta V_{\mathrm{Kadec,ret}}^*\approx0.0154$; $N = 2319$ retained crossings over $T_{\mathrm{obs}} = 20$~s), we inject controlled levels of TDC jitter $\sigma_j$, input noise $\sigma_n$, and threshold residual imprecision $\sigma_q$, varying each independently while preserving the clean/noisy event correspondence assumed in Section~\ref{sec:hardware_error}. The retained-phase diagnostic at this spacing gives $\delta_{\mathrm{ret}}\in[0.022,0.038]$, close to the curvature margin $0.0355$, so the experiment is interpreted through the measured retained Kadec geometry rather than the scalar retained-safe sufficient margin. For the threshold-residual branch, each residual $r_m$ both shifts the physical crossing time by the first-order rule $\delta t\approx r_m/x'(t_n)$ and leaves a value residual because reconstruction still receives only the reported calibrated level.

Fig.~\ref{fig:hwerror}(a) plots the measured normalized Kadec perturbation against the leading-order diagnostic $L_{\mathrm{hw}}^{(1)}$ versus TDC jitter $\sigma_j$, together with the clean-topology theoretical upper bound from~\eqref{eq:L_hw_leading}: $L_{\mathrm{hw}}^{(1)} \le \Delta V f_{max}/v_{min} + 2f_{max}\sigma_j$ obtained by setting $\bar{\alpha}_{\mathrm{hw}}^{(\mathrm{ret})}=1/2$ and $\sigma_n=\sigma_q=0$. The deterministic certificate would replace this diagnostic by $L_{\mathrm{hw}}^{\mathrm{det}}$ and add the certified $2f_{max}R_{2,\mathrm{hw}}$ term. The retained-safe scalar bound would replace the first term by $3\Delta V f_{max}/v_{min}$. The measured perturbation grows linearly with $\sigma_j$ and remains well below the clean-topology bound at all tested values, consistent with Proposition~\ref{prop:impairment_augmented_perturbation} under the checked event-correspondence regime. The conservatism reflects the worst-case nature of the $v_{min}$ denominator: the optimal grid offset $\tau_0$ avoids the near-extremal windows where $|x'|$ is close to $v_{min}$. The measured perturbation remains below the Kadec limit $L = 1/4$ out to $\sigma_j \approx 0.10$~s, which should be read as an empirical finite-instance margin rather than a retained-safe deterministic design guarantee.

Fig.~\ref{fig:hwerror}(b) plots the clean-topology hardware-adjusted admissible spacing $\Delta V_{\mathrm{hw}}^*$ from~\eqref{eq:hw_design} with $\bar{\alpha}_{\mathrm{hw}}^{(\mathrm{ret})}=1/2$ as a function of $\sigma_j$ (blue), $\sigma_n$ (red), and $\sigma_q$ (magenta). Under the numerical normalization used here, the amplitude-domain penalties have steeper slopes than the jitter penalty after conversion through $v_{min}\approx0.185<1$; this is consistent with~\eqref{eq:hw_design}. The retained-safe plot would have the same zero-crossing feasibility condition but a three-times smaller admissible spacing scale.

Fig.~\ref{fig:hwerror}(c) is consistent with the leading-order finite-dimensional prediction of Theorem~\ref{thm:error_budget}: over the tested range, the reconstruction NRMSE depends approximately \emph{linearly} on each impairment parameter. The plotted NRMSE is a finite-window full-waveform error; its zero-impairment baseline $2.0 \times 10^{-3}$ is the model/window residual, while the fitted slopes isolate the additional hardware perturbation around that baseline. Thus this panel should be read as a model-space perturbation experiment superposed on a small fixed residual, not as a standalone full-signal SNR certificate. The finite sampling matrix uses $M=33$ sinc basis functions and $N_{\mathrm{in}}=1929$ interior retained crossings, with $\sigma_{\min}(H_{\mathrm{hw},M})\approx 3.636$ at the baseline geometry, satisfying the finite-record acceptance criterion~\eqref{eq:finite_acceptance} for any floor below this value. The fitted empirical slopes are $C_j = 0.44$, $C_n = 0.05$, and $C_q = 0.08$. The finite-dimensional Theorem~\ref{thm:error_budget} worst-case constants are $C_j^{\mathrm{bound}}=35.28$, $C_n^{\mathrm{bound}}=160.95$, and $C_q^{\mathrm{bound}}=164.72$, which are about $80.5$, $2938$, and $2025$ times larger than the empirical slopes, respectively. This gap is expected because the bound substitutes worst-case quantities and matrix conditioning for the more benign local statistics of this signal instance; the experiment is therefore evidence for first-order scaling, not for sharpness or engineering optimality of the constants. We also evaluated the finite-matrix perturbation certificate from Proposition~\ref{prop:finite_matrix_perturb}: over the tested range $\sigma\le0.05$, the maximum values of $\mu_M=\|H_M^\dagger\|_2\|\delta H_M\|_2$ were $0.401$ for jitter, $0.301$ for input-noise timing shifts, and $0.261$ for threshold residuals. Hence the sufficient condition $\mu_M<1$ is satisfied in this experiment, confirming that the tighter finite-record certificate is applicable on this perturbation range even though the plotted constants remain the conservative Theorem~\ref{thm:error_budget} worst-case constants.

\begin{figure*}[!t]
    \centering
    \includegraphics[width=\textwidth]{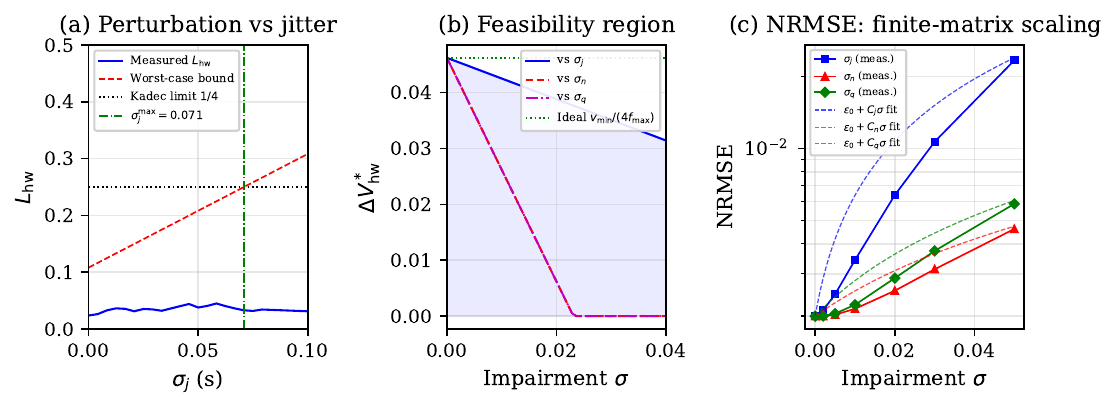}
    \caption{Hardware error theory evaluation ($\Delta V = 0.02$, $N = 2319$ retained crossings). (a)~Global Kadec perturbation versus TDC jitter: the measured perturbation (blue) stays well below the leading-order worst-case upper bound $L_{\mathrm{hw}}^{(1)}$ (red dashed). (b)~Hardware-adjusted admissible spacing $\Delta V_{\mathrm{hw}}^*$ versus impairment level, including the residual threshold branch $\sigma_q$. (c)~Finite-dimensional reconstruction NRMSE versus impairment level: measured values (markers) and linear fits (dashed), with threshold residuals modeled as both timing shifts and reported-value residuals.}
    \label{fig:hwerror}
\end{figure*}

\subsection{Experiment~7: Same-Threshold Suppression Stress Test}
\label{sec:exp7}
The preceding experiments mostly operate in clean retained topologies. This final experiment isolates the mechanism that makes the retained-safe factor $3/2$ in Section~\ref{sec:velocity_kadec} substantively necessary. Consider the bandlimited waveform
\[
    x_{\mathrm{sup}}(t)=-0.25+1.2\sin(2\pi\!\cdot\!0.25\,t),
\]
with $\Delta V=1$. Its local maximum is $0.95$, so after the upward crossing of the zero threshold the trajectory does not reach the $+1$ threshold. On the subsequent descent, the zero crossing is a same-threshold return and is suppressed by the recorded adjacent-transition convention; the next retained adjacent-level event is the crossing of the $-1$ threshold.

Fig.~\ref{fig:retained_suppression} chooses a Nyquist anchor $g_n$ on this descending branch with $x_{\mathrm{sup}}(g_n)=\Delta V/2$. The nearest raw threshold is then the zero level, at voltage distance $\Delta V/2$, but that hit is not retained. The nearest retained threshold is the $-1$ level, at voltage distance $3\Delta V/2$. Thus a clean-topology bound using only $\Delta V/2$ would certify a sample that the recorded MT stream does not contain, while the retained-safe factor $3/2$ exactly covers this local topology. This is a local event-geometry stress test, not a reconstruction experiment; its purpose is to validate the retained-event model used in the sufficient condition.

\begin{figure}[!t]
    \centering
    \includegraphics[width=\columnwidth]{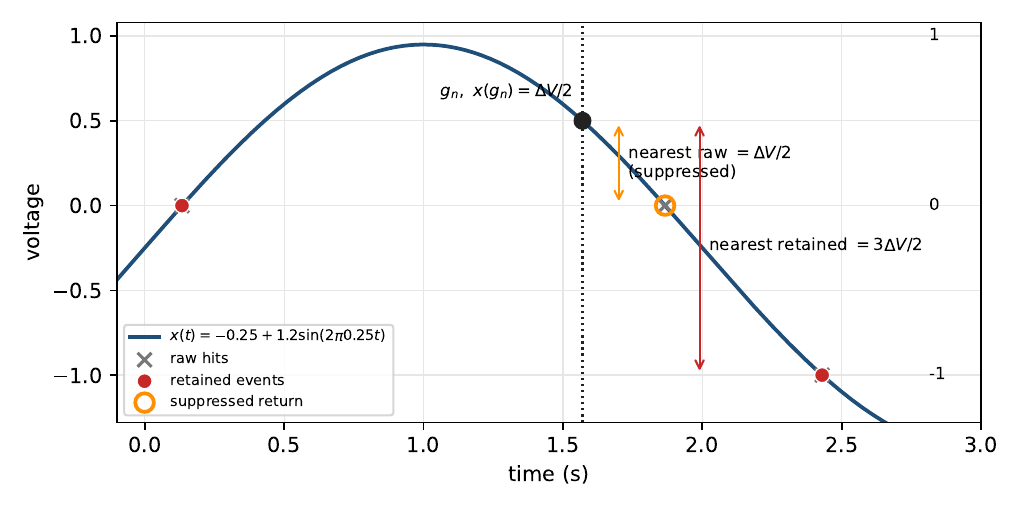}
    \caption{Same-threshold suppression stress test for $x_{\mathrm{sup}}(t)=-0.25+1.2\sin(2\pi\!\cdot\!0.25t)$ and $\Delta V=1$. The nearest raw hit to the anchor has voltage excursion $\Delta V/2$ but is a suppressed return to the zero threshold. The nearest retained event is the adjacent $-1$ crossing, requiring voltage excursion $3\Delta V/2$.}
    \label{fig:retained_suppression}
\end{figure}

\section{Discussion: Meaning of the Post-Realization Characterization}
\label{sec:discussion}

Section~\ref{sec:exact_iff} fills the conceptual gap left by the one-sided Kadec and density results only after the event times have been realized: necessary-and-sufficient conditions do exist for the fixed point set $\Lambda_{\Delta V}(x)$. For a single realized set, the broadest one is the Beurling weak-limit condition. For an a priori signal class, Theorem~\ref{thm:apriori_class_hull_mvt} gives the corresponding class-level iff: the prior object is the full MT weak-limit hull, not the bare threshold density $1/\Delta V$. The closest analogue of the classical Shannon formula is the complete-interpolation subcase in Subsection~\ref{sec:critical_complete_interp}, where a selected MT-generated subsequence is encoded by a canonical product and reconstructed by Lagrange kernels. Inside the periodic-gap subclass, that canonical-product condition becomes the finite-block equality $H=q$ together with the explicit product~\eqref{eq:periodic_generating_function}; if one asks only for stable sampling, the same periodic geometry gives the density rule $q/H\ge1$, with strict inequality corresponding to redundant sampling. Corollary~\ref{cor:finite_periodic_template_mvt} is the computable finite-orbit version of the class-hull criterion: finitely many periodic weak-limit orbits reduce the iff to finitely many density/rank checks, and Corollary~\ref{cor:periodic_template_threshold_density} turns scalar threshold density into an iff only after finite-orbit velocity calibration. Corollary~\ref{cor:finite_state_threshold_density} gives the analogous noncritical criterion for a finite-state graph gap prior through the maximum cycle mean. Section~\ref{sec:apriori_strict_density_certificate} gives a different scalar iff: it is valid only after exact calibration has converted threshold density and retained losses into the actual lower time-event density, and only away from the critical boundary. The class weak-limit hull theorem is the most general a priori necessary-and-sufficient condition available within the nonuniform-sampling framework. Scalar density conditions become exact only after additional structure is imposed: exact retained-density calibration away from the critical boundary, finitely many periodic weak-limit orbits, or a finite-state gap prior. For unrestricted finite-energy static MT, the correct general statement remains the negative one given by Theorem~\ref{thm:no_go}. Outside calibrated or structured subclasses, neither post-realization nor class-level admissibility is a scalar inequality in $\Delta V$, $f_{max}$, and a few signal-instance parameters.

\subsection{Why the Post-Realization Condition Is Not a Design Formula}

The equivalence
\begin{align*}
    &\Lambda_{\Delta V}(x)\text{ samples }PW_\Omega \\
    &\qquad\Longleftrightarrow\quad
    \Gamma\text{ is a uniqueness set for }\mathcal{B}_\Omega\\
    &\qquad\qquad\text{for all }\Gamma\in W(\Lambda_{\Delta V}(x))
\end{align*}
is fully necessary and sufficient for the fixed point set. Its cost is that it retains the complete microscopic geometry of that event set. Every translate weak limit must be checked, including limits that arise by looking far out in time or by centering on rare local configurations. This is exactly what is needed at critical density, where two point sets may have the same density, separation, and maximum gap behavior over many windows, yet differ in sampling status. What the equivalence does not decide is whether the nonlinear map $x\mapsto E_{\Delta V}(x)$ is stable under perturbations that can change event topology.

The removed Shannon-grid example illustrates why this issue cannot generally be avoided. The set $\{n\Delta T:n\in\mathbb{Z}\}$ samples $PW_\Omega$, while deleting one point leaves the lower Beurling density unchanged but removes uniqueness through the corresponding sinc function. Thus no finite list of macroscopic descriptors can decide the boundary case. Away from the boundary, density can decide stability only after the retained time-event density itself has been exactly calibrated, as in Theorem~\ref{thm:apriori_strict_density_mvt}. Pavlov/Lyubarskii--Seip theory makes the critical point more positively: in complete-interpolation regimes the fixed-set condition is expressible through a Muckenhoupt $A_2$ condition on a canonical product, i.e., through an object encoding all zeros of the sequence, not through density alone~\cite{LyubarskiiSeip1997}.

\subsection{What the Shannon-Type Formula Adds}

The weak-limit condition decides whether a separated MT set is a stable sampling frame. The Shannon-type theorem answers a narrower but more explicit question: when can a selected MT event subsequence play the role of the Nyquist grid itself? The answer is complete interpolation. In that regime the sine function in the classical Shannon theorem is replaced by the canonical product $G$ generated by the normalized MT event times, and the sinc kernel is replaced by
\[
    \frac{G(\Omega t/\pi)}
         {G'(\mu_n)(\Omega t/\pi-\mu_n)}.
\]
This gives both a reconstruction formula and an exact criterion in the stated complete-interpolation setting: if the Pavlov/Lyubarskii--Seip $A_2$ condition fails, then the desired non-redundant Shannon-type MT expansion fails.

The price is that complete interpolation is stricter than sampling. A redundant MT record may be an excellent frame without containing a distinguished complete interpolating subsequence that the hardware or algorithm has selected. Conversely, if such a subsequence is available, the remaining MT samples are redundant information that may improve finite-window conditioning but are not needed for exact recovery in the ideal infinite theorem.

\subsection{What the Periodic-Gap Formula Adds}

The periodic-gap results isolate the part of the Shannon-type theory that can be made as explicit as the classical Nyquist--Shannon theorem. Once the normalized gap word
\[
    (h_0,\ldots,h_{q-1})
\]
repeats forever, the full weak-limit hull and the Pavlov $A_2$ weight no longer need to be tested directly. The density of the normalized set is exactly $q/H$. For stable sampling, the finite fiber/Vandermonde calculation gives the exact periodic rule $q/H\ge1$, equivalently $q/T_{\mathrm{per}}\ge\Omega/\pi=2f_{max}$. The strict case $q/H>1$ is a redundant sampling frame. It supplies more samples than the critical number of degrees of freedom per period, so it need not be a complete-interpolation sequence or have a Shannon-type bijective cardinal formula.

The complete-interpolation theorem is the critical subcase. Complete interpolation is simultaneously sampling and interpolation, so Landau's two density inequalities force $q/H=1$, hence $H=q$. Conversely, if $H=q$, the nodes are a union of $q$ arithmetic progressions $q\mathbb{Z}+\beta_r$ with distinct residues, and the finite Vandermonde matrix of those residues gives a Riesz-basis proof. The generating function is then the finite sine-type product~\eqref{eq:periodic_generating_function}, and the reconstruction kernels are the corresponding finite-product Lagrange kernels.

The periodic-gap theorem provides a structured MT analogue of Shannon-type reconstruction within a repeated finite-gap regime. It is still not a universal threshold-spacing rule: the hypothesis is the entire normalized event-gap pattern, not merely a Nyquist average event rate. Under that hypothesis, complete interpolation is expressed by the transparent finite equality
\[
    \sum_{r=0}^{q-1}h_r=q,
\]
or equivalently one period of $q$ MT events spans exactly $q$ Nyquist intervals. Stable sampling allows the period to be no longer than this critical span, and the adjacent-threshold velocity form~\eqref{eq:periodic_threshold_density_condition} shows that the scalar threshold density is meaningful only after it is scaled by the harmonic mean velocity over the repeated gap word. Thus the theorem shows both sides of the explicitness boundary: a finite condition is possible when the microscopic geometry is finite-state and repeated, while unstructured critical MT sequences still require the canonical-product or weak-limit geometry.

\subsection{Why Static Finite-Energy MT Still Fails Globally}

The weak-limit theorem also sharpens the no-go message. For finite-energy $x\in PW_\Omega$, the retained static MT event stream eventually stops producing informative adjacent-threshold transitions, so the retained recorded set has $D^-=0$ as in~\eqref{eq:finite_energy_mvt_density_zero}. Consequently it cannot satisfy the Landau necessary condition and cannot satisfy the weak-limit uniqueness condition. In this precise sense, the global infinite-dimensional problem has a negative answer for static MT sampling under the finite-energy model unless additional information is supplied.

The Nyquist-padded construction in Subsection~\ref{sec:finite_bridge} should therefore be read as an auxiliary geometric certificate, not as an acquisition claim. Padding creates a bi-infinite Kadec-compatible set whose outside values would make a frame theorem applicable, but those outside values are not measured by the MT hardware. The actual finite-record reconstruction claim passes through either a tail prior or the finite-dimensional rank condition~\eqref{eq:finite_exact_rank_condition}.

\subsection{Role of the Computable Certificates}

The Kadec, density, structured-prior, and hardware sections occupy the gap between exact theory and usable MT certificates. The weak-limit condition tells us what must be true for a general sampling frame. The MT--Pavlov theorem tells us what must be true for a non-redundant Shannon-type cardinal formula. The periodic-gap theorem shows that this formula can become a finite-block MT--Shannon rule when the realized normalized gaps repeat exactly, while its stable-sampling corollary permits redundant periodic frames above the critical density. The Kadec condition supplies a local, easily checked sufficient certificate. The density analysis supplies finite-window retained-density rules and their long-window design scales once extremum-density and suppression-density certificates are available; when an exact retained-density calibration is part of the model and the density is noncritical, Theorem~\ref{thm:apriori_strict_density_mvt} upgrades that density comparison to an iff. The hardware section then tracks how close the realized event geometry remains to the certified ideal geometry under timing, noise, and threshold perturbations.

Within the retained-event point-sampling formulation studied here, the MT problem is resolved at three levels. For a realized record, the exact object is the translate weak-limit hull of its event times. For a signal class, the exact a priori object is the class weak-limit hull. For computable structured priors, periodic-template orbits and finite-state gap graphs reduce the weak-limit condition to finite algebraic or graph-theoretic checks away from the critical boundary. Scalar threshold-density rules are exact only after this additional calibration or structure has converted voltage traversal into retained time-event density. Without such prior information, the finite-energy static no-go theorem is the correct universal statement.

This layered structure is especially important because MT crossing sets are signal-generated. The signal determines $\Lambda_{\Delta V}(x)$, but the sampling property of $\Lambda_{\Delta V}(x)$ is tested against every $f\in PW_\Omega$, not merely against the generating waveform $x$. A finite-dimensional implementation resolves this circularity by declaring a model space $V_M$ and testing the realized matrix $H_{\mathcal{I},M}$ directly. In that finite setting, the necessary-and-sufficient condition is no longer implicit: $\sigma_{\min}(H_{\mathcal{I},M})>0$.

\subsection{Scope-Preserving Extensions}

The results above resolve the static-threshold retained-event point-sampling problem at the level of exact geometry, structured priors, and the static finite-energy boundary. The following extensions would sharpen application-specific priors but are not needed for the main theorem. One can study the weak-limit hull of MT-generated sets more directly, because MT events are not arbitrary separated sequences: on monotone segments their gaps are controlled by $\Delta V/|x'|$, and near non-degenerate extrema they follow the phase-aware square-root law of Proposition~\ref{prop:phase_aware_extremum_pattern}. Application-specific pulse libraries, periodic templates, finite-state gap graphs, and interval-uncertain graph priors are natural places to look for sharper certificates.

\section{Limitations and Future Directions}
\label{sec:limitations}

The main limitation is the deliberate restriction to static, uniformly spaced thresholds. The theory is developed under an ideal bandlimited model, a retained adjacent-transition convention, and deterministic finite-record certificates. Time-varying or adaptive threshold sets would require update latency, calibration dynamics, and closed-loop stability assumptions, and are outside the theory developed here. The present paper answers the corresponding static-threshold point-sampling question within that scope: once the retained event set or its admissible class is specified, nonuniform-sampling theory gives the exact geometry needed for stable reconstruction.

The hardware perturbation results are conditional. They apply when clean and noisy records remain in event correspondence, or when an isolation certificate such as Proposition~\ref{prop:event_correspondence_isolation} is available. Stochastic topology-changing effects, including spurious, missed, or mislabeled crossings, are outside these small-perturbation bounds and require a separate event-process and robust-reconstruction theory.

The ideal $PW_\Omega$ model also needs application-level validation. Real detector pulses are only effectively bandlimited after analog shaping, and the effective bandwidth $f_{max}$ must be defined conservatively. Experimental validation should compare static MT records against high-resolution reference digitizers, verify the retained topology, and measure the realized finite-matrix conditioning and hardware timing budget.

These limitations delimit the scope of the static-threshold retained-event theory rather than leaving its central sampling question unresolved.

\section{Conclusion}
This paper gives a nonuniform-sampling theory for retained Multi-Threshold data generated by static, uniformly spaced thresholds. The central conclusion is geometric: the threshold set first produces a retained event set, and stable reconstruction is determined by the nonuniform sampling geometry of that set.

For a realized retained time set, Beurling weak limits give the exact necessary-and-sufficient sampling criterion for $PW_\Omega$. For a signal class, the corresponding exact a priori object is the class weak-limit hull. These two results identify the general solution of the problem inside classical nonuniform sampling theory.

The remaining results explain when this abstract criterion becomes computable. Periodic-template hulls, periodic-gap streams, finite-state gap priors, and noncritical calibrated density reduce the weak-limit criterion to finite rank, density, or graph-theoretic tests under explicit structure. Kadec, retained-density, finite-matrix, and hardware perturbation bounds provide practical sufficient certificates for active records and implementations.

The theory also proves its own boundary. In the unrestricted static finite-energy model, retained MT records generated by static, uniformly spaced thresholds have zero full-line lower density and cannot sample the whole space $PW_\Omega$ without additional information. Finite records therefore require a tail prior or a declared finite-dimensional reconstruction space. Within this retained-event point-sampling formulation, the MT recovery question is resolved in the following sense: exact recovery is governed by event geometry, computable guarantees require explicit structure or finite-model assumptions, and the prior-free finite-energy full-line problem is impossible.

\appendices
\numberwithin{equation}{section}
\renewcommand{\theequation}{\thesection.\arabic{equation}}
\renewcommand{\[}{\begin{equation}}
\renewcommand{\]}{\end{equation}}

\section{Landau Necessity and Beurling Sufficiency Inputs}
\label{app:landau}
This appendix records the density inputs used in Section~\ref{sec:preliminaries}. The necessary density bound is derived from Landau's concentration argument. The strict-density sufficient bound is used as a classical external harmonic-analysis theorem; we record only the normalization and frame-bound transfer needed in this paper. Throughout the Landau subsection, $K(t,s)$ denotes the reproducing kernel of $PW_\Omega$ introduced later in Appendix~\ref{app:trace}.

\subsection{Landau Necessity: Every Sampling Set Must Satisfy $D^-(\Lambda) \ge \Omega/\pi$}
Assume that $\Lambda$ is a sampling set for $PW_\Omega$, so that the frame lower bound
\begin{equation}
    A\|f\|_{L^2}^2 \le \sum_{\lambda \in \Lambda} |f(\lambda)|^2
    \label{eq:landau_lower_appendix}
\end{equation}
holds for every $f \in PW_\Omega$.

\textbf{Step 1: The concentration operator and its high-energy subspaces.} For $a>0$, define the time-truncation operator $D_a f = f\,\mathbf{1}_{[-a,a]}$ and the bandlimiting projector $B_\Omega f = \mathcal{F}^{-1}(\mathbf{1}_{[-\Omega,\Omega]}\hat{f})$, where $\mathbf{1}_E$ denotes the indicator function of the set $E$. Their composition
\begin{equation}
    T_a = B_\Omega D_a B_\Omega
\end{equation}
acts on $PW_\Omega$ and is compact, self-adjoint, and positive. Indeed, on $PW_\Omega$ it is an integral operator with kernel $K(t,s)\mathbf{1}_{[-a,a]}(s)$, which is square-integrable on $\mathbb{R}\times[-a,a]$; hence $T_a$ is Hilbert-Schmidt and therefore compact. If $\{\phi_k^{(a)}\}_{k \ge 1}$ is an orthonormal eigenbasis with eigenvalues $1 \ge \mu_1^{(a)} \ge \mu_2^{(a)} \ge \cdots \ge 0$, then
\begin{equation}
    \langle T_a f, f \rangle_{L^2} = \int_{-a}^{a} |f(t)|^2\,dt
    \label{eq:landau_energy_identity}
\end{equation}
for every $f \in PW_\Omega$. In particular, the trace of $T_a$ (derived in Appendix~\ref{app:trace}) is
\begin{equation}
    \mathrm{Tr}(T_a) = \sum_{k \ge 1} \mu_k^{(a)} = \frac{2\Omega a}{\pi}.
    \label{eq:trace}
\end{equation}

Fix any $\alpha \in (0,1)$ and define the high-energy subspace
\begin{equation}
    E_a(\alpha) = \mathrm{span}\{\phi_k^{(a)} : \mu_k^{(a)} \ge \alpha\}.
\end{equation}
By the classical prolate-spheroidal concentration theory~\cite{Slepian1961,Slepian1978} and the Widom--Landau eigenvalue asymptotics~\cite{Widom1964,LandauWidom1980} (see Appendix~\ref{app:prolate} for the precise statement), for every fixed $\alpha \in (0,1)$,
\begin{equation}
    \dim E_a(\alpha) = \frac{2\Omega a}{\pi} + o(a), \qquad a \to \infty.
    \label{eq:landau_dim_asymp}
\end{equation}
Moreover, if $f \in E_a(\alpha)$, then all eigenvalues contributing to $f$ are at least $\alpha$, so by \eqref{eq:landau_energy_identity}
\begin{align}
    \int_{|t|>a} |f(t)|^2\,dt
     & = \|f\|_{L^2}^2 - \int_{-a}^{a} |f(t)|^2\,dt     \\
     & = \|f\|_{L^2}^2 - \langle T_a f, f \rangle_{L^2} \\
     & \le (1-\alpha)\|f\|_{L^2}^2.
    \label{eq:landau_leakage}
\end{align}
Thus functions in $E_a(\alpha)$ are quantitatively concentrated on $[-a,a]$.

\textbf{Step 2: Sampling sets have bounded local multiplicity.} To use \eqref{eq:landau_leakage} at the level of samples, we need a localization estimate for point evaluations. This requires a finite overlap bound on the sample set. Let
\begin{equation}
    \kappa_x(t) = \frac{K(t,x)}{\sqrt{K(x,x)}} = \sqrt{\frac{\pi}{\Omega}}\,K(t,x),
\end{equation}
so that $\kappa_x \in PW_\Omega$, $\|\kappa_x\|_{L^2}=1$, and $\kappa_x(x)=\sqrt{\Omega/\pi}$. Because $K(t,x)=\sin[\Omega(t-x)]/[\pi(t-x)]$ depends only on $t-x$ and is continuous, there exist constants $r_{\mathrm{loc}}>0$ and $c_0>0$ such that
\begin{equation}
    |\kappa_x(t)| \ge c_0 \qquad \text{whenever } |t-x| \le r_{\mathrm{loc}},
    \label{eq:landau_kernel_lower}
\end{equation}
uniformly in $x \in \mathbb{R}$. Applying the upper frame bound in \eqref{eq:frame} to $\kappa_x$ gives
\begin{equation}
    c_0^2\,\#\bigl(\Lambda \cap [x-r_{\mathrm{loc}}, x+r_{\mathrm{loc}}]\bigr)
    \le \sum_{\lambda \in \Lambda} |\kappa_x(\lambda)|^2
    \le B.
\end{equation}
Hence
\begin{equation}
    N_\Lambda \triangleq \sup_{x \in \mathbb{R}} \#\bigl(\Lambda \cap [x-r_{\mathrm{loc}}, x+r_{\mathrm{loc}}]\bigr) < \infty.
    \label{eq:landau_local_multiplicity}
\end{equation}
That is, $\Lambda$ is relatively separated. The standard localized Plancherel-P\'olya inequality for relatively separated sets~\cite{Pogany2023} then implies the existence of a constant $C_{\Lambda,\Omega}>0$ such that for every measurable set $E \subseteq \mathbb{R}$ and every $f \in PW_\Omega$,
\begin{equation}
    \sum_{\lambda \in \Lambda \cap E} |f(\lambda)|^2
    \le C_{\Lambda,\Omega} \int_{E + [-r_{\mathrm{loc}},r_{\mathrm{loc}}]} |f(t)|^2\,dt.
    \label{eq:landau_local_pp}
\end{equation}

\textbf{Step 3: Contradiction if the lower density is below critical.} Assume for contradiction that
\begin{equation}
    D^-(\Lambda) < \frac{\Omega}{\pi}.
\end{equation}
Choose $\varepsilon>0$ such that $D^-(\Lambda) \le \Omega/\pi - 3\varepsilon$. By definition of $D^-$, there exist sequences $R_j \to \infty$ and $a_j \in \mathbb{R}$ such that
\begin{equation}
    N_j \triangleq \#\bigl(\Lambda \cap [a_j-R_j, a_j+R_j]\bigr)
    \le \frac{2\Omega R_j}{\pi} - 2\varepsilon R_j
    \label{eq:landau_low_density_interval}
\end{equation}
for all sufficiently large $j$.

Translation preserves both the space $PW_\Omega$ and the frame bounds: if $\Lambda_j = \Lambda - a_j$, then $\Lambda_j$ is also a sampling set for $PW_\Omega$ with the same constants $A$ and $B$. Therefore we may replace $\Lambda$ by $\Lambda_j$ and work with the centered interval $[-R_j,R_j]$.

Now choose $\alpha \in (0,1)$ so close to $1$ that
\begin{equation}
    C_{\Lambda,\Omega}(1-\alpha) < \frac{A}{2}.
    \label{eq:landau_alpha_choice}
\end{equation}
Set $a_j' = R_j - r_{\mathrm{loc}}$. By \eqref{eq:landau_dim_asymp},
\begin{equation}
    \dim E_{a_j'}(\alpha) = \frac{2\Omega(R_j-r_{\mathrm{loc}})}{\pi} + o(R_j)
    = \frac{2\Omega R_j}{\pi} + o(R_j).
\end{equation}
Combining this with \eqref{eq:landau_low_density_interval}, we obtain for all sufficiently large $j$ (see Appendix~\ref{app:landau_step3_detail} for the explicit counting step),
\begin{equation}
    \dim E_{a_j'}(\alpha) > N_j.
    \label{eq:landau_dimension_excess}
\end{equation}
Therefore there exists a nonzero function $f_j \in E_{a_j'}(\alpha)$ that vanishes at every sample point of $\Lambda_j$ inside $[-R_j,R_j]$.

Applying the lower frame bound \eqref{eq:landau_lower_appendix} to $f_j$ and using the vanishing just obtained yields
\begin{equation}
    A\|f_j\|_{L^2}^2
    \le \sum_{\lambda \in \Lambda_j} |f_j(\lambda)|^2
    = \sum_{\lambda \in \Lambda_j \setminus [-R_j,R_j]} |f_j(\lambda)|^2.
    \label{eq:landau_lower_on_outside}
\end{equation}
Now apply the localized Plancherel-P\'olya estimate \eqref{eq:landau_local_pp} with $E = \mathbb{R} \setminus [-R_j,R_j]$:
\begin{equation}
    \sum_{\lambda \in \Lambda_j \setminus [-R_j,R_j]} |f_j(\lambda)|^2
    \le C_{\Lambda,\Omega} \int_{|t|>R_j-r_{\mathrm{loc}}} |f_j(t)|^2\,dt.
    \label{eq:landau_outside_samples_bound}
\end{equation}
Since $f_j \in E_{R_j-r_{\mathrm{loc}}}(\alpha)$, the leakage estimate \eqref{eq:landau_leakage} gives
\begin{equation}
    \int_{|t|>R_j-r_{\mathrm{loc}}} |f_j(t)|^2\,dt
    \le (1-\alpha)\|f_j\|_{L^2}^2.
    \label{eq:landau_outside_energy_bound}
\end{equation}
Combining \eqref{eq:landau_lower_on_outside}--\eqref{eq:landau_outside_energy_bound} with \eqref{eq:landau_alpha_choice}, we obtain
\begin{equation}
    A\|f_j\|_{L^2}^2
    \le C_{\Lambda,\Omega}(1-\alpha)\|f_j\|_{L^2}^2
    < \frac{A}{2}\|f_j\|_{L^2}^2,
\end{equation}
which is impossible because $f_j \neq 0$. This contradiction proves that every sampling set must satisfy
\begin{equation}
    D^-(\Lambda) \ge \frac{\Omega}{\pi} = 2f_{max}.
\end{equation}

\subsection{Beurling Strict-Density Sufficiency as an External Input}
\label{app:beurling_strict_density_input}

The sufficient direction used in the main text is the classical Beurling
strict-density theorem~\cite{beurling1966local,Seip2004,GrochenigRomeroStoeckler2018}.
We do not reprove the balayage theorem here; it is used as an external
harmonic-analysis input. What is needed in the present paper is the normalization
from general bandwidth $\Omega$ to $PW_\pi$ and the transfer from the classical
statement to the sampling-frame inequality.

\begin{theorem}[Beurling strict-density theorem]
\label{thm:beurling_strict_density_external}
Let $\Lambda\subset\mathbb R$ be uniformly discrete. If
\[
    D^-(\Lambda)>\frac{\Omega}{\pi},
\]
then $\Lambda$ is a sampling set for $PW_\Omega$; that is, there exist
$0<A\le B<\infty$ such that
\[
    A\|f\|_{L^2}^2
    \le
    \sum_{\lambda\in\Lambda}|f(\lambda)|^2
    \le
    B\|f\|_{L^2}^2,
    \qquad
    f\in PW_\Omega .
\]
\end{theorem}

\textit{Normalization.}
Let $c_\Omega=\Omega/\pi$ and define
\[
    \widetilde\Lambda=c_\Omega\Lambda.
\]
If $f\in PW_\Omega$ and
\[
    g(u)=f(u/c_\Omega)=f(\pi u/\Omega),
\]
then $g\in PW_\pi$,
\[
    \|g\|_{L^2}^2=c_\Omega\|f\|_{L^2}^2,
\]
and
\[
    g(c_\Omega\lambda)=f(\lambda).
\]
Moreover,
\[
    D^-(\widetilde\Lambda)
    =
    \frac{\pi}{\Omega}D^-(\Lambda)>1.
\]
Thus the normalized Beurling theorem for $PW_\pi$ applies to
$\widetilde\Lambda$. Scaling back gives the displayed frame inequality on
$PW_\Omega$. The upper frame bound follows from the Plancherel--P\'olya inequality
for uniformly discrete sets.

\section{Details for the Weak-Limit MT Characterization}
\label{app:exact_iff}
This appendix expands the reductions used in Section~\ref{sec:exact_iff}. The only non-elementary input is Beurling's weak-limit sampling theorem for separated sets, stated below in the precise form needed here. All remaining steps are bookkeeping: scaling the Paley-Wiener convention, translating MT labels into point samples, and reducing finite-record recovery to matrix rank.
The roadmap is as follows: Appendix~\ref{app:weak_limit_theorem} states the external weak-limit theorem, Appendix~\ref{app:mvt_to_sampling_operator} converts retained MT labels into ordinary point samples, Appendix~\ref{app:fourier_frame_equiv} checks the Fourier-frame normalization, Appendix~\ref{app:exact_mvt_sep} proves the MT separation estimate, Appendix~\ref{app:finite_energy_density_zero} records the finite-energy density obstruction, and Appendix~\ref{app:finite_rank_equiv} gives the finite-dimensional rank equivalence.

\subsection{External Weak-Limit Theorem}
\label{app:weak_limit_theorem}
For a separated set $\Lambda\subset\mathbb{R}$, let $W(\Lambda)$ denote the set of all translate weak limits as defined in Section~\ref{sec:weak_limit_condition}. Beurling's theorem states:
\begin{align}
    \Lambda \text{ samples } PW_\Omega
    &\Longleftrightarrow
    \Gamma \text{ is a uniqueness set}\\
    &\quad\text{for }\mathcal{B}_\Omega
    \text{ for every }\Gamma\in W(\Lambda).
\end{align}
This theorem is a genuine infinite-dimensional sampling theorem~\cite{beurling1966local,Seip2004,GrochenigRomeroStoeckler2018}. It is stronger than the density pair used in Appendix~\ref{app:landau}: Landau's theorem gives a necessary lower-density inequality, Beurling's strict-density theorem gives a convenient sufficient condition, while the weak-limit theorem decides the critical and noncritical cases by keeping the full local limiting geometry.

The class-hull theorem uses the corresponding uniform-family form. If $\mathcal F$ is a uniformly separated family of point sets, then $\mathcal F$ has common sampling frame bounds for $PW_\Omega$ if and only if every translate weak limit of every sequence $\Lambda_j\in\mathcal F$ is a uniqueness set for $\mathcal B_\Omega$. This is also used as an external Beurling theorem; the present paper applies it to the family $\{\Lambda_{\Delta V}(x):x\in\mathcal X\}$ after the uniform-separation hypothesis has been verified.

The normalization in many references uses $PW_\pi$ or the Fourier support interval $[-1/2,1/2]$. Under the convention of this paper,
\[
    \mathrm{supp}\,\hat f\subseteq[-\Omega,\Omega],
    \qquad
    \Omega=2\pi f_{max},
\]
the critical density is
\[
    \frac{\Omega}{\pi}=2f_{max}.
\]
The scaling map $\widetilde{\Lambda}=(\Omega/\pi)\Lambda$ carries $PW_\Omega$ to $PW_\pi$ by $g(u)=f(\pi u/\Omega)$; hence a statement proved at bandwidth $\pi$ transfers to the present convention with the density threshold multiplied by $\Omega/\pi$.

\subsection{From MT Events to a Sampling Operator}
\label{app:mvt_to_sampling_operator}
Let the ideal MT record be
\[
    \{(t_n,m_n)\}_{n\in\mathcal{N}},
    \qquad
    x(t_n)=m_n\Delta V.
\]
The acquisition map $x\mapsto\Lambda_{\Delta V}(x)$ is nonlinear because the event times depend on $x$. However, once the event set is realized, the reconstruction stage sees the linear point-evaluation operator
\[
    S_\Lambda f = (f(t_n))_{n\in\mathcal{N}}.
\]
If $S_\Lambda$ satisfies the frame inequality
\[
    A\|f\|_{L^2}^2
    \le
    \|S_\Lambda f\|_{\ell^2}^2
    \le
    B\|f\|_{L^2}^2,
    \qquad f\in PW_\Omega,
\]
then $S_\Lambda$ is bounded below and has a bounded inverse on its range. Conversely, if stable recovery from the samples is possible by a bounded reconstruction operator $R$ satisfying $R S_\Lambda f=f$, then
\[
    \|f\|_{L^2}
    =
    \|R S_\Lambda f\|_{L^2}
    \le
    \|R\|\,\|S_\Lambda f\|_{\ell^2},
\]
which gives the lower frame bound $A=\|R\|^{-2}$. The upper bound is the ordinary Plancherel--P\'olya boundedness of point evaluations on separated sets. Thus stable recoverability and the sampling frame inequality are equivalent.

\subsection{Equivalence with a Fourier Frame}
\label{app:fourier_frame_equiv}
Let
\[
    h(\omega)=\frac{1}{\sqrt{2\pi}}\hat f(\omega),
    \qquad \omega\in[-\Omega,\Omega].
\]
By Plancherel, $\|h\|_{L^2[-\Omega,\Omega]}=\|f\|_{L^2(\mathbb{R})}$. For every sample point $t_n$,
\begin{align}
    f(t_n)
    &=
    \frac{1}{2\pi}
    \int_{-\Omega}^{\Omega}
        \hat f(\omega)e^{j\omega t_n}\,d\omega  \\
    &=
    \frac{1}{\sqrt{2\pi}}
    \int_{-\Omega}^{\Omega}
        h(\omega)e^{j\omega t_n}\,d\omega \\
    &=
    \frac{1}{\sqrt{2\pi}}
    \langle h,e^{-j\omega t_n}\rangle_{L^2[-\Omega,\Omega]},
\end{align}
where the inner product is linear in its first argument. Therefore
\[
    \sum_n |f(t_n)|^2
    =
    \frac{1}{2\pi}
    \sum_n
    |\langle h,e^{-j\omega t_n}\rangle|^2.
\]
The sampling inequality for $f$ is thus equivalent, after multiplying the frame constants by $2\pi$, to the Fourier-frame inequality for the exponential system $\{e^{-j\omega t_n}\}$ on $L^2[-\Omega,\Omega]$.

\subsection{MT Separation}
\label{app:exact_mvt_sep}
Assume consecutive retained MT events are adjacent threshold transitions, as enforced by the retained adjacent-transition rule in~\eqref{eq:mvt_set}. Then
\[
    |x(t_{n+1})-x(t_n)|=\Delta V.
\]
By the mean value theorem there exists $\xi_n\in(t_n,t_{n+1})$ such that
\[
    \Delta V
    =
    |x'(\xi_n)|\,|t_{n+1}-t_n|.
\]
Since $x\in PW_\Omega$ extends to an entire function of type at most $\Omega$, Bernstein's inequality gives
\[
    |x'(\xi_n)|\le\|x'\|_\infty\le\Omega\|x\|_\infty.
\]
Hence
\[
    |t_{n+1}-t_n|
    \ge
    \frac{\Delta V}{\Omega\|x\|_\infty}.
\]
This separation is the hypothesis needed to invoke the weak-limit theorem and the Plancherel--P\'olya upper sampling bound. If the event topology changes because of noise-induced spurious events or mislabeled thresholds, this deterministic separation statement must be replaced by the topology certificate of Section~\ref{sec:event_topology}.

\subsection{Why Finite-Energy Static MT Has Zero Global Density}
\label{app:finite_energy_density_zero}
This appendix records the analytic tail input used in Theorem~\ref{thm:no_go}. Let $x\in PW_\Omega$. Since $\hat{x}\in L^2[-\Omega,\Omega]$ and the support interval has finite measure,
\[
    \|\hat{x}\|_{L^1[-\Omega,\Omega]}
    \le
    (2\Omega)^{1/2}\|\hat{x}\|_{L^2[-\Omega,\Omega]}<\infty.
\]
The inverse Fourier integral is therefore the Fourier transform of an $L^1$ function, so the Riemann--Lebesgue lemma gives $x(t)\to0$ as $|t|\to\infty$. Given $\Delta V>0$, choose $T$ such that $|x(t)|<\Delta V$ outside $[-T,T]$.

Outside this interval, nonzero thresholds cannot be crossed. The retained adjacent-transition convention suppresses repeated same-threshold returns once the waveform remains in the single cell $(-\Delta V,\Delta V)$, and the compact part is controlled by the separation estimate of Appendix~\ref{app:exact_mvt_sep}. The resulting density-zero conclusion~\eqref{eq:finite_energy_mvt_density_zero}, together with the Landau exclusion, is proved in Theorem~\ref{thm:no_go}.

\subsection{Finite-Dimensional Rank Equivalence}
\label{app:finite_rank_equiv}
Let $V_M=\mathrm{span}\{\phi_1,\ldots,\phi_M\}$ and write
\[
    p(t)=\sum_{k=1}^M c_k\phi_k(t).
\]
The finite MT samples are
\[
    y_n=p(t_n)=\sum_{k=1}^M (H_{\mathrm{fin}})_{n,k}c_k,
    \qquad
    (H_{\mathrm{fin}})_{n,k}=\phi_k(t_n),
\]
or $\mathbf{y}=H_{\mathrm{fin}}\mathbf{c}$.

If $\mathrm{rank}(H_{\mathrm{fin}})=M$, then $H_{\mathrm{fin}}^*H_{\mathrm{fin}}$ is positive definite. The least-squares inverse
\[
    H_{\mathrm{fin}}^\dagger=(H_{\mathrm{fin}}^*H_{\mathrm{fin}})^{-1}H_{\mathrm{fin}}^*
\]
satisfies $H_{\mathrm{fin}}^\dagger H_{\mathrm{fin}}=I_M$, so $\mathbf{c}=H_{\mathrm{fin}}^\dagger\mathbf{y}$ is uniquely recovered. Moreover,
\[
    \sigma_{\min}(H_{\mathrm{fin}})\|\mathbf{c}\|_2
    \le
    \|H_{\mathrm{fin}}\mathbf{c}\|_2
    \le
    \sigma_{\max}(H_{\mathrm{fin}})\|\mathbf{c}\|_2,
\]
which gives finite frame bounds $A_M=\sigma_{\min}(H_{\mathrm{fin}})^2$ and $B_M=\sigma_{\max}(H_{\mathrm{fin}})^2$ in an orthonormal model basis.

Conversely, if $\mathrm{rank}(H_{\mathrm{fin}})<M$, then $\ker H_{\mathrm{fin}}$ contains some nonzero $\mathbf{c}_0$. The nonzero model signal
\[
    p_0(t)=\sum_{k=1}^M(c_0)_k\phi_k(t)
\]
satisfies $p_0(t_n)=0$ for all recorded samples. Hence $p$ and $p+p_0$ produce identical finite MT sample values on the record. No algorithm can distinguish them from those data alone. This proves the finite-record equivalence~\eqref{eq:finite_exact_rank_condition}.

\subsection{Critical-Density Complete Interpolation}
\label{app:pavlov_a2}
This subsection supplies the details behind the Shannon-type MT--Pavlov theorem in Subsection~\ref{sec:critical_complete_interp}. It is useful to separate three notions. A sampling sequence may be redundant; then its natural frequency-domain object is a frame of exponentials. A complete interpolating sequence is non-redundant; then the exponential system is a Riesz basis and one obtains a cardinal Lagrange formula. The Shannon grid is the special complete-interpolation sequence whose generating function is $\sin\pi z$.

\textbf{Step 1: Scaling the MT nodes.}
Let $S=\{s_n\}_{n\in\mathbb{Z}}\subseteq\Lambda_{\Delta V}(x)$ and define
\[
    \mu_n=\frac{\Omega}{\pi}s_n.
\]
For $f\in PW_\Omega$, set
\[
    g(u)=f\!\left(\frac{\pi}{\Omega}u\right).
\]
We verify directly that $g\in PW_\pi$. Write the inverse Fourier representation of $f$:
\[
    f(t)=\frac{1}{2\pi}\int_{-\Omega}^{\Omega}
        \hat f(\omega)e^{j\omega t}\,d\omega .
\]
Substituting $t=\pi u/\Omega$ gives
\[
    g(u)=\frac{1}{2\pi}\int_{-\Omega}^{\Omega}
        \hat f(\omega)e^{j(\pi\omega/\Omega)u}\,d\omega .
\]
Let $\xi=\pi\omega/\Omega$, so $\omega=\Omega\xi/\pi$ and $d\omega=(\Omega/\pi)d\xi$. Then
\[
    g(u)=
    \frac{1}{2\pi}
    \int_{-\pi}^{\pi}
        \frac{\Omega}{\pi}
        \hat f\!\left(\frac{\Omega}{\pi}\xi\right)
        e^{j\xi u}\,d\xi .
\]
Thus
\[
    \widehat g(\xi)
    =
    \frac{\Omega}{\pi}
    \hat f\!\left(\frac{\Omega}{\pi}\xi\right),
    \qquad
    \operatorname{supp}\widehat g\subseteq[-\pi,\pi].
\]
The samples agree under the same scaling:
\[
    g(\mu_n)
    =
    f\!\left(\frac{\pi}{\Omega}\mu_n\right)
    =
    f(s_n).
\]
The $L^2$ norm changes only by a constant:
\[
    \|g\|_{L^2(\mathbb{R})}^2
    =
    \int_{\mathbb{R}}
        \left|f\!\left(\frac{\pi}{\Omega}u\right)\right|^2\,du
    =
    \frac{\Omega}{\pi}
    \|f\|_{L^2(\mathbb{R})}^2.
\]
Consequently $S$ is a complete interpolating sequence for $PW_\Omega$ exactly when $\{\mu_n\}$ is a complete interpolating sequence for $PW_\pi$, after the harmless norm scaling above.

\textbf{Step 2: Complete interpolation and exponential Riesz bases.}
For $g\in PW_\pi$,
\[
    g(\mu_n)
    =
    \frac{1}{2\pi}
    \int_{-\pi}^{\pi}
        \widehat g(\xi)e^{j\xi\mu_n}\,d\xi .
\]
Set $h(\xi)=\widehat g(\xi)/\sqrt{2\pi}$. By Plancherel, $\|h\|_2=\|g\|_2$. With the $L^2[-\pi,\pi]$ inner product linear in the first argument,
\[
    g(\mu_n)
    =
    \frac{1}{\sqrt{2\pi}}
    \langle h,e^{-j\mu_n\xi}\rangle_{L^2[-\pi,\pi]}.
\]
Therefore the restriction map $g\mapsto(g(\mu_n))$ is, up to the constant $(2\pi)^{-1/2}$, the analysis operator of the exponential system
\[
    E_\mu=\{e^{-j\mu_n\xi}\}_{n\in\mathbb{Z}}.
\]
If $E_\mu$ is a Riesz basis, this analysis map is an isomorphism from $L^2[-\pi,\pi]$ onto $\ell^2(\mathbb{Z})$, hence the sample map is an isomorphism from $PW_\pi$ onto $\ell^2(\mathbb{Z})$. Conversely, if the sample map is an isomorphism, then the exponential analysis map is an isomorphism, so $E_\mu$ is the image of the standard orthonormal basis under a bounded invertible operator and is a Riesz basis. This proves the equivalence between complete interpolation and the Riesz-basis statement used in the Shannon-type MT--Pavlov theorem.

\textbf{Step 3: Pavlov/Lyubarskii--Seip criterion.}
Let $G$ be the canonical product with real zero set $\{\mu_n\}$, and assume the standard hypotheses under which $G$ is the generating function of the sequence: the zeros are separated, $G$ has the correct exponential type $\pi$, and the product has the usual convergence properties for real separated complete-interpolation sequences. Pavlov's theorem, in the Lyubarskii--Seip formulation, states that $\{\mu_n\}$ is a complete interpolating sequence for $PW_\pi$ if and only if
\[
    W_G(u)^2
    =
    \left(
        \frac{|G(u)|}{\mathrm{dist}(u,\{\mu_n\})}
    \right)^2
\]
belongs to the Muckenhoupt class $A_2$, namely if and only if~\eqref{eq:muckenhoupt_a2_condition} holds~\cite{LyubarskiiSeip1997}. This condition is exact in the complete-interpolation regime, but it is still global: evaluating it requires the canonical product, hence the complete zero geometry of the selected MT subsequence.

\textbf{Step 4: Lagrange functions and reconstruction formula.}
Assume the equivalent conditions above hold. Define
\[
    \ell_n(u)=
    \frac{G(u)}
         {G'(\mu_n)(u-\mu_n)}.
\]
For $m\ne n$, the numerator vanishes at $u=\mu_m$ while the denominator is finite, so
\[
    \ell_n(\mu_m)=0.
\]
For $m=n$, both numerator and denominator vanish to first order. Since $\mu_n$ is a simple zero of the generating function in the complete-interpolation setting,
\[
    G(u)=G'(\mu_n)(u-\mu_n)+\mathcal{O}((u-\mu_n)^2),
    \qquad u\to\mu_n.
\]
Hence
\[
    \ell_n(\mu_n)
    =
    \lim_{u\to\mu_n}
    \frac{G(u)}
         {G'(\mu_n)(u-\mu_n)}
    =
    1.
\]
Thus $\ell_n(\mu_m)=\delta_{nm}$. Since the restriction map is an isomorphism onto $\ell^2$, the inverse image of the standard basis vector $e_n\in\ell^2$ is precisely $\ell_n$, and every $g\in PW_\pi$ has the expansion
\[
    g(u)
    =
    \sum_{n\in\mathbb{Z}}g(\mu_n)\ell_n(u)
    =
    \sum_{n\in\mathbb{Z}}
    g(\mu_n)
    \frac{G(u)}
         {G'(\mu_n)(u-\mu_n)}.
\]
The convergence is the convergence supplied by the Riesz-basis inverse in $PW_\pi$; standard Paley-Wiener estimates also give local uniform convergence on finite horizontal strips. Returning to physical time $u=\Omega t/\pi$ and using $g(\mu_n)=f(s_n)$ yields
\[
    f(t)
    =
    \sum_{n\in\mathbb{Z}}
    f(s_n)
    \frac{G(\Omega t/\pi)}
         {G'(\mu_n)(\Omega t/\pi-\mu_n)}.
\]
This is exactly~\eqref{eq:mvt_shannon_pavlov_formula}. For the generating waveform measured by ideal MT,
\[
    f(s_n)=x(s_n)=m_n\Delta V,
\]
so the reconstruction coefficients are the recorded threshold labels multiplied by $\Delta V$.

\textbf{Step 5: Recovery of the Shannon formula.}
Let the MT subsequence coincide with a shifted Nyquist grid,
\[
    s_n=\tau+\frac{\pi}{\Omega}n.
\]
Then
\[
    \mu_n=\frac{\Omega}{\pi}s_n
    =
    \frac{\Omega\tau}{\pi}+n
    =
    \alpha+n,
    \qquad
    \alpha=\frac{\Omega\tau}{\pi}.
\]
Choose
\[
    G(u)=\sin\pi(u-\alpha).
\]
Then $G(\alpha+n)=0$ and
\[
    G'(\alpha+n)
    =
    \pi\cos\pi n
    =
    \pi(-1)^n.
\]
The $n$-th Lagrange kernel becomes
\begin{align}
    \ell_n(u)
    &=
    \frac{\sin\pi(u-\alpha)}
         {\pi(-1)^n(u-\alpha-n)}  \\
    &=
    \frac{\sin\{\pi[(u-\alpha-n)+n]\}}
         {\pi(-1)^n(u-\alpha-n)}  \\
    &=
    \frac{(-1)^n\sin\pi(u-\alpha-n)}
         {\pi(-1)^n(u-\alpha-n)}  \\
    &=
    \mathrm{sinc}(u-\alpha-n).
\end{align}
Substituting $u=\Omega t/\pi$ gives
\[
    \ell_n(\Omega t/\pi)
    =
    \mathrm{sinc}\!\left(
        \frac{\Omega t}{\pi}
        -\frac{\Omega\tau}{\pi}
        -n
    \right)
    =
    \mathrm{sinc}\!\left(
        \frac{\Omega(t-\tau)}{\pi}-n
    \right).
\]
Thus the MT--Pavlov formula reduces to the shifted Shannon expansion~\eqref{eq:shifted_shannon_from_pavlov}. In this special case the canonical product $G$ is just the sine generating function, and the general Lagrange kernels are the ordinary sinc kernels.

\textbf{Step 6: Why density alone is not the theorem.}
The complete-interpolation theorem gives an explicit kernel, but its condition is not merely $D^-=1$ after normalization. The normalized Shannon grid $\mathbb{Z}$ has critical density and is a complete interpolating sequence. Removing one point gives $\mathbb{Z}\setminus\{0\}$, whose lower Beurling density is still one. However,
\[
    q(u)=\mathrm{sinc}(u)
    =
    \frac{\sin\pi u}{\pi u}
\]
belongs to $PW_\pi$, satisfies $q(0)=1$, and satisfies
\[
    q(k)=0,\qquad k\in\mathbb{Z}\setminus\{0\}.
\]
Therefore $\mathbb{Z}\setminus\{0\}$ is not a uniqueness set for $PW_\pi$, and the sampling lower bound is zero even though the density has not changed. This example explains why a Shannon-type MT theorem must retain the canonical product or an equivalent global geometric object. Threshold spacing $\Delta V$ and average retained event density influence the node set generated by MT, but they cannot by themselves decide the critical non-redundant case.

\subsection{Periodic-Gap Complete Interpolation}
\label{app:periodic_gap_shannon}
This subsection proves the periodic-gap MT--Shannon theorem stated in Subsection~\ref{sec:periodic_gap_shannon}. The proof is independent of the full Pavlov $A_2$ criterion once the normalized gaps are periodic. The only external sampling-theoretic input is Landau's density necessity: a complete interpolating sequence for $PW_\pi$ is both a stable sampling set and an interpolation set, and therefore must have lower and upper Beurling density equal to one.

\textbf{Step 1: Normalize the physical problem.}
The scaling argument in Appendix~\ref{app:pavlov_a2} already shows that
\[
    f(t)\in PW_\Omega
    \quad\Longleftrightarrow\quad
    g(u)=f\!\left(\frac{\pi}{\Omega}u\right)\in PW_\pi,
\]
with
\[
    g(\mu_n)=f(s_n),
    \qquad
    \mu_n=\frac{\Omega s_n}{\pi}.
\]
Thus $S=\{s_n\}$ is a Shannon-type complete-interpolation sampling sequence for $PW_\Omega$ if and only if $\mu=\{\mu_n\}$ is a complete interpolating sequence for $PW_\pi$. Under the periodic-gap hypothesis
\[
    \mu_{n+1}-\mu_n=h_{n\bmod q},
    \qquad h_r>0,
\]
we have, for every $n$,
\[
    \mu_{n+q}-\mu_n
    =
    \sum_{r=0}^{q-1}h_r
    =
    H.
\]
Therefore every interval of normalized length $H$ contains exactly $q$ points, up to endpoint conventions.

\textbf{Step 2: Density forces $H=q$.}
The normalized periodic set has uniform Beurling density, where $D^+$ denotes the upper Beurling density (the limsup/sup analogue of $D^-$),
\[
    D^-(\mu)=D^+(\mu)=\frac{q}{H}.
\]
Indeed, for a long interval $[a,a+R]$, the number of complete periods contained in the interval is $R/H+\mathcal{O}(1)$, and each complete period contributes $q$ points. Hence
\[
    \#(\mu\cap[a,a+R])
    =
    \frac{q}{H}R+\mathcal{O}(1),
    \qquad R\to\infty,
\]
uniformly in the starting point $a$, which gives both lower and upper densities equal to $q/H$.

If $\mu$ is a complete interpolating sequence for $PW_\pi$, then it is a sampling set and an interpolation set. Landau's necessary density conditions for $PW_\pi$ give
\[
    D^-(\mu)\ge 1,
    \qquad
    D^+(\mu)\le 1.
\]
Substituting $D^-(\mu)=D^+(\mu)=q/H$ yields
\[
    \frac{q}{H}\ge 1,
    \qquad
    \frac{q}{H}\le 1,
\]
and hence
\[
    \frac{q}{H}=1,
    \qquad
    H=q.
\]
This proves the necessity of the finite block condition~\eqref{eq:periodic_mvt_shannon_condition}.

\textbf{Step 3: Periodic nodes become finitely many lattice translates.}
Assume now that $H=q$. Choose the residue representative $r\in\{0,\ldots,q-1\}$ and write $n=kq+r$. Define
\[
    \beta_0=\mu_0,
    \qquad
    \beta_r=\mu_0+\sum_{\ell=0}^{r-1}h_\ell,
    \quad r=1,\ldots,q-1.
\]
By summing the gaps,
\[
    \mu_{kq+r}
    =
    \mu_0+kH+\sum_{\ell=0}^{r-1}h_\ell.
\]
Since $H=q$,
\[
    \mu_{kq+r}=qk+\beta_r.
\]
Therefore
\[
    \mu
    =
    \bigcup_{r=0}^{q-1}(q\mathbb{Z}+\beta_r).
\]
Because all gaps $h_r$ are strictly positive and the full period has length $q$, the representatives satisfy
\[
    \beta_0<\beta_1<\cdots<\beta_{q-1}<\beta_0+q
\]
after the natural ordering along one period. Hence the residues $\beta_r\bmod q$ are distinct.

\textbf{Step 4: The finite fiber matrix is invertible.}
It remains to prove that
\[
    E_\mu=\{e^{-j(qk+\beta_r)u}:k\in\mathbb{Z},\ r=0,\ldots,q-1\}
\]
is a Riesz basis for $L^2[-\pi,\pi]$. A translation of the integration interval only multiplies each exponential by a unit-modulus constant, so it suffices to work on $L^2[0,2\pi]$.

Take finitely supported coefficients $c_{k,r}$ and define
\[
    F(u)=
    \sum_{r=0}^{q-1}\sum_{k\in\mathbb{Z}}
        c_{k,r}e^{-j(qk+\beta_r)u}.
\]
Fiber the interval $[0,2\pi]$ into $q$ subintervals by writing
\[
    u=\theta+\frac{2\pi m}{q},
    \qquad
    0\le\theta<\frac{2\pi}{q},
    \qquad
    m=0,\ldots,q-1.
\]
For each residue $r$, set
\[
    C_r(\theta)=
    \sum_{k\in\mathbb{Z}}c_{k,r}e^{-jqk\theta}.
\]
Then
\begin{align}
    F\!\left(\theta+\frac{2\pi m}{q}\right)
    &=
    \sum_{r=0}^{q-1}\sum_k
    c_{k,r}
    e^{-j(qk+\beta_r)\theta}
    e^{-j(qk+\beta_r)2\pi m/q}       \\
    &=
    \sum_{r=0}^{q-1}
    e^{-j\beta_r\theta}
    e^{-j2\pi m\beta_r/q}
    \sum_k c_{k,r}e^{-jqk\theta},
\end{align}
because $e^{-j2\pi mk}=1$. In vector form,
\[
    \mathbf{F}(\theta)
    =
    V_{\mathrm{per}}D(\theta)\mathbf{C}(\theta),
\]
where
\[
    \mathbf{F}_m(\theta)=
    F\!\left(\theta+\frac{2\pi m}{q}\right),
    \qquad m=0,\ldots,q-1,
\]
the matrix $D(\theta)$ is diagonal with
\[
    (D(\theta))_{r,r}=e^{-j\beta_r\theta},
    \qquad r=0,\ldots,q-1,
\]
and
\[
    (V_{\mathrm{per}})_{m,r}
    =
    e^{-j2\pi m\beta_r/q},
    \qquad
    m,r=0,\ldots,q-1.
\]
The matrix $V_{\mathrm{per}}$ is a Vandermonde matrix in the distinct nodes
\[
    z_r=e^{-j2\pi\beta_r/q}.
\]
Indeed,
\[
    (V_{\mathrm{per}})_{m,r}=z_r^m.
\]
Its determinant is
\[
    \det V_{\mathrm{per}}
    =
    \prod_{0\le r<s\le q-1}(z_s-z_r),
\]
which is nonzero because the residues $\beta_r\bmod q$ are distinct. Hence $V_{\mathrm{per}}$ is invertible. Let $\sigma_{\min}$ and $\sigma_{\max}$ denote its smallest and largest singular values. Since $D(\theta)$ is unitary for real $\theta$,
\[
    \sigma_{\min}^2\|\mathbf{C}(\theta)\|_2^2
    \le
    \|\mathbf{F}(\theta)\|_2^2
    \le
    \sigma_{\max}^2\|\mathbf{C}(\theta)\|_2^2.
\]
Integrating over the base interval gives
\[
    \int_0^{2\pi}\!|F(u)|^2\,du
    =
    \int_0^{2\pi/q}\!\|\mathbf{F}(\theta)\|_2^2\,d\theta.
\]
The functions $\{e^{-jqk\theta}\}_{k\in\mathbb{Z}}$ are orthogonal on $[0,2\pi/q]$ and have squared norm $2\pi/q$, so
\[
    \int_0^{2\pi/q}\!\|\mathbf{C}(\theta)\|_2^2\,d\theta
    =
    \frac{2\pi}{q}
    \sum_{r=0}^{q-1}\sum_{k\in\mathbb{Z}}|c_{k,r}|^2.
\]
Combining the last three displays yields the Riesz inequalities
\[
    \frac{2\pi}{q}\sigma_{\min}^2
    \sum_{r,k}|c_{k,r}|^2
    \le
    \|F\|_{L^2[0,2\pi]}^2
    \le
    \frac{2\pi}{q}\sigma_{\max}^2
    \sum_{r,k}|c_{k,r}|^2.
\]
Thus $E_\mu$ is a Riesz sequence.

It is also complete. The fiber map identifies $L^2[0,2\pi]$ with vector-valued functions $\mathbf{F}(\theta)\in L^2([0,2\pi/q];\mathbb{C}^q)$. The scalar Fourier systems $\{e^{-jqk\theta}\}_{k\in\mathbb{Z}}$ span each component $L^2[0,2\pi/q]$, while the matrix $V_{\mathrm{per}}D(\theta)$ is invertible for every $\theta$. Therefore the closed span of $E_\mu$ is all of $L^2[0,2\pi]$, hence also all of $L^2[-\pi,\pi]$ after the harmless interval translation. This proves that $E_\mu$ is a Riesz basis.

\textbf{Step 5: Riesz basis implies complete interpolation.}
For $g\in PW_\pi$,
\[
    g(\mu_n)
    =
    \frac{1}{2\pi}
    \int_{-\pi}^{\pi}\widehat g(\xi)e^{j\mu_n\xi}\,d\xi.
\]
Therefore the sample map $g\mapsto(g(\mu_n))$ is, up to the constant $(2\pi)^{-1/2}$, the analysis map of the exponential system $E_\mu$. Since $E_\mu$ is a Riesz basis, that analysis map is an isomorphism onto $\ell^2(\mathbb{Z})$. Hence $\mu$ is a complete interpolating sequence for $PW_\pi$. Scaling back by $s_n=(\pi/\Omega)\mu_n$ proves Shannon-type stable exact recovery on $PW_\Omega$. This proves sufficiency of $H=q$ and completes the equivalence of statements (i)--(iv).

\textbf{Step 6: The explicit generating function and Lagrange kernels.}
For the lattice union
\[
    \mu=\bigcup_{r=0}^{q-1}(q\mathbb{Z}+\beta_r),
\]
define
\[
    G_{\mathrm{per}}(u)
    =
    \prod_{r=0}^{q-1}
    \sin\!\left(\frac{\pi}{q}(u-\beta_r)\right).
\]
Each factor has zeros at $\beta_r+q\mathbb{Z}$. Since the residues $\beta_r\bmod q$ are distinct, the zero sets of the factors do not overlap, so every zero of $G_{\mathrm{per}}$ is simple and the full zero set is exactly $\mu$. Each sine factor has exponential type $\pi/q$, hence the product has exponential type $\pi$, the correct type for $PW_\pi$ complete interpolation.

For each node $\mu_n$, define
\[
    \ell_n^{\mathrm{per}}(u)
    =
    \frac{G_{\mathrm{per}}(u)}
         {G_{\mathrm{per}}'(\mu_n)(u-\mu_n)}.
\]
If $m\ne n$, then $G_{\mathrm{per}}(\mu_m)=0$ while the denominator is finite, so
\[
    \ell_n^{\mathrm{per}}(\mu_m)=0.
\]
At the node $\mu_n$, the zero is simple:
\[
    G_{\mathrm{per}}(u)
    =
    G_{\mathrm{per}}'(\mu_n)(u-\mu_n)
    +
    \mathcal{O}((u-\mu_n)^2),
    \qquad u\to\mu_n.
\]
Therefore
\[
    \ell_n^{\mathrm{per}}(\mu_n)
    =
    \lim_{u\to\mu_n}
    \frac{G_{\mathrm{per}}(u)}
         {G_{\mathrm{per}}'(\mu_n)(u-\mu_n)}
    =
    1.
\]
Thus $\ell_n^{\mathrm{per}}(\mu_m)=\delta_{nm}$. Since $\mu$ is complete interpolating, these cardinal functions are the inverse images of the standard basis vectors of $\ell^2$, and every $g\in PW_\pi$ has the expansion
\[
    g(u)=
    \sum_{n\in\mathbb{Z}}
    g(\mu_n)
    \ell_n^{\mathrm{per}}(u)
    =
    \sum_{n\in\mathbb{Z}}
    g(\mu_n)
    \frac{G_{\mathrm{per}}(u)}
         {G_{\mathrm{per}}'(\mu_n)(u-\mu_n)}.
\]
Returning to physical time $u=\Omega t/\pi$ gives
\[
    f(t)
    =
    \sum_{n\in\mathbb{Z}}
    f(s_n)
    \frac{G_{\mathrm{per}}(\Omega t/\pi)}
         {G_{\mathrm{per}}'(\mu_n)(\Omega t/\pi-\mu_n)}.
\]
For the ideal MT waveform itself,
\[
    f(s_n)=x(s_n)=m_n\Delta V,
\]
so the coefficients are exactly the recorded threshold labels multiplied by the threshold spacing.

\textbf{Step 7: Reduction to the classical Shannon formula.}
When $q=1$, the periodic condition is
\[
    H=h_0=1.
\]
Hence
\[
    \mu_n=\mu_0+n.
\]
The product~\eqref{eq:periodic_generating_function} becomes
\[
    G_{\mathrm{per}}(u)=\sin\pi(u-\mu_0),
\]
and, for $\mu_n=\mu_0+n$,
\[
    G_{\mathrm{per}}'(\mu_n)
    =
    \pi\cos\pi n
    =
    \pi(-1)^n.
\]
The Lagrange kernel is
\begin{align}
    \ell_n^{\mathrm{per}}(u)
    &=
    \frac{\sin\pi(u-\mu_0)}
         {\pi(-1)^n(u-\mu_0-n)} \\
    &=
    \frac{\sin\pi(u-\mu_0-n)}
         {\pi(u-\mu_0-n)}
    =
    \mathrm{sinc}(u-\mu_0-n).
\end{align}
Thus the periodic theorem reduces exactly to the shifted Shannon formula.

\textbf{Step 8: MT form of the block condition.}
The normalized gap is
\[
    h_{n\bmod q}
    =
    \mu_{n+1}-\mu_n
    =
    \frac{\Omega}{\pi}(s_{n+1}-s_n)
    =
    \frac{s_{n+1}-s_n}{\Delta T}.
\]
Therefore $H=q$ is equivalent to
\[
    \sum_{r=0}^{q-1}
    \frac{s_{kq+r+1}-s_{kq+r}}{\Delta T}
    =
    q,
\]
or
\[
    \sum_{r=0}^{q-1}(s_{kq+r+1}-s_{kq+r})
    =
    q\Delta T
    =
    \frac{q\pi}{\Omega}.
\]
If the consecutive retained events are adjacent-threshold crossings on monotone branches, then
\[
    |x(s_{n+1})-x(s_n)|=\Delta V.
\]
The ordinary mean-value theorem gives a point $\xi_n$ between the two crossings such that
\[
    |x'(\xi_n)|
    =
    \frac{|x(s_{n+1})-x(s_n)|}{s_{n+1}-s_n}
    =
    \frac{\Delta V}{s_{n+1}-s_n}.
\]
Substituting this into the normalized gap yields
\[
    h_{n\bmod q}
    =
    \frac{\Omega\Delta V}{\pi |x'(\xi_n)|}.
\]
Summing over one period gives the velocity-domain form
\[
    \sum_{r=0}^{q-1}
    \frac{\Delta V}{|x'(\xi_{k,r})|}
    =
    \frac{q\pi}{\Omega}.
\]
This is the promised explicit MT interpretation: in a periodic normalized gap stream, Shannon-type exact reconstruction holds exactly when one period of $q$ consecutive adjacent-threshold time gaps spans $q$ Nyquist intervals.

\section{Reproducing Kernel and Trace Formula}
\label{app:trace}
We derive the reproducing kernel $K(t,s)$ of $PW_\Omega$ and verify the trace formula \eqref{eq:trace}.

\textbf{Reproducing kernel.} For any $x \in PW_\Omega$, its inverse Fourier representation is:
\begin{equation}
    x(t) = \frac{1}{2\pi} \int_{-\Omega}^{\Omega} \hat{x}(\omega)\, e^{j\omega t}\, d\omega
    \label{eq:inverse_ft}
\end{equation}
A reproducing kernel $K(t,s)$ is a function satisfying $x(s) = \langle x, K(\cdot, s) \rangle_{L^2}$ for all $x \in PW_\Omega$. By Parseval's theorem, the $L^2$ inner product can be expressed in the frequency domain:
\begin{equation}
    \langle x, K(\cdot,s) \rangle_{L^2} = \frac{1}{2\pi} \int_{-\Omega}^{\Omega} \hat{x}(\omega)\, \overline{\widehat{K}(\omega, s)}\, d\omega
\end{equation}
Comparing with \eqref{eq:inverse_ft}, we identify $\overline{\widehat{K}(\omega, s)} = e^{j\omega s}$ for $|\omega| \le \Omega$, i.e., $\widehat{K}(\omega, s) = e^{-j\omega s}$. Computing the inverse Fourier transform:
\begin{align}
    K(t, s) & = \frac{1}{2\pi} \int_{-\Omega}^{\Omega} e^{j\omega(t-s)}\, d\omega
    = \frac{1}{2\pi} \left[\frac{e^{j\omega(t-s)}}{j(t-s)}\right]_{-\Omega}^{\Omega} \\
            & = \frac{e^{j\Omega(t-s)} - e^{-j\Omega(t-s)}}{2\pi j(t-s)}
    = \frac{\sin[\Omega(t-s)]}{\pi(t-s)}
    \label{eq:kernel_derivation}
\end{align}
which confirms \eqref{eq:repro_kernel}.

\textbf{Trace formula.} The time-frequency concentration operator $T_R = B_\Omega D_R B_\Omega$ acts on $PW_\Omega$ via the integral kernel $K(t,s)$ restricted to $[-R, R]$:
\begin{equation}
    (T_R f)(t) = \int_{-R}^{R} K(t,s)\, f(s)\, ds
\end{equation}
The trace of such an integral operator equals the integral of the kernel on the diagonal:
\begin{equation}
    \mathrm{Tr}(T_R) = \int_{-R}^{R} K(t,t)\, dt
\end{equation}
Evaluating $K(t,t)$ by L'H\^{o}pital's rule (or directly from the integral representation):
\begin{equation}
    K(t,t) = \lim_{s \to t} \frac{\sin[\Omega(t-s)]}{\pi(t-s)} = \frac{\Omega}{\pi}
\end{equation}
Therefore:
\begin{equation}
    \mathrm{Tr}(T_R) = \int_{-R}^{R} \frac{\Omega}{\pi}\, dt = \frac{2\Omega R}{\pi}
\end{equation}
confirming \eqref{eq:trace}. Since $\Omega = 2\pi f_{max}$, this equals $4 f_{max} R$, representing the number of Nyquist intervals within $[-R, R]$.

\section{MT Crossing Density Bounds}
\label{app:density}
We provide a detailed derivation of the retained event-density bounds used in Section~\ref{sec:structure_density}.

\textbf{Setup.} Consider the retained MT crossing set $\Lambda_{\Delta V}$ restricted to a time window $[t, t+r] \subseteq [0, T_{\text{obs}}]$. Let $n_{\mathrm{ret}}(t,r)$ denote the number of retained crossings in this window, and let $N_{\mathrm{supp}}(t,r)$ denote the number of monotone pieces whose first raw threshold hit is suppressed as a same-threshold return.

\textbf{Step 1: Upper bound on $n_{\mathrm{ret}}(t,r)$.} Let the retained events in $[t,t+r]$ be
\[
    u_1<\cdots<u_N.
\]
For every consecutive retained pair, the adjacent-transition convention gives threshold labels differing by one. Therefore
\[
    |x(u_{i+1})-x(u_i)|=\Delta V,
\]
and the total variation between these two retained times obeys
\[
    \int_{u_i}^{u_{i+1}}|x'(\tau)|\,d\tau\ge\Delta V.
\]
Summing over $i=1,\ldots,N-1$ gives
\[
    (N-1)\Delta V
    \le
    \int_t^{t+r}|x'(\tau)|\,d\tau .
\]
Thus $N\le \mathrm{TV}(x;[t,t+r])/\Delta V+1$; we keep the slightly looser $+2$ used in the main text to cover endpoint conventions and extracted subsequences. Therefore:
\begin{equation}
    n_{\mathrm{ret}}(t, r) \le \frac{1}{\Delta V} \int_t^{t+r} |x'(\tau)|\, d\tau + 2
\end{equation}
This proof uses retained adjacent transitions. It is not a raw-crossing bound: repeated up/down crossings of the same threshold can have small vertical travel between successive raw hits.

\textbf{Step 2: Lower bound on $n_{\mathrm{ret}}(t,r)$.} Decompose $[t, t+r]$ into $N_{\mathrm{mono}}(t,r)$ maximal monotone segments $I_1, \ldots, I_{N_{\mathrm{mono}}(t,r)}$, where $x$ is strictly monotone on each $I_k$. Let $\mathrm{TV}_k = \int_{I_k} |x'(\tau)|\, d\tau$ denote the total variation on segment $I_k$. Because the threshold phase and the endpoints of $I_k$ can remove at most one guaranteed interval from a monotone traversal, the raw threshold-hit count on $I_k$ is bounded below by $\mathrm{TV}_k/\Delta V-1$. Equivalently, using $\lfloor y \rfloor \ge y-1$ for $y\ge0$:
\begin{equation}
    \lfloor \mathrm{TV}_k / \Delta V \rfloor \ge \mathrm{TV}_k / \Delta V - 1
\end{equation}
On a strictly monotone segment, same-threshold suppression can remove at most the first raw hit: after the first retained adjacent-level transition, all later threshold indices continue monotonically and differ from the previously retained index. Thus the retained count is lower bounded by the raw floor count minus one additional unit on each segment counted by $N_{\mathrm{supp}}(t,r)$. Summing over all $N_{\mathrm{mono}}(t,r)$ segments:
\begin{align}
    n_{\mathrm{ret}}(t, r)
     & \ge \sum_{k=1}^{N_{\mathrm{mono}}(t,r)} \left(\frac{\mathrm{TV}_k}{\Delta V} - 1\right) \\
     &\quad -N_{\mathrm{supp}}(t,r)\\
     & = \frac{1}{\Delta V} \sum_{k=1}^{N_{\mathrm{mono}}(t,r)} \mathrm{TV}_k - N_{\mathrm{mono}}(t,r)-N_{\mathrm{supp}}(t,r)
\end{align}
Since $\sum_k \mathrm{TV}_k = \int_t^{t+r} |x'(\tau)|\, d\tau$ (the total variations of monotone segments sum to the total variation of the signal; see Appendix~\ref{app:tv_additivity} for a detailed derivation):
\begin{equation}
    n_{\mathrm{ret}}(t, r) \ge \frac{1}{\Delta V} \int_t^{t+r} |x'(\tau)|\, d\tau - N_{\mathrm{mono}}(t,r)-N_{\mathrm{supp}}(t,r)
\end{equation}

\textbf{Step 3: Bounding $N_{\mathrm{mono}}(t,r)$ by an explicit extremum-density certificate.} The number of maximal monotone segments $N_{\mathrm{mono}}(t,r)$ equals the number of sign changes of $x'(t)$ in $[t, t+r]$ plus one. Since sign changes occur only at zeros of $x'$, we have $N_{\mathrm{mono}}(t,r) \le Z(t,r) + 1$, where $Z(t,r)$ is the number of zeros of $x'$ in $[t,t+r]$.

No bandwidth-only bound of the form $Z(t,r)\le \Omega r/\pi+\mathcal{O}(1)$ holds uniformly on finite intervals for all of $PW_\Omega$: superoscillatory bandlimited functions can force arbitrarily many local oscillations while keeping the same spectral support~\cite{Kempf2000}. The density estimate used in Section~\ref{sec:density_analysis} therefore assumes an explicit certificate
\begin{equation}
    Z(t,r) \le \eta_0 r + C_0,
    \label{eq:app_zero_density_assumption}
\end{equation}
where $\eta_0$ and $C_0$ are signal-instance or signal-class constants. Thus
\begin{equation}
    N_{\mathrm{mono}}(t,r) \le \eta_0 r+\mathcal{O}(1).
\end{equation}
Similarly assume or verify the suppression-density certificate
\begin{equation}
    N_{\mathrm{supp}}(t,r) \le \eta_{\mathrm{supp},0}r+\mathcal{O}(1),
\end{equation}
with the retained-safe fallback $\eta_{\mathrm{supp},0}\le\eta_0$ because $N_{\mathrm{supp}}\le N_{\mathrm{mono}}$.

\textbf{Calibration example.} For the non-superoscillatory pure tone $x(t)=A\sin(\Omega t)$, $x'(t)=A\Omega\cos(\Omega t)$ has zeros spaced $\pi/\Omega=1/(2f_{max})$ apart. This particular signal admits the certificate $\eta_0=2f_{max}$ up to endpoint constants. The example calibrates one common signal family; it is not a universal Paley-Wiener zero-counting theorem.

\textbf{Step 4: Assembly.} Dividing the crossing count bounds by $r$:
\begin{align}
    \frac{\bar{\nu}(t,r)}{\Delta V}
    - \eta_0-\eta_{\mathrm{supp},0} - \mathcal{O}(1/r)
    &\le \frac{n_{\mathrm{ret}}(t,r)}{r}\\
    &\le \frac{\bar{\nu}(t,r)}{\Delta V} + \mathcal{O}(1/r)
\end{align}
where $\bar{\nu}(t,r) = \frac{1}{r}\int_t^{t+r}|x'(\tau)|\,d\tau$ is the local mean absolute velocity. For interval lengths $r \gg \Delta T$, the endpoint term $\mathcal{O}(1/r)$ becomes negligible. Taking the infimum over such long sub-intervals and then using the definition of $\nu_{min}$ yields the asymptotic envelopes
\begin{equation}
    \frac{\nu_{min}}{\Delta V} - \eta_0-\eta_{\mathrm{supp},0} \;\lesssim\; \inf \frac{n_{\mathrm{ret}}(t,r)}{r} \;\lesssim\; \frac{\nu_{min}}{\Delta V},
\end{equation}
which is the conditional long-window statement used in Section~\ref{sec:density_analysis}.

For a finite-window claim over all $r\ge R$, keep the additive constants instead of absorbing them into $\mathcal{O}(1/r)$. If
\begin{align}
    N_{\mathrm{mono}}(t,r)
    &\le\eta_0 r+C_{\mathrm{mono}},\\
    N_{\mathrm{supp}}(t,r)
    &\le\eta_{\mathrm{supp},0}r+C_{\mathrm{supp}},
\end{align}
then the same calculation gives
\begin{align}
    \frac{n_{\mathrm{ret}}(t,r)}{r}
    &\ge
    \frac{\nu_{min}^{(R)}}{\Delta V}
    -\eta_0-\eta_{\mathrm{supp},0}\\
    &\quad
    -\frac{C_{\mathrm{mono}}+C_{\mathrm{supp}}}{R},
    \qquad r\ge R,
\end{align}
which is the finite-window certificate~\eqref{eq:finite_R_density_bound}.

\section{Kadec's $1/4$ Theorem and Centered-Window Local Constraint Details}
\label{app:constraints}
We collect the appendical material associated with Kadec's theorem. The first subsection gives a detailed proof of the classical $1/4$ theorem used in Section~\ref{sec:kadec}. The remaining subsections give the centered-window extraction arguments specific to the MT analysis of Section~\ref{sec:sufficient}.

\subsection{Detailed Proof of Kadec's $1/4$ Theorem}
We prove the sufficient direction used in Section~\ref{sec:kadec}.

\textbf{Step 1: Normalize the bandwidth and reformulate the claim.} Let $\Delta T = 1/(2f_{max})$ and assume
\begin{equation}
    \sup_n |t_n - n\Delta T| < \frac{\Delta T}{4}.
\end{equation}
Define the normalized nodes
\begin{equation}
    \lambda_n = \frac{t_n}{\Delta T} = n + \delta_n,
    \qquad \sup_n |\delta_n| \le L < \frac{1}{4}.
\end{equation}
Under the unitary dilation
\begin{equation}
    (\mathcal{D}_{\Delta T}x)(u) = \Delta T^{1/2} x(\Delta T u),
\end{equation}
the space $PW_\Omega$ is mapped onto $PW_\pi$, and the samples $x(t_n)$ become the normalized samples $(\mathcal{D}_{\Delta T}x)(\lambda_n)$. Hence it is enough to prove the normalized statement: if $\lambda_n = n + \delta_n$ with $\sup_n |\delta_n| < 1/4$, then the exponential system
\begin{equation}
    f_n(t) = \frac{1}{\sqrt{2\pi}} e^{j\lambda_n t}, \qquad t \in [-\pi,\pi],
\end{equation}
is a Riesz basis for $L^2[-\pi,\pi]$.

\textbf{Step 2: A perturbation lemma for orthonormal bases.} Let
\begin{equation}
    e_n(t) = \frac{1}{\sqrt{2\pi}} e^{jnt}, \qquad n \in \mathbb{Z},
\end{equation}
so that $\{e_n\}_{n \in \mathbb{Z}}$ is the standard orthonormal basis of $L^2[-\pi,\pi]$.

\begin{lemma}
If there exists $\theta < 1$ such that for every finitely supported sequence $\{c_n\}$,
\begin{equation}
    \left\| \sum_n c_n (f_n - e_n) \right\|_{L^2[-\pi,\pi]}^2 \le \theta^2 \sum_n |c_n|^2,
    \label{eq:pw_perturbation}
\end{equation}
then $\{f_n\}$ is a Riesz basis for $L^2[-\pi,\pi]$.
\end{lemma}

\textbf{Proof.} Define an operator $U$ on the finite linear span of $\{e_n\}$ by
\begin{equation}
    U\!\left(\sum_n c_n e_n\right) = \sum_n c_n f_n.
\end{equation}
Then
\begin{equation}
    (I-U)\!\left(\sum_n c_n e_n\right) = \sum_n c_n (e_n - f_n),
\end{equation}
so \eqref{eq:pw_perturbation} gives $\|I-U\| \le \theta < 1$ on a dense subspace, hence on all of $L^2[-\pi,\pi]$ by continuity. The Neumann series
\begin{equation}
    U^{-1} = \sum_{k=0}^{\infty} (I-U)^k
\end{equation}
therefore converges in operator norm, so $U$ is bounded and invertible. Since $f_n = Ue_n$, the sequence $\{f_n\}$ is the image of an orthonormal basis under a bounded invertible operator and is therefore a Riesz basis.

More explicitly, for every $g \in L^2[-\pi,\pi]$ we have
\begin{equation}
    \|Ug\| \le \|g\| + \|(U-I)g\| \le (1+\theta)\|g\|,
\end{equation}
while the reverse triangle inequality gives
\begin{equation}
    \|Ug\| \ge \|g\| - \|(U-I)g\| \ge (1-\theta)\|g\|.
\end{equation}
Hence, for all finitely supported $\{c_n\}$,
\begin{equation}
    (1-\theta)^2 \sum_n |c_n|^2
    \le \left\| \sum_n c_n f_n \right\|_{L^2[-\pi,\pi]}^2
    \le (1+\theta)^2 \sum_n |c_n|^2,
\end{equation}
which is exactly the Riesz-basis inequality. \hfill $\square$

\textbf{Step 3: Kadec's trigonometric estimate.} The main harmonic-analysis input is the following classical estimate.

\begin{lemma}[Kadec estimate]
If $|\delta_n| \le L < 1/4$ for all $n$, then for every finitely supported sequence $\{c_n\}$ and every trigonometric polynomial $F(t)=\sum_n c_n e^{jnt}$,
\begin{equation}
    \left\| \sum_n c_n \bigl(e^{j(n+\delta_n)t} - e^{jnt}\bigr) \right\|_{L^2[-\pi,\pi]}
    \le \theta(L) \|F\|_{L^2[-\pi,\pi]},
    \label{eq:kadec_ineq}
\end{equation}
where
\begin{equation}
    \theta(L) = 1 - \cos(\pi L) + \sin(\pi L).
    \label{eq:theta}
\end{equation}
\end{lemma}

\textbf{Proof.} Set $\eta_n = |\delta_n| \in [0,L]$ and $\varsigma_n = \operatorname{sgn}(\delta_n) \in \{-1,0,1\}$. Since
\begin{equation}
    e^{j(n+\delta_n)t} - e^{jnt} = e^{jnt}\bigl((\cos(\delta_n t)-1) + j\sin(\delta_n t)\bigr),
\end{equation}
the triangle inequality gives
\begin{equation}
    \left\| \sum_n c_n \bigl(e^{j(n+\delta_n)t} - e^{jnt}\bigr) \right\|_{L^2}
    \le R + S,
\end{equation}
where
\begin{align}
    R & = \left\| \sum_n c_n e^{jnt}\bigl(1-\cos(\eta_n t)\bigr) \right\|_{L^2[-\pi,\pi]}, \\
    S & = \left\| \sum_n c_n e^{jnt}\sin(\delta_n t) \right\|_{L^2[-\pi,\pi]}.
\end{align}

We now estimate $R$ and $S$ separately by expanding the scalar multipliers in orthogonal trigonometric systems. This explicit coefficient proof follows the classical route of Young~\cite{Young2001}; it is self-contained at the level of the trigonometric expansion and avoids any appeal to Hilbert-transform or Hardy-space machinery.

For $0 \le \eta < 1$, the even function $1-\cos(\eta t)$ has the cosine expansion
\begin{equation}
    1-\cos(\eta t) = a_0(\eta) + \sum_{j=1}^{\infty} a_j(\eta) \cos(jt),
    \label{eq:kadec_even_expansion}
\end{equation}
with coefficients obtained directly from orthogonality (the explicit evaluation of each coefficient is recorded in Appendix~\ref{app:kadec_coeffs}):
\begin{align}
    a_0(\eta)
     & = \frac{1}{\pi} \int_{0}^{\pi} \bigl(1-\cos(\eta t)\bigr)\,dt
    = 1 - \frac{\sin(\pi \eta)}{\pi \eta},                                         \\
    a_j(\eta)
     & = \frac{2}{\pi} \int_{0}^{\pi} \bigl(1-\cos(\eta t)\bigr)\cos(jt)\,dt       \\
     & = \frac{2\eta\sin(\pi\eta)}{\pi} \frac{(-1)^j}{j^2-\eta^2}, \qquad j \ge 1.
\end{align}
For the odd function $\sin(\eta t)$, one expands instead in the orthogonal half-integer sine system. Thus, for $0 \le \eta < 1/2$,
\begin{equation}
    \sin(\eta t) = \sum_{j=1}^{\infty} b_j(\eta)\sin\!\left(\left(j-\frac{1}{2}\right)t\right),
    \label{eq:kadec_odd_expansion}
\end{equation}
where
\begin{align}
    b_j(\eta)
     & = \frac{2}{\pi} \int_{0}^{\pi} \sin(\eta t)\sin\!\left(\left(j-\frac{1}{2}\right)t\right)dt \\
     & = \frac{2\eta\cos(\pi\eta)}{\pi}
    \frac{(-1)^{j+1}}{\left(j-\frac{1}{2}\right)^2-\eta^2}, \qquad j \ge 1.
\end{align}
Because $L<1/4$, all denominators above are positive for $0 \le \eta \le L$. Introduce the nonnegative majorants
\begin{align}
    a_j^{\sharp}(\eta) & = \frac{2\eta\sin(\pi\eta)}{\pi(j^2-\eta^2)}, \qquad j \ge 1,                                     \\
    b_j^{\sharp}(\eta) & = \frac{2\eta\cos(\pi\eta)}{\pi\left(\left(j-\frac{1}{2}\right)^2-\eta^2\right)}, \qquad j \ge 1.
\end{align}

The coefficient functions are increasing on $[0,L]$. Indeed,
\begin{equation}
    a_0'(\eta) = \frac{\sin(\pi\eta)-\pi\eta\cos(\pi\eta)}{\pi\eta^2} > 0,
\end{equation}
because $\tan x > x$ for $x \in (0,\pi/2)$. Also $\eta\sin(\pi\eta)$ is increasing on $[0,L]$, while $j^2-\eta^2$ is decreasing, so each $a_j^{\sharp}$ is increasing. For $b_j^{\sharp}$, the numerator $\eta\cos(\pi\eta)$ has derivative
\begin{equation}
    \cos(\pi\eta)-\pi\eta\sin(\pi\eta)
    = \cos(\pi\eta)\bigl(1-\pi\eta\tan(\pi\eta)\bigr) > 0
\end{equation}
on $[0,L]$, because the function $x\tan x$ is increasing and $x\tan x < 1$ for $0 \le x < \pi/4$. Hence each $b_j^{\sharp}$ is also increasing on $[0,L]$.

We first estimate $R$. For $j \ge 1$, set
\begin{equation}
    d_{n,j}=c_n a_j(\eta_n),
    \qquad
    U_j^{\pm}=\left\|\sum_n d_{n,j}e^{j(n\pm j)t}\right\|_{L^2}.
\end{equation}
Using \eqref{eq:kadec_even_expansion},
\begin{align}
    R
     & \le a_0(L)\|F\|_{L^2}
    + \frac{1}{2}\sum_{j=1}^{\infty} U_j^+
    + \frac{1}{2}\sum_{j=1}^{\infty} U_j^-.
    \label{eq:kadec_R_first}
\end{align}
For each fixed $j$, the two shifted families $\{e^{j(n+j)t}\}_{n\in\mathbb{Z}}$ and $\{e^{j(n-j)t}\}_{n\in\mathbb{Z}}$ remain orthogonal on $[-\pi,\pi]$, so Parseval gives
\begin{equation}
    U_j^{\pm} \le a_j^{\sharp}(L)\|F\|_{L^2}.
\end{equation}
Substituting this into \eqref{eq:kadec_R_first} yields
\begin{align}
    R    & \le A(L)\|F\|_{L^2},
    \label{eq:kadec_R_bound}                                                \\
    A(L) & \triangleq a_0(L)+\sum_{j=1}^{\infty} a_j^{\sharp}(L).
\end{align}
Now use the classical partial-fraction identity (derived in Appendix~\ref{app:partial_frac})
\begin{equation}
    \pi\cot(\pi z) = \frac{1}{z} - 2z\sum_{j=1}^{\infty} \frac{1}{j^2-z^2},
    \qquad z \notin \mathbb{Z},
\end{equation}
with $z=L$. Rearranging gives
\begin{equation}
    C(L) \triangleq \sum_{j=1}^{\infty} \frac{1}{j^2-L^2}
    = \frac{1}{2L^2} - \frac{\pi}{2L}\cot(\pi L).
\end{equation}
Therefore
\begin{align}
    A(L)
     & = 1 - \frac{\sin(\pi L)}{\pi L}
    + \frac{2L\sin(\pi L)}{\pi}
    C(L)                     \\
     & = 1 - \cos(\pi L).
\end{align}
Combining with \eqref{eq:kadec_R_bound},
\begin{equation}
    R \le \bigl(1-\cos(\pi L)\bigr)\|F\|_{L^2}.
    \label{eq:kadec_real_part_bound}
\end{equation}

We next estimate $S$. Using \eqref{eq:kadec_odd_expansion} and the identity $\sin(\delta_n t)=\varsigma_n\sin(\eta_n t)$,
\begin{align}
    S
     & = \left\| \sum_{j=1}^{\infty} \sum_n c_n\varsigma_n b_j(\eta_n)
    e^{jnt}\sin\!\left(\left(j-\frac{1}{2}\right)t\right) \right\|_{L^2} \\
     & \le \frac{1}{2}\sum_{j=1}^{\infty}
    \left\| \sum_n c_n\varsigma_n b_j(\eta_n)e^{j(n+j-1/2)t} \right\|_{L^2} \\
     & \quad + \frac{1}{2}\sum_{j=1}^{\infty}
    \left\| \sum_n c_n\varsigma_n b_j(\eta_n)e^{j(n-j+1/2)t} \right\|_{L^2}.
    \label{eq:kadec_S_first}
\end{align}
For each fixed $j$, set $\widetilde{d}_{n,j}=c_n\varsigma_n b_j(\eta_n)$. The frequencies $n\pm(j-1/2)$ are all distinct half-integers, so the corresponding exponential families are orthogonal on $[-\pi,\pi]$. Hence
\begin{equation}
    \left\| \sum_n \widetilde{d}_{n,j}e^{j(n\pm j\mp 1/2)t} \right\|_{L^2}
    \le b_j^{\sharp}(L)\|F\|_{L^2},
\end{equation}
and \eqref{eq:kadec_S_first} becomes
\begin{align}
    S    & \le B(L)\|F\|_{L^2},
    \label{eq:kadec_imag_part_pre}                                   \\
    B(L) & \triangleq \sum_{j=1}^{\infty} b_j^{\sharp}(L).
\end{align}
Now invoke the companion partial-fraction formula
\begin{equation}
    \pi\tan(\pi z) = 2z\sum_{j=1}^{\infty}
    \frac{1}{\left(j-\frac{1}{2}\right)^2-z^2},
    \qquad |z|<\frac{1}{2}.
\end{equation}
Taking $z=L$ yields
\begin{align}
    \sum_{j=1}^{\infty} b_j^{\sharp}(L)
     & = \frac{2L\cos(\pi L)}{\pi}
    \sum_{j=1}^{\infty}
    \frac{1}{\left(j-\frac{1}{2}\right)^2-L^2} \\
     & = \sin(\pi L).
\end{align}
Thus \eqref{eq:kadec_imag_part_pre} becomes
\begin{equation}
    S \le \sin(\pi L)\|F\|_{L^2}.
    \label{eq:kadec_imag_part_bound}
\end{equation}

Finally, combining \eqref{eq:kadec_real_part_bound} and \eqref{eq:kadec_imag_part_bound} gives
\begin{multline}
    \left\| \sum_n c_n \bigl(e^{j(n+\delta_n)t} - e^{jnt}\bigr) \right\|_{L^2[-\pi,\pi]} \\
    \le \bigl(1-\cos(\pi L)+\sin(\pi L)\bigr)\|F\|_{L^2[-\pi,\pi]},
\end{multline}
which is exactly \eqref{eq:kadec_ineq}. \hfill $\square$

\textbf{Step 4: Complete the proof of the normalized theorem.} By \eqref{eq:kadec_ineq}, for finitely supported $\{c_n\}$,
\begin{align}
    \left\| \sum_n c_n (f_n - e_n) \right\|_{L^2[-\pi,\pi]}
     & \le \frac{\theta(L)}{\sqrt{2\pi}}\|F\|_{L^2[-\pi,\pi]} \\
     & = \theta(L) \left(\sum_n |c_n|^2\right)^{1/2},
\end{align}
so the perturbation lemma applies with this same $\theta(L)$. Thus $\{f_n\}$ is a Riesz basis provided $\theta(L) < 1$.

It remains to determine when this happens. Using \eqref{eq:theta},
\begin{align}
    \theta(L) < 1
     & \iff 1 - \cos(\pi L) + \sin(\pi L) < 1 \\
     & \iff \sin(\pi L) < \cos(\pi L)         \\
     & \iff \tan(\pi L) < 1.
    \label{eq:tan_ineq}
\end{align}
Since $L \in [0,1/4)$, this is equivalent to $L < 1/4$. Therefore the normalized exponential system $\{e^{j(n+\delta_n)t}\}_{n \in \mathbb{Z}}$ is a Riesz basis for $L^2[-\pi,\pi]$ whenever $\sup_n |\delta_n| < 1/4$.

\textbf{Step 5: Undo the normalization.} Returning to the original sampling sequence $\{t_n\}$, the unitary dilation $\mathcal{D}_{\Delta T}$ transports the normalized Riesz-basis statement back to bandwidth $\Omega = 2\pi f_{max}$. Hence any sequence satisfying
\begin{equation}
    \sup_n |t_n - n\Delta T| < \frac{\Delta T}{4} = \frac{1}{8f_{max}}
\end{equation}
is a sampling sequence for $PW_\Omega$. This completes the proof of the sufficient direction stated in Section~\ref{sec:kadec}.

\subsection{Subsequence Extraction from Centered Windows}
\textbf{Goal.} Show that one crossing in every centered window implies existence of a Kadec-compatible subsequence.

Fix a grid offset $\tau_0$ and consider the centered windows $J_n(\tau_0)$ from \eqref{eq:centered_window}. Suppose each window that lies inside the active observation interval contains at least one MT crossing. The interiors of adjacent windows are disjoint, and neighboring windows meet only at a single endpoint. Choose one crossing $s_n \in \Lambda_{\Delta V} \cap J_n(\tau_0)$ for each $n$. If a crossing lies exactly on a common boundary, assign it by a deterministic tie-breaking rule; under the strict inequalities used in Section~\ref{sec:sufficient}, the crossing produced by the range or curvature argument lies in the interior, so this boundary case is nongeneric.

By construction,
\begin{equation}
    \left|s_n - (\tau_0 + n\Delta T)\right| < \frac{\Delta T}{4},
\end{equation}
which is precisely Kadec's admissible perturbation bound. Therefore the selected subsequence is Kadec-compatible.

\subsection{Macroscopic Constraint: Detailed Argument}
\textbf{Goal.} Show that if $2\Delta V < R_x^{(c)}(\tau)$, then the centered window $J_\tau = [\tau-\Delta T/4,\, \tau+\Delta T/4]$ contains at least one retained adjacent-transition crossing under the conservative retained model. The weaker condition $\Delta V<R_x^{(c)}(\tau)$ guarantees only a raw threshold hit.

\textbf{Proof.} Fix an arbitrary centered window $J_\tau$. Let
\begin{equation}
    x_{min} = \min_{t \in J_\tau} x(t), \qquad x_{max} = \max_{t \in J_\tau} x(t),
\end{equation}
so that $R_x^{(c)}(\tau) = x_{max} - x_{min}$.

Since $2\Delta V < R_x^{(c)}(\tau)$, there are at least two threshold levels between $x_{min}$ and $x_{max}$. Indeed, the number of integer multiples of $\Delta V$ lying in $(x_{min},x_{max})$ is bounded below by
\begin{equation}
    \left\lfloor \frac{x_{max}}{\Delta V} \right\rfloor - \left\lceil \frac{x_{min}}{\Delta V} \right\rceil + 1,
\end{equation}
which is at least $2$ whenever $x_{max}-x_{min} > 2\Delta V$.

Choose two such distinct threshold indices. By the intermediate value theorem, since $x$ is continuous on the compact interval $J_\tau$ and attains both $x_{min}$ and $x_{max}$ there, both levels are crossed on monotone branches inside the window. At most one of these two crossings can be suppressed as a same-threshold return to the previously retained level; the other has a different threshold index and is retained. Hence $J_\tau$ contains at least one retained MT sample.

Applying the same argument to every centered window inside the active interval proves the retained range/no-gap version of the velocity condition. Together with the subsequence-extraction step above, this gives a rigorous retained Kadec-compatible subsequence.

\subsection{Microscopic Constraint: Detailed Taylor Analysis}
\textbf{Goal.} Near a local extremum $t^*$ where $x'(t^*) = 0$ and $x''(t^*) \neq 0$, derive a local condition ensuring a threshold crossing before the trajectory exits the centered window that contains $t^*$.

\textbf{Taylor expansion.} Expand $x(t)$ about $t^*$:
\begin{equation}
    x(t) = x(t^*) + \underbrace{x'(t^*)}_{=\,0}(t - t^*) + \frac{1}{2}x''(t^*)(t-t^*)^2 + R_3(t)
\end{equation}
where the remainder is $R_3(t) = \frac{1}{6}x'''(\xi)(t-t^*)^3$ for some $\xi$ between $t$ and $t^*$.

\textbf{Leading-order voltage excursion with extremum isolation.} Let $J_n(\tau_0)$ be the centered window containing $t^*$. Denote by $\ell_-$ and $\ell_+$ the distances from $t^*$ to the left and right boundaries of that window. Since $\ell_- + \ell_+ = \Delta T/2$, at least one of them is at least $\Delta T/4$. However, this window geometry alone is insufficient: the monotone branch on that side may end at a neighboring extremum before the first adjacent threshold is crossed. We therefore use the isolation radius $\rho(t^*)$ from~\eqref{eq:extremum_isolation_radius}. Since $\rho(t^*)\le \Delta T/4$ and $\rho(t^*)\le \rho_\pm(t^*)$ on both sides, the side of the window with boundary distance at least $\Delta T/4$ contains a monotone segment of length at least $\rho(t^*)$. Write $\Delta t=|t-t^*|$ along that side. Then
\begin{equation}
    |x(t) - x(t^*)| = \frac{1}{2}|x''(t^*)|(\Delta t)^2 + R_3
\end{equation}
At the certified isolation scale $\Delta t \le \rho(t^*)$, the remainder obeys
\begin{equation}
    |R_3| = \frac{1}{6}|x'''(\xi)|(\Delta t)^3 \le \frac{1}{6}M_3(t^*)(\Delta t)^3,
\end{equation}
so the quadratic term dominates whenever
\begin{equation}
    \Delta t \ll \frac{3|x''(t^*)|}{M_3(t^*)}.
\end{equation}
For non-degenerate extrema this dominance assumption is compatible with the Kadec scale only when the isolation radius is not too small and the third-order envelope is moderate; otherwise the curvature branch simply has no positive certificate.

\textbf{Crossing before leaving the window.} Let $\delta_{\mathrm{ret}}^{(\sigma)}(t^*;\Delta V)$ be the retained directional phase gap in~\eqref{eq:retained_extremum_gap} on the selected side. Neglecting the third-order remainder at leading order, the excursion reaches the first retained threshold once
\begin{equation}
    \frac{1}{2}|x''(t^*)|\rho(t^*)^2
    \ge
    \delta_{\mathrm{ret}}^{(\sigma)}(t^*;\Delta V)
\end{equation}
which yields the realized retained-phase form~\eqref{eq:curvature_leading}. With the third-order remainder retained, the conservative excursion at the isolation radius is
\begin{equation}
    \frac{1}{2}|x''(t^*)|\rho(t^*)^2
    -\frac{1}{6}M_3(t^*)\rho(t^*)^3,
\end{equation}
and requiring this quantity to exceed $\delta_{\mathrm{ret}}^{(\sigma)}(t^*;\Delta V)$ gives~\eqref{eq:curvature_exact}. Since $\delta_{\mathrm{ret}}^{(\sigma)}(t^*;\Delta V)\le2\Delta V$, the stronger requirement that the same quantity exceed $2\Delta V$ gives the retained phase-uniform condition~\eqref{eq:curvature_phase_uniform}. Because the chosen side has both window room and monotone-branch room for distance $\rho(t^*)$, the corresponding retained threshold crossing occurs before the trajectory leaves the centered window or hits the neighboring extremum.

Taking the infimum over all extrema gives the curvature term in~\eqref{eq:unified_condition}. Combining this local argument with the centered-window extraction step proves the combined sufficient statement from Section~\ref{sec:sufficient}.

\section{Convergence of the Iterative Projection Method}
\label{app:iterative}
We derive the convergence properties of the iterative algorithm \eqref{eq:iterative}.

\textbf{Error evolution.} Define the error function at iteration $k$ by $e_k(t) = x_k(t) - x(t)$. Substituting this into \eqref{eq:iterative} gives
\begin{align}
    e_{k+1}(t) & = x_{k+1}(t) - x(t)                                                                          \\
               & = x_k(t) - x(t)                                                                              \\
               & \quad + \lambda\, P_\Omega\!\left[\sum_n \bigl(y_n - x_k(t_n)\bigr)\delta(t - t_n)\right]    \\
               & = \bigl(x_k(t) - x(t)\bigr)                                                                  \\
               & \quad + \lambda\, P_\Omega\!\left[\sum_n \bigl(x(t_n) - x_k(t_n)\bigr)\delta(t - t_n)\right] \\
               & = e_k(t) - \lambda\, P_\Omega\!\left[\sum_n e_k(t_n)\,\delta(t - t_n)\right].
\end{align}
The projection term can be written as
\begin{align}
    P_\Omega\!\left[\sum_n e_k(t_n)\,\delta(t-t_n)\right]
     & = \sum_n e_k(t_n)\, K(t, t_n) \\
     & = (\mathfrak{S}_{\Lambda} e_k)(t),
\end{align}
where $\mathfrak{S}_{\Lambda}$ is the frame operator. Therefore:
\begin{equation}
    e_{k+1} = (I - \lambda \mathfrak{S}_{\Lambda})\, e_k = (I - \lambda \mathfrak{S}_{\Lambda})^{k+1}\, e_0
\end{equation}

\textbf{Spectral analysis.} By the frame condition \eqref{eq:frame}, the operator $\mathfrak{S}_{\Lambda}: PW_\Omega \to PW_\Omega$ is bounded, positive, and self-adjoint, with spectrum contained in
\begin{equation}
    \sigma(\mathfrak{S}_{\Lambda}) \subseteq [A, B],
\end{equation}
where $A > 0$. By spectral mapping, the spectrum of $\lambda\mathfrak{S}_{\Lambda}$ is contained in
\begin{equation}
    \sigma(\lambda \mathfrak{S}_{\Lambda}) \subseteq [\lambda A, \lambda B].
\end{equation}
Applying spectral mapping again to $I - \lambda \mathfrak{S}_{\Lambda}$ gives
\begin{equation}
    \sigma(I - \lambda \mathfrak{S}_{\Lambda}) \subseteq [1 - \lambda B, 1 - \lambda A].
\end{equation}
Since $I - \lambda \mathfrak{S}_{\Lambda}$ is self-adjoint, its operator norm equals the maximum absolute value attained on its spectrum:
\begin{equation}
    \begin{aligned}
        \|I - \lambda \mathfrak{S}_{\Lambda}\|
         & = \sup_{\mu \in \sigma(\mathfrak{S}_{\Lambda})} |1 - \lambda \mu|                         \\
         & = \max(|1 - \lambda A|,\; |1 - \lambda B|) \triangleq \rho(\lambda).
    \end{aligned}
\end{equation}

\textbf{Convergence condition.} For $\rho < 1$, we need both $|1 - \lambda A| < 1$ and $|1 - \lambda B| < 1$:
\begin{align}
    |1 - \lambda A| < 1 & \iff 0 < \lambda A < 2 \iff \lambda < 2/A \\
    |1 - \lambda B| < 1 & \iff 0 < \lambda B < 2 \iff \lambda < 2/B
\end{align}
Since $A \le B$, the binding constraint is $\lambda < 2/B$. For any $\lambda \in (0, 2/B)$:
\begin{equation}
    \|e_k\|_{L^2} = \|(I - \lambda \mathfrak{S}_{\Lambda})^k e_0\|_{L^2} \le \rho^k \|e_0\|_{L^2}
\end{equation}
establishing exponential convergence.

\textbf{Optimal relaxation.} The optimal $\lambda$ minimizes $\rho(\lambda) = \max(|1 - \lambda A|,\, |1 - \lambda B|)$. Setting $1 - \lambda A = -(1 - \lambda B)$ (equating the two branches):
\begin{equation}
    1 - \lambda^* A = \lambda^* B - 1 \implies \lambda^* = \frac{2}{A + B}
\end{equation}
The resulting optimal convergence rate is:
\begin{equation}
    \rho^* = 1 - \lambda^* A = 1 - \frac{2A}{A+B} = \frac{B - A}{B + A}
\end{equation}
This equals the condition number mismatch of the frame bounds. When $A = B$ (tight frame), $\rho^* = 0$ and convergence is achieved in a single iteration.

\section{Equation-by-Equation Derivations for Main-Text Formulae}
\label{app:main_eqs}
This appendix records representative equation-level derivations for the main algebraic and computational formulae in the text. When a formula is a definition rather than a theorem, we state that explicitly and explain why that construction is the natural one for the MT problem. The appendix provides a direct equation-to-equation bridge for formulae whose derivation is not already contained in the surrounding proof, complementing Appendices~\ref{app:landau}--\ref{app:iterative}.

\subsection{Formulae from the Preliminaries and Local Conditions}
\textbf{Main-text Eq.~\eqref{eq:pw_space}.} Equation~\eqref{eq:pw_space} is a definition. The physical assumptions are:
\begin{align}
     & \text{finite energy} \quad \Longrightarrow \quad x \in L^2(\mathbb{R}),                              \\
     & \text{bandlimit } \Omega \quad \Longrightarrow \quad \hat{x}(\omega)=0 \text{ for } |\omega|>\Omega.
\end{align}
The support statement $\hat{x}(\omega)=0$ for $|\omega|>\Omega$ is written compactly as $\mathrm{supp}(\hat{x})\subseteq[-\Omega,\Omega]$. Imposing both conditions simultaneously gives
\[
    PW_\Omega = \left\{x\in L^2(\mathbb{R}) : \mathrm{supp}(\hat{x})\subseteq[-\Omega,\Omega]\right\},
\]
which is exactly Eq.~\eqref{eq:pw_space}.

\textbf{Main-text Eq.~\eqref{eq:tv_def}.} For a continuously differentiable function, the classical definition of total variation is
\[
    \mathrm{TV}(x;[a,b]) = \sup_{P}\sum_{i=1}^{N}|x(t_i)-x(t_{i-1})|,
\]
where $P: a=t_0<t_1<\cdots<t_N=b$ ranges over all partitions. By the fundamental theorem of calculus,
\[
    x(t_i)-x(t_{i-1}) = \int_{t_{i-1}}^{t_i}x'(t)\,dt.
\]
Taking absolute values gives
\[
    |x(t_i)-x(t_{i-1})| \le \int_{t_{i-1}}^{t_i}|x'(t)|\,dt.
\]
Summing over $i$ yields
\[
    \sum_{i=1}^{N}|x(t_i)-x(t_{i-1})| \le \sum_{i=1}^{N}\int_{t_{i-1}}^{t_i}|x'(t)|\,dt = \int_a^b|x'(t)|\,dt.
\]
Since this holds for every partition $P$, one obtains
\[
    \mathrm{TV}(x;[a,b]) \le \int_a^b|x'(t)|\,dt.
\]
For the reverse inequality, refine the partition so that each subinterval lies inside a region where $x'$ does not change sign. On such a subinterval,
\[
    \left|\int_{t_{i-1}}^{t_i}x'(t)\,dt\right| = \int_{t_{i-1}}^{t_i}|x'(t)|\,dt.
\]
Summing over the refined partition gives the opposite inequality, hence
\[
    \mathrm{TV}(x;[a,b]) = \int_a^b|x'(t)|\,dt,
\]
which is Eq.~\eqref{eq:tv_def}.

\textbf{Main-text Eq.~\eqref{eq:mvt_set}.} Equation~\eqref{eq:mvt_set} is also a definition. Fix one maximal monotone interval $I_k$. On that interval the restriction $x|_{I_k}$ is strictly monotone, so for any threshold level $m\Delta V$ the equation $x(t)=m\Delta V$ has at most one solution in $I_k$. Therefore the raw set of horizontal intersections on that branch is
\[
    \{t\in I_k : x(t)\in \Delta V\mathbb{Z}\}.
\]
Collecting these raw crossings from all monotone branches gives
\[
    \Lambda_{\Delta V}^{\mathrm{raw}} = \bigcup_k \{t\in I_k : x(t)\in \Delta V\mathbb{Z}\}.
\]
The recorded set used in the sampling theory is obtained by applying the retained adjacent-transition map $\mathcal{A}_{\Delta V}$, which suppresses same-threshold returns and retains adjacent threshold-index transitions:
\[
    \Lambda_{\Delta V}=\mathcal{A}_{\Delta V}(\Lambda_{\Delta V}^{\mathrm{raw}}),
\]
which is Eq.~\eqref{eq:mvt_set}.

\textbf{Main-text Eq.~\eqref{eq:frame}.} Let $S_\Lambda:PW_\Omega\to \ell^2$ be the sampling operator defined by $S_\Lambda x = \{x(t_n)\}_{n\in\mathbb{Z}}$. Stable sampling means that the sample norm and the signal norm are comparable, namely that there exist constants $c_1,c_2>0$ such that
\[
    c_1\|x\|_{L^2} \le \|S_\Lambda x\|_{\ell^2} \le c_2\|x\|_{L^2}.
\]
Now compute
\[
    \|S_\Lambda x\|_{\ell^2}^2 = \sum_{n\in\mathbb{Z}}|x(t_n)|^2.
\]
Squaring the norm comparison and renaming $A=c_1^2$ and $B=c_2^2$ gives
\[
    A\|x\|_{L^2}^2 \le \sum_{n\in\mathbb{Z}}|x(t_n)|^2 \le B\|x\|_{L^2}^2,
\]
which is Eq.~\eqref{eq:frame}.

\textbf{Main-text Eq.~\eqref{eq:density_def}.} Equation~\eqref{eq:density_def} is the standard definition of lower Beurling density. For a window $[t,t+r]$, the average sampling rate inside that window is
\[
    \frac{n(t,r)}{r}.
\]
To capture the most poorly sampled window of length $r$, take the infimum over all $t\in\mathbb{R}$:
\[
    \inf_{t\in\mathbb{R}}\frac{n(t,r)}{r}.
\]
To extract the asymptotic worst-case rate on large windows, take the lower limit as $r\to\infty$:
\[
    D^-(\Lambda) = \liminf_{r\to\infty}\inf_{t\in\mathbb{R}}\frac{n(t,r)}{r},
\]
which is Eq.~\eqref{eq:density_def}.

\textbf{Main-text Eq.~\eqref{eq:kadec_base}.} Appendix~\ref{app:constraints} proves the normalized statement that if
\[
    \sup_n|\lambda_n-n| < \frac{1}{4},
\]
then the perturbed exponential system is a Riesz basis for $PW_\pi$. To translate this back to physical units, write
\[
    \lambda_n = \frac{t_n}{\Delta T}.
\]
Then
\[
    \left|\lambda_n-n\right| = \left|\frac{t_n}{\Delta T}-n\right| = \frac{|t_n-n\Delta T|}{\Delta T}.
\]
The normalized condition becomes
\[
    \frac{\sup_n|t_n-n\Delta T|}{\Delta T} < \frac{1}{4}.
\]
Multiplying by $\Delta T$ gives
\[
    \sup_n|t_n-n\Delta T| < \frac{\Delta T}{4}.
\]
Finally, substitute $\Delta T=1/(2f_{max})$ to obtain
\[
    \frac{\Delta T}{4} = \frac{1}{4}\cdot\frac{1}{2f_{max}} = \frac{1}{8f_{max}},
\]
which yields Eq.~\eqref{eq:kadec_base}.

\textbf{Main-text Eq.~\eqref{eq:centered_window}.} Equation~\eqref{eq:centered_window} is a definition chosen to encode the admissible Kadec neighborhood around the reference grid point $\tau_0+n\Delta T$. The center is $\tau_0+n\Delta T$ and the allowed radius is $\Delta T/4$. Therefore the left endpoint is center minus radius and the right endpoint is center plus radius:
\[
    \tau_0+n\Delta T-\frac{\Delta T}{4},
    \qquad
    \tau_0+n\Delta T+\frac{\Delta T}{4}.
\]
Putting them together gives Eq.~\eqref{eq:centered_window}.

\textbf{Range $R_x^{(c)}(\tau)$ (used in the Remark after Eq.~\eqref{eq:global_retained_velocity_condition}).} The local range is defined as the vertical excursion of the signal inside the centered window $J_\tau$. The highest signal value in the window is $\max_{t\in J_\tau}x(t)$ and the lowest is $\min_{t\in J_\tau}x(t)$. Their difference is the local range,
\[
    R_x^{(c)}(\tau) = \max_{t\in J_\tau}x(t)-\min_{t\in J_\tau}x(t).
\]

\textbf{Main-text Eqs.~\eqref{eq:global_velocity_condition}--\eqref{eq:global_retained_velocity_condition}.} The local retained monotone-window Kadec requirement is Eq.~\eqref{eq:velocity_kadec_condition},
\[
    \Delta V < \frac{v_n}{8\bar{\alpha}_n^{(\mathrm{ret})}f_{max}}.
\]
To impose it simultaneously over all centered windows that lie inside monotone portions of the active interval, first enumerate the monotone windows $\mathcal{N}_{\mathrm{mono}}(\tau_0)$ from~\eqref{eq:mono_window_set}, then define
\[
    \begin{aligned}
    v_{min}^{(\mathrm{mono})}(x;\tau_0)
    &= \min_{n\in\mathcal{N}_{\mathrm{mono}}(\tau_0)} v_n,\\
    \alpha_{\mathrm{mono}}^{(\mathrm{ret})}(x;\tau_0)
    &=
    \max_{n\in\mathcal{N}_{\mathrm{mono}}(\tau_0)}
    \alpha_n^{(\mathrm{ret})},\\
    \bar{\alpha}_{\mathrm{mono}}^{(\mathrm{ret})}(x;\tau_0)
    &=
    \max_{n\in\mathcal{N}_{\mathrm{mono}}(\tau_0)}
    \bar{\alpha}_n^{(\mathrm{ret})}.
    \end{aligned}
\]
Replacing the local $v_n$ and $\bar{\alpha}_n^{(\mathrm{ret})}$ by these envelopes gives
\[
    \Delta V <
    \frac{v_{min}^{(\mathrm{mono})}(x;\tau_0)}
    {8\bar{\alpha}_{\mathrm{mono}}^{(\mathrm{ret})}(x;\tau_0)f_{max}},
\]
which is Eq.~\eqref{eq:global_retained_velocity_condition}.

\textbf{Main-text Eq.~\eqref{eq:extremum_pattern}.} Start from the second-order Taylor formula with remainder around the extremum $t^*$:
\begin{align}
    x(t)
     & = x(t^*) + x'(t^*)(t-t^*) + \frac{1}{2}x''(t^*)(t-t^*)^2 \\
     & \quad + \frac{1}{6}x'''(\xi)(t-t^*)^3,
\end{align}
for some $\xi$ between $t$ and $t^*$. Because $t^*$ is an extremum, $x'(t^*)=0$. Dropping the cubic remainder when only the leading local geometry is retained gives
\[
    x(t) \approx x(t^*) + \frac{1}{2}x''(t^*)(t-t^*)^2,
\]
The first retained threshold in the excursion direction is not always the first raw threshold: the first same-threshold return may be suppressed. The retained directional phase gap is $\delta_{\mathrm{ret}}^{(\sigma)}(t^*;\Delta V)$ from~\eqref{eq:retained_extremum_gap}. The $j$-th subsequent retained threshold excursion is therefore $\delta_{\mathrm{ret}}^{(\sigma)}+j\Delta V$. Equating the quadratic excursion to that value gives
\[
    |t_{j,\mathrm{ret}}^{(\sigma)}-t^*|
    =
    \sqrt{\frac{2(\delta_{\mathrm{ret}}^{(\sigma)}+j\Delta V)}{|x''(t^*)|}},
\]
with the cubic remainder producing the multiplicative error in Eq.~\eqref{eq:extremum_pattern}.

\textbf{Main-text Eq.~\eqref{eq:curvature_leading}.} From Eq.~\eqref{eq:extremum_pattern}, the leading-order excursion over time offset $\Delta t$ is
\[
    \delta_{\mathrm{ret}}^{(\sigma)}(t^*;\Delta V) \approx \frac{1}{2}|x''(t^*)|(\Delta t)^2.
\]
To force a crossing before leaving both the Kadec-compatible window and the monotone branch, impose
\[
    \Delta t < \rho(t^*),
\]
where $\rho(t^*)$ is the isolation radius in~\eqref{eq:extremum_isolation_radius}. Substituting this strict upper bound for $\Delta t$ into the pure quadratic excursion equality algebraically guarantees:
\begin{align}
    \delta_{\mathrm{ret}}^{(\sigma)}(t^*;\Delta V)
     & < \frac{1}{2}|x''(t^*)|\rho(t^*)^2,
\end{align}
which is Eq.~\eqref{eq:curvature_leading}.

\textbf{Main-text Eq.~\eqref{eq:curvature_exact}.} Including the third-order remainder in the Taylor expansion around $t^*$ (see Eq.~\eqref{eq:extremum_pattern}), the guaranteed excursion at time offset $\Delta t$ is bounded below by
\[
    |x(t^*+\Delta t)-x(t^*)|
    \ge
    \frac{|\kappa|}{2}(\Delta t)^2
    - \frac{M_3(t^*)}{6}(\Delta t)^3.
\]
The realized retained-phase sufficient condition requires this conservative quadratic--cubic lower excursion at $\Delta t=\rho(t^*)$ to exceed $\delta_{\mathrm{ret}}^{(\sigma)}(t^*;\Delta V)$. Substituting gives:
\begin{align}
    \delta_{\mathrm{ret}}^{(\sigma)}(t^*;\Delta V)
     &< \frac{|x''(t^*)|}{2}\rho(t^*)^2
      - \frac{M_3(t^*)}{6}\rho(t^*)^3,
\end{align}
which is Eq.~\eqref{eq:curvature_exact}. Since $\delta_{\mathrm{ret}}^{(\sigma)}(t^*;\Delta V)\le2\Delta V$, replacing the left side by $2\Delta V$ gives the retained phase-uniform condition~\eqref{eq:curvature_phase_uniform}.

\textbf{Main-text Eq.~\eqref{eq:extremum_perturbation}.} Let $g_n=\tau_0+n\Delta T$, let $d_n=|t^*-g_n|$, and set $\varsigma_n=\operatorname{sgn}(g_n-t^*)$. On the side of $t^*$ pointing toward $g_n$, let $h_n$ denote the distance from the extremum to the first retained threshold crossing. The quadratic Taylor law gives
\[
    h_n=\sqrt{\frac{2\delta_{\mathrm{ret}}^{(\varsigma_n)}(t^*;\Delta V)}{|\kappa|}}\,(1+o(1)).
\]
Since the chosen crossing lies between the extremum and the grid center direction, its displacement from the grid anchor is the difference between the extremum-anchor offset and the extremum-crossing offset:
\[
    |s_n-g_n|\le |d_n-h_n|.
\]
Dividing by $\Delta T$ gives Eq.~\eqref{eq:extremum_perturbation}. This derivation also explains why the offset term $d_n$ cannot be dropped unless the grid is aligned with the extremum to the same order as $h_n$.

\subsection{Formulae from the Density Section}
\textbf{Main-text Eq.~\eqref{eq:local_spacing}.} Consecutive MT crossings differ by one threshold step, so
\[
    |x(t_{k+1})-x(t_k)| = \Delta V.
\]
Apply the mean value theorem to $x$ on $[t_k,t_{k+1}]$. There exists $\xi\in(t_k,t_{k+1})$ such that
\[
    x(t_{k+1})-x(t_k) = x'(\xi)(t_{k+1}-t_k).
\]
Taking absolute values gives
\[
    |x(t_{k+1})-x(t_k)| = |x'(\xi)|\,|t_{k+1}-t_k|.
\]
Since the left-hand side equals $\Delta V$, solving for the time gap gives
\[
    |t_{k+1}-t_k| = \frac{\Delta V}{|x'(\xi)|},
\]
which is Eq.~\eqref{eq:local_spacing}.

\textbf{Main-text Eq.~\eqref{eq:spacing_expansion}.} On a monotone branch let $\sigma_k=\operatorname{sgn}(x')$ and enumerate threshold levels in traversal order, $v_{m+1}=v_m+\sigma_k\Delta V$. With $t_m=\tau(v_m)$ and $\tau=x^{-1}$,
\[
    t_{m+1}-t_m=\tau(v_m+\sigma_k\Delta V)-\tau(v_m).
\]
Taylor expansion gives
\[
    t_{m+1}-t_m
    =\sigma_k\Delta V\,\tau'(v_m)+\frac{(\Delta V)^2}{2}\tau''(\eta_m).
\]
The inverse-function identities
\[
    \tau'(v)=\frac{1}{x'(\tau(v))},
    \qquad
    \tau''(v)=-\frac{x''(\tau(v))}{(x'(\tau(v)))^3}
\]
and $\sigma_k\tau'(v_m)=1/|x'(t_m)|$ yield Eq.~\eqref{eq:spacing_expansion}. The traversal-oriented index is essential on decreasing branches; otherwise $t_{m+1}-t_m$ would carry the wrong sign.

\textbf{Main-text Eq.~\eqref{eq:sep_bound}.} Starting from Eq.~\eqref{eq:local_spacing}, use the elementary bound $|x'(\xi)|\le \|x'\|_\infty$ to obtain
\[
    |t_{k+1}-t_k| = \frac{\Delta V}{|x'(\xi)|} \ge \frac{\Delta V}{\|x'\|_\infty}.
\]
Now take the infimum over $k$:
\[
    \inf_k|t_{k+1}-t_k| \ge \frac{\Delta V}{\|x'\|_\infty}.
\]
Bernstein's inequality gives $\|x'\|_\infty \le \Omega\|x\|_\infty = 2\pi f_{max}\|x\|_\infty$, hence
\[
    \delta_{\mathrm{sep}} = \inf_k|t_{k+1}-t_k| \ge \frac{\Delta V}{2\pi f_{max}\|x\|_\infty},
\]
which is Eq.~\eqref{eq:sep_bound}.

\textbf{Main-text Eq.~\eqref{eq:cross_upper}.} Let
\[
    V_{\mathrm{tot}}(t,r) = \int_t^{t+r}|x'(\tau)|\,d\tau
\]
be the total variation over $[t,t+r]$. The upper bound is derived from consecutive retained events, not from raw hits. Consecutive retained labels differ by one ladder step; hence the total variation between two consecutive retained times is at least $\Delta V$. Therefore the number of retained crossings cannot exceed $V_{\mathrm{tot}}(t,r)/\Delta V$ plus endpoint corrections. If the observation window cuts through a partial retained traversal at one or both ends, that creates at most a constant endpoint correction. Hence
\[
    n_{\mathrm{ret}}(t,r) \le \frac{1}{\Delta V}\int_t^{t+r}|x'(\tau)|\,d\tau + 2,
\]
which is Eq.~\eqref{eq:cross_upper}. The constant $2$ is explicit: the left endpoint can intersect at most one unfinished threshold traversal, and the right endpoint can intersect at most one more.

\textbf{Main-text Eq.~\eqref{eq:cross_lower}.} Decompose $[t,t+r]$ into maximal monotone pieces $I_1,\dots,I_{N_{\mathrm{mono}}(t,r)}$. On piece $I_k$, the total variation is
\[
    \mathrm{TV}_k = \int_{I_k}|x'(\tau)|\,d\tau.
\]
The proof needs only a guaranteed lower count, not an exact raw-hit formula. Threshold phase and endpoint conventions may change the exact number of raw hits, but a monotone traversal of total variation $\mathrm{TV}_k$ guarantees at least $\mathrm{TV}_k/\Delta V-1$ threshold intervals. This is the same floor-loss estimate as $\lfloor y\rfloor\ge y-1$:
\[
    \left\lfloor \frac{\mathrm{TV}_k}{\Delta V}\right\rfloor \ge \frac{\mathrm{TV}_k}{\Delta V}-1.
\]
On each monotone piece, at most one first raw hit can be suppressed as a same-threshold return. Summing over $k=1,\dots,N_{\mathrm{mono}}(t,r)$ and subtracting the suppression count gives
\begin{align}
    n_{\mathrm{ret}}(t,r)
     & \ge \sum_{k=1}^{N_{\mathrm{mono}}(t,r)}\left(\frac{\mathrm{TV}_k}{\Delta V}-1\right)           \\
     &\quad -N_{\mathrm{supp}}(t,r)\\
     & = \frac{1}{\Delta V}\sum_{k=1}^{N_{\mathrm{mono}}(t,r)}\mathrm{TV}_k - N_{\mathrm{mono}}(t,r)-N_{\mathrm{supp}}(t,r).
\end{align}
Because the total variations of the monotone pieces add up to the total variation of the whole interval,
\[
    \sum_{k=1}^{N_{\mathrm{mono}}(t,r)}\mathrm{TV}_k = \int_t^{t+r}|x'(\tau)|\,d\tau.
\]
Substituting this identity yields Eq.~\eqref{eq:cross_lower}.

\textbf{Main-text Eq.~\eqref{eq:density_profile}.} Equation~\eqref{eq:density_profile} is the definition of the instantaneous threshold-traversal envelope. Locally, a waveform with velocity magnitude $|x'(t)|$ traverses voltage distance $|x'(t)|\,dt$ during a small interval $dt$. Since each threshold interval has height $\Delta V$, the first-order number of threshold intervals traversed per unit time is
\[
    \rho_{MT}(t)=\frac{|x'(t)|}{\Delta V},
\]
which is Eq.~\eqref{eq:density_profile}. Same-threshold suppression is not included in this continuum envelope; it enters through $N_{\mathrm{supp}}$ in~\eqref{eq:cross_lower}.

\textbf{Main-text Eq.~\eqref{eq:velocity}.} Equation~\eqref{eq:velocity} defines the mean absolute speed on one interval and then takes its lower envelope. For a fixed interval,
\[
    \bar{\nu}(t,r) = \frac{1}{r}\int_t^{t+r}|x'(\tau)|\,d\tau,
\]
and the slowest such value is obtained by taking the infimum over all admissible intervals:
\[
    \nu_{min} = \inf_{\substack{[t,t+r]\subseteq[0,T_{\text{obs}}]\\ r\ge 1/(2f_{max})}}\bar{\nu}(t,r),
\]
which is Eq.~\eqref{eq:velocity}. Because that infimum is taken over subintervals of a fixed waveform, $\nu_{min}$ is signal-dependent rather than a universal constant of the full model class.

\textbf{Main-text Eq.~\eqref{eq:density_window_improved}.} Divide the upper estimate \eqref{eq:cross_upper} by $r$:
\begin{align}
    \frac{n_{\mathrm{ret}}(t,r)}{r}
     & \le \frac{1}{\Delta V}\cdot \frac{1}{r}\int_t^{t+r}|x'(\tau)|\,d\tau + \frac{2}{r} \\
     & = \frac{\bar{\nu}(t,r)}{\Delta V}+\frac{2}{r}.
\end{align}
Now divide the lower estimate \eqref{eq:cross_lower} by $r$:
\begin{align}
    \frac{n_{\mathrm{ret}}(t,r)}{r}
     & \ge \frac{1}{\Delta V}\cdot \frac{1}{r}\int_t^{t+r}|x'(\tau)|\,d\tau \\
     &\quad - \frac{N_{\mathrm{mono}}(t,r)}{r}-\frac{N_{\mathrm{supp}}(t,r)}{r} \\
     & = \frac{\bar{\nu}(t,r)}{\Delta V} - \frac{N_{\mathrm{mono}}(t,r)}{r}-\frac{N_{\mathrm{supp}}(t,r)}{r}.
\end{align}
Using the certified extremum-density bound $N_{\mathrm{mono}}(t,r)/r \le \eta_0+\mathcal{O}(1/r)$ and suppression-density bound $N_{\mathrm{supp}}(t,r)/r \le \eta_{\mathrm{supp},0}+\mathcal{O}(1/r)$ gives
\begin{align}
    \frac{\bar{\nu}(t,r)}{\Delta V} - \eta_0-\eta_{\mathrm{supp},0} - \mathcal{O}(1/r)
     & \le \frac{n_{\mathrm{ret}}(t,r)}{r}                               \\
     & \le \frac{\bar{\nu}(t,r)}{\Delta V} + \frac{2}{r},
\end{align}
which is Eq.~\eqref{eq:density_window_improved}.

\textbf{Main-text Eq.~\eqref{eq:density_bounds_improved}.} Start from Eq.~\eqref{eq:density_window_improved}. On intervals with $r\gg \Delta T$, the term $\mathcal{O}(1/r)$ is asymptotically negligible. Then
\[
    \frac{\bar{\nu}(t,r)}{\Delta V} - \eta_0-\eta_{\mathrm{supp},0} \lesssim \frac{n_{\mathrm{ret}}(t,r)}{r} \lesssim \frac{\bar{\nu}(t,r)}{\Delta V}.
\]
Now take the lower envelope over long subintervals inside the active window. By definition of $\nu_{min}$, the lowest achievable value of $\bar{\nu}(t,r)$ is $\nu_{min}$. Replacing $\bar{\nu}(t,r)$ by that lower envelope yields
\[
    \frac{\nu_{min}}{\Delta V} - \eta_0-\eta_{\mathrm{supp},0}
    \lesssim \inf \frac{n_{\mathrm{ret}}(t,r)}{r}
    \lesssim \frac{\nu_{min}}{\Delta V},
\]
which is Eq.~\eqref{eq:density_bounds_improved}. For a single verified instance one may use $\eta_0=\eta_x$ and $\eta_{\mathrm{supp},0}=\eta_{\mathrm{supp},x}$.

\textbf{Main-text Eq.~\eqref{eq:finite_R_density_design}.} For finite-window statements, keep the additive constants in the count certificates. If $r\ge R$ and~\eqref{eq:finite_R_count_constants} holds, then
\[
    \frac{N_{\mathrm{mono}}(t,r)}{r}
    +\frac{N_{\mathrm{supp}}(t,r)}{r}
    \le
    \eta_0+\eta_{\mathrm{supp},0}
    +\frac{C_{\mathrm{mono}}+C_{\mathrm{supp}}}{R}.
\]
Substituting this into the lower count estimate and using $\bar{\nu}(t,r)\ge\nu_{min}^{(R)}$ gives~\eqref{eq:finite_R_density_bound}. Requiring that lower bound to exceed $2f_{max}$ and solving for $\Delta V$ yields~\eqref{eq:finite_R_density_design}.

\textbf{Main-text Eq.~\eqref{eq:finite_R_retained_safe_fallback}.} If no suppression-topology certificate is available before acquisition, use the deterministic retained-safe estimate~\eqref{eq:supp_count_bound}. Together with $N_{\mathrm{mono}}(t,r)\le \eta_0 r+C_{\mathrm{mono}}$, it gives
\[
    N_{\mathrm{supp}}(t,r)
    \le
    N_{\mathrm{mono}}(t,r)
    \le
    \eta_0 r+C_{\mathrm{mono}}.
\]
Thus $\eta_{\mathrm{supp},0}=\eta_0$ and $C_{\mathrm{supp}}=C_{\mathrm{mono}}$ are valid finite-window choices in~\eqref{eq:finite_R_density_design}. Substitution gives~\eqref{eq:finite_R_retained_safe_fallback}.

\textbf{Main-text Eq.~\eqref{eq:strict_suff_improved}.} The displayed rule is the long-window scale obtained by demanding that the asymptotic lower estimate in Eq.~\eqref{eq:density_bounds_improved} exceed the Nyquist benchmark $2f_{max}$:
\[
    \frac{\nu_{min}}{\Delta V} - \eta_0-\eta_{\mathrm{supp},0} > 2f_{max}.
\]
Add $\eta_0+\eta_{\mathrm{supp},0}$ to both sides:
\[
    \frac{\nu_{min}}{\Delta V} > 2f_{max}+\eta_0+\eta_{\mathrm{supp},0}.
\]
Multiply by $\Delta V>0$:
\[
    \nu_{min} > (2f_{max}+\eta_0+\eta_{\mathrm{supp},0})\Delta V.
\]
Now divide by $2f_{max}+\eta_0+\eta_{\mathrm{supp},0}>0$ to obtain
\[
    \Delta V < \frac{\nu_{min}}{2f_{max}+\eta_0+\eta_{\mathrm{supp},0}},
\]
which is Eq.~\eqref{eq:strict_suff_improved}.

\textbf{Main-text Eq.~\eqref{eq:ultimate_ns}.} In a certified high-density regime, the correction term $\eta_0+\eta_{\mathrm{supp},0}$ in Eq.~\eqref{eq:density_bounds_improved} is lower order compared with $\nu_{min}/\Delta V$. The leading traversal-density proxy is therefore
\[
    \frac{\nu_{min}}{\Delta V}.
\]
Matching that leading term to the Nyquist benchmark gives
\[
    \frac{\nu_{min}}{\Delta V} \gtrsim 2f_{max}.
\]
Multiply by $\Delta V>0$ and divide by $2f_{max}>0$ to obtain
\[
    \Delta V \lesssim \frac{\nu_{min}}{2f_{max}},
\]
which is Eq.~\eqref{eq:ultimate_ns}. This is a leading-density benchmark under the stated extremum-density certificate; it is not a universal sharp threshold over $PW_\Omega$.

\subsection{Formulae from the Hardware Section}
\textbf{Main-text Eqs.~\eqref{eq:tdc_spec}--\eqref{eq:dac_spec}.} Let
\begin{align}
    C&=2\pi f_{max}\|x\|_\infty,\\
    \beta&=\frac{C}{v_{min}},\\
    \gamma_A^{\mathrm{det}}
    &=
    \frac{
    \sqrt{A_{\mathrm{hw}}^{\mathrm{det}}}
    \left(\|x\|_{L^2}/\sqrt{S_{\mathrm{tar}}}-E_{\mathrm{tc}}\right)
    -E_R}
    {3\sqrt{N_{\mathrm{in}}}}.
\end{align}
The deterministic sample-error norm~\eqref{eq:total_sample_error} has three leading components, $C\sigma_j$, $\beta\sigma_n$, and $(1+\beta)\sigma_q$, plus the reserved Taylor remainder budget $E_R$. Under the simultaneous-budget convention in Corollary~\ref{cor:hw_specs}, each leading component is required to be at most $\gamma_A^{\mathrm{det}}$. Because the tail-prior amplification factor is $1/\sqrt{A_{\mathrm{hw}}^{\mathrm{det}}}$, this guarantees that the sample-error contribution and the externally assigned tail/collar budget remain within the target full-signal perturbation amplitude. Solving the three scalar inequalities gives
\begin{align}
    \sigma_j\le\frac{\gamma_A^{\mathrm{det}}}{C},
    \qquad
    \sigma_n\le\frac{\gamma_A^{\mathrm{det}}}{\beta},
    \qquad\\
    \sigma_q\le\frac{\gamma_A^{\mathrm{det}}}{1+\beta},
\end{align}
which are exactly Eqs.~\eqref{eq:tdc_spec}--\eqref{eq:dac_spec}. The factor $1+\beta$ in the DAC rule accounts for the direct value error from the unreported residual threshold offset $r_{m_k}$ in addition to the velocity-amplified timing shift.

If $\theta(L_{\mathrm{hw}}^{\mathrm{det}})\le1-1/\sqrt{2}$, then $A_{\mathrm{hw}}^{\mathrm{det}}\ge1/(2\Delta T)$. Substituting this lower bound for $A_{\mathrm{hw}}^{\mathrm{det}}$ in $\gamma_A^{\mathrm{det}}$ gives the simplified explicit rules~\eqref{eq:tdc_spec_simple}--\eqref{eq:dac_spec_simple}. These simplified rules are therefore a stronger-margin corollary, not a consequence of $L_{\mathrm{hw}}^{\mathrm{det}}<1/4$ alone.

For the finite-dimensional specifications~\eqref{eq:tdc_spec_finite}--\eqref{eq:dac_spec_finite}, repeat the same algebra with
\[
    \gamma_M^{\mathrm{det}}
    =
    \frac{\sigma_{\min}(H_{\mathrm{hw},M})\|p\|_{L^2}/\sqrt{S_{\mathrm{tar}}}-E_R}
    {3\sqrt{N_{\mathrm{in}}}},
\]
because the finite-record amplification factor is $1/\sigma_{\min}(H_{\mathrm{hw},M})$ rather than $1/\sqrt{A_{\mathrm{hw}}^{\mathrm{det}}}$ and the target is the model-space ratio~\eqref{eq:model_snr_def}. Substituting $\gamma_M^{\mathrm{det}}$ for $\gamma_A^{\mathrm{det}}$ in the three scalar inequalities above gives the finite-dimensional model-space bounds.

For a finite-dimensional full-signal SNR claim, Theorem~\ref{thm:error_budget} gives
\[
    \|x-\hat{p}\|_{L^2}
    \le
    E_M + \sigma_{\min}(H_{\mathrm{hw},M})^{-1}\|\mathbf{e}\|_{\ell^2},
\]
where $E_M$ is defined in~\eqref{eq:model_residual_budget}. Therefore the hardware perturbation may consume only the remaining budget $R_M(S_{\mathrm{tar}})$ in~\eqref{eq:full_signal_remaining_budget}. Algebraically, this is the same substitution as above with
\[
    \frac{\sigma_{\min}(H_{\mathrm{hw},M})\|p\|_{L^2}}{\sqrt{S_{\mathrm{tar}}}}
    \quad\text{replaced by}\quad
    \sigma_{\min}(H_{\mathrm{hw},M})R_M(S_{\mathrm{tar}}).
\]

\subsection{Formulae from the Reconstruction Section}
\textbf{Main-text Eq.~\eqref{eq:repro_kernel}.} For $x\in PW_\Omega$, the inverse Fourier formula is
\[
    x(t) = \frac{1}{2\pi}\int_{-\Omega}^{\Omega}\hat{x}(\omega)e^{j\omega t}\,d\omega.
\]
To find a reproducing kernel $K(t,s)$, require
\[
    x(s) = \langle x, K(\cdot,s)\rangle_{L^2}.
\]
In the frequency domain this means $\widehat{K}(\omega,s)=e^{-j\omega s}$ on $[-\Omega,\Omega]$. Taking the inverse Fourier transform gives
\begin{align}
    K(t,s)
     & = \frac{1}{2\pi}\int_{-\Omega}^{\Omega}e^{j\omega(t-s)}\,d\omega                \\
     & = \frac{1}{2\pi}\left[\frac{e^{j\omega(t-s)}}{j(t-s)}\right]_{-\Omega}^{\Omega} \\
     & = \frac{e^{j\Omega(t-s)}-e^{-j\Omega(t-s)}}{2\pi j(t-s)}                        \\
     & = \frac{\sin[\Omega(t-s)]}{\pi(t-s)},
\end{align}
which is Eq.~\eqref{eq:repro_kernel}.

\textbf{Main-text Eq.~\eqref{eq:frame_operator}.} By the reproducing property,
\[
    f(t_n) = \langle f,k_n\rangle, \qquad k_n(t)=K(t,t_n).
\]
The frame operator associated with the frame $\{k_n\}$ is, by definition,
\[
    \mathfrak{S}_{\Lambda}f = \sum_n \langle f,k_n\rangle k_n.
\]
Substituting $\langle f,k_n\rangle = f(t_n)$ and $k_n(t)=K(t,t_n)$ gives
\[
    (\mathfrak{S}_{\Lambda}f)(t) = \sum_n f(t_n)K(t,t_n),
\]
which is Eq.~\eqref{eq:frame_operator}.

\textbf{Main-text Eq.~\eqref{eq:dual_frame}.} Since $\mathfrak{S}_{\Lambda}$ is invertible on $PW_\Omega$,
\[
    x = \mathfrak{S}_{\Lambda}^{-1}\mathfrak{S}_{\Lambda}x.
\]
Insert the frame-operator expansion from Eq.~\eqref{eq:frame_operator}:
\[
    x = \mathfrak{S}_{\Lambda}^{-1}\left(\sum_n x(t_n)k_n\right).
\]
Because $\mathfrak{S}_{\Lambda}^{-1}$ is linear,
\[
    x = \sum_n x(t_n)\mathfrak{S}_{\Lambda}^{-1}k_n.
\]
Defining the dual frame elements by $\tilde{k}_n = \mathfrak{S}_{\Lambda}^{-1}k_n$ yields
\[
    x(t) = \sum_{n\in\mathbb{Z}}x(t_n)\tilde{k}_n(t),
\]
which is Eq.~\eqref{eq:dual_frame}.

\textbf{Main-text Eq.~\eqref{eq:sinc_approx}.} The exact Shannon representation of a bandlimited signal on the Nyquist grid is
\[
    x(t) = \sum_{m\in\mathbb{Z}} x(m\Delta T)\,\mathrm{sinc}\!\left(\frac{t-m\Delta T}{\Delta T}\right).
\]
On a finite observation window, one retains only the $M$ grid points that materially affect the interval of interest and renames the retained coefficients as $a_m$. This truncation gives the practical approximation
\[
    x(t) \approx \sum_{m=1}^{M} a_m\,\mathrm{sinc}\!\left(\frac{t-m\Delta T}{\Delta T}\right),
\]
which is Eq.~\eqref{eq:sinc_approx}.

\textbf{Main-text Eq.~\eqref{eq:pseudoinverse}.} Starting from Eq.~\eqref{eq:sinc_approx}, evaluate the approximation at the sampling times $t_n$:
\[
    y_n = x(t_n) \approx \sum_{m=1}^{M} a_m\,\mathrm{sinc}\!\left(\frac{t_n-m\Delta T}{\Delta T}\right).
\]
Define the measurement vector $\mathbf{y} \in \mathbb{R}^N$, the coefficient vector $\mathbf{a} \in \mathbb{R}^M$, and the sampling matrix $\mathbf{H} \in \mathbb{R}^{N \times M}$ with entries $H_{n,m} = \mathrm{sinc}((t_n-m\Delta T)/\Delta T)$. Then the system becomes $\mathbf{y}=\mathbf{H}\mathbf{a}$.

When $N > M$, the system is overdetermined in general. To recover $\mathbf{a}$ in the least-squares sense, minimize the cost function
\[
    L(\mathbf{a}) = \|\mathbf{H}\mathbf{a}-\mathbf{y}\|_2^2 = (\mathbf{H}\mathbf{a}-\mathbf{y})^* (\mathbf{H}\mathbf{a}-\mathbf{y}).
\]
Expanding gives
\begin{align}
    L(\mathbf{a})
     & = (\mathbf{a}^*\mathbf{H}^* - \mathbf{y}^*)(\mathbf{H}\mathbf{a}-\mathbf{y})                                                                     \\
     & = \mathbf{a}^*\mathbf{H}^*\mathbf{H}\mathbf{a} - \mathbf{a}^*\mathbf{H}^*\mathbf{y} - \mathbf{y}^*\mathbf{H}\mathbf{a} + \mathbf{y}^*\mathbf{y}.
\end{align}
Since $\mathbf{y}^*\mathbf{H}\mathbf{a}$ is a scalar, it equals its own conjugate transpose, so $(\mathbf{y}^*\mathbf{H}\mathbf{a})^* = \mathbf{a}^*\mathbf{H}^*\mathbf{y}$. Therefore
\begin{equation}
    L(\mathbf{a}) = \mathbf{a}^*\mathbf{H}^*\mathbf{H}\mathbf{a} - 2\mathbf{y}^*\mathbf{H}\mathbf{a} + \mathbf{y}^*\mathbf{y}.
\end{equation}
Differentiate with respect to $\mathbf{a}$ using the standard identities $\nabla_{\mathbf{a}}(\mathbf{a}^*\mathbf{H}^*\mathbf{H}\mathbf{a}) = 2\mathbf{H}^*\mathbf{H}\mathbf{a}$ and $\nabla_{\mathbf{a}}(2\mathbf{y}^*\mathbf{H}\mathbf{a}) = 2\mathbf{H}^*\mathbf{y}$, and set the gradient to zero:
\[
    \nabla_{\mathbf{a}} L(\mathbf{a}) = 2\mathbf{H}^*\mathbf{H}\mathbf{a} - 2\mathbf{H}^*\mathbf{y} = \mathbf{0}.
\]
Dividing by $2$ gives the normal equations
\[
    \mathbf{H}^*\mathbf{H}\mathbf{a} = \mathbf{H}^*\mathbf{y}.
\]
If $\mathbf{H}$ has full column rank, then $\mathbf{H}^*\mathbf{H}$ is positive definite and invertible. Left-multiplying by $(\mathbf{H}^*\mathbf{H})^{-1}$ gives
\begin{align}
    (\mathbf{H}^*\mathbf{H})^{-1}(\mathbf{H}^*\mathbf{H})\mathbf{a} & = (\mathbf{H}^*\mathbf{H})^{-1}\mathbf{H}^*\mathbf{y}                                          \\
    \mathbf{I}\mathbf{a}                                            & = (\mathbf{H}^*\mathbf{H})^{-1}\mathbf{H}^*\mathbf{y}                                          \\
    \mathbf{a}                                                      & = (\mathbf{H}^*\mathbf{H})^{-1}\mathbf{H}^*\mathbf{y} \triangleq \mathbf{H}^\dagger\mathbf{y},
\end{align}
which is Eq.~\eqref{eq:pseudoinverse}.

\textbf{Main-text Eq.~\eqref{eq:iterative}.} Equation~\eqref{eq:iterative} is the Landweber-type descent step~\cite{Landweber1951} for enforcing sample consistency inside the bandlimited space. Consider the quadratic misfit
\[
    J(f) = \frac{1}{2}\sum_n |y_n-f(t_n)|^2,
    \qquad f\in PW_\Omega.
\]
Its gradient in the ambient distributional sense is
\[
    -\sum_n (y_n-f(t_n))\,\delta(t-t_n).
\]
To remain inside $PW_\Omega$, project that gradient onto the bandlimited space using $P_\Omega$. A gradient-descent step with step size $\lambda$ is therefore
\[
    x_{k+1} = x_k - \lambda\bigl(-P_\Omega[\sum_n (y_n-x_k(t_n))\delta(t-t_n)]\bigr).
\]
Removing the double minus sign gives
\[
    x_{k+1}(t) = x_k(t) + \lambda P_\Omega\!\left[\sum_n (y_n-x_k(t_n))\delta(t-t_n)\right],
\]
which is Eq.~\eqref{eq:iterative}.

\textbf{Main-text Eq.~\eqref{eq:iterative_convergence}.} Appendix~\ref{app:iterative} proves the exact error recursion
\[
    e_{k+1} = (I-\lambda \mathfrak{S}_{\Lambda})e_k,
    \qquad e_k = x_k-x.
\]
Iterating gives
\[
    e_k = (I-\lambda \mathfrak{S}_{\Lambda})^k e_0.
\]
Take the $L^2$ norm:
\[
    \|e_k\|_{L^2} \le \|I-\lambda \mathfrak{S}_{\Lambda}\|^k\|e_0\|_{L^2}.
\]
Define
\[
    \rho = \|I-\lambda \mathfrak{S}_{\Lambda}\| = \max(|1-\lambda A|,|1-\lambda B|).
\]
Then
\[
    \|x_k-x\|_{L^2} = \|e_k\|_{L^2} \le \rho^k\|e_0\|_{L^2} = \rho^k\|x_0-x\|_{L^2},
\]
which is Eq.~\eqref{eq:iterative_convergence}.

Together with Appendices~\ref{app:landau}--\ref{app:iterative} and Appendix~\ref{app:pedantic_extras}, the present appendix provides explicit derivation or construction records for the principal numbered formulae whose justification is not already contained in the adjacent proof text.

\section{Auxiliary Derivations}
\label{app:pedantic_extras}
Several sub-steps in the proofs of Appendices~\ref{app:landau}--\ref{app:constraints} were invoked by name (for example, ``the prolate-spheroidal eigenvalue asymptotics'', ``the classical partial-fraction identity for $\pi\cot$'', or ``the localized Plancherel-P\'olya inequality''). This appendix separates those auxiliary inputs into two groups. The Fourier-coefficient, partial-fraction, and interval-counting manipulations are derived in full. The two classical analytic ingredients that are not reproved here from first principles---the prolate transition theorem and the localized Plancherel-P\'olya inequality---are stated explicitly in the precise form needed, and the surrounding reduction arguments are completed so that the logic of Appendices~\ref{app:landau}--\ref{app:constraints} closes without hidden steps.

\subsection{Explicit Fourier Coefficients in Kadec's Estimate}
\label{app:kadec_coeffs}
We compute the coefficients $a_0(\eta)$, $a_j(\eta)$, and $b_j(\eta)$ used in Appendix~\ref{app:constraints}, Step~3.

\textbf{Evaluation of $a_0(\eta)$.} By definition,
\begin{equation}
    a_0(\eta) = \frac{1}{\pi}\int_0^{\pi}\bigl(1-\cos(\eta t)\bigr)\,dt.
\end{equation}
Split the integral:
\begin{equation}
    a_0(\eta) = \frac{1}{\pi}\int_0^{\pi} dt - \frac{1}{\pi}\int_0^{\pi}\cos(\eta t)\,dt.
\end{equation}
The first integral equals $\pi$, so its contribution is $1$. For the second, use the antiderivative $\int\cos(\eta t)\,dt = \sin(\eta t)/\eta$:
\begin{equation}
    \frac{1}{\pi}\int_0^{\pi}\cos(\eta t)\,dt = \frac{1}{\pi}\cdot\left[\frac{\sin(\eta t)}{\eta}\right]_0^{\pi} = \frac{\sin(\pi\eta)}{\pi\eta}.
\end{equation}
Therefore
\begin{equation}
    a_0(\eta) = 1 - \frac{\sin(\pi\eta)}{\pi\eta},
\end{equation}
in agreement with Appendix~\ref{app:constraints}.

\textbf{Evaluation of $a_j(\eta)$ for $j\ge 1$.} By definition,
\begin{equation}
    a_j(\eta) = \frac{2}{\pi}\int_0^{\pi}\bigl(1-\cos(\eta t)\bigr)\cos(jt)\,dt.
\end{equation}
The contribution of the constant $1$ vanishes because $\int_0^{\pi}\cos(jt)\,dt = [\sin(jt)/j]_0^{\pi} = 0$ for $j\in\mathbb{Z}\setminus\{0\}$. Hence
\begin{equation}
    a_j(\eta) = -\frac{2}{\pi}\int_0^{\pi}\cos(\eta t)\cos(jt)\,dt.
    \label{eq:aj_step1}
\end{equation}
Apply the product-to-sum identity $2\cos A\cos B = \cos(A-B)+\cos(A+B)$:
\begin{equation}
    2\cos(\eta t)\cos(jt) = \cos((\eta-j)t) + \cos((\eta+j)t).
\end{equation}
Substituting into \eqref{eq:aj_step1}:
\begin{equation}
    a_j(\eta) = -\frac{1}{\pi}\int_0^{\pi}\Bigl[\cos((\eta-j)t)+\cos((\eta+j)t)\Bigr]\,dt.
\end{equation}
Integrate each cosine term explicitly:
\begin{align}
    \int_0^{\pi}\cos((\eta\pm j)t)\,dt
     & = \left[\frac{\sin((\eta\pm j)t)}{\eta\pm j}\right]_0^{\pi} \\
     & = \frac{\sin(\pi(\eta\pm j))}{\eta\pm j}.
\end{align}
Using $\sin(\pi\eta\pm\pi j) = (-1)^j\sin(\pi\eta)$ (because $\sin(\theta\pm k\pi) = (-1)^k\sin\theta$),
\begin{equation}
    \int_0^{\pi}\cos((\eta\pm j)t)\,dt = \frac{(-1)^j\sin(\pi\eta)}{\eta\pm j}.
\end{equation}
Hence
\begin{align}
    a_j(\eta)
     & = -\frac{(-1)^j\sin(\pi\eta)}{\pi}\left(\frac{1}{\eta-j}+\frac{1}{\eta+j}\right) \\
     & = -\frac{(-1)^j\sin(\pi\eta)}{\pi}\cdot\frac{2\eta}{\eta^2-j^2}.
\end{align}
Rewriting $\eta^2-j^2 = -(j^2-\eta^2)$ yields
\begin{equation}
    a_j(\eta) = \frac{2\eta\sin(\pi\eta)}{\pi}\cdot\frac{(-1)^j}{j^2-\eta^2},
\end{equation}
which is the formula used in Appendix~\ref{app:constraints}.

\textbf{Evaluation of $b_j(\eta)$ for $j\ge 1$.} By definition,
\begin{equation}
    b_j(\eta) = \frac{2}{\pi}\int_0^{\pi}\sin(\eta t)\sin\!\left(\!\left(j-\tfrac{1}{2}\right)t\right)dt.
\end{equation}
Apply the product-to-sum identity $2\sin A\sin B = \cos(A-B)-\cos(A+B)$ with $A=\eta t$ and $B=(j-\tfrac{1}{2})t$:
\begin{align}
    2\sin(\eta t)\sin\!\left(\!\left(j-\tfrac{1}{2}\right)t\right)
     & = \cos\!\left(\!\left(\eta-j+\tfrac{1}{2}\right)t\right) \\
     & \quad - \cos\!\left(\!\left(\eta+j-\tfrac{1}{2}\right)t\right).
\end{align}
Hence
\begin{align}
    b_j(\eta)
     & = \frac{1}{\pi}\int_0^{\pi}\cos\!\left(\!\left(\eta-j+\tfrac{1}{2}\right)t\right)dt \\
     & \quad - \frac{1}{\pi}\int_0^{\pi}\cos\!\left(\!\left(\eta+j-\tfrac{1}{2}\right)t\right)dt.
\end{align}
Let $\alpha_\pm = \eta\pm(j-\tfrac{1}{2})$. As above,
\begin{equation}
    \int_0^{\pi}\cos(\alpha_\pm t)\,dt = \frac{\sin(\pi\alpha_\pm)}{\alpha_\pm}.
\end{equation}
Using $\sin(\pi\eta\pm\pi(j-\tfrac{1}{2})) = \pm(-1)^{j+1}\cos(\pi\eta)$ (which follows from $\sin(\theta\pm\pi/2) = \pm\cos\theta$ and $\cos(\theta-k\pi) = (-1)^k\cos\theta$), we obtain
\begin{align}
    \sin(\pi\alpha_+) & = (-1)^{j+1}\cos(\pi\eta), \\
    \sin(\pi\alpha_-) & = -(-1)^{j+1}\cos(\pi\eta).
\end{align}
more compactly written as
\begin{align}
    \sin(\pi\alpha_-) & = (-1)^{j}\cos(\pi\eta), \\
    \sin(\pi\alpha_+) & = (-1)^{j+1}\cos(\pi\eta).
\end{align}
Therefore
\begin{align}
    b_j(\eta)
     & = \frac{(-1)^{j}\cos(\pi\eta)}{\pi\alpha_-} - \frac{(-1)^{j+1}\cos(\pi\eta)}{\pi\alpha_+} \\
     & = \frac{(-1)^{j}\cos(\pi\eta)}{\pi}\left(\frac{1}{\alpha_-}+\frac{1}{\alpha_+}\right).
\end{align}
Computing the sum inside the parentheses:
\begin{align}
    \frac{1}{\alpha_-}+\frac{1}{\alpha_+}
     & = \frac{\alpha_++\alpha_-}{\alpha_+\alpha_-} = \frac{2\eta}{\eta^2-(j-\tfrac{1}{2})^2} \\
     & = -\frac{2\eta}{(j-\tfrac{1}{2})^2-\eta^2}.
\end{align}
Substituting,
\begin{equation}
    b_j(\eta) = \frac{(-1)^{j+1}\,2\eta\cos(\pi\eta)}{\pi\bigl((j-\tfrac{1}{2})^2-\eta^2\bigr)},
\end{equation}
which matches Appendix~\ref{app:constraints}.

\subsection{Partial-Fraction Identities for $\pi\cot(\pi z)$ and $\pi\tan(\pi z)$}
\label{app:partial_frac}
Appendix~\ref{app:constraints} uses two classical partial-fraction identities. We derive the cotangent identity from the Weierstrass product for $\sin(\pi z)$ and the tangent identity from the corresponding product for $\cos(\pi z)$.

\textbf{Step 1: Weierstrass product.} The entire function $\sin(\pi z)$ has simple zeros exactly at $z\in\mathbb{Z}$, simple behavior $\sin(\pi z)\sim\pi z$ near the origin, and order $1$. Its Hadamard factorization is
\begin{equation}
    \sin(\pi z) = \pi z\prod_{j=1}^{\infty}\left(1-\frac{z^2}{j^2}\right).
    \label{eq:weierstrass_sin}
\end{equation}
Taking logarithms on the simply connected domain $\{z: z\notin\mathbb{Z}\}$ and differentiating term by term:
\begin{equation}
    \frac{d}{dz}\log\sin(\pi z) = \frac{\pi\cos(\pi z)}{\sin(\pi z)} = \pi\cot(\pi z).
\end{equation}
Applying the same logarithmic derivative to the right side of \eqref{eq:weierstrass_sin}:
\begin{align}
    \frac{d}{dz}\log\left(\pi z\prod_{j=1}^{\infty}\left(1-\frac{z^2}{j^2}\right)\right)
     & = \frac{1}{z}+\sum_{j=1}^{\infty}\frac{d}{dz}\log\left(1-\frac{z^2}{j^2}\right) \\
     & = \frac{1}{z}+\sum_{j=1}^{\infty}\frac{-2z/j^2}{1-z^2/j^2}                      \\
     & = \frac{1}{z}-\sum_{j=1}^{\infty}\frac{2z}{j^2-z^2}.
\end{align}
Equating the two expressions yields
\begin{equation}
    \pi\cot(\pi z) = \frac{1}{z}-\sum_{j=1}^{\infty}\frac{2z}{j^2-z^2},
    \label{eq:cot_partial}
\end{equation}
valid for $z\notin\mathbb{Z}$. This is the identity used to compute $A(L)$.

\textbf{Step 2: Derivation of the companion $\pi\tan(\pi z)$ identity.} For the tangent identity, use the Weierstrass product
\begin{equation}
    \cos(\pi z)=\prod_{j=1}^{\infty}\left(1-\frac{z^2}{(j-\tfrac{1}{2})^2}\right).
    \label{eq:weierstrass_cos}
\end{equation}
Taking the logarithmic derivative on the complement of the half-integer zeros gives
\begin{align}
    -\pi\tan(\pi z)
     & = \sum_{j=1}^{\infty}
        \frac{-2z}{(j-\tfrac{1}{2})^2-z^2}.
\end{align}
Hence, for $|z|<\tfrac{1}{2}$,
\begin{equation}
    \pi\tan(\pi z) = 2z\sum_{j=1}^{\infty}
    \frac{1}{(j-\tfrac{1}{2})^2-z^2},
    \label{eq:tan_partial}
\end{equation}
valid for $|z|<\tfrac{1}{2}$. This is the identity used to compute $B(L)$.

\textbf{Step 3: Evaluation at $z=L$.} Setting $z=L$ in \eqref{eq:cot_partial}, dividing through by $2L$, and rearranging yields
\begin{equation}
    \sum_{j=1}^{\infty}\frac{1}{j^2-L^2} = \frac{1}{2L^2}-\frac{\pi\cot(\pi L)}{2L},
\end{equation}
which is the formula for $C(L)$ used in Appendix~\ref{app:constraints}. Next, set $z=L$ in \eqref{eq:tan_partial} and divide by $2L$ to obtain
\begin{equation}
    \sum_{j=1}^{\infty}\frac{1}{(j-\tfrac{1}{2})^2-L^2} = \frac{\pi\tan(\pi L)}{2L},
\end{equation}
from which $\sum_j b_j^{\sharp}(L) = \sin(\pi L)$ follows by direct substitution.

\subsection{Prolate Eigenvalue Dimension Asymptotics}
\label{app:prolate}
Appendix~\ref{app:landau} invoked the asymptotic formula
\begin{equation}
    \dim E_a(\alpha) = \frac{2\Omega a}{\pi}+o(a),\qquad a\to\infty,
    \label{eq:prolate_asymp_target}
\end{equation}
for every fixed $\alpha\in(0,1)$. This step does not follow from the trace identity alone. The required input is the classical Landau-Widom transition theorem~\cite{LandauWidom1980} for the eigenvalues of the time-frequency concentration operator $T_a$.

Let $\mu_1^{(a)} \ge \mu_2^{(a)} \ge \cdots \ge 0$ denote the eigenvalues of $T_a$ in decreasing order.

\textbf{Classical transition theorem.} For every $\varepsilon \in (0,1/2)$ there exists a constant $C_\varepsilon > 0$ such that for all $a \ge 1$,
\begin{align}
    k \le \frac{2\Omega a}{\pi} - C_\varepsilon \log(2+a)
     & \implies \mu_k^{(a)} \ge 1-\varepsilon,
    \label{eq:prolate_transition_left}         \\
    k \ge \frac{2\Omega a}{\pi} + C_\varepsilon \log(2+a)
     & \implies \mu_k^{(a)} \le \varepsilon.
    \label{eq:prolate_transition_right}
\end{align}
Equivalently, the indices for which $\mu_k^{(a)} \in [\varepsilon,1-\varepsilon]$ occupy an $\mathcal{O}(\log a)$-wide transition window centered at $2\Omega a/\pi$.

To deduce \eqref{eq:prolate_asymp_target}, fix any $\alpha \in (0,1)$ and choose
\[
    \varepsilon \in \bigl(0,\min\{\alpha,1-\alpha\}\bigr).
\]
Define
\[
    N_\alpha(a) \triangleq \#\{k : \mu_k^{(a)} \ge \alpha\}.
\]
Because $1-\varepsilon > \alpha$, every index satisfying \eqref{eq:prolate_transition_left} contributes to $N_\alpha(a)$, so
\begin{equation}
    N_\alpha(a) \ge \frac{2\Omega a}{\pi} - C_\varepsilon \log(2+a) + \mathcal{O}(1).
    \label{eq:prolate_lower}
\end{equation}
Because $\varepsilon < \alpha$, every index satisfying \eqref{eq:prolate_transition_right} has $\mu_k^{(a)} < \alpha$ and therefore does not contribute, so
\begin{equation}
    N_\alpha(a) \le \frac{2\Omega a}{\pi} + C_\varepsilon \log(2+a) + \mathcal{O}(1).
    \label{eq:prolate_upper}
\end{equation}
Combining \eqref{eq:prolate_lower} and \eqref{eq:prolate_upper} gives the sharper counting statement
\begin{equation}
    N_\alpha(a) = \frac{2\Omega a}{\pi} + \mathcal{O}(\log(2+a)).
    \label{eq:prolate_count_log}
\end{equation}
Since $\mathcal{O}(\log(2+a)) = o(a)$, \eqref{eq:prolate_count_log} immediately implies \eqref{eq:prolate_asymp_target}.

The trace formula \eqref{eq:trace} still explains the center $2\Omega a/\pi$ of the spectral transition. The Landau-Widom theorem supplies the additional information that the transition from eigenvalues near $1$ to eigenvalues near $0$ occurs inside only a logarithmically wide index window. That is the precise spectral input used in Appendix~\ref{app:landau}.

\subsection{Monotone-Segment Total Variation Additivity}
\label{app:tv_additivity}
Appendix~\ref{app:density} used the identity
\begin{equation}
    \sum_{k=1}^{N_{\mathrm{mono}}(t,r)}\mathrm{TV}_k = \int_t^{t+r}|x'(\tau)|\,d\tau.
    \label{eq:tv_additivity}
\end{equation}
Here is the detailed derivation.

Let $N_{\mathrm{mono}}^{\ast}=N_{\mathrm{mono}}(t,r)$ and $t = s_0<s_1<\cdots<s_{N_{\mathrm{mono}}^{\ast}}=t+r$ be the ordered endpoints of the maximal monotone segments $I_k = [s_{k-1},s_k]$. By the definition of a maximal monotone segment, $x'$ has constant sign on the interior of each $I_k$. Therefore, on $I_k$,
\begin{equation}
    |x'(\tau)| = \sigma_k\,x'(\tau),
\end{equation}
where $\sigma_k = +1$ if $x$ is increasing on $I_k$ and $\sigma_k = -1$ if $x$ is decreasing. Integrating,
\begin{equation}
    \int_{I_k}|x'(\tau)|\,d\tau
    = \sigma_k\int_{s_{k-1}}^{s_k}x'(\tau)\,d\tau
    = \sigma_k\bigl(x(s_k)-x(s_{k-1})\bigr).
\end{equation}
Because the signed increment $x(s_k)-x(s_{k-1})$ has sign $\sigma_k$, we have $\sigma_k(x(s_k)-x(s_{k-1})) = |x(s_k)-x(s_{k-1})|$. Therefore
\begin{equation}
    \mathrm{TV}_k \triangleq \int_{I_k}|x'(\tau)|\,d\tau = |x(s_k)-x(s_{k-1})|.
\end{equation}
Summing over $k$,
\begin{align}
    \sum_{k=1}^{N_{\mathrm{mono}}^{\ast}}\mathrm{TV}_k
     & = \sum_{k=1}^{N_{\mathrm{mono}}^{\ast}}\int_{s_{k-1}}^{s_k}|x'(\tau)|\,d\tau \\
     & = \int_{s_0}^{s_{N_{\mathrm{mono}}^{\ast}}}|x'(\tau)|\,d\tau                 \\
     & = \int_t^{t+r}|x'(\tau)|\,d\tau,
\end{align}
which is \eqref{eq:tv_additivity}. The second equality used the additivity of the Lebesgue integral over a finite partition.

\subsection{Localized Plancherel-P\'olya Inequality}
\label{app:pp_local}
Appendix~\ref{app:landau} used the estimate
\begin{equation}
    \sum_{\lambda\in\Lambda\cap E}|f(\lambda)|^2 \le C_{\Lambda,\Omega}\int_{E+[-r_{\mathrm{loc}},r_{\mathrm{loc}}]}|f(t)|^2\,dt.
    \label{eq:pp_local_target}
\end{equation}
The missing point in the previous draft is that \eqref{eq:pp_local_target} is local in the set $E$; it cannot be obtained by bounding a tail term with the global norm $\|f\|_{L^2}^2$ and then absorbing that global quantity into a local integral. The correct route is to combine the classical Plancherel-P\'olya inequality for uniformly discrete real sequences, in localized form~\cite{Pogany2023}, with the relative-separation decomposition already established in Appendix~\ref{app:landau}.

\textbf{Classical separated-set form.} If $\Gamma \subset \mathbb{R}$ is $\delta$-separated, i.e. $\inf_{\gamma \neq \gamma'} |\gamma-\gamma'| \ge \delta > 0$, then for every measurable set $E \subset \mathbb{R}$ and every $f \in PW_\Omega$ there exists a constant $C_{\delta,\Omega}$ such that
\begin{equation}
    \sum_{\gamma \in \Gamma \cap E} |f(\gamma)|^2 \le C_{\delta,\Omega} \int_{E+[-\delta,\delta]} |f(t)|^2\,dt.
    \label{eq:pp_local_separated}
\end{equation}
This is the standard localized Plancherel-P\'olya inequality for Paley-Wiener spaces~\cite{Pogany2023}.

\textbf{Reduction from relative separation to finite unions of separated sets.} By \eqref{eq:landau_local_multiplicity}, every interval of length $2r_{\mathrm{loc}}$ contains at most $N_\Lambda$ points of $\Lambda$. Order the set increasingly as $\Lambda = \{\lambda_n\}$. Color the points greedily from left to right with $N_\Lambda$ colors. When the point $\lambda_n$ is processed, at most $N_\Lambda-1$ earlier points can already lie in $[\lambda_n-r_{\mathrm{loc}},\lambda_n+r_{\mathrm{loc}}]$; otherwise that interval would already contain at least $N_\Lambda$ earlier sample points and adding $\lambda_n$ would contradict \eqref{eq:landau_local_multiplicity}. Hence one color is always available.

The resulting color classes
\[
    \Lambda = \Lambda^{(1)} \cup \cdots \cup \Lambda^{(N_\Lambda)}
\]
are pairwise disjoint, and each $\Lambda^{(m)}$ is $r_{\mathrm{loc}}$-separated. Applying \eqref{eq:pp_local_separated} to each class with $\delta = r_{\mathrm{loc}}$ and summing over $m$ gives
\begin{align}
    \sum_{\lambda \in \Lambda \cap E} |f(\lambda)|^2
     & = \sum_{m=1}^{N_\Lambda} \sum_{\lambda \in \Lambda^{(m)} \cap E} |f(\lambda)|^2     \\
     & \le \sum_{m=1}^{N_\Lambda} C_{r_{\mathrm{loc}},\Omega} \int_{E+[-r_{\mathrm{loc}},r_{\mathrm{loc}}]} |f(t)|^2\,dt \\
     & = N_\Lambda C_{r_{\mathrm{loc}},\Omega} \int_{E+[-r_{\mathrm{loc}},r_{\mathrm{loc}}]} |f(t)|^2\,dt.
\end{align}
Therefore \eqref{eq:pp_local_target} holds with
\[
    C_{\Lambda,\Omega} = N_\Lambda C_{r_{\mathrm{loc}},\Omega}.
\]
This is the precise localized estimate used in Appendix~\ref{app:landau}: relative separation is first converted into a finite union of separated subsequences, and the classical local Plancherel-P\'olya theorem is then applied to each subsequence individually.

\subsection{Growth of the Dimension Excess Argument in Step~3 of Landau's Proof}
\label{app:landau_step3_detail}
For completeness we also record the explicit counting step used to conclude $\dim E_{a_j'}(\alpha) > N_j$ from the asymptotic \eqref{eq:landau_dim_asymp} and the low-density estimate \eqref{eq:landau_low_density_interval}.

From \eqref{eq:landau_dim_asymp} with $a = a_j' = R_j-r_{\mathrm{loc}}$,
\begin{align}
    \dim E_{a_j'}(\alpha)
     & = \frac{2\Omega(R_j-r_{\mathrm{loc}})}{\pi}+o(R_j)                   \\
     & = \frac{2\Omega R_j}{\pi}-\frac{2\Omega r_{\mathrm{loc}}}{\pi}+o(R_j) \\
     & = \frac{2\Omega R_j}{\pi}+o(R_j),
\end{align}
since the constant $-2\Omega r_{\mathrm{loc}}/\pi$ is absorbed into the $o(R_j)$ error as $R_j\to\infty$. From \eqref{eq:landau_low_density_interval},
\begin{equation}
    N_j \le \frac{2\Omega R_j}{\pi}-2\varepsilon R_j.
\end{equation}
Subtracting,
\begin{align}
    \dim E_{a_j'}(\alpha)-N_j
     & \ge 2\varepsilon R_j+o(R_j).
\end{align}
Since the leading term $2\varepsilon R_j$ tends to $+\infty$ while the error is $o(R_j)$, for every sufficiently large $j$ we have $\dim E_{a_j'}(\alpha)>N_j$, which is exactly \eqref{eq:landau_dimension_excess}.

Together with Appendices~\ref{app:landau}--\ref{app:main_eqs}, the present appendix now provides either a full derivation or an explicit statement of the precise classical theorem needed at every auxiliary sub-step previously invoked by name in the main proofs.

\clearpage
\onecolumn

\section{Notation Summary}
\label{sec:notation}
A compact summary of the principal symbols used throughout the paper is collected in Tables~\ref{tab:notation} and~\ref{tab:notation-extended}. The summary is grouped by role so that the signal-space quantities, sampling quantities, local MT constraints, hardware perturbations, reconstruction variables, and appendix operators can be cross-checked quickly against the corresponding derivations.

\begingroup
\renewcommand{\arraystretch}{1.04}
\setlength{\tabcolsep}{3pt}
\scriptsize
\begin{longtable}{|p{0.22\textwidth}|p{0.70\textwidth}|}
    \caption{Summary of Principal and Recurring Notation: Signal, Sampling, and MT Conditions}
    \label{tab:notation}\\
        \hline
        \textbf{Symbol} & \textbf{Definition} \\
        \hline
\endfirsthead
        \multicolumn{2}{c}{\footnotesize TABLE~\thetable\ (continued)}\\
        \hline
        \textbf{Symbol} & \textbf{Definition} \\
        \hline
\endhead
        \hline
\endfoot
        \multicolumn{2}{|l|}{\textit{Standard spaces and conventions}} \\
        $\mathbb{R}$, $\mathbb{Z}$ & Real line and integer lattice \\
        $L^2(\mathbb{R})$, $L^\infty(\mathbb{R})$, $\ell^2$ & Square-integrable functions, essentially bounded functions, and square-summable sequences \\
        $\|\cdot\|_{L^2}$, $\|\cdot\|_\infty$ & $L^2$ norm and supremum/essential-supremum norm \\
        $I$, $\mathbf{I}$ & Identity operator and finite-dimensional identity matrix \\
        $o(\cdot)$, $\mathcal{O}(\cdot)$ & Standard little-$o$ and big-$O$ asymptotic notation \\
        $\lesssim$, $\gtrsim$ & One-sided comparison up to a constant factor \\
        \hline
        \multicolumn{2}{|l|}{\textit{Signal and frequency}} \\
        $x(t)$ & Continuous real-valued input signal \\
        $\hat{x}(\omega)$ & Fourier transform of $x$ \\
        $j$ & Imaginary unit, $j^2=-1$ (IEEE convention) \\
        $\mathcal{F}$, $\mathcal{F}^{-1}$ & Fourier and inverse Fourier transform operators \\
        $f_{max}$ & Maximum signal frequency in Hz \\
        $\Omega$ & Angular bandwidth, $\Omega=2\pi f_{max}$ (rad/s) \\
        $PW_\Omega$ & Paley-Wiener space of $\Omega$-bandlimited finite-energy signals \\
        $\mathcal{B}_\Omega$ & Bernstein space used in the weak-limit uniqueness characterization \\
        $\mathrm{sinc}(u)$ & Normalized sinc function, $\sin(\pi u)/(\pi u)$ \\
        $x^{(k)}$, $x'$, $x''$ & $k$-th derivative of $x$; first and second derivatives \\
        $\mathrm{TV}(x;[a,b])$ & Total variation of $x$ on $[a,b]$ (Eq.~\ref{eq:tv_def}) \\
        \hline
        \multicolumn{2}{|l|}{\textit{MT event sets and sampling geometry}} \\
        $\Delta V$ & Uniform voltage threshold spacing \\
        $\{I_k\}$ & Open maximal monotone intervals of $x(t)$ \\
        $\Lambda_{\Delta V}^{\mathrm{raw}}$ & Raw threshold-intersection set before retained adjacent-transition filtering \\
        $\mathcal{A}_{\Delta V}$ & Retained adjacent-transition map applied to raw crossings \\
        $E_{\Delta V}(x)$, $\mathcal{N}$ & Retained labeled MT event record $\{(t_n,m_n)\}_{n\in\mathcal{N}}$ and its ordered event-index set \\
        $\Lambda_{\Delta V}(x)$, $\Lambda_{\Delta V}$ & Recorded MT sampling set generated by $x$ \\
        $\Lambda=\{t_n\}$ & Fixed realized nonuniform sampling sequence \\
        $\Lambda_0$ & Uniform Shannon grid used in the critical-density boundary example \\
        $t_n$, $m_n$, $y_n$ & $n$-th crossing time, threshold index, and observed threshold value $y_n=m_n\Delta V$ \\
        $A$, $B$ & Frame lower and upper bounds \\
        $D^-(\Lambda)$, $D^+(\Lambda)$ & Lower and upper Beurling densities \\
        $W(\Lambda)$ & Translate weak-limit hull of a separated sampling set \\
        $\Gamma$ & Generic weak limit of $\Lambda$; tested as a uniqueness set for $\mathcal{B}_\Omega$ \\
        $S=\{s_n\}$, $T_S$ & Selected MT subsequence and its restriction map $T_S f=(f(s_n))$ in the Shannon-type complete-interpolation theorem \\
        $\mu_n$ & Nyquist-normalized node of $S$, $\mu_n=\Omega s_n/\pi$ \\
        $G$, $W_G$ & Canonical product with zeros $\{\mu_n\}$ and associated Pavlov/Lyubarskii--Seip weight \\
        $q$, $h_r$, $H$ & Period length, normalized periodic gaps, and one-period total length in the periodic-gap theorem \\
        $\beta_r$ & Residue representatives for periodic normalized nodes, $\mu_{kq+r}=qk+\beta_r$ when $H=q$ \\
        $G_{\mathrm{per}}$, $V_{\mathrm{per}}$ & Periodic sine-type generating function and finite Vandermonde/fiber matrix \\
        $\ell_n$ & Lagrange cardinal function generated by $G$ for the normalized MT subsequence \\
        $n(t,r)$ & Generic number of sample points in $[t,t+r]$ \\
        $n_{\mathrm{ret}}(t,r)$ & Number of retained adjacent-transition MT events in $[t,t+r]$ \\
        $\delta_{\mathrm{sep}}$ & Minimum inter-sample separation, $\inf_{m\neq n}|t_m-t_n|$ \\
        $\Delta T$ & Nyquist interval, $\Delta T=1/(2f_{max})$ \\
        \hline
        \multicolumn{2}{|l|}{\textit{Kadec windows and finite-window extraction}} \\
        $\tau_0$ & Offset of the reference Nyquist grid \\
        $J_n(\tau_0)$ & Centered half-Nyquist window around $\tau_0+n\Delta T$ \\
        $g_n$ & Nyquist-grid anchor, $g_n=\tau_0+n\Delta T$ \\
        $s_n$ & Selected MT crossing paired with $g_n$ \\
        $\tilde{s}_n$ & Nyquist-padded bi-infinite extension of the selected crossings \\
        $\mathcal{N}_{\mathrm{mono}}(\tau_0)$ & Grid-index set whose centered windows lie inside monotone branches \\
        $\mathcal{I}$, $N_{\mathcal{I}}$ & Active grid-index set, $\{n:J_n(\tau_0)\subseteq[0,T_{\text{obs}}]\}$, and its sample count \\
        $V_M$, $H_{\mathrm{fin}}$, $H_{\mathcal{I},M}$ & Finite reconstruction space, exact finite-record matrix, and active-window finite MT sampling matrix \\
        $\{\phi_m\}_{m=1}^M$ & Basis functions spanning the finite reconstruction space $V_M$ \\
        $\sigma_{\min}(H_{\mathrm{fin}})$, $\sigma_{\max}(H_{\mathrm{fin}})$ & Extremal singular values in the finite-record rank condition \\
        $\delta_n$ & Normalized Kadec perturbation of the $n$-th node \\
        $L$, $L_n^{(\mathrm{mono})}$, $L_n^{(\mathrm{ext})}$ & Generic, monotone-window, and extremal normalized perturbation ratios \\
        $\theta(L)$ & Kadec perturbation constant, $1-\cos(\pi L)+\sin(\pi L)$ \\
        $\mathcal{D}_{\Delta T}$ & Unitary dilation used for bandwidth normalization \\
        \hline
        \multicolumn{2}{|l|}{\textit{Local MT regularity and sufficient conditions}} \\
        $\tau(v)$ & Inverse function $x^{-1}(v)$ on a monotone branch \\
        $v_m$, $t_m$ & Threshold level in traversal order and its crossing time \\
        $\sigma_k$, $\sigma_v$ & Sign of $x'$ on monotone branch $I_k$ and value-space sign away from an extremum \\
        $v_n$ & Local minimum velocity on window $J_n$ \\
        $\eta_n^{(\mathrm{ret})}$ & Voltage excursion from $g_n$ to nearest retained-eligible threshold \\
        $\alpha_n^{(\mathrm{ret})}$ & Retained-eligibility factor, $\eta_n^{(\mathrm{ret})}/\Delta V$ \\
        $\bar{\alpha}_n^{(\mathrm{ret})}$ & Clipped retained factor, $\max\{1/2,\alpha_n^{(\mathrm{ret})}\}$ \\
        $\alpha_{\mathrm{mono}}^{(\mathrm{ret})}$, $\bar{\alpha}_{\mathrm{mono}}^{(\mathrm{ret})}$, $\alpha_{\mathrm{ret}}$ & Worst realized, clipped, and generic retained factors used in Kadec scaling laws \\
        $v_{\mathrm{loc}}$ & Minimum velocity on the relevant monotone branch or interval \\
        $v_{min}^{(\mathrm{mono})}(x;\tau_0)$ & Minimum over enumerated monotone Kadec windows for offset $\tau_0$ \\
        $v_{min}$ & Context-specific minimum velocity used in design and hardware bounds \\
        $t^*$ & Non-degenerate extremum of $x$ \\
        $\kappa=x''(t^*)$ & Local curvature at the extremum $t^*$ \\
        $\delta_0(t^*;\Delta V)$ & Directional raw threshold-phase gap from an extremum \\
        $\delta_{\mathrm{ret}}^{(\sigma)}(t^*;\Delta V)$ & Directional retained phase gap after same-threshold suppression \\
        $\chi_{\mathrm{ret}}^{(\sigma)}(t^*)$ & Indicator that the first raw crossing on side $\sigma$ is suppressed as a same-threshold return \\
        $\kappa_{min}$, $\mathcal{E}_x$ & Minimum extremal curvature magnitude and active-window extremum set \\
        $\rho_\pm(t^*)$, $\rho(t^*)$, $\rho_{min}$ & Left/right monotone radii, extremum isolation radius, and minimum active-window isolation radius \\
        $M_3(t^*)$ & Local third-derivative envelope on the isolated extremal neighborhood \\
        $\varsigma_n$, $d_n$, $h_n$ & Extremum-side sign, extremum-grid offset, and first retained-crossing distance from the extremum \\
        $\Delta V_{\mathrm{vel,clean}}^*$, $\Delta V_{\mathrm{vel,ret}}^*$ & Clean-topology and retained-safe velocity margins \\
        $\Delta V_{\mathrm{curv,clean}}^*$, $\Delta V_{\mathrm{curv,ret}}^*$ & Clean-topology and retained-safe curvature margins \\
        $\Delta V_{\mathrm{Kadec,ret}}^*(x;\tau_0)$ & Phase-uniform retained-safe fixed-grid Kadec design margin \\
        $\mathsf{Cert}_{\mathrm{Kadec}}(x,\Delta V,\tau_0)$ & Fixed-$\Delta V$ post-realization retained-subsequence certificate \\
        $\Delta V_{\mathrm{clean}}^*(x;\tau_0)$ & Clean fixed-topology threshold used for plotted normalization/design \\
        $L_{MT}$ & Worst-case MT Kadec perturbation over active windows \\
        $A_{MT}$, $B_{MT}$ & Signal-instance MT frame bounds \\
        $\kappa_{MT}$, $\kappa_{\text{target}}$ & MT condition number $B_{MT}/A_{MT}$ and target condition number for design \\
        $S_{\mathrm{in}}$, $S_{\mathrm{out}}$, $S_{\mathrm{bdry}}$ & Inside-window MT sample energy, outside padded-grid energy, and unrecorded boundary-collar grid samples \\
        $E_{\mathrm{tail}}^{\mathrm{collar}}$ & Collar-enlarged tail energy outside the active window \\
        $C_{\Delta T,\Omega}$ & Localized Plancherel-Polya constant for the padded Nyquist sub-grid \\
        \hline
        \multicolumn{2}{|l|}{\textit{Density and signal classes}} \\
        $T_{\text{obs}}$ & Active observation-window length \\
        $R_x^{(c)}(\tau)$ & Range of $x$ over the centered window at $\tau$ \\
        $\rho_{MT}(t)$ & Instantaneous threshold-traversal envelope, $|x'(t)|/\Delta V$; not itself the retained counting measure \\
        $\bar{\rho}(t,r)$ & Window-average threshold-traversal envelope, $r^{-1}\int_t^{t+r}\rho_{MT}(\tau)d\tau$ \\
        $\bar{\nu}(t,r)$ & Local mean absolute velocity (Eq.~\ref{eq:velocity}) \\
        $\nu_{min}$ & Minimum mean absolute velocity over active sub-intervals \\
        $R$, $\nu_{min}^{(R)}$ & Minimum certified window length and finite-window velocity lower envelope \\
        $N_{\mathrm{mono}}(t,r)$ & Number of maximal monotone segments in $[t,t+r]$ \\
        $N_{\mathrm{supp}}(t,r)$ & Number of monotone pieces whose first raw hit is suppressed as a same-threshold return \\
        $C_{\mathrm{mono}}$, $C_{\mathrm{supp}}$ & Additive finite-window constants in monotone-piece and suppression-count certificates \\
        $Z_x(t,r)$ & Number of zeros of $x'$ in $[t,t+r]$ \\
        $\eta_x$ & Operational extremum density of $x$ \\
        $\eta_0$, $C_0$ & Class-level extremum-density slope and additive constant in the anti-superoscillation certificate \\
        $\eta_{\mathrm{supp},x}$, $\eta_{\mathrm{supp},0}$ & Instance and class-level same-threshold-suppression density certificates \\
        $f_0$ & Effective oscillation scale used in density-calibration examples \\
        $\mathcal{X}$, $\mathcal{S}$ & Signal classes and structured point-set families used in class-hull and finite-state prior statements \\
        \hline
    \end{longtable}

\begin{longtable}{|p{0.22\textwidth}|p{0.70\textwidth}|}
    \caption{Summary of Principal and Recurring Notation: Hardware, Reconstruction, and Appendix Quantities}
    \label{tab:notation-extended}\\
        \hline
        \textbf{Symbol} & \textbf{Definition} \\
        \hline
\endfirsthead
        \multicolumn{2}{c}{\footnotesize TABLE~\thetable\ (continued)}\\
        \hline
        \textbf{Symbol} & \textbf{Definition} \\
        \hline
\endhead
        \hline
\endfoot
        \multicolumn{2}{|l|}{\textit{Hardware error theory}} \\
        $\hat{t}_k$, $\epsilon_k^{(j)}$ & Reported crossing time and TDC jitter error \\
        $\tilde{x}(t)$, $w(t)$ & Noisy comparator input and additive noise waveform, $\tilde{x}=x+w$ \\
        $\tilde{t}_k$ & Noise-shifted crossing time before TDC jitter is added \\
        $v_k$ & Local velocity at an implicit point between clean and perturbed crossings \\
        $\widehat{\Lambda}$, $\mathcal{M}$ & Reported event set and partial matching from clean to reported crossings \\
        $\Lambda_{\mathrm{spur}}$, $\Lambda_{\mathrm{miss}}$, $\mathcal{J}$ & Spurious events, missed clean events, and certified matched subset after topology changes \\
        $V_m^{(\mathrm{actual})}$ & Actual physical threshold level of comparator $m$ \\
        $\hat{V}_{m_k}$ & Calibrated threshold level reported for event $k$ \\
        $\bar{q}_m$ & Calibrated threshold offset reported by the threshold table \\
        $r_m$, $r_{m_k}$ & Unreported residual threshold offset for level $m$ and event $k$ \\
        $\sigma_j$ & Worst-case TDC timing jitter bound \\
        $\sigma_n$ & Front-end additive noise amplitude bound \\
        $\sigma_q$ & Threshold imprecision bound (DAC/comparator/reference) \\
        $\hat{s}_n$ & Hardware-perturbed selected crossing paired with $g_n$ \\
        $\alpha_{\mathrm{hw}}^{(\mathrm{ret})}$, $\bar{\alpha}_{\mathrm{hw}}^{(\mathrm{ret})}$ & Realized and clipped retained-eligibility factors used in hardware perturbation bounds \\
        $R_{2,n}$, $R_{2,\mathrm{hw}}$ & Certified second-order implicit time-shift remainder and its selected-window supremum \\
        $\Delta V_{\mathrm{hw}}^*(v_n,\bar{\alpha}_n^{(\mathrm{ret})},R_{2,n})$ & Hardware-reduced admissible threshold spacing with implicit-shift remainder \\
        $L_{\mathrm{hw},n}^{(1)}$, $L_{\mathrm{hw}}^{(1)}$ & Local and global leading-order hardware Kadec diagnostics \\
        $L_{\mathrm{hw},n}^{\mathrm{det}}$, $L_{\mathrm{hw}}^{\mathrm{det}}$ & Local and global deterministic hardware Kadec perturbations (Eqs.~\ref{eq:L_hw_local}--\ref{eq:L_hw}) \\
        $A_{\mathrm{hw}}^{\mathrm{det}}$, $B_{\mathrm{hw}}^{\mathrm{det}}$ & Auxiliary deterministic hardware-perturbed frame bounds under tail-prior use \\
        $\kappa_{\mathrm{hw}}^{\mathrm{det}}$ & Deterministic hardware condition number, $B_{\mathrm{hw}}^{\mathrm{det}}/A_{\mathrm{hw}}^{\mathrm{det}}$ \\
        $\hat{y}_k$ & Reported crossing value $m_k\Delta V+\bar{q}_{m_k}$ \\
        $e_k$, $\mathbf{e}$, $E_R$ & Per-sample value error, the corresponding error vector, and the $\ell^2$ Taylor-remainder budget on $\mathcal{I}$ \\
        $H_{\mathrm{hw},M}$, $s_M$ & Hardware-perturbed finite sampling matrix and $s_M=\sigma_{\min}(H_{\mathrm{hw},M})$ \\
        $H_M$, $\delta H_M$, $\mu_M$ & Clean finite matrix, hardware matrix perturbation, and $\|H_M^\dagger\|_2\|\delta H_M\|_2$ \\
        $\boldsymbol{\eta}$ & Residual value-error vector in the finite-matrix perturbation bound \\
        $N_{\mathrm{in}}$ & Number of inside-window samples used in the error budget \\
        $\mathrm{SNR}_{\mathrm{recon}}$, $\mathrm{SNR}_{M}$, $S_{\mathrm{tar}}$ & Full-signal SNR, model-space SNR, and target linear SNR \\
        $\Gamma_A^{\mathrm{det}}$, $\gamma_A^{\mathrm{det}}$, $\Gamma_M^{\mathrm{det}}$, $\gamma_M^{\mathrm{det}}$ & Deterministic tail-prior and model-space hardware budgets after reserving $E_R$ \\
        $E_M$, $R_M(S_{\mathrm{tar}})$, $E_{\mathrm{tc}}$ & Finite-model residual allowance, remaining full-signal hardware budget, and external tail/collar error budget \\
        $\beta$ & Velocity amplification factor, $2\pi f_{max}\|x\|_\infty/v_{min}$ \\
        \hline
        \multicolumn{2}{|l|}{\textit{Reconstruction variables and operators}} \\
        $K(t,s)$ & Reproducing kernel of $PW_\Omega$ (Eq.~\ref{eq:repro_kernel}) \\
        $k_n(t)$ & Kernel function at $t_n$, $k_n(t)=K(t,t_n)$ \\
        $\mathfrak{S}_{\Lambda}$ & Frame operator on $PW_\Omega$ associated with $\Lambda$ \\
        $\{\tilde{k}_n\}$ & Dual frame functions, $\tilde{k}_n=\mathfrak{S}_{\Lambda}^{-1}k_n$ \\
        $x_{\mathrm{rec}}$, $x_{\mathrm{tail}}$ & Finite-record and tail-conditional reconstructed signals \\
        $a_m$, $M$, $N$ & Sinc coefficient, number of retained grid coefficients, and number of samples \\
        $P_{V_M}$ & Projection onto the declared finite reconstruction space $V_M$ \\
        $\mathbf{y}$, $\mathbf{a}$ & Measurement vector and coefficient vector \\
        $\mathbf{H}$ & Sampling matrix, $(\mathbf{H})_{n,m}=\mathrm{sinc}((t_n-m\Delta T)/\Delta T)$ \\
        $\mathbf{H}^{\dagger}$ & Moore-Penrose pseudo-inverse of $\mathbf{H}$ \\
        $s_{\mathcal{I},M}$, $s_{\mathrm{floor}}$, $\kappa_{\max}$ & Realized finite-matrix singular-value certificate and acceptance thresholds \\
        $\mathbf{I}$ & Identity matrix in finite-dimensional normal equations \\
        $P_\Omega$ & Bandlimiting projection onto $PW_\Omega$ \\
        $x_k(t)$ & $k$-th iterative reconstruction estimate \\
        $\lambda$, $\lambda^*$ & Relaxation parameter and optimal relaxation parameter \\
        $\rho$, $\rho^*$ & Iterative convergence rate and optimal convergence rate \\
        $\delta(\cdot)$ & Dirac delta distribution \\
        \hline
        \multicolumn{2}{|l|}{\textit{Appendix operators and auxiliary recurring symbols}} \\
        $\mathbf{1}_E$ & Indicator function of a set $E$ \\
        $T_a$, $T_R$ & Time-frequency concentration operators on windows $[-a,a]$ and $[-R,R]$ \\
        $D_a$, $D_R$, $B_\Omega$ & Time-truncation and bandlimiting operators \\
        $c_\Omega$ & Bandwidth-normalization scale factor $\Omega/\pi$ in Appendix~\ref{app:landau} \\
        $\widetilde{\Lambda}$ & Scaled sampling set, $\widetilde{\Lambda}=(\Omega/\pi)\Lambda$ \\
        $E_a(\alpha)$ & High-energy eigenspace of $T_a$ with eigenvalues at least $\alpha$ \\
        $\mu_k^{(a)}$, $\phi_k^{(a)}$ & Eigenvalues and orthonormal eigenfunctions of $T_a$ \\
        $\kappa_x(t)$ & Normalized reproducing kernel centered at $x$ \\
        $S_\Lambda$, $S_{\widetilde{\Lambda}}$ & Sampling operators on $\Lambda$ and on the scaled set $\widetilde{\Lambda}$ \\
        $\lambda_n$ & Normalized Kadec node $t_n/\Delta T=n+\delta_n$ \\
        $e_n(t)$, $f_n(t)$ & Reference and perturbed exponential basis functions in Kadec's proof \\
        $F(t)$ & Trigonometric polynomial $\sum_n c_n e^{jnt}$ in Kadec's proof \\
        $\mathrm{Tr}(T_a)$, $\mathrm{Tr}(T_R)$ & Trace of the indicated concentration operator \\
        $N_\Lambda$, $\Lambda^{(m)}$ & Local multiplicity bound and separated color classes in Appendix~\ref{app:pp_local} \\
        $r_{\mathrm{loc}}$, $C_{\Lambda,\Omega}$, $C_{\delta,\Omega}$, $C_{r_{\mathrm{loc}},\Omega}$ & Localization radius and positive constants in localized Plancherel-Polya estimates \\
        \hline
    \end{longtable}

Purely local dummy variables, summation indices, and one-step proof constants are defined in situ and are intentionally omitted from the tables.
\endgroup

\clearpage
\twocolumn
\flushbottom

\bibliographystyle{IEEEtran}
\bibliography{references}

\end{document}